\documentclass{dissertation}

\usepackage{xltabular} 
\usepackage{booktabs} 
\usepackage{balance}
\usepackage{lipsum}
\usepackage{xcolor}
\usepackage{listings}
\usepackage{enumitem}
\usepackage{graphicx}
\usepackage{caption}
\usepackage{subcaption}
\usepackage{hyperref}
\usepackage{tikz}
\usepackage{multirow}
\usepackage{comment}
\usepackage{makecell}
\usepackage{cleveref}
\usepackage{amsfonts}
\usepackage{pifont}
\usepackage[normalem]{ulem}
\usepackage[ruled,vlined,linesnumbered]{algorithm2e}

\newtheorem{lemma}{Lemma}
\newtheorem{proof}{Proof}

\definecolor{circlecolor}{HTML}{004d80}
\definecolor{redcirclecolor}{HTML}{a62b17}
\definecolor{circlecolor2}{HTML}{ff6701}

\newcommand{\para}[1]{\vspace{1mm}\noindent\textbf{#1.}}
\newcommand{\parait}[1]
{\vspace{0mm}\noindent\textit{\underline{#1.}}}

\newcommand{\entity}{\texttt{entity}}
\newcommand{\entities}{\texttt{entities}}
\newcommand{\functions}{\texttt{functions}}
\newcommand{\sstate}{\texttt{state}}
\newcommand{\key}{\texttt{key}}

\newcommand{\context}{\texttt{context}}

\newcommand{\rqs}[1]{\hangindent=15pt\textbf{#1:}}

\newcommand*\circleb[1]{\tikz[baseline=(char.base)]{
            \node[shape=circle,line width=0mm,inner sep=1pt,fill = circlecolor,text=white] (char) {\textsf{{#1}}};}}

\newcommand*\circler[1]{\tikz[baseline=(char.base)]{
            \node[shape=circle,line width=0mm,inner sep=1pt,fill = redcirclecolor,text=white] (char) {\textsf{{#1}}};}}

\newcommand*\circlem[1]{\tikz[baseline=(char.base)]{
            \node[shape=circle,line width=0mm,inner sep=1pt,fill = circlecolor2,text=white] (char) {\textsf{{#1}}};}}

\definecolor{green}{HTML}{008800}
\definecolor{nicerblue}{HTML}{0066bb}
\definecolor{deepblue}{HTML}{0000dd}
\definecolor{velvetred}{HTML}{bb0066}
\definecolor{CC}{HTML}{666666}
\definecolor{darkgray}{HTML}{2F4F4F}
\definecolor{deepgreen}{rgb}{0,0.5,0}

\lstdefinestyle{pythonlang}{
  frame=single,
  language=Python,
  alsoletter={.},
  showstringspaces=false,
  emph={Hotel, Flight, \@reservation_operator.register, reserve, StatefulFunction, Operator, NotEnoughSpace, buy_item, buy_item_0, buy_item_1, serializable_transfer, sagas_transfer, stock, payment, cart, NotEnoughCredit, NotEnoughStock, \@stock.register, \@payment.register,\@cart.register},
  emphstyle={\color{velvetred}},
  columns=flexible,
  basicstyle={\fontfamily{lmtt}\scriptsize},
  identifierstyle={\color{black}},
  numbers=left,
  numberstyle=\tiny\color{deepblue},
  numbersep=5pt,
  keywordstyle=\color{blue},
  morekeywords={async, await, n_partitions, operator, function_name, key, Item, self, Transfer}, 
  commentstyle=\color{CC},
  stringstyle=\color{deepgreen},
  breaklines=false,
  breakatwhitespace=false,
  tabsize=1,
  backgroundcolor=\color{white},
  xleftmargin=.0in,
  captionpos=b,
  keepspaces=true,
  escapechar=|
}

\begin{document}

\title[Transactional Stateful Functions on Streaming Dataflows]{Democratizing Scalable Cloud Applications}

\author{Kyriakos}{Psarakis}

\frontmatter

\begin{titlepage}

\begin{center}

\vspace*{2\bigskipamount}

{\makeatletter
\titlestyle\bfseries\LARGE\@title
\makeatother}

{\makeatletter
\ifx\@subtitle\undefined\else
    \bigskip
    \titlefont\titleshape\Large\@subtitle
\fi
\makeatother}

\end{center}

\cleardoublepage
\thispagestyle{empty}

\begin{center}


\vspace*{2\bigskipamount}

{\makeatletter
\titlestyle\bfseries\LARGE\@title
\makeatother}

{\makeatletter
\ifx\@subtitle\undefined\else
    \bigskip
    \titlefont\titleshape\Large\@subtitle
\fi
\makeatother}

\vfill


{\Large\titlefont\bfseries Dissertation}

\bigskip
\bigskip

for the purpose of obtaining the degree of doctor

at Delft University of Technology,

by the authority of the Rector Magnificus prof.~dr.~ir.~T.H.J.J.~van~der~Hagen,

chair of the Board for Doctorates,

to be defended publicly on

Wednesday 14 January 2026 at 17.30 o’clock

\bigskip
\bigskip

by

\bigskip
\bigskip

\makeatletter
{\Large\titlefont\bfseries\@firstname\ \titleshape{\MakeUppercase{\@lastname}}}
\makeatother

\bigskip
\bigskip

Master of Science in Computer Science, \\
Delft University of Technology, Kingdom of the Netherlands, \\
and Master of Engineering in Electrical and Computer Engineering, \\
Technical University of Crete, Hellenic Republic, \\

born in Chania, Hellenic Republic.

\vspace*{2\bigskipamount}

\end{center}

\clearpage
\thispagestyle{empty}

\noindent This dissertation has been approved by the promotors.

\bigskip
\noindent Composition of the doctoral committee:

\medskip\noindent
\begin{tabular}{p{4.5cm}l}
    Rector Magnificus, & \textit{chairperson} \\
    Prof.\ dr.\ ir.\ G.J.P.M.\ Houben, & Delft University of Technology, \textit{promotor}\\
    Dr.\ A.\ Katsifodimos, & Delft University of Technology, \textit{copromotor} \\
    \medskip
    \mbox{\emph{Independent members:}} & \\
    Prof.\ dr.\ S. Madden, & Massachusetts Institute of Technology \\
    Prof.\ dr.\ G. Smaragdakis, & Delft University of Technology \\
    Dr.\ P. Carbone, & KTH Royal Institute of Technology \\
    Dr.\ A. Klimović, & ETH Zurich \\
    Prof.\ dr.\ A. van Deursen, & Delft University of Technology, \textit{reserve member} \\ \\

\end{tabular}

\medskip

\noindent SIKS Dissertation Series No. 2026-03\\
\noindent The research reported in this thesis has been carried out under the auspices of SIKS, the Dutch Research School for Information and Knowledge Systems.

\medskip

\vfill
\begin{center}
    \includegraphics[height=0.5in]{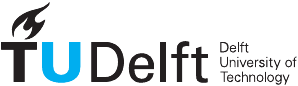}
    \hspace{2em}
    \includegraphics[height=0.5in]{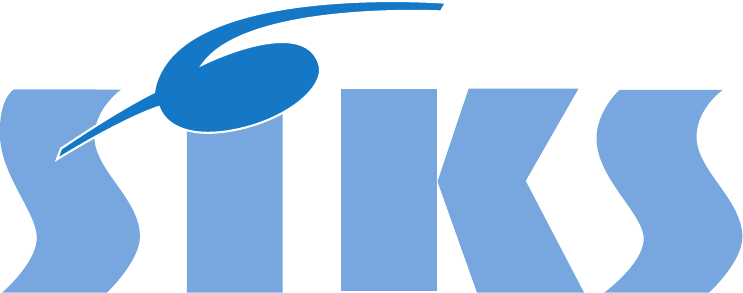}
\end{center}
\vfill

\noindent
\begin{tabular}{@{}p{0.2\textwidth}@{}p{0.8\textwidth}}
  \textit{Keywords:} & stream processing, transactions, stateful functions, state migration, deterministic transactions, dataflow programming, microservice architectures, event-driven programming, fault tolerance \\[\medskipamount]
  \noalign{\vskip \medskipamount}
      \textit{Cover:} & Cover designed by \href{https://www.fiverr.com/nicolasol1}{Nicolás Oliveros} \\[\medskipamount]
\end{tabular}

\medskip

\vspace{\bigskipamount}

\noindent Copyright \textcopyright\ 2026 by Kyriakos Psarakis

\medskip

\noindent ISBN/EAN:  978-94-6518-185-1 

\medskip
\noindent An electronic version of this dissertation is available at: \url{http://repository.tudelft.nl}

\end{titlepage}

\dedication{
    \epigraph{
    Life is Short, Art long, Opportunity fleeting, Experience deceitful, and Judgment difficult. \\
    \vspace{2.5mm}
    \textgreek{
    Ὁ βίος βραχύς,
    ἡ δὲ τέχνη μακρή,
    ὁ δὲ καιρὸς ὀξύς,
    ἡ δὲ πεῖρα σφαλερή,
    ἡ δὲ κρίσις χαλεπή.}
    }
{The Aphorisms of Hippocrates, 5th Century B.C.}}

\tableofcontents

\chapter*{Summary}
\addcontentsline{toc}{chapter}{Summary}
\setheader{Summary}

\noindent Web applications power almost every aspect of our digitalized society, from entertainment to web shopping, vacation planning and booking, online games, communication, work, and social interaction. However, building scalable and consistent Web applications in modern cloud environments requires extensive and diverse expertise in multiple domains, such as cloud computing, software development, distributed and database systems, and domain knowledge. These requirements make the development of such applications possible only by a few highly talented individuals that only large corporations can hire. In this thesis, we aim at democratizing the development and maintenance of such cloud applications by identifying and addressing three key challenges: \emph{i)} programmability of cloud applications; \emph{ii}) high-performance serializable transactions with fault tolerance guarantees; and \emph{iii}) serverless semantics. To address those, we created Stateflow, a high-level, object-oriented, easy-to-use programming model that operates alongside Styx, a novel deterministic dataflow engine that provides high-performance serializable transactions and serverless semantics.

While investigating the challenge of democratizing scalable cloud applications, we discovered that they closely resemble the principles behind the streaming dataflow execution model. In \Cref{chapter1}, we highlight the similarities of streaming dataflow processing and the current state-of-the-art event-driven microservice architectures and lay a path towards the ideal cloud application runtime. To validate our hypothesis, we have created T-Statefun, presented in \Cref{chapter2}, by adapting an existing dataflow system to support transactional cloud applications. At the time, the best candidate appeared to be Apache Flink Statefun, a stateful function as a service system (SFaaS), to which we added transactional support with coordinator functions. With T-Statefun, we showed that a dataflow system can support transactional cloud applications through a SFaaS API. Furthermore, its development helped us identify two significant issues: \emph{i)} it was challenging to program, especially after the addition of the coordinator functions; and \emph{ii)} due to the disaggregation of state and processing and an inefficient transactional protocol, T-Statefun was lacking in performance.

In this thesis, to address the programmability issue, in \Cref{chapter3} we introduce Stateflow, a user-friendly programming model where software developers code in the well-established object-oriented programming style with \emph{zero} boilerplate code, and Stateflow transforms it into an intermediate representation based on stateful dataflow graphs. While experimenting with Stateflow, we verified the inefficiencies detected in \Cref{chapter2} regarding messaging and state, or the lack of transactional support in the rest of Stateflow's supported backends. Thus, in \Cref{chapter4}, we present all the details behind the design of Styx, a distributed streaming dataflow system that supports multi-partition deterministic transactions with serializable isolation guarantees through a high-level, standard Python programming model that obviates transaction failure management. Our design choices and novel algorithms allow Styx to outperform the state-of-the-art systems by at least one order of magnitude in all tested workloads regarding throughput.

Styx demonstrates that it is possible to build a high-performance SFaaS system that provides transactional and fault-tolerance guarantees while offering an intuitive programming model with minimal boilerplate. Building on this foundation, we extend Styx with the ability to dynamically and efficiently adapt to varying workloads. To enable this, Chapter~\ref{chapter5} explores how Styx can migrate state transactionally, a necessary capability for elasticity, given that Styx maintains application state in-memory.

We conclude this thesis by summarizing the key findings and reflecting on the contributions, critically examining the limitations of the proposed methods, and considering their broader ethical and societal implications. Moreover, based on the insights we gained from creating the Stateflow programming model and the Styx runtime, we lay out the new challenges and future directions in the field.

\chapter*{Samenvatting}
\addcontentsline{toc}{chapter}{Samenvatting}
\setheader{Samenvatting}

{\selectlanguage{dutch}

  Webapplicaties ondersteunen nagenoeg elk aspect van onze sterk gedigitaliseerde samenleving: van entertainment en online winkelen tot vakantieplanning, videospellen, communicatie, werk en sociale interactie. Ze spelen een cruciale rol in ons dagelijks leven. Het bouwen van dergelijke applicaties in de moderne cloudomgeving vereist echter diepgaande en diverse expertise in verschillende domeinen, zoals cloud computing, softwareontwikkeling, gedistribueerde systemen, databasesystemen en domeinspecifieke kennis. Deze vereisten maken de ontwikkeling van dergelijke applicaties enkel haalbaar voor een beperkt aantal uiterst getalenteerde individuen of grote bedrijven die over de middelen beschikken om dergelijk talent aan te trekken. In deze thesis beogen we de ontwikkeling en het onderhoud van cloudapplicaties te democratiseren door drie belangrijke uitdagingen te identificeren en aan te pakken: \emph{i)} programmeerbaarheid van cloudapplicaties; \emph{ii)} hoog-performante, seriële transacties met fouttolerantie; en \emph{iii)} serverless semantiek. Om dit te bereiken hebben we \textlatin{Stateflow} ontwikkeld, een hoog-niveau, objectgeoriënteerd en gebruiksvriendelijk programmeermodel dat werkt naast \textlatin{Styx}, een vernieuwende deterministische dataflow-engine die seriële transacties met hoge prestaties en serverless semantiek ondersteunt.

Tijdens het onderzoeken van de uitdaging om schaalbare cloudapplicaties te democratiseren, ontdekten we dat deze nauw aansluiten bij de principes van het streaming dataflow-uitvoeringsmodel. In \Cref{chapter1}, benadrukken we de overeenkomsten tussen streaming dataflow-verwerking en de huidige state-of-the-art event-driven microservice-architecturen. We schetsen een pad richting een ideale runtime voor cloudapplicaties. Om onze hypothese te toetsen, hebben we \textlatin{T-Statefun} ontwikkeld, zoals besproken in \Cref{chapter2}, door een bestaand dataflowsysteem aan te passen om transactionele cloudapplicaties te ondersteunen. Destijds leek \textlatin{Apache Flink Statefun} — een \textlatin{Stateful Function as a Service (SFaaS)}-systeem — de beste kandidaat, waaraan we transactionele ondersteuning toevoegden via coördinerende functies. Met \textlatin{T-Statefun} toonden we aan dat een dataflowsysteem transactionele cloudapplicaties kan ondersteunen via een \textlatin{SFaaS API}. Daarnaast hielp de ontwikkeling ervan ons twee belangrijke problemen te identificeren: \emph{i)} het systeem was moeilijk te programmeren, zeker na het toevoegen van de coördinerende functies; en \emph{ii)} vanwege de scheiding van toestand en verwerking en een inefficiënt transactioneel protocol, presteerde \textlatin{T-Statefun} ondermaats.

Om het programmeerprobleem aan te pakken, introduceren we in \Cref{chapter3} \textlatin{Stateflow}, een gebruiksvriendelijk programmeermodel waarbij ontwikkelaars software schrijven in een vertrouwde objectgeoriënteerde stijl, zonder enige \textlatin{boilerplate}-code. \textlatin{Stateflow} vertaalt dit vervolgens naar een intermediaire representatie gebaseerd op toestandsgebaseerde dataflow-grafen. Tijdens het experimenteren met \textlatin{Stateflow} bevestigden we de inefficiënties uit \Cref{chapter2} met betrekking tot messaging, toestand en het gebrek aan transactionele ondersteuning in de andere ondersteunde backends. Daarom presenteren we in \Cref{chapter4} het ontwerp van \textlatin{Styx}, een gedistribueerd dataflow-systeem dat multi-partitie deterministische transacties ondersteunt met seriële isolatiegaranties, via een hoog-niveau Python programmeermodel dat het afhandelen van mislukte transacties overbodig maakt. Onze ontwerpkeuzes en nieuwe algoritmes stelden \textlatin{Styx} in staat om de meest geavanceerde systemen van dat moment met minstens een orde van grootte te overtreffen qua doorvoer in alle geteste workloads.

\textlatin{Styx} toonde aan dat het mogelijk is om een hoog-performant \textlatin{SFaaS}-systeem te bouwen dat transactionele en fouttolerante garanties biedt, terwijl het een intuïtief programmeermodel met minimale \textlatin{boilerplate} behoudt. Voortbouwend op deze basis was de volgende stap om \textlatin{Styx} uit te breiden met de mogelijkheid om zich dynamisch en efficiënt aan te passen aan wisselende workloads. In \Cref{chapter5} verkennen we hoe \textlatin{Styx} toestand transactioneel kan migreren — een noodzakelijke eigenschap voor elasticiteit, aangezien \textlatin{Styx} applicatietoestand in het geheugen houdt.

Tot slot vatten we deze thesis samen door de belangrijkste bevindingen te bespreken en kritisch te reflecteren op de bijdragen. We analyseren de beperkingen van de voorgestelde methodes en overwegen hun bredere ethische en maatschappelijke implicaties. Tot slot formuleren we, op basis van de inzichten die we opdeden tijdens de ontwikkeling van het \textlatin{Stateflow}-programmeermodel en de \textlatin{Styx}-runtime, nieuwe uitdagingen en richtingen voor toekomstig onderzoek in dit vakgebied.

}

{\selectlanguage{greek}
\chapter*{Περίληψη}
\addcontentsline{toc}{chapter}{Περίληψη}
\setheader{Περίληψη}

Οι διαδικτυακές εφαρμογές υποστηρίζουν σχεδόν κάθε πτυχή της έντονα ψηφιοποιημένης κοινωνίας μας, από την ψυχαγωγία και τις αγορές στο διαδίκτυο, μέχρι τον προγραμματισμό και την κράτηση διακοπών, τα διαδικτυακά παιχνίδια, την επικοινωνία, την εργασία και την κοινωνική αλληλεπίδραση, παίζοντας κρίσιμο ρόλο στην καθημερινότητά μας. Ωστόσο, η ανάπτυξη τέτοιων εφαρμογών μεγάλης κλίμακας απαιτεί εκτενή και πολυδιάστατη τεχνογνωσία σε πολλούς τομείς, όπως το υπολογιστικό νέφος, η ανάπτυξη λογισμικού, τα κατανεμημένα συστήματα και τα συστήματα βάσεων δεδομένων, καθώς και εξειδίκευση στον εκάστοτε επιχειρησιακό τομέα. Αυτές οι απαιτήσεις καθιστούν την ανάπτυξη τέτοιων εφαρμογών εφικτή μόνο από λίγα εξαιρετικά ταλαντούχα άτομα ή από μεγάλες εταιρείες που έχουν τη δυνατότητα να προσλάβουν τέτοιο προσωπικό. Σε αυτή τη διδακτορική διατριβή, στοχεύουμε στη δημοκρατικοποίηση της ανάπτυξης και της συντήρησης τέτοιων εφαρμογών στο νέφος, εντοπίζοντας και επιλύοντας τρεις βασικές προκλήσεις: \textlatin{\emph{i)}} την προγραμματισιμότητα των εφαρμογών νέφους, \textlatin{\emph{ii)}} τις υψηλής απόδοσης σειριοποιήσιμες συναλλαγές βάσεων δεδομένων με εγγυήσεις ανοχής σε σφάλματα, και \textlatin{\emph{iii)}} την εκτέλεση των εφαρμογών σε αρχιτεκτονική χωρίς διακομιστή \textlatin{(serverless)}. Για να τις αντιμετωπίσουμε, δημιουργήσαμε το \textlatin{Stateflow}, ένα υψηλού επιπέδου αφαίρεσης, αντικειμενοστρεφές, εύχρηστο προγραμματιστικό μοντέλο. Το προγραμματιστικό μοντέλο επιτρέπει την ανάπτυξη εφαρμογών που εκτελούνται από το \textlatin{Styx}, μια καινοτόμα ντετερμινιστική μηχανή επεξεργασίας ροών δεδομένων που παρέχει συναλλαγές με σειριοποιήσιμη εγγύηση απόδοσης.

Κατά τη μελέτη της πρόκλησης της δημοκρατικοποίησης των εφαρμογών νέφους, διαπιστώσαμε ότι αυτές ταιριάζουν στενά με τις αρχές του μοντέλου εκτέλεσης ροών δεδομένων (\textlatin{streaming dataflow}). Στο Κεφάλαιο \ref{chapter1}, επισημαίνουμε τις ομοιότητες μεταξύ της επεξεργασίας ροών δεδομένων και των σύγχρονων αρχιτεκτονικών μικροϋπηρεσιών που βασίζονται σε γεγονότα (\textlatin{event-driven}) και σκιαγραφούμε το ιδανικό περιβάλλον εκτέλεσης για εφαρμογές νέφους. Για να επαληθεύσουμε την υπόθεσή μας, δημιουργήσαμε το \textlatin{T-Statefun}, όπως παρουσιάζεται στο Κεφάλαιο \ref{chapter2}, προσαρμόζοντας ένα υπάρχον σύστημα επεξεργασίας ροών δεδομένων ώστε να υποστηρίζει συναλλακτικές εφαρμογές νέφους. Εκείνη τη χρονική περίοδο, το καταλληλότερο σύστημα φάνηκε να είναι το \textlatin{Apache Flink Statefun}, ένα σύστημα που προσομοιάζει μια υπηρεσία εκτέλεσης συναρτήσεων με δεδομένα κατάστασης στο υπολογιστικό νέφος \textlatin{Stateful Function as a Service (SFaaS)}, στο οποίο προσθέσαμε υποστήριξη για συναλλαγές μέσω συντονιστικών συναρτήσεων (\textlatin{coordinator functions}). Με το \textlatin{T-Statefun}, δείξαμε ότι ένα σύστημα επεξεργασίας ροών δεδομένων μπορεί να υποστηρίξει συναλλακτικές εφαρμογές νέφους μέσω ενός \textlatin{API} τύπου \textlatin{SFaaS}. Επιπλέον, η ανάπτυξή του μας βοήθησε να εντοπίσουμε δύο σημαντικά προβλήματα: \textlatin{\emph{i)}} ήταν δύσκολο στον προγραμματισμό, ιδιαίτερα μετά την προσθήκη των συντονιστικών συναρτήσεων, και \textlatin{\emph{ii)}} λόγω του διαχωρισμού κατάστασης και επεξεργασίας και ενός μη αποδοτικού πρωτοκόλλου συναλλαγών, το \textlatin{T-Statefun} υστερούσε σε απόδοση.

Για την αντιμετώπιση του προβλήματος προγραμματισιμότητας, στο Κεφάλαιο \ref{chapter3} παρουσιάζουμε το \textlatin{Stateflow}, ένα φιλικό προς τον χρήστη προγραμματιστικό μοντέλο, στο οποίο οι προγραμματιστές λογισμικού γράφουν σε καθιερωμένο αντικειμενοστρεφές στυλ προγραμματισμού, χωρίς \emph{καθόλου} \textlatin{boilerplate} κώδικα, ενώ το \textlatin{Stateflow} μετατρέπει αυτόματα τον κώδικα σε ενδιάμεση αναπαράσταση βασισμένη σε ροές δεδομένων με κατάσταση. Κατά την πειραματική χρήση του \textlatin{Stateflow}, επιβεβαιώσαμε τις αναποτελεσματικότητες που εντοπίστηκαν στο Κεφάλαιο \ref{chapter2}, αναφορικά με τη διαχείριση μηνυμάτων, την κατάσταση, ή την απουσία υποστήριξης συναλλαγών στα υπόλοιπα υποσυστήματα του \textlatin{Stateflow}. Έτσι, στο Κεφάλαιο \ref{chapter4} παρουσιάζουμε όλες τις λεπτομέρειες του σχεδιασμού του \textlatin{Styx}, ενός κατανεμημένου συστήματος ροής δεδομένων που υποστηρίζει πολυ-κατατμημένες, ντετερμινιστικές συναλλαγές με εγγυήσεις σειριοποιήσιμης απομόνωσης μέσω ενός υψηλού επιπέδου προγραμματιστικού μοντέλου βασισμένου στην \textlatin{Python}, το οποίο απαλλάσσει τον προγραμματιστή από τη διαχείριση αποτυχημένων συναλλαγών. Οι σχεδιαστικές μας επιλογές και οι καινοτόμοι αλγόριθμοι που αναπτύξαμε, επέτρεψαν στο \textlatin{Styx} να ξεπεράσει τα πιο προηγμένα συστήματα της εποχής του, τουλάχιστον κατά μία τάξη μεγέθους υψηλότερη απόδοση σε όλες τις δοκιμασμένες περιπτώσεις φόρτου.

Το \textlatin{Styx} απέδειξε ότι είναι εφικτή η ανάπτυξη ενός υψηλής απόδοσης \textlatin{SFaaS} συστήματος που προσφέρει εγγυήσεις συναλλακτικότητας και ανοχής σε σφάλματα, διατηρώντας ταυτόχρονα ένα διαισθητικό προγραμματιστικό μοντέλο με ελάχιστο απαιτούμενο σκελετό κώδικα λογισμικού στην υλοποίηση των εφαρμογών. Βασιζόμενοι σε αυτό το θεμέλιο, το επόμενο βήμα ήταν η επέκταση του \textlatin{Styx} με τη δυνατότητα να προσαρμόζεται δυναμικά και αποδοτικά σε μεταβαλλόμενα φορτία εργασίας. Για να το επιτύχουμε αυτό, στο Κεφάλαιο \ref{chapter5} διερευνούμε πώς το \textlatin{Styx} μπορεί να μεταφέρει την κατάσταση με συναλλακτικό τρόπο, μια αναγκαία δυνατότητα για ελαστικότητα, δεδομένου ότι το \textlatin{Styx} διατηρεί την κατάσταση των εφαρμογών στη μνήμη.

Ολοκληρώνουμε αυτή τη διατριβή συνοψίζοντας τα βασικά ευρήματα και αναλογιζόμενοι τις συνεισφορές της, εξετάζοντας κριτικά τους περιορισμούς των προτεινόμενων μεθόδων και λαμβάνοντας υπόψη τις ευρύτερες ηθικές και κοινωνικές τους επιπτώσεις. Επιπλέον, βασισμένοι στις εμπειρίες μας από τη δημιουργία του προγραμματιστικού μοντέλου \textlatin{Stateflow} και της πλατφόρμας εκτέλεσης \textlatin{Styx}, σκιαγραφούμε τις νέες προκλήσεις και τις μελλοντικές κατευθύνσεις στον τομέα.

}

\mainmatter

\thumbtrue

\chapter{Introduction}
\label{introduction}

\vfill

\begin{abstract}
The primary objective of this thesis is to democratize the development lifecycle of large-scale cloud applications. At present, only very few people have expertise in cloud application development, infrastructure, distributed systems, and data management combined. This thesis argues that to enable anyone to code such applications, an easy-to-use distributed programming model and a computing system to serve it are required.

This chapter introduces the fundamental concepts of contemporary large-scale cloud applications and summarizes the contributions of this thesis. \Cref{intro:sec:fmtc} presents the transition from on-premise servers to modern cloud offerings. 
Following that, \Cref{intro:sec:scaa} lays out the fundamentals of scalable cloud applications, and \Cref{intro:sec:wis} completes the fundamentals with a dive into the different aspects of serverless technology and this thesis interpretation. \Cref{intro:sec:rq,intro:sec:contrib,intro:sec:origins,intro:sec:vis} highlight the main research questions, this thesis's contributions, the publications included, and a visual outline of the thesis document.
\end{abstract}

\vfill

\newpage

\dropcap{R}{}apid technological advancements, driven by large-scale cloud applications such as social networks, e-commerce platforms, multimedia streaming services, and AI-driven tools, have fundamentally reshaped how we interact with digital services and data. These applications handle massive volumes of data and requests, ensuring seamless user experiences globally with high availability and minimal downtime. However, building and maintaining such systems is an inherently complex task, often reserved for individuals with extensive technical expertise or large corporations with significant resources.

One of the key challenges in developing these applications lies in their scale, where reliability, performance, and scalability are critical factors. These systems must manage thousands of concurrent users, integrate with diverse data sources, and perform under unpredictable loads, all the while maintaining low latency and fault tolerance. As a result, the design and development of these cloud-based applications demand specialized skills in distributed systems, data management, cloud architectures, and scalable infrastructure.

Architecturally, the journey to this level of sophistication has been gradual but profound. In the early days of software development, monolithic applications were the norm, self-contained systems with tightly integrated components \cite{Gray1978}. While monolithic architectures offered simplicity in development, they became difficult to scale and maintain as applications grew. To address these limitations, the industry transitioned toward microservice architectures, where different functionalities are decoupled into independently deployable services \cite{rodrigosurvey}. This architectural shift enabled teams to scale, update, and manage services more efficiently.

At first sight, microservices appear to be the obvious step for replacing monolithic applications and migrating to the cloud. However, microservices dismiss an important advantage that monolithic applications enjoyed for almost five decades: state management, failure management, and state consistency have been the responsibility of database systems. Today's microservice architectures depart from the amenities that were once provided by database management systems (DBMSs) by integrating state management, service messaging, and coordination with application logic. From the database community's point of view, the microservice architectural pattern resembles a situation long ago \cite{papadimitriou1979serializability}, when developers implemented ad-hoc application-level transactions to ensure database consistency. 

For instance, in a shopping cart application, to complete a checkout, we first need to ensure there is enough stock of the selected products, reserve them, then receive payment before shipping the products. In the microservice paradigm, each service (Cart, Stock, Payment) has its own API, database, and application logic, and communicates with other services through API calls. The main issue with the microservice is that both atomicity (i.e., update stock \textit{and} get paid for an order, or cancel both actions) and state consistency across workflow steps (i.e., the stock counts should reflect the successfully paid orders) must be implemented in the application code.

Unsurprisingly, the easy-to-code Function-as-a-Service (FaaS)~\cite{Lambda,GoogleFunctions,Azure} paradigm embraces the same general architecture pattern as microservices: stateless application, stateful database, and communication via messages or external storage. An orchestration layer on top of FaaS allows composing more complex workflows to build service-oriented applications. However, workflow orchestrators solve only part of the problem. For atomicity and consistency, developers adopt the SAGA pattern or the Two-Phase Commit (2PC)~\cite{gray1992transaction} protocol. For applications requiring \textit{transactional} state consistency across services~\cite{rodrigosurvey}, important challenges remain open, including data consistency across service calls and exactly-once guarantees.

By studying these challenges, we identify that the event-driven microservice and orchestrated FaaS architectures inherently form stateful streaming dataflow graphs with partitioned state co-located with the application logic. In short, this architectural pattern is the same as that followed by streaming dataflow systems such as Apache Flink~\cite{flink} and Spark Streaming~\cite{ArmbrustDT18}. However, modern dataflow systems are primarily built for streaming analytics and do not have an API or important features like transactional support necessary to create general-purpose cloud applications.

This thesis addresses these critical gaps by extending the Stateful Function-as-a-Service (SFaaS)~\cite{durable_functions,cloudburst,boki,beldi,tstatefunjournal,styx} paradigm with serializable transactional guarantees, local state management with state migration for elasticity, and fault-tolerant execution. Our contributions significantly simplify the development of scalable, consistent, and reliable cloud applications, democratizing access to sophisticated distributed system features without requiring developers to manage underlying complexity explicitly.

\begin{figure}[t]
    \centering
    \captionsetup{justification=centering}
    \includegraphics[width=\textwidth]{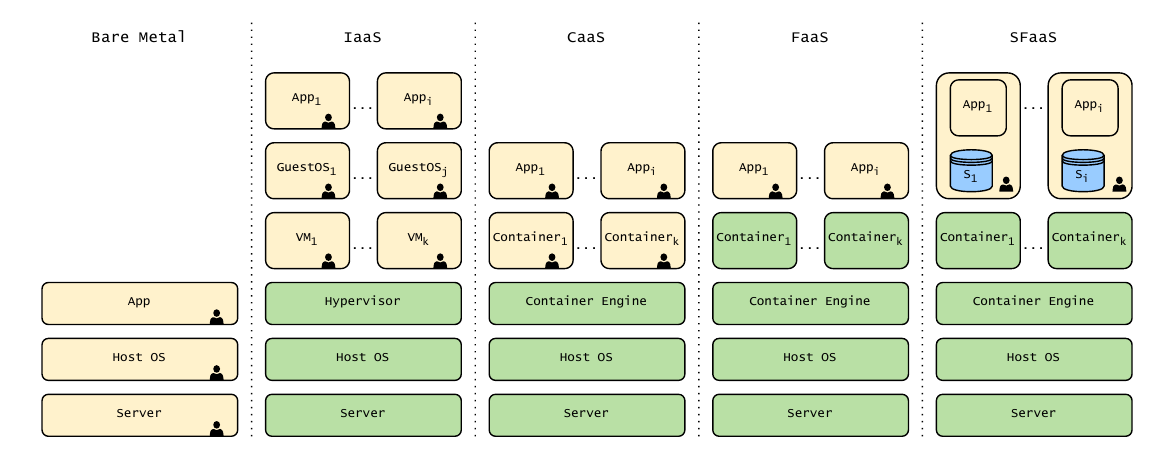}
    \caption{Evolution of cloud computing abstractions and the human involvement required}
    \label{fig:metal_to_clouds}
\end{figure}

\section{From the Metal to the Clouds}\label{intro:sec:fmtc}

Database and distributed systems have experienced significant shifts in deployment models over the last decades, evolving from tightly coupled bare metal hardware/software to highly abstracted cloud-based environments, as illustrated in \Cref{fig:metal_to_clouds}. Initially, web applications and database systems operated directly on physical servers, providing predictable performance and resource allocation, but with limited scalability and significant operational overhead~\cite{Srinivasan2014}. Administrators were tasked with manually managing infrastructure components, dealing with hardware failures, and optimizing performance at the physical level. This manual approach often led to downtime during maintenance and limited agility in responding to evolving business requirements.

\para{Virtual Machines (2000s)} The introduction of Virtual Machines~\cite{vm_arch_1,vm_arch_2, microkernels_vm} (VMs), as an Infrastructure-as-a-Service (IaaS) solution, marked a key transition. VMs abstract physical hardware, enabling multiple isolated operating systems and database instances to coexist on a single physical host. Virtualization improved resource utilization, simplified infrastructure management, and allowed easier scaling through VM replication and migration. Both on-premise virtualization and cloud providers became commonplace, democratizing access to flexible infrastructure and reducing administrative burdens. However, VMs introduced new challenges, such as performance overhead due to hypervisor abstraction, complexities in VM management, and slower startup times.

\para{Containerization (2010s)} The emergence of containerization technologies~\cite{container,container2}, notably Docker and Kubernetes, represented another substantial advancement. Containers encapsulate applications and their dependencies in lightweight, portable units,  reducing overhead compared to VMs. The benefits of containerized environments are faster startup times, simplified versioning, and streamlined deployment pipelines. Container orchestrators such as Kubernetes further introduced automated scaling and self-healing capabilities. Moreover, containers enhanced consistency across development, testing, and production environments. These solutions are referred to as Containers-as-a-Service (CaaS).

\para{Serverless (2020s)} Most recently, serverless computing has transformed the landscape by abstracting away infrastructure management entirely. Serverless architectures enable database systems to scale dynamically and transparently, responding to demand fluctuations without explicit provisioning of resources. Databases such as DynamoDB~\cite{dynamoDB}, Aurora~\cite{aurora}, or Firestore~\cite{firestore} exemplify this trend by automatically scaling compute and storage independently. Furthermore, serverless computing reduces the barrier to entry for developers and businesses by eliminating the need to manage the underlying infrastructure, enabling a stronger focus on the application logic. 

Regarding application logic, serverless computing promotes a shift towards a stateless execution model, exemplified by Function-as-a-Service (FaaS) offerings such as AWS Lambda~\cite{Lambda}, Azure Functions~\cite{Azure}, or Google Cloud Functions~\cite{GoogleFunctions}. In FaaS, developers break down applications into fine-grained, event-driven functions that can be triggered by external events such as HTTP requests or asynchronous messages. These functions are stateless, with each invocation operating independently and relying on external storage to persist state. This paradigm encourages modularity and fine-grained scalability, as each function can be deployed and scaled independently. However, FaaS introduces new challenges in managing state, coordinating execution across functions, and reasoning about correctness.

To overcome the limitations of FaaS, the Stateful Function-as-a-Service (SFaaS) paradigm has emerged. SFaaS allows functions to maintain state across invocations, enabling richer application semantics without compromising the benefits of serverless infrastructure. In this model, functions can encapsulate state locally or access stateful abstractions through tightly integrated state management systems. Examples include systems like Apache Flink Statefun~\cite{statefun} and Azure Durable Functions~\cite{AzureDurableFunctions}, which provide mechanisms for long-lived workflows, stateful coordination, and reliable function orchestration. By bridging the gap between stateless scalability and stateful logic, SFaaS makes a step toward making serverless computing viable for general-purpose cloud applications. Furthermore, alongside SFaaS, (Virtual) Actors~\cite{orleans} are a closely related paradigm that addresses the same core challenge: enabling stateful, long-lived interactions in cloud applications while preserving the benefits of serverless infrastructure. Both paradigms aim to abstract away the complexities of distributed state management, fault tolerance, and scalability from the developer, allowing them to focus on business logic.

\para{This Thesis} In this thesis, we extend SFaaS, the latest serverless paradigm, with serializable transactional guarantees, local state support with state migration support for elasticity, and fault tolerance (\Cref{chapter3,chapter4,chapter5}). Making SFaaS closer to being feature-complete for scalable cloud applications.

\section{Scalable Cloud Application Aspects}\label{intro:sec:scaa}

Building scalable cloud applications involves navigating the design space of system and programmability tradeoffs. This section highlights key aspects that influence how modern cloud applications are developed and operated at scale. We begin by examining cloud runtimes and programming models, which define how developers express application logic and how systems execute it across distributed infrastructure. We then discuss the importance of transactional guarantees for preserving application correctness. Next, we cover high availability and fault tolerance mechanisms that ensure services remain responsive and resilient despite failures. Finally, adaptivity is addressed, enabling systems to adjust to workload changes and infrastructure constraints dynamically.

\subsection{Cloud Runtimes \&\ Programming Models}

Programming models for distributed systems have been a long-standing subject of research~\cite{pl_for_distributed, alvaro2010boom, dryadlinq, hilda, bal1992orca}. In the context of developing cloud applications, the programming model and runtime abstraction greatly influence the design of the underlying system and vice versa. At the moment of writing this thesis, the most common approach taken by software engineers is the use of microservice frameworks (e.g., Java Spring~\cite{java_spring}, Python Flask~\cite{python_flask}). Alongside microservices lie some emerging programming models for cloud applications. These are: \emph{i)} Actors (e.g., Akka~\cite{akka}, Orleans~\cite{bykov2011orleans}) and \emph{ii)} Stateful Functions (e.g., Flink Statefun~\cite{statefun}, Azure Durable Functions~\cite{durable_functions}), all differing substantially in system design, abstractions, and guarantees offered to developers. 

\para{Microservices} To reap the benefits of parallel processing and loose coupling, the prevalent approach is functionally partitioning the application logic and state into independent components that communicate with each other via synchronous or asynchronous messages~\cite{Base2008}, called microservices. Microservice frameworks provide libraries and tooling to help developers build microservices. Provided libraries might offer Object-Relational Mapping for database interactions, service communication using REST or message passing, and retrying/revoking features to ensure correctness. Each microservice built with such frameworks often employs a multi-threaded application server. If a given microservice is stateful, the paired database handles data consistency based on the configured isolation level. However, on the level of a global microservice deployment, no consistency guarantees can be given because the databases are separate, and the consistency guarantees of a single distributed and consistent database cannot be used anymore; thus, the developers need to implement them (e.g., ad-hoc transactions~\cite{adhoc_transactions}), adding to their complexity. 

\para{Actors} The actor model is a programming model for concurrent and parallel computation in distributed systems~\cite{agha_actors_1986}. An actor models a sequential process that performs transformations on the local state based on incoming asynchronous messages. Actor systems achieve concurrency by pipelining and dynamic actor creation~\cite{agha_actors_1986}. Traditional actor frameworks allow developers to program systems using low-level primitives, such as actor IDs and prescriptions of their physical locations.

Virtual actors~\cite{bykov2011orleans} are an extension of the traditional actor model that provides location transparency without forcing developers to deal with actor allocation/scheduling in a cluster, life-cycle management, explicitly creating and tearing down actor instances, as well as failure transparency. Virtual actors are implemented in popular distributed application frameworks like Orleans~\cite{bykov2011orleans} and Dapr~\cite{dapr}.

\para{Cloud Functions} With the emergence of serverless computing~\cite{serverless}, a new cloud paradigm called Function-as-a-Service (FaaS)~\cite{bench_faas,durable_functions} rose in popularity. In FaaS, developers build applications as a collection of functions. Function executions are triggered by external events or invocations from other functions, allowing for function workflow compositions.

Initially, FaaS offerings targeted workloads with small to moderate I/O and communication, demotivating offering data models and consistency guarantees on operations within a single function or cutting across functions~\cite{beldi,cloudburst}.

Due to these limitations, there has been increasing interest in extending the FaaS paradigm to applications that require frequent state access with some consistency guarantees, called Stateful-FaaS (SFaaS)~\cite{beldi, boki, cloudburst, tstatefun, styx}. In SFaaS, developers write programs based on composing functions and enjoy a key-value interface to access the application state. Apart from the shared state interface, the programming, execution, and deployment model resembles Virtual Actors.

\para{Stateful Dataflows} The dataflow model prescribes that an application is represented as a data flow graph. That involves decomposing programs into independent processing units. Organized as Directed Acyclic Graphs (DAGs), processing units (nodes) exchange data via message streams (edges). Dataflows have been mainly applied as the programming model for analytical batch and stream processing systems like Spark~\cite{DBLP:conf/hotcloud/ZahariaCFSS10} and  Flink~\cite{flink}. In these systems, processing units are framed as operators that perform either stateful (e.g., joins, aggregates) or stateless (e.g., map, filter) operations. Message streams can be partitioned and assigned to different concurrent operator instances. Stateful operators typically do not share state, preventing concurrency issues and enhancing parallelism.

However, the dataflow model has two main issues regarding its use for transactional cloud applications. First, dataflow systems are typically programmed using functional programming-style dataflow APIs, requiring developers to rewrite cloud applications to align with the event-driven dataflow model. While many cloud applications can be adapted to this paradigm, doing so demands significant developer training and effort. Second, implementing \textit{transactions} on top of dataflows, namely transactions that span multiple services with serializable guarantees, is not trivial~\cite{tstatefun, zhang2024survey, tstatefunjournal, styxcidr}.

\para{This Thesis} In this thesis, we utilize the dataflow execution model, address the transactionality issues (\Cref{chapter2,chapter4}), and provide a stateful entity domain-specific language for ease of programming that is close to the Virtual Actor programming model (\Cref{chapter3}).

\subsection{Transactional Guarantees} \label{intro:sec:transactional_guarantees}

\begin{figure*}[t]
    \centering
    \captionsetup{justification=raggedright,singlelinecheck=false}
    \begin{subfigure}{0.9\textwidth}
        \centering
        \captionsetup{justification=centering}
        \includegraphics[width=\linewidth]{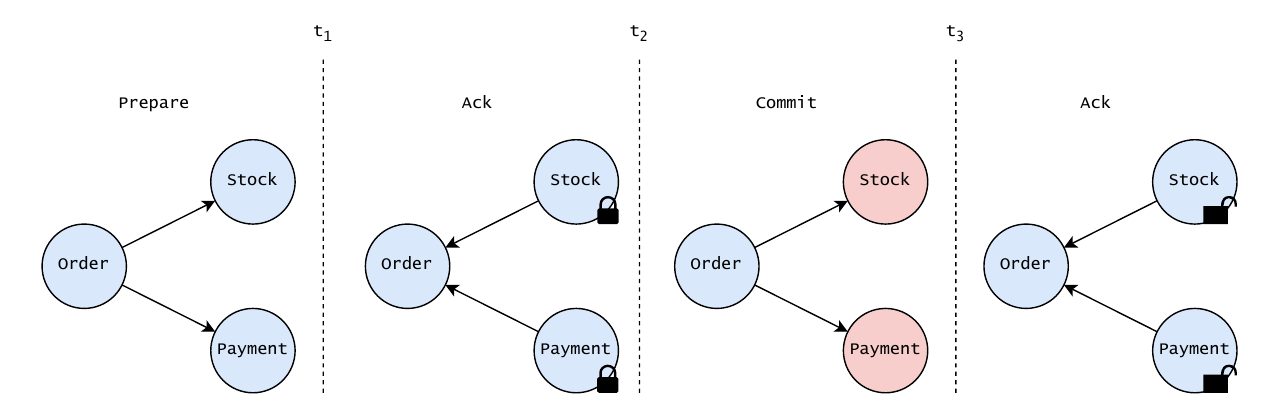}
        \caption{Two-Phase Commit}
        \label{intro:fig:2pc}
    \end{subfigure}
    \hfill
    \begin{subfigure}{0.75\textwidth}
        \centering
        \captionsetup{justification=centering}
        \includegraphics[width=\linewidth]{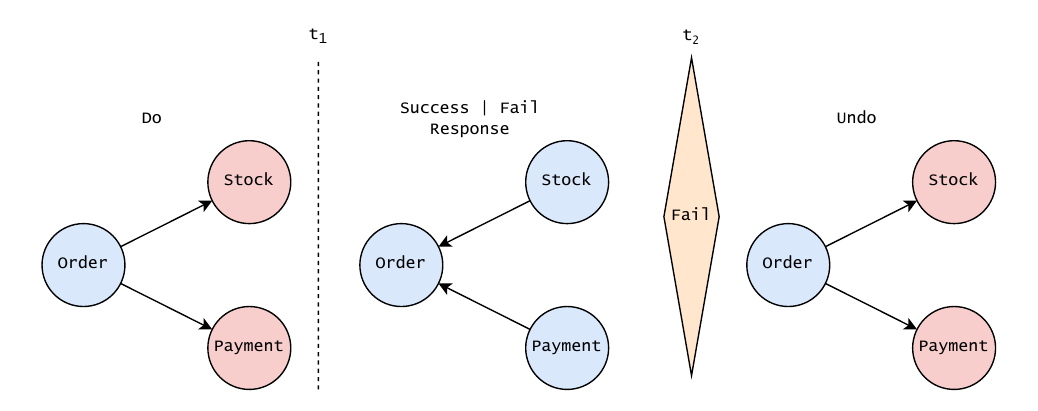}
        \caption{Saga}
        \label{intro:fig:sagas}
    \end{subfigure}
    \caption{Comparison between the two-phase commit and Saga transactional protocols in the checkout scenario that spans three microservices (order, stock, and payment). Phases with state mutations are marked in red. In this example, the two-phase commit requires four timeslots to commit, while the Saga only requires two.}
    \label{intro:fig:transactions}
    \vspace{-2mm}
\end{figure*}

While developing cloud applications, maintaining correctness is essential. In traditional database settings, this usually refers to ACID transactional guarantees~\cite{Gray1978}. \emph{Atomicity} ensures that either all operations within a transaction succeed or have no effect on the state. \emph{Consistency} guarantees that a transaction moves the database from one valid state to another, preserving integrity constraints. \emph{Isolation} ensures that concurrently executing transactions do not interfere with each other. \emph{Durability} ensures that once a transaction commits, its effects are permanent, even in the case of system crashes.

While atomicity and durability are relatively straightforward, consistency and isolation come in different flavors. Recently, researchers have argued for \emph{serializability}~\cite{blanastransactions,transactionsmakedebuggingeasy}. The isolation levels offered by post-NoSQL systems also reflect this notion, such as Google's Spanner~\cite{spanner} and, more recently, CockroachDB~\cite{cockroachdb}. The primary argument for serializability is that, even in very complex distributed deployments, engineers can reason about correctness in the presence of state inconsistencies~\cite{blanastransactions,transactionsmakedebuggingeasy}. The standard way of guaranteeing serializability in a distributed setting is the two-phase commit protocol~\cite{Gray1978} (illustrated in \Cref{intro:fig:2pc}). In the first phase, the transaction manager prepares based on a deadlock detection mechanism (i.e., wait-die or wound-wait) and locks the requested keys or returns a failure. Once all parties acknowledge they hold the locks, the transaction can commit and unlock the keys. 

The counterargument against serializability is the difficulty of providing such a strong guarantee and maintaining high performance, especially when long-running transactions are present in the workload. The most prominent technique for eventual consistency without isolation is the Saga pattern, as shown in \Cref{intro:fig:sagas}. A Saga decomposes a transaction into a sequence of local transactions (steps), each performed by a different service or process. If one of the local steps fails, the system executes compensating actions to undo the preceding successful steps. This approach allows for coordination without distributed locking~\cite{2pl, specification1991distributed} or global consensus~\cite{lamport2001paxos,omnipaxos}.

As cloud applications grow in scale and complexity, the need for strong guarantees around data correctness becomes increasingly pronounced. However, it is currently observed that a significant number of developers overlook the need for transactions and attempt to create custom solutions~\cite{blanastransactions, adhoc_transactions} called Ad Hoc transactions, which are developer-coordinated sequences of database operations embedded in application code rather than expressed using traditional database transactions or ORM invariant validations. Ad hoc transactions attempt to simulate isolation through manual concurrency control, often with partial correctness and performance trade-offs.

\para{This Thesis} The facts mentioned above strengthen our argument for democratizing scalable cloud applications and the need for a simple programming model~\cite{stateflow} (\Cref{chapter3}) and Styx~\cite{styx} that guarantees serializable ACID transactions without any developer involvement (\Cref{chapter4}).

\subsection{High Availability \&\ Fault Tolerance}

Two critical requirements for reliable systems are high-availability and fault tolerance~\cite{highavailability, Gray1978, MalviyaWM14}, particularly for mission-critical applications such as telecommunications, financial services, transportation, and healthcare. These systems are expected to operate continuously with near-zero downtime. A typical service-level agreement (SLA) that cloud providers offer software engineers specifies 5 minutes or less of unavailability per year, corresponding to the so-called "five nines" (99.999\%) of availability.

To narrow the scope, we will focus on higher-level techniques that database or distributed systems employ to give such guarantees. For high availability, the primary mechanism is replication (i.e., maintaining copies of the database across multiple nodes or data centers). In practice, most systems implement failover mechanisms; for example, automatic switching to a standby replica when the active/primary node fails. Moreover, load balancing incoming requests and partitioning/sharding the database state help distribute the load more evenly so that more nodes share the load at peak times.

Database systems use a plethora of mechanisms to maintain fault tolerance. Some of them are: \emph{i)} periodically capturing state to enable rollback and recovery after failures (snapshotting/checkpointing~\cite{chandy1985distributed}), \emph{ii)} logging state mutations before applying them to ensure durability (write-ahead logging\cite{wal}),  \emph{iii)} transaction mechanisms, as mentioned in \Cref{intro:sec:transactional_guarantees},  \emph{iv)} requiring agreement from a majority before committing operations in the presence of replicas (quorum and voting~\cite{raft, paxos, omnipaxos, lamport2001paxos}).

\para{This Thesis} The primary contribution of this thesis, Styx (\Cref{chapter4}), utilizes periodic coordinated snapshots~\cite{chandy1985distributed, checkmate} and a heartbeat failure detection mechanism for fault tolerance. Furthermore, we use minimal write-ahead logging to maintain determinism throughout the system. Although not directly addressed, high availability is straightforward to implement using async replication, as discussed in \Cref{concl:sec:lim_fw}.

\subsection{Adaptivity}

Adaptivity is a critical requirement for scalable cloud systems. As workloads evolve due to user demand, temporal access patterns, or changes in system topology, systems must continuously adapt to maintain performance and correctness. Adaptivity manifests in several forms, most notably in the system's ability to autoscale~\cite{surveyAutoscaling,googleAutoscaling,awareAutoscaling,siachamisAutoscalers} and, in the case of stateful systems (e.g., stream processing engines, databases, or stateful microservices/functions), migrate state across machines~\cite{meces,rhino,megaphone,squall}.

\para{Autoscaling} Adaptive systems offer autoscaling to automatically adjust the number of compute instances based on workload metrics such as CPU usage, request latency, or custom application-level signals. Autoscaling is relatively straightforward for stateless systems by spawning or removing instances, whereas for stateful components, it requires careful orchestration to ensure seamless continuity and consistency. For example, spawning a new instance might require retrieving its state from a shared store or other replicas, potentially introducing latency or consistency issues if not handled properly. Furthermore, fine-grained function-level autoscaling in serverless systems must balance responsiveness and the overhead of startup, warmup, or reconfiguration~\cite{serverlessAutoscaling}.

\para{State Migration} To support adaptivity in stateful systems, state migration is essential. It enables redistributing the application state across nodes, whether to balance load, scale up/down, or recover from failures. State migration often involves moving partitions/shards or individual keys in data-intensive systems. Effective migration mechanisms must minimize downtime, ensure no data loss, and avoid violating consistency or transactional guarantees.

\para{This Thesis} Styx (\Cref{chapter4}) provides state migration functionality (\Cref{chapter5}) while maintaining zero downtime, consistency, and transactional guarantees. Since autoscaling is orthogonal, it is left as future work.

\section{What is Serverless?}\label{intro:sec:wis}

Serverless computing~\cite{serverless-open-problems, serverless_architecture18, serverless_uta18, hellerstein_serverless19, patterson_serverless21, serverless_iosup23} represents a cloud computing paradigm characterized by abstracting infrastructure management, automated scaling, and allowing granular, pay-per-use billing. By shifting operational concerns away from developers, serverless computing enables a simpler development model, more efficient resource utilization, and more flexible economic practices. In \Cref{intro:sec:fmtc}, we discussed the implications of serverless for data management and scalable cloud applications; this section explores the three dimensions of serverless computing: the developer experience, underlying system considerations, and associated economic trade-offs.

\subsection{Developers' Perspective}

From the developers' viewpoint, serverless computing significantly simplifies cloud application development by abstracting away infrastructure concerns and allowing developers to focus exclusively on business logic~\cite{serverless-open-problems,patterson_serverless21}. Serverless applications are composed of event-driven, stateless functions triggered by external events such as HTTP requests, database updates, or messages. The complexity associated with resource provisioning, container management, and scaling is entirely managed by cloud providers.

This paradigm dramatically reduces developers' effort, facilitating rapid prototyping, faster iteration, and shorter deployment cycles. Nevertheless, developers must adapt to constraints intrinsic to serverless architectures, including function statelessness, ephemeral execution environments, and runtime limitations~\cite{hellerstein_serverless19,serverless-open-problems}. Despite these limitations, the overall productivity gains and reduced operational complexity are considerable.

\subsection{Systems' Considerations}

Serverless computing demands sophisticated automation and efficient resource management from the underlying execution systems. Cloud platforms dynamically allocate resources and scale functions in response to changing demand with minimal intervention. Elasticity is achieved through mechanisms such as rapid cold-start initialization, function reuse, and fine-grained resource isolation across tenants~\cite{patterson_serverless21,serverless-open-problems}. However, these properties impose significant system challenges.

One critical challenge is mitigating cold-start latency, the delay experienced when initializing function execution engines, which significantly impacts performance for latency-sensitive applications. Moreover, efficient state and data locality handling presents additional complexity, as functions must interact with external storage services due to inherent statelessness~\cite{hellerstein_serverless19,patterson_serverless21}. Another challenge is achieving predictable performance despite variable, unpredictable workloads, which require sophisticated scheduling and orchestration strategies.

\begin{figure}[t]
    \centering
    \captionsetup{justification=raggedright,singlelinecheck=false}
    \includegraphics[width=0.5\textwidth]{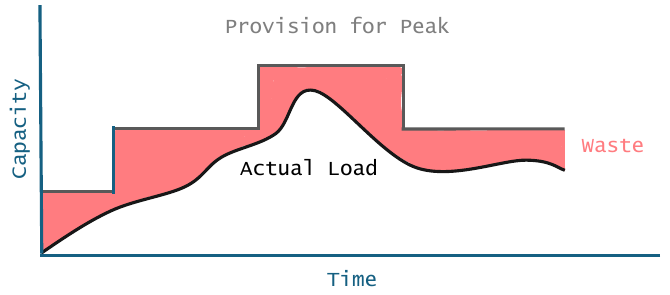}
    \caption{In a typical scenario, operators provision machines based on the expected peak, which results in much waste. With Serverless solutions, operators do not need to provision any machines, which matches actual load and reduces waste.}
    \label{fig:serverless_pay_as_you_go}
\end{figure}

\subsection{Economics}

Serverless computing introduces a fundamental shift in cloud economics by aligning resource costs directly with consumption. Unlike traditional Infrastructure-as-a-Service (IaaS) or Platform-as-a-Service (PaaS) models, where developers pre-allocate resources and bear the costs of idle infrastructure, serverless charges only for the execution time and resources actively utilized by applications~\cite{patterson_serverless21,serverless-open-problems}.

This fine-grained billing model significantly reduces intermittent or unpredictable workload costs, enabling efficient resource utilization without upfront capital commitments. Consequently, organizations can experiment, iterate, and innovate more freely without the economic risk associated with resource over-provisioning~\cite{hellerstein_serverless19}. Thus, serverless democratizes access to cloud-scale computing, benefiting smaller teams, startups, or applications with highly variable usage patterns.

However, this consumption-based model may become economically unfavorable for sustained or predictable workloads. Traditional reserved resource models may offer superior cost efficiency in scenarios with constant or highly predictable resource usage. Therefore, enterprises must evaluate workload characteristics carefully to determine the optimal economic strategy when adopting serverless computing~\cite{hellerstein_serverless19,patterson_serverless21}.

\subsection{This Thesis Interpretation}

This thesis primarily focuses on the developer's perspective, adding transactional guarantees, fault tolerance, stateful functions, and determinism. Our approach provides lower latency, higher throughput, and stronger consistency guarantees while maintaining ease of use. We explore elasticity with state migration mechanisms, and although critical to serverless adoption, we leave the economics out of scope.

\section{Main Research Questions}\label{intro:sec:rq}

In this thesis, the primary research question is to investigate the possibility of creating a system and framework with an easy-to-use programming model that allows non-expert developers to write complex cloud applications free of concurrency or machine failure considerations while maintaining high performance.
\vspace{2mm}

\noindent To that end, we have split this into five research questions:

\vspace{2mm}

\rqs{RQ-1} What would be the optimal substrate for a system serving complex cloud applications?

\vspace{2mm}

\noindent In \Cref{chapter1}, we answer \textbf{RQ-1} by first laying out the path towards the ideal cloud runtime from a software developer's perspective. Furthermore, pointing out the similarities between modern cloud architectures (i.e., event-driven microservices) and stateful dataflow graphs. To test our deduction, we proceed with the following research question.

\vspace{2mm}

\rqs{RQ-2} Is it possible to use an existing SFaaS dataflow system for this purpose? If so, what are the limitations?

\vspace{2mm}

\noindent To answer \textbf{RQ-2}, as shown in \Cref{chapter2}, we first attempted to use Flink-Statefun~\cite{statefun}, a well-established SFaaS system, and add all the required transaction orchestration. Although it outperformed the state-of-the-art, it had a few limitations. The serializable protocol suffered from low throughput in high-contention scenarios because of its implemented deadlock-prevention mechanism. The architecture of Flink-Statefun transfers the state to remote processing workers instead of processing it within the stream processor, increasing latency. Also, its API had a lot of boilerplate code, making it difficult for non-expert developers. Leading to the third research question:

\vspace{2mm}

\rqs{RQ-3} Can we design a domain-specific language that runs on top of all stream processing systems, providing a simple, easy-to-use object-oriented API?

\vspace{2mm}

\noindent Answering \textbf{RQ-3} is essential for the democratization of the development of large-scale cloud applications. In \Cref{chapter3}, we present \textit{Stateflow}, a programming model and intermediate representation (IR) that compiles imperative, transactional object-oriented applications into distributed dataflow graphs and executes them on existing dataflow systems. 
Instead of designing an external Domain-Specific Language (DSL) for our needs, we opted for an internal DSL embedded in Python, which is already popular for cloud programming and is easy to use.
Specifically, a given Python program is first compiled into an IR, an enriched stateful dataflow graph independent of the target execution engine. The choice of execution engine is entirely independent of the application layer, which allows switching to different ones with no changes to the application code. However, all current systems had limitations, leading to poor performance, and our fourth research question:

\vspace{2mm}

\rqs{RQ-4} Can we build a system that enables developers to write transactional, data-intensive cloud applications without requiring expertise in distributed systems?

\vspace{2mm}

\noindent To that end, in \Cref{chapter4}, we showcase how we built Styx, a novel dataflow-based runtime for SFaaS that ensures exactly-once execution while enabling arbitrary function orchestrations with end-to-end serializability guarantees, leveraging concepts from deterministic databases to avoid costly 2PCs. Our work stems from two critical observations. First, modern streaming dataflow systems such as Apache Flink~\cite{flink} guarantee exactly once processing~\cite{flink,carbone2017state,silvestre2021clonos} by hiding failures from their developers. However, they cannot be used to execute cloud applications such as microservices, let alone guarantee transactional SFaaS orchestrations. Second, deterministic database protocols~\cite{calvin,aria} that can avoid expensive 2PC invocations have not been designed for complex function orchestrations and call-graphs. Thus, they are not directly applicable to the needs of SFaaS. While Styx solved the high-performance requirements, it is not flexible resource-wise, leading to the fifth and final research question:

\vspace{2mm}

\rqs{RQ-5} Can we give Styx elasticity properties, such as state migration, allowing it to become serverless? 

\vspace{2mm}

\noindent Much work has been carried out in dynamic reconfiguration~\cite{noria, ds2, dhalion} and state migration~\cite{megaphone, meces, rhino} of streaming dataflow systems over the last few years. These advancements are necessary for providing serverless elasticity in the case of state and compute collocation and enable dataflow graphs as an execution model for \textit{serverless} stateful cloud applications, which is presented in \Cref{chapter5}.

\section{Contributions}\label{intro:sec:contrib}

The main contributions of this thesis, alongside their open-source code artifacts,  are summarized as follows:

\begin{enumerate}
    \item We characterize cloud application runtimes and lay a path toward the "ideal" runtime in the modern setting. We deduce that the modern event-driven microservice paradigm closely matches the fundamentals of dataflow engines and argue that a dataflow engine can serve as one. (\Cref{chapter1})
    \item To validate our deduction from \Cref{chapter1}, we explore the possibility of adapting an existing SFaaS system, Apache Flink Statefun, to the cloud application runtime requirements. To that end, we implemented coordinator functions that provide transactional support within Apache Flink Statefun with varying consistency guarantees. Our new system is called T-Statefun\footnote{\url{https://github.com/delftdata/flink-statefun-transactions}} and outperforms the current state-of-the-art transactional SFaaS systems by an order of magnitude. It has also distributed OLTP databases by at least 1.5x. (\Cref{chapter2})
    \item In \Cref{chapter2}, we addressed all the functional requirements of a dataflow runtime that serves scalable cloud applications. Next, we created an easy-to-use domain-specific language called Stateflow\footnote{\url{https://github.com/delftdata/stateflow}}. Stateflow takes object-oriented code, where an object is a stateful entity, and transforms it into the dataflow execution model. We have proven the ease of integration of Stateflow with existing dataflow systems and the minimal overhead it adds to those.  (\Cref{chapter3})
    \item We created a new distributed dataflow engine, Styx\footnote{\url{https://github.com/delftdata/styx}}, that serves as a runtime for scalable cloud applications. Based on lessons learned from \Cref{chapter2,chapter3}, we ensure that the state is local to the dataflow operator and allows for direct addressing of operators since other systems had to go through the ingress if they wanted to respond to another operator. These design changes required a few algorithmic changes and optimizations that enabled Styx to outperform the T-Statefun (\Cref{chapter2}) and the state of the art by an order of magnitude while providing serializable transactional guarantees and coarse-grained fault tolerance. (\Cref{chapter4})
    \item We extended Styx with state-of-the-art state migration capabilities, a step towards Styx becoming elastic, leading to it becoming serverless. Our experiments show minimal impact of the migrating actions on Styx's throughput and latency. (\Cref{chapter5})
\end{enumerate}

\section{Thesis Origins}\label{intro:sec:origins}

\noindent The main body of the thesis consists of five main chapters based on the research papers listed below:
\\

\noindent\textbf{\Cref{chapter1}} is based on the following publication:
\vspace{2mm}

\faFileTextO~\hangindent=15pt\emph{K. Psarakis, G. Christodoulou, M. Fragkoulis, and A. Katsifodimos. Transactional Cloud Applications Go with the (Data)Flow, CIDR'25}~\cite{styxcidr}.\\

\noindent\textbf{\Cref{chapter2}} is based on the following publications:
\vspace{2mm}

\faFileTextO~\faTrophy\footnote{The trophy icon indicates that the paper won the best paper award}~\hangindent=15pt\emph{M. de Heus, K. Psarakis, M. Fragkoulis, and A. Katsifodimos. Distributed transactions on serverless stateful functions, ACM DEBS'21}~\cite{tstatefun}. 
\vspace{1mm}

\faFileTextO~\hangindent=15pt\emph{M. de Heus, K. Psarakis, M. Fragkoulis, and A. Katsifodimos. Transactions Across Serverless Functions Leveraging Stateful Dataflows. In Elsevier's Information Systems, Volume 108, September 2022}~\cite{tstatefunjournal}.\\

\noindent\textbf{\Cref{chapter3}} is based on the following publications:
\vspace{2mm}

\faFileTextO~\hangindent=15pt\emph{K. Psarakis, W. Zorgdrager,  M. Fragkoulis, G. Salvaneschi, and A. Katsifodimos. Stateful Entities: Object-oriented Cloud Applications as Distributed Dataflows (Abstr.), CIDR'23}~\cite{stateflowcidr}.\\


\faFileTextO~\hangindent=15pt\emph{K. Psarakis, W. Zorgdrager,  M. Fragkoulis, G. Salvaneschi, and A. Katsifodimos. Stateful Entities: Object-oriented Cloud Applications as Distributed Dataflows (Vision), EDBT'24}~\cite{stateflow}.\\

\noindent\textbf{\Cref{chapter4}} is based on the following publication:
\vspace{2mm}

\faFileTextO~\hangindent=15pt\emph{K. Psarakis, G. Christodoulou, G. Siachamis, M. Fragkoulis, and A. Katsifodimos. Styx: Transactional Stateful Functions on Streaming Dataflows, ACM SIGMOD'25}~\cite{styx}.\\

\faFileTextO~\hangindent=15pt\emph{K. Psarakis, O. Mraz, G. Christodoulou, G. Siachamis, M. Fragkoulis, and A. Katsifodimos. Styx in Action: Transactional Cloud Applications Made Easy (Demo), VLDB'25}~\cite{styxdemo}.\\

\noindent\textbf{\Cref{chapter5}} is based on the following publication:
\vspace{2mm}

\faFileTextO~\hangindent=15pt\emph{K. Psarakis, G. Christodoulou, G. Siachamis, M. Fragkoulis, and A. Katsifodimos. State Migration in Styx: Towards Serverless Transactional Functions (Under Review)}.\\

\noindent Additionally, this dissertation benefits from the following research papers: \\

\faFileTextO~\hangindent=15pt\emph{R. Laigner, G. Christodoulou, K. Psarakis, A. Katsifodimos, Y. Zhou. Transactional Cloud Applications: Status Quo, Challenges, and Opportunities (Tutorial), ACM SIGMOD'25}~\cite{rodrigotutorial}.\\

\faFileTextO~\hangindent=15pt\emph{G. Siachamis, K. Psarakis, M. Fragkoulis, A. van Deursen, P. Carbone, A. Katsifodimos. CheckMate: Evaluating Checkpointing Protocols for Streaming Dataflows, IEEE ICDE'24}~\cite{checkmate}.\\

\faFileTextO~\hangindent=15pt\emph{G. Siachamis, G. Christodoulou, K. Psarakis, M. Fragkoulis, A. van Deursen and A. Katsifodimos. Evaluating Stream Processing Autoscalers, ACM DEBS'24}~\cite{siachamisAutoscalers}.\\

\newpage

\section{Visual Outline} \label{intro:sec:vis}

\begin{figure}[h]
    \centering
    \captionsetup{justification=centering}
    \includegraphics[width=\textwidth]{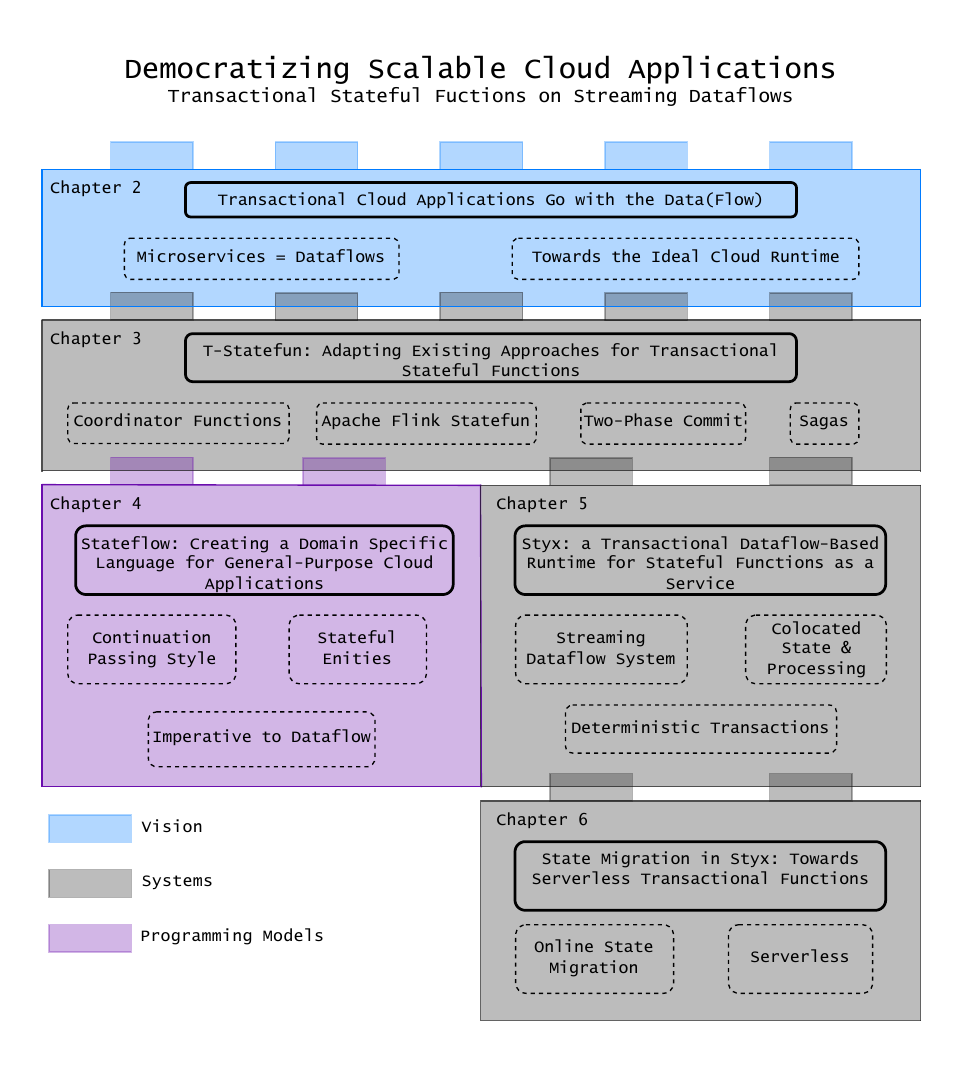}
    \caption{Visual outline of this thesis with the chapters and main ideas put into them.}
    \label{fig:vco}
\end{figure}
\chapter{Transactional Cloud Applications Go with the Data(Flow)}
\label{chapter1}

\vfill

\begin{abstract}

Traditional monolithic applications are migrated to the cloud, typically using a microservice-like architecture. Although this migration offers significant benefits, such as scalability and development agility, it also leaves behind the transactional guarantees that database systems have provided to monolithic applications for decades. In the cloud era, developers build transactional, fault-tolerant distributed applications by explicitly implementing transaction protocols at the application level.

This chapter presents the main argument of this thesis and outlines our approach: the principles underlying the streaming dataflow execution model and deterministic transactional protocols provide a powerful and suitable substrate for executing transactional cloud applications.


\end{abstract}

\vfill

\blfootnote{Parts of this chapter have been published in:\\ \faFileTextO~\hangindent=15pt\emph{K. Psarakis, G. Christodoulou, M. Fragkoulis, and A. Katsifodimos. Transactional Cloud Applications Go with the (Data)Flow, CIDR'25}~\cite{styxcidr}. }

\newpage


\dropcap{O}{}ver the last decades, enterprises have migrated applications such as order management systems, banking systems, game-backend services, and supply chain management to the cloud. The transition from monolithic applications follows an architectural pattern that favors a stateless application layer supported by a stateful database layer. All the stateless and stateful components communicate with each other via REST calls or message queues. Microservice architectures are well-known instances of this pattern.

At first sight, microservices are an obvious candidate for replacing monolithic applications and migrating to the cloud. Microservices offer code modularity, scalability, and development agility. However, microservices lose an important advantage that monolithic applications enjoyed for almost five decades: state management, failure management, and state consistency were the responsibility of database systems. Today's microservice architectures depart from these DBMS amenities by intermingling state management, service messaging, and coordination with application logic~\cite{DBLP:conf/cidr/Helland24}. From the database community's point of view, the microservice architectural pattern resembles the situation described long ago~\cite{papadimitriou1979serializability}, when developers implemented ad hoc application-level transactions to ensure database consistency. 
Worse, managing communication and state in a distributed cloud environment increases complexity.

For instance, in a shopping cart application, to complete a checkout, we first need to ensure there is enough stock of the selected products and then receive payment before shipping the products. In the microservice paradigm, each service (Cart, Stock, Payment) has its own API, database, and application logic, and communicates with other services through API calls. The main issue with microservices is that both atomicity (i.e., update stock \textit{and} get paid for an order, or cancel both actions) and state consistency across workflows (i.e., the stock counts should reflect the successfully paid orders) must be implemented in application code. Similarly, Function-as-a-Service (FaaS) follows the same general architecture pattern as microservices: a stateless application, an external database, and message-based communication. An orchestration layer on top of FaaS enables the composition of complex workflows to build service-oriented applications. 

However, orchestrators~\cite{GoogleCloudRunfunctions, stepfunctions, AzureDurableFunctions} solve only part of the problem, namely the atomicity of a workflow's execution. Moreover, achieving atomicity typically requires developers to handcraft compensating actions to roll back changes correctly using the SAGA pattern~\cite{sagas}. To address these concerns, a line of research~\cite{beldi, boki} proposes FaaS systems for workflow orchestration with transactional guarantees at the expense of performance and high-level programming primitives.
For applications requiring low-latency transaction execution and state consistency across services \cite{rodrigosurvey}, important challenges remain open.

\begin{figure*}[t]
    \centering
    \includegraphics[width=0.8\textwidth]{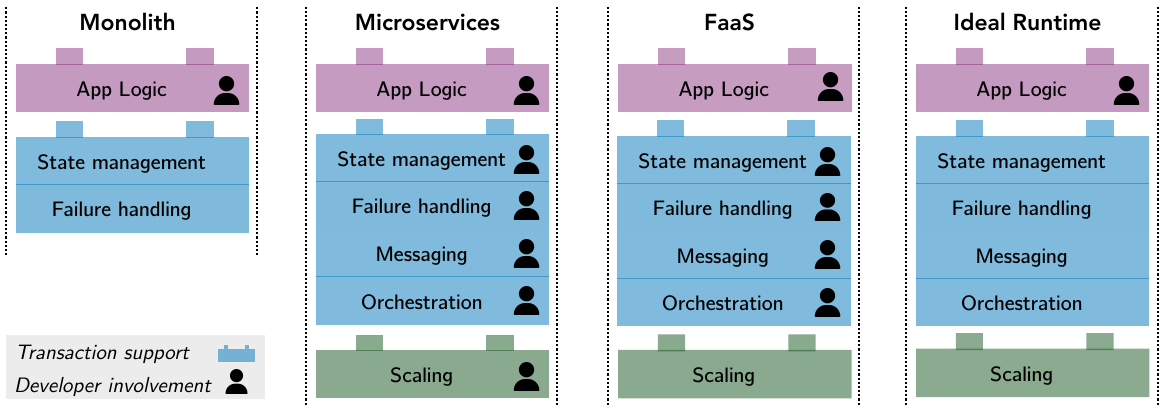}
    \caption{In monolithic applications, developers focused on application logic while a transactional database handled state management and failure recovery. In distributed cloud applications, development involves more challenges (e.g., failures, exactly-once messaging, and orchestration for atomicity and scalability). The ideal runtime should offer the same state consistency and ease of programming as monoliths, with improved scalability, without developer involvement.}
    \label{fig:mono_micro_sfaas}
\end{figure*}

In this chapter, we first identify the limitations and shortcomings of microservice-like architectures for implementing transactional applications and then motivate the need for dedicated runtimes to support transactional cloud applications. We argue that to remove transaction- and failure-handling code from the application level, we need to address complex orchestration, service calls, and state management in a \textit{holistic} manner at the system level, i.e., via a dedicated runtime. During the last few years, we have been developing a runtime for transactional applications called Styx (\Cref{chapter4})~\cite{styx}. Styx automatically partitions state, parallelizes function execution, and enables arbitrary transactional workflows with low latency. Most importantly, Styx's programming model (\Cref{chapter3})~\cite{stateflow} allows for application development that resembles a single-node application/monolith while transparently handling the serializable execution of massively parallel workflows in the cloud.

Our work is in line with recent research, such as Orleans \cite{orleans}, DBOS~\cite{dbos}, Hydroflow~\cite{hydro}, and SSMSs~\cite{li2024serverless}. Contrary to these systems, our work adopts the streaming dataflow execution model while exposing an object-oriented/actor-like programming model on top \cite{stateflow} and guarantees serializability \textit{across} services. 

\vspace{2mm}

\noindent To summarize, in this chapter, we make the following contributions: 

\begin{itemize}
        \item We analyze the shortcomings of modern cloud applications by exemplifying issues with current architectures and requirements for future systems (\Cref{ch1:sec:frommono_to_micro}).
        \item We provide arguments on the suitability of the stateful streaming dataflow paradigm for transactional cloud applications (\Cref{sec:dataflow-everything}).
        \item We introduce a novel approach that combines ideas from deterministic databases, dataflow systems, and serverless architectures (\Cref{ch2:sec:design}).
\end{itemize}

\section{From Monoliths to Microservices}\label{ch1:sec:frommono_to_micro}

As illustrated in \Cref{fig:mono_micro_sfaas}, developers in monolithic architectures were primarily responsible for the application logic. At the same time, with the adoption of microservices, they need to deal with messaging and failures (Section~\ref{sec:messaging-idempotency}), state management and orchestration (Section~\ref{sec:transactions-orchestration}), and scaling techniques (Section~\ref{sec:rescaling}). Interestingly, in \Cref{fig:monolith-to-micro}, we observe that these aspects are not orthogonal. The conversion to a partitioned, event-driven architecture  (\Cref{fig:monolith-to-micro}b to \Cref{fig:monolith-to-micro}c) requires state migration, coordination, and fault-tolerance.

\begin{figure*}[t]
    \centering    
    \includegraphics[width=\textwidth]{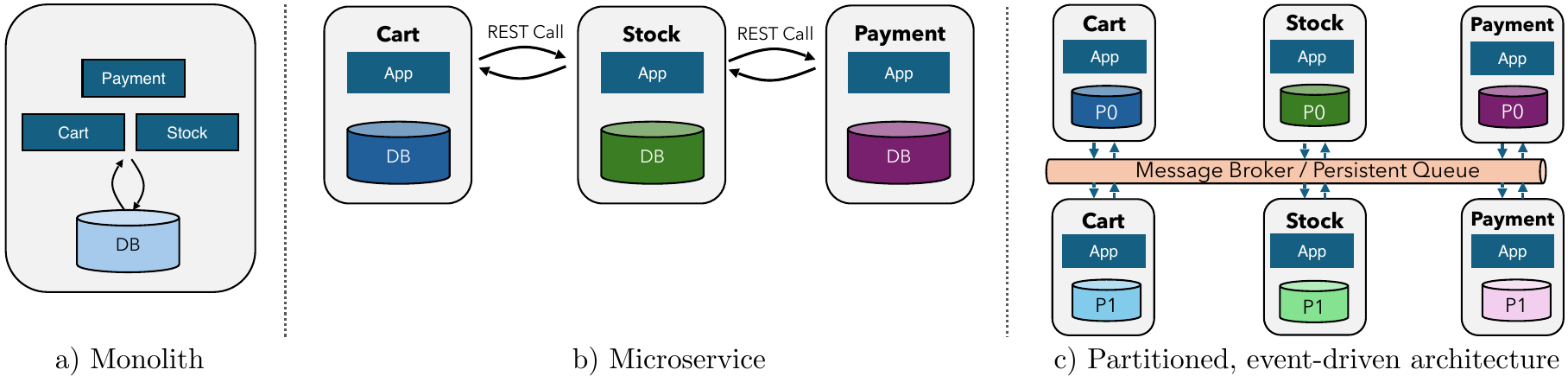}
    \captionsetup{justification=centering,margin=0cm}
    \caption{Three-step process of converting a monolith to a scalable, low-latency service architecture.}
    \label{fig:monolith-to-micro}
\end{figure*}

\Cref{fig:monolith-to-micro} depicts the process of breaking down a monolithic application (\Cref{fig:monolith-to-micro}a) into three microservices, each with its own database (\Cref{fig:monolith-to-micro}b). In the microservice architecture, direct access to a single database and DBMS-based transactions are no longer possible. Instead, the microservices split functionality and maintain their own database. Each service's database is partitioned to scale out, as shown in (\Cref{fig:monolith-to-micro}c). REST API calls are also transformed into messages that asynchronously trigger those calls.

\para{Microservices Implement Dataflows} A critical observation is that the architecture depicted in \Cref{fig:monolith-to-micro}c closely resembles a streaming dataflow graph with the partitioned state co-located with the application logic. While we elaborate on this in \Cref{sec:dataflows-adapted}, in short, this architectural pattern is the same pattern that is followed by streaming dataflow systems such as Apache Flink~\cite{flink} and Spark Streaming~\cite{ZahariaDLHSS13}.

\subsection{Messaging, Idempotency \&\ Consistency}
\label{sec:messaging-idempotency}
Traditional monoliths achieved workflow execution atomicity (e.g., a shopping cart checkout) by combining state mutations across subsystems (cart, payment, stock) in a single transaction. If the transaction fails, the database rolls back to the previous state, and the application retries the checkout.

\para{Idempotency in Services} To achieve the same effect, a stateful service or function must be idempotent, meaning that calling the service multiple times should have the same effect on the global state of an application as calling it exactly once. Considering that various issues can arise \cite{DBLP:conf/cidr/Helland07} when two services communicate (such as network failures, rescaling, or service restarts), currently ensuring idempotency works as follows: the sender service generates an \textit{idempotency-key}\footnote{\url{https://datatracker.ietf.org/doc/html/draft-idempotency-header-00}} that is persisted in the state of the sender, right before the call is performed. Suppose the sender sends a message twice (e.g., because of an intermittent network issue or a failure). In that case, the \textit{idempotency-key} must be recognized and safely ignored by the receiving service. It is important to note here that idempotency cannot be achieved without \textit{persisting} the idempotency-key to durable storage (e.g., a database) in the \textit{same local transaction} as the one that mutates the state of the receiver. At the moment, \textit{idempotency-keys} are managed by the developers, adding to the complexity of developing cloud applications. 

\subsection{Transactions \&\ Orchestration} 
\label{sec:transactions-orchestration}

\para{Serializability in Services} Multiple works advocate that serializable guarantees are preferred \cite{blanastransactions,transactionsmakedebuggingeasy}. This is also reflected in the offered isolation levels of post-NoSQL systems such as Google's Spanner~\cite{spanner} and, more recently, CockroachDB~\cite{cockroachdb}, which all provide serializability. 
Serializability has been highly important in monolithic applications, but in distributed service deployments, it is virtually impossible to reason about correctness in the presence of state inconsistencies \cite{transactionsmakedebuggingeasy}. Transactional service architectures must address message delivery guarantees. 

\para{SAGAs and Two-Phase Commit} Popular solutions to this challenge, known for years, involve the Saga pattern and two-phase commit protocols orchestrated by a transaction coordinator implementing XA transactions \cite{specification1991distributed}. However, both of them present significant drawbacks. Implementing the Saga pattern involves managing the execution of compensating actions to reverse the partial state effects of a failed workflow while the offered consistency level is eventual. Alternatively, 2PC protocols coupled with two-phase locking provide atomicity and isolation at the expense of blocking the progress of service orchestrations involved in a transaction. We need a new way to architect cloud applications with support for transactional workflows that span multiple components of an application.

\para{Orchestrators} Currently, several commercial orchestrators are available for executing SAGAs. Those orchestrators ensure atomicity only: they make sure that a given sequence of service calls eventually comes to completion. While we do see the value of orchestrators for analytics applications (e.g., as Apache Airflow~\cite{airflow}, AWS Step Functions~\cite{stepfunctions}), orchestrators are not suitable for transactional applications, as they are all \textit{oblivious} of the state of the functions/services that they are orchestrating. 

\subsection{Application (Re-)Scaling}
\label{sec:rescaling}
Scaling microservices requires scaling the stateless business logic and the state management system that serves the stateless part of an application. Scaling stateless services is relatively straightforward: one needs to rescale the application logic instance, assuming that the database behind the stateless instance can handle the new load. However, when optimizing for latency, the database is partitioned and preferably co-located with the application logic. In that case, rescaling an application becomes a hurdle: the database has to migrate state and possibly keep replicas. Soon enough, application developers re-implement some version of database state migration and rescaling \cite{clay} protocol.

While current FaaS cloud offerings do allow for stateless functions to scale on demand, they still provide no transaction management primitives that take into account service orchestrations and state consistency issues during the rescaling process. An ideal runtime should be able to perform the rescaling of applications without forcing operations teams and developers to perform rescaling by hand while keeping the state across services transactionally consistent.

\section{Streaming Dataflows to the Rescue}
\label{sec:dataflow-everything}
In this section, we highlight the key aspects and advantages of streaming dataflow systems design and argue that they can be extended to encapsulate the primitives required for executing transactional cloud applications consistently and efficiently. Moreover, we argue that combining deterministic databases and dataflow systems can create a runtime that ensures atomicity, consistency, and scalability. Finally, we show how deterministic databases can be extended for SFaaS, where transaction boundaries are unknown, unlike online transaction processing (OLTP). 

\vspace{-1mm}
\subsection{Dataflows as an Architectural Abstraction}

Stateful dataflows are the execution model implemented by virtually all modern stream processors \cite{fragkoulis2024survey}. Streaming systems owe their wide adoption in the last decade to a set of key system design aspects: exactly-once processing, consistent fault tolerance, co-location of state and compute, and data-parallel scale-out architecture. We elaborate on these characteristics below.

\para{Exactly-once Processing} Message-delivery guarantees are fundamentally hard to deal with in the general case, with the root of the problem being the well-known Byzantine Generals problem. However, in the closed world of dataflow systems, exactly-once processing is possible \cite{flink,carbone2017state}. In principle, to achieve exactly-once processing, the processing layer records the outcome of each message's state effects, the networking layer ensures message delivery in FIFO order, and the fault tolerance layer guarantees that no message that is already reflected in the state will be processed again. Note that the guarantee of exactly-once processing significantly simplifies programming. The APIs of popular streaming dataflow systems, such as Apache Flink, require no error management code (e.g., message retries or duplicate elimination with \textit{idempotency-keys}).

\para{Fault Tolerance} Exactly-once processing extends to the system's fault tolerance approach. The two can be gracefully combined using Chandy-Lamport's distributed snapshot protocol~\cite{chandy1985distributed}  adapted for streaming systems~\cite{carbone2017state, checkmate}. The approach involves periodically circulating special messages called checkpoint markers into the streaming dataflow system, instructing its operators to snapshot their state. Because checkpoint markers coexist with common data-related messages on the same channel, they enforce a global order that creates a consistent cut of the system's state. In case of a failure, the system can automatically roll back to the latest checkpoint of its distributed state and resume processing from that point, assuming the input is delivered from a replayable source, such as Apache Kafka~\cite{kreps2011kafka}. This fault tolerance approach ensures that the system's state remains consistent under failures.

\para{Co-location of State with Compute} Streaming dataflow systems have demonstrated their capacity to process millions of events per second~\cite{flink}. One main design decision that enables this level of sustainable performance is that the system's operators maintain the state of their computations in their local memory space. The state is periodically snapshotted to persistent storage, securing the progress of continuous computations against failures. Notably, this coarse-grained approach bears a low overhead to the system's regular operation.

\para{Data-parallel Scale-out Architecture} Continuing from the previous point, the system's architecture enables high-throughput at scale.
Each operator in the logical dataflow graph is instantiated as several operator instances deployed in distributed nodes. Each instance holds a partition of the operator's state, enabling input data to be distributed and processed in parallel across instances.

\subsection{Dataflows for Transactional Applications}
\label{sec:dataflows-adapted}

The aforementioned advantages of streaming dataflow systems do not apply to transactional cloud applications. To begin with, typical transactional workloads in the cloud manifest as workflows of functions that arbitrarily call one another. This computation pattern is markedly different from analytics functions that populate the operators of streaming systems. Second, streaming dataflow systems lack support for transactions as prescribed in the database literature~\cite{papadimitriou1979serializability}. Finally, the development of workflows of functions entails a programming model that can convey transactional semantics, form workflows, and support custom business logic. This programming model departs from the typical way of programming stream-processing jobs as chained functional transformations. 

\para{Dataflows for Arbitrary-Workflow Execution} The prime use case for dataflow systems nowadays is streaming analytics, which typically involves executing a chain of standalone functions. By comparison, transactional cloud applications involve arbitrary workflows of functions calling each other. To enable the execution of arbitrary workflows in a dataflow system, we connect operators at the system level such that an operator can directly invoke a computation in another operator. In addition, we allow such nested computations to be executed in parallel. Finally, we devised an approach for identifying the transaction boundaries of a workflow, which we briefly describe next.

\para{Deterministic transactions} Deterministic transactional protocols have two properties that make them coexist harmoniously with dataflow systems. First, given a set of sequenced transactions, a deterministic database \cite{abadi2018overview, calvin} will reach the same final state with serializable guarantees despite node failures and possible concurrency issues. This property is essential because it allows a deterministic transactional protocol to align with a dataflow system without changing the stream processor's checkpointing mechanism.

Second, unlike 2PC, which requires rollbacks in case of failures, deterministic database protocols \cite{aria,calvin} are "forward-only": once the locking order \cite{calvin} or read/write set \cite{aria} of a batch of transactions has been determined, the transactions will be executed and reflected on the database state, without the need to rollback changes. This alignment between deterministic databases and the dataflow execution model is the primary motivation to support a deterministic transaction protocol on top of a dataflow system.


Still, supporting deterministic transactions in a streaming dataflow system is not trivial and poses two main challenges that we address in our prototype system presented in \Cref{ch2:sec:design}. The first challenge is determining transaction boundaries. This is not required in deterministic databases, where each transaction is encapsulated in a single-threaded function that can execute remote reads and writes across partitions \cite{calvin,aria}. In SFaaS, however, arbitrary function calls to remote partitions are common because they enable developers to leverage both the separation-of-concerns principle widely applied in microservice architectures \cite{rodrigosurvey} and code modularity. Therefore, to determine the boundaries of a transactional workflow, we introduce an accounting scheme for function calls nested within a workflow. The scheme, which also supports calls to remote operators and cycles, signals the termination of a workflow's execution once all function calls complete.

The second challenge is deciding when to commit to durable storage and reply to users. Traditionally, a transactional system can respond to a client only when $i)$ the requested transaction has been committed to a persistent, durable state or $ii)$~the write-ahead log is flushed and replicated. Within the scope of a dataflow system, this would require completing a snapshot, leading to prohibitive latency. However, a deterministic transactional protocol executes an agreed-upon sequence of transactions among the workers; after a failure, the system would run the same transactions with exactly the same effects. This determinism allows for early commit replies: the client can receive a reply before a persistent snapshot is stored.

\para{Programming Models} Currently, dataflow systems are only programmable through functional-programming style dataflow APIs: a given cloud application needs to be rewritten by developers to match the event-driven dataflow paradigm. Although it is possible to rewrite many applications in this paradigm, it takes a considerable amount of programmer training and effort to do so. Therefore, we have introduced an object-oriented programming abstraction that encapsulates functions into actor-like entities. We present the programming model as a whole in Section~\ref{sec:programming-model}. We argue that this programming model is suitable for developing transactional cloud applications like microservices.

\section{The Stateflow/Styx Approach} \label{ch2:sec:design}

Styx (\Cref{chapter4})~\cite{styx} is a transactional distributed dataflow system that executes workflows of stateful functions with serializable guarantees. Styx adopts Stateflow (\Cref{chapter3})~\cite{stateflow} as a higher-level programming abstraction, enabling users to code in a pure object-oriented style without state management or fault tolerance considerations. In this section, we briefly describe the programming model (Section~\ref{sec:programming-model}) and underlying system (Section~\ref{sec:design-system}).


\subsection{Programming Model} 
\label{sec:programming-model}

The Stateflow/Styx framework provides developers with two levels of abstraction: a high-level actor-like programming interface based on Stateflow~\cite{stateflow} and a lower-level dataflow API~\cite{styx}. 

\para{High-level}
Users can code transactional cloud applications in Python object-oriented code where an entity is an object with a unique key and class functions that mutate the entity's state (similar to actor programming). Additionally, when an entity calls a function of another entity, Stateflow automatically creates an edge in the dataflow graph. We describe Stateflow's workings and how it uses continuation-passing style programming to transform calls between different entities into a distributed dataflow graph in \cite{stateflow}.

\para{Low-level} Styx follows the operator API of dataflow systems (e.g., Apache Flink~\cite{flink}). In Styx, a streaming operator can hold multiple entities based on a partitioning scheme, on functions that act upon the operator as a whole (allowing range queries), or on the entities themselves (allowing point queries). To communicate across operators, developers can call remote operator functions using Styx's API.

\subsection{The Styx Runtime}
\label{sec:design-system}

Styx employs a typical worker/coordinator architecture. It is complemented by a messaging system, such as Apache Kafka, that propagates input to Styx, including the replay of unprocessed messages following a failure.
The coordinator's responsibilities are to deploy a user-defined dataflow graph to the workers, monitor the cluster's health while collecting useful metrics, and trigger the fault tolerance pipeline in case of failure. 

The workers are responsible for a subset of the dataflow graph's operator state partitions, which are 1-to-1 aligned with the partitions of the replayable input source, say Apache Kafka. First, each worker ingests client requests through Kafka and sequences them (Styx uses a non-replicated sequencer partitioned per worker). Then it receives a batch of transactions from the sequencer and executes them as coroutines on a single CPU to improve efficiency. To execute transactions deterministically, Styx extends a deterministic transactional protocol similar to Calvin~\cite{calvin} and Aria \cite{aria}. Determinism is required by the dataflow snapshotting mechanism to guarantee the same state mutations after a replay in case of failure. Transactions are executed in parallel across workers, and nested function calls are transparently scheduled for execution by local or remote operators. Finally, Styx's acknowledgment-sharing scheme signals the end of a transaction's execution.

\para{Fault Tolerance} To recover from failures, Styx relies on a replayable input source to perform deterministic message replay based on recorded offsets. This design ensures that the sequencer will re-create the same transaction sequence post-recovery and enables early replies (before the state commits to durable storage). Finally, Styx utilizes a blob store to persist incremental snapshots of worker states.


\section{Related Work}
\label{sec:related}

Our system shares motivation with projects such as Hydroflow \cite{hydro} and DBOS \cite{dbos}. DBOS takes a DB-centric approach, where functions can be translated into stored procedures within a database (co-location of state and processing) or on the server, where state needs to be transferred, and workflows form a database transaction with ACID guarantees. Hydroflow, at its present state, does not support transactional end-to-end workflows and focuses primarily on cloud-native stream processing for analytics.
Cloudburst \cite{cloudburst} provides causal consistency guarantees within a single Directed Acyclic Graph (DAG) workflow. Netherite \cite{netherite} offers exactly-once execution guarantees and a high-level programming model, though it does not ensure transactional serializability across functions. Orleans \cite{orleans} introduces virtual actors decoupling applications from the underlying architecture, but does not guarantee exactly-once message delivery. 
Finally, transactional SFaaS paradigms with serializability guarantees (Beldi~\cite{beldi}, Boki~\cite{boki}, and T-Statefun~\cite{tstatefun}) do support transactional end-to-end workflows but suffer from poor performance and fail to decouple user code from their transactional primitives.

\chapter{T-Statefun: Adapting Existing Approaches for Transactional Stateful Functions}
\label{chapter2}

\vfill

\begin{abstract}

\Cref{introduction,chapter1} introduced the motivation for supporting general-purpose cloud applications with strong consistency guarantees. This chapter investigates whether an open-source platform can be adapted to support transactions in stateful cloud functions. 

Before building a custom system from scratch, we sought to understand whether existing open-source platforms could be adapted to meet our goals. In this chapter, we present \textit{T-Statefun}, our extension to Apache Flink StateFun, a Stateful Function-as-a-Service (SFaaS) platform built atop a stream processing engine that already offers exactly-once processing guarantees.

T-Statefun introduces two complementary models for transactional coordination across stateful functions: the Saga pattern for eventual consistency and two-phase commit (2PC) for serializability. By implementing both on StateFun's dataflow runtime, we explored how far a general-purpose streaming engine can be stretched to support transactional workflows typically required in microservices and cloud-native applications.

Finally, the limitations we observed with T-Statefun informed the design of our system, Styx (\Cref{chapter4}). Thus, this chapter serves as both a feasibility study and a key design probe.

\end{abstract}

\vfill

\blfootnote{This chapter is based on the following research paper and its extended journal version:\\
\faFileTextO~\faTrophy~\hangindent=15pt\emph{M. de Heus, K. Psarakis, M. Fragkoulis, and A. Katsifodimos. Distributed transactions on serverless stateful functions, DEBS '21}~\cite{tstatefun}. \\
\faFileTextO~\hangindent=15pt\emph{M. de Heus, K. Psarakis, M. Fragkoulis, and A. Katsifodimos. Transactions Across Serverless Functions Leveraging Stateful Dataflows. In Elsevier's Information Systems, Volume 108, September 2022}~\cite{tstatefunjournal}.
}

\newpage


\dropcap{T}{}he idea of democratizing distributed systems programming is not new. Approaches such as Distributed ML~\cite{krumvieda1993distributed} and Erlang~\cite{armstrong2013programming} aim to simplify the programming and deployment of distributed applications. Erlang \cite{armstrong2013programming} first introduced the actor model, which Akka~\cite{wyatt2013akka} implemented later in Scala, offering a low-level programming model. Following that, Virtual Actors \cite{bykov2011orleans,orleans} try to abstract away the low-level primitives.

Serverless computing~\cite{jonas2019cloud} is a cloud computing execution model promising to simplify the programming, deployment, and operation of scalable cloud applications.
In the serverless model, developer teams upload their code written in a high-level API, and the cloud platform handles application deployment and operations.
Serverless computing aims to substantially increase cloud adoption by addressing the status quo in the cloud landscape, where developer teams need to possess skills in distributed systems, data management, and the internals of cloud execution models to use the cloud effectively.

\para{Function-as-a-Service \& Messaging} The most prominent serverless offering is Function-as-a-Service (FaaS), in which users write functions, and cloud providers automate deployment and operation.
However, FaaS offerings lack support for state management and the ability to execute transactional workflows across multiple functions~\cite{hellerstein_serverless19, SFaaS-in-action}, which are needed by general-purpose cloud applications. In addition, none of the current FaaS approaches offers message-delivery guarantees, failing to support \emph{exactly-once} processing: the ability of a function to mutate its state exactly once per incoming message.

When a system does not guarantee exactly-once processing, the burden of debugging and handling system errors (e.g., machine failures, network partitions, or stragglers) falls on developers~\cite{tech-debt-microservices}. These developers then have to ``pollute'' their business logic with extra consistency checks, state rollbacks, timeouts, or recovery mechanisms, for example. \cite{microservices-drawbacks}. The result is that the majority of the application code is not comprised of business logic but error checking, management, and mitigation~\cite{rodrigosurvey}. Sooner or later, programming distributed systems at the application level leads to problems with state consistency, bugs, and eventually significant service outages.

Message-delivery guarantees are fundamentally hard to handle in the general case, with the root of the problem being the well-known Byzantine Generals Problem~\cite{lamport1982byzantine}. However, in the closed world of dataflow systems, exactly-once processing is possible\cite{flink,carbone2017state,silvestre2021clonos} as in stateful dataflows, the system has \emph{full control over both messaging and state management}. Apache Flink's StateFun~\cite{statefun} is, to the best of our knowledge, the first approach to build a FaaS execution engine on top of a streaming dataflow system offering exactly-once processing guarantees even under complex failure scenarios. However, StateFun's approach can also be implemented on top of other dataflow systems~\cite{Carbone2020beyond, flink, ArmbrustDT18, akidau2013millwheel, jet}.

Such dataflow systems can execute stateful functions as follows: incoming events represent function execution requests routed to continuous stateful operators that execute the corresponding functions. With proper, consistent fault tolerance mechanisms~\cite{carbone2017state,silvestre2021clonos}, state-of-the-art stream processing systems operate at high throughput and low latency. At the same time, they guarantee the correctness of execution even in the presence of failures. As we show in this paper, this set of properties can support \textit{transactions} with minimal involvement from application developers.

\para{Transactional SFaaS} Although there is ongoing work on supporting stateful FaaS (SFaaS) applications that mutate state transactionally, \textit{across} functions, remains an open problem. The only approach addressing distributed transactions in an SFaaS setting is Beldi~\cite{beldi}, which provides fault-tolerant ACID transactions on stateful workflows across functions by logging the functions' operations to a serverless cloud database. Cloudburst~\cite{cloudburst} with HydroCache~\cite{causal-consistency-cloudburst} provides causal consistency on function workflows forming a DAG by leveraging Anna~\cite{anna-kvs}, a key-value store with conflict resolution policies in place. Cloudburst does not provide isolation between DAG workflows.

In contrast with the aforementioned approaches, developer teams in the microservices and cloud applications landscape go to extreme lengths when they need to implement transactional workflows across the boundaries of a single service or function. The most common approach adopted is the Saga pattern~\cite{sagas}. The Saga pattern separates a transaction into sub-transactions that proceed independently with the benefit of improved performance, but at the risk of having to undo or compensate the changes of successful sub-transactions when at least one of the involved sub-transactions fails.
In addition, compensating actions can be challenging when concurrent changes are applied to the state because Sagas do not require any means of isolation. For this reason, state consistency needs to be dealt with at the application level. On the other hand, applications that prioritize consistency over performance implement distributed transactions using the two-phase commit protocol. Two-phase commit (TPC) \cite{Gray1978} offers ACID, serializable transactions, but imposes blocking operations across functions participating in a transaction, which penalizes performance in return for strict atomicity.

In this chapter, we draw inspiration from best practices in developing microservices and cloud applications and offer developers a programming model that supports both Sagas and distributed transactions with two-phase commit.
Our implementation for authoring workflows across stateful functions in FaaS with transactional guarantees is publicly available on GitHub\footnote{\url{https://github.com/delftdata/flink-statefun-transactions}}.
We implement the two approaches on an open-source stateful FaaS system, Apache Flink~\cite{flink} StateFun~\cite{statefun}, and call our extension T-Statefun. \\

\noindent In summary, this chapter makes the following contributions:

\begin{itemize}
    \item We argue for implementing transactional workflows on a stateful dataflow engine and outline its advantages.
    \item We propose a programming model for transactional workflows across stateful serverless functions.
    \item We implement the two main approaches used by cloud application practitioners to achieve transactional guarantees: two-phase commit and Saga workflows.
    \item We evaluate two transactional schemes using an extended version of the YCSB benchmark on a cloud infrastructure.
    \item We compare against the state-of-the-art academic SFaaS proposal that supports serializable transactions and one of the most popular transactional distributed database systems.
\end{itemize}


\section{Transactions on Streaming Dataflows} \label{ch2:sec:transactionsOnDataflow}

Serverless platforms come in different flavors. One breed of SFaaS systems (e.g., Apache Flink StateFun and \cite{SFaaS-in-action}) is built on top of a stateful streaming dataflow engine. This architecture bears important implications for supporting transactions because of how distribution, state management, and fault tolerance work. 

Network communication between distributed components in a typical streaming dataflow engine is implemented via FIFO network channels that guarantee exactly-once processing and preserve delivery order. 
In a serverless FaaS system, this characteristic obviates the need to handle lost messages and implement retry logic concerning function invocations in transactional workflows. Messaging errors and retries are a significant source of friction and development effort at the application level, and those are offered by the underlying dataflow system.

State management in state-of-the-art streaming systems achieves exactly-once processing guarantees by taking consistent snapshots of the system's distributed state periodically~\cite{carbone2017state}. The snapshots capture a globally consistent state of the system at a specific point in time and are used to recover the system's state upon failure. Exactly-once means that the changes brought by each function execution instance are recorded in the system's state exactly once, even in the face of failures.
For transactions, this capability is essential because fault recovery of transactions can piggyback on the underlying fault tolerance mechanism with zero effort and knowledge by the application.
Given that a big part of code and effort is spent on failure handling, fault tolerance, and virtual resiliency \cite{goldstein2020ambrosia} provided at the system level can play a significant role.

Furthermore, unlike traditional streaming queries, where the computations are fully encapsulated within the system's operators, it is common to have nondeterministic side effects (typically calls to external services or remote key-value stores) in microservices and cloud applications. However, the traditional fault tolerance mechanisms of streaming dataflow systems were not designed to support non-determinism prevalent in general-purpose applications. Thus, the consistency of applications and the integrity of transactions are at risk when transactions involve nondeterministic operations. Extending the fault tolerance approach of streaming dataflow systems to support nondeterministic computations~\cite{silvestre2021clonos} is an important step towards opening their adoption for executing general-purpose applications. Recent work~\cite{hydro, cloudburst} also recognizes the dataflow model as a key enabler for the SFaaS systems of the future.

In short, we believe that stateful streaming dataflows and the associated research that has been proposed so far\cite{fernandez2014making, Katsifodimos2019operational, affetti2017flowdb} can alleviate the burden of building rich stateful and transactional applications on top of streaming dataflows. This paper presents a step in this direction.

\section{Preliminaries} \label{ch2:sec:background}

In this section, we first present our transactional model (\Cref{ch2:sec:transactions-definition}). Then, in \Cref{ch2:sec:flinkstatefun}, we describe the functionality and internals of Apache Flink StateFun, which forms the backbone of our proposed solution. Lastly, in \Cref{ch2:sec:requirements}, we list the requirements that an SFaaS system should satisfy in order to be considered as a backend for our work.

\subsection*{Transaction Model}
\label{ch2:sec:transactions-definition}

In the context of this work, a transaction is an atomic execution of a set of stateful function invocations. 
More specifically, the transactional model introduced in this paper considers transactions defined up-front. This is referred to as \textit{single-shot} \cite{single-shot-aws} or \textit{one-shot} \cite{kallman2008h} \textit{transactions} in prior works. We follow the definition of H-Store's~\cite{kallman2008h} \textit{one-shot transactions}, which states that the output of a function (query) cannot be used as input to subsequent functions (queries) in the same transaction. Since the output of functions is not used by subsequent ones, the execution of functions involved in a transactional workflow is independent of one another. This simplifies coordination of the transaction across the system while still providing a practical model for transactions. Widely used database services, such as Amazon's DynamoDB  \cite{sivasubramanian2012amazon}, support one-shot transactions \cite{single-shot-aws}. In an SFaaS system, one-shot transactions provide a significant advantage: functions can implement arbitrary business logic in a general-purpose programming language such as Java or Python instead of being limited to the API supported by a specific database, such as DynamoDB. Thus, this advantage translates to considerable flexibility in the programming model.

\subsection*{Apache Flink StateFun}
\label{ch2:sec:flinkstatefun}

Apache Flink StateFun\footnote{\url{https://flink.apache.org/stateful-functions.html}} offers an abstraction and runtime for users to implement stateful cloud functions. A stateful function implemented by user code is referred to as a function type and describes the state it holds. Multiple instances based on the same function type can exist in parallel and are identified by an ID. Each of these function instances encapsulates its own state and can be uniquely addressed by its type and ID. Function instances can be invoked from other function instances or through ingress points such as Kafka. Function instances can have four different controlled side effects: (1) state updates, (2) function invocations, (3) delayed function invocations, (4) egress messages (for example, Kafka). StateFun supports end-to-end exactly-once guarantees from ingress to egress, including any state updates. 

\begin{figure}
    \centering
    \captionsetup{justification=centering,margin=0cm}
    \includegraphics[width=0.6\columnwidth]{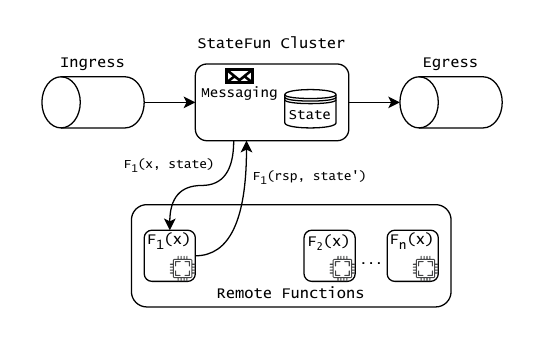}
    \caption{Flink StateFun Cluster Architecture}
    \label{ch2:fig:statefun-architecture}
\end{figure}

\para{Architecture} In \Cref{ch2:fig:statefun-architecture} we present the general system architecture of Apache Flink StateFun. The interface with the system is based on the ingress/egress pattern (e.g., ingest/produce Kafka messages). The Apache Flink StateFun cluster lies at the core of the system, consisting of multiple workers that manage both messaging and partitioned state, enabling stateless remote functions. However, this means the state must be transferred along with the request to each specific function for processing. After processing, both the response and the new state are returned to the StateFun cluster. This architecture's primary benefit is that since StateFun manages both messaging and state exactly-once semantics is easier to achieve than other architectures.

\para{Embedded vs. Remote Functions} Functions can be deployed both inside the StateFun workers (referred to as embedded functions) and outside the StateFun cluster (co-located and remote functions). Embedded functions are simply an abstraction on top of stateful streaming operators in Flink, therefore providing exactly-once and fault-tolerance guarantees. StateFun allows dynamic communication between these streaming operators by introducing a cycle in the streaming graph. The co-located and remote functions are entirely stateless because the state is persisted within StateFun. This paper focuses on remote functions as these can leverage existing FaaS services such as AWS Lambda to auto-scale the compute layer. \Cref{ch2:fig:original-communication-flow} shows how remote functions work. Each function instance is represented by an embedded stateful function in the StateFun cluster. This standardized embedded function is responsible for managing the state of the function instance and communicating with the remote function, which may be deployed anywhere. The persisted data in the embedded stateful function with the communication pattern for remote functions are shown in \Cref{ch2:fig:original-communication-flow}.

\begin{figure}[t]
    \centering
    \captionsetup{justification=centering,margin=0cm}
    \includegraphics[width=0.7\columnwidth]{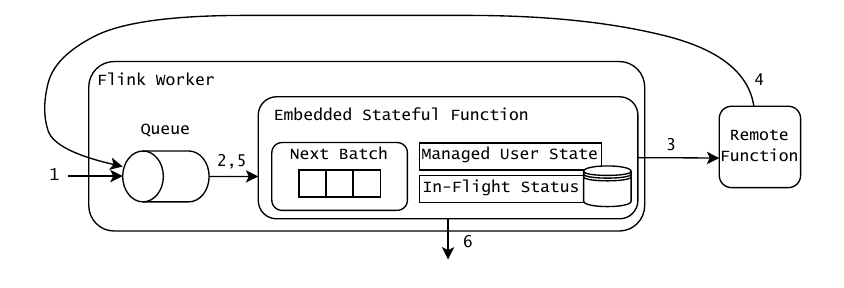}
    \caption{Original communication flow of Flink Statefun}
    \label{ch2:fig:original-communication-flow}
\end{figure}

\para{Function Invocations as Dataflow Messages} Invocations that are sent to a function instance arrive in a queue, as shown in step 1 of  \Cref{ch2:fig:original-communication-flow}. If the embedded stateful function is ready to process the next invocation, it pulls a message (invocation parameters) from the queue (step 2). When no invocation is being executed at the remote function, the remote function is called. However, if the remote function is busy with a previous function call, the current invocation message is appended to the next batch. Batching is used as an optimization in order to avoid multiple remote calls to external functions at the expense of latency (see \Cref{ch2:sec:exp1}). Batches are also used to preserve the invocation order and the order of state access (the batch must wait until the state updates caused by the previous batch have been applied), thus ensuring linearizability at the function instance level.

In step 3, the stateless remote function is called through a Protobuf interface that contains both the (keyed) state required for the remote function to operate and the invocation parameters of the function. The stateless remote function can execute the (batch of) invocations and will be ready to return the updated state back to the Flink worker that made the call. In step 4, the response of the stateless remote function is appended to the queue of incoming messages to the function. The response includes any side effects caused by the invocation(s), including updates to the user-defined state.

When the response from the stateless function is processed (step 5), the side effects caused by the invocation(s) are applied to the state of the embedded function, updating the managed state in the embedded stateful function. If any invocations are batched, the next batch of invocations is sent to the remote stateless function, and the batch is truncated. When there are no batched invocations, the in-flight status is cleared. Finally, any outgoing function invocations are sent to the queues of their respective function instances, and egress messages are sent to their respective egresses (step 6).

\subsection*{Assumptions \& Requirements}
\label{ch2:sec:requirements}
As we describe in the next section, our coordinator functions rely on an underlying SFaaS system for bookkeeping the state of ongoing transactions and reliable messaging. To allow this, the underlying system should satisfy two requirements.

\para{Exactly-once Processing Guarantees}
Firstly, all communication must be reliable and executed with exactly-once processing guarantees. Thus, we require that the underlying system be fault-tolerant~\cite{flink} to ensure transaction atomicity in the event of a failure. This also means the state is durable across snapshots/checkpoints, even in the event of failures. If we can rely on exactly-once processing guarantees, message replay, and error handling, a significant part of transaction coordination can be simplified. Flink StateFun does guarantee exactly-once processing.

\para{Linearizable Operations}
The second requirement is that the operations for any specific function instance should be linearizable, meaning that there is a well-defined order in which operations are performed on the instance and the state it encapsulates. Accordingly, a function invocation will always have the correct state of the function instance to implement transactions. Since Flink StateFun's function instances use a single replica of the state per function instance and a single process executes function invocations for that function instance in a sequential FIFO manner, this ensures linearizable operations per function instance.

\section{Coordinator Functions \&\ the T-Statefun API} \label{ch2:sec:apiPlusOverview}

\begin{table*}[t]
\captionsetup{justification=centering,margin=0cm}
\resizebox{\textwidth}{!}{
\centering
\normalsize{
\begin{tabular}{l r}
\textbf{Function} & \textbf{Description} \\
\toprule
\multicolumn{2}{c}{Shared coordinator function methods} \\
\midrule
\textit{send\_on\_success}(\textit{type}, \textit{id}, \textit{message}) & Sends a message to another function instance if the transaction is successful \\
\textit{send\_after\_on\_success}(\textit{delay}, \textit{type}, \textit{id}, \textit{message}) & Sends a delayed message if the transaction is successful \\
\textit{send\_egress\_on\_success}(\textit{type}, \textit{egress\_message}) & Sends a message to an egress if the transaction is successful \\
\textit{send\_on\_failure}(\textit{type}, \textit{id}, \textit{message}) & Sends a message to another function instance if the transaction failed \\
\textit{send\_after\_on\_failure}(\textit{delay}, \textit{type}, \textit{id}, \textit{message}) & Sends a delayed message if the transaction failed \\
\textit{send\_egress\_on\_failure}(\textit{type}, \textit{egress\_message}) & Sends a message to an egress if the transaction failed \\
\midrule
\multicolumn{2}{c}{Two-phase commit function methods} \\
\midrule
\textit{tpc\_invocation}(\textit{type}, \textit{id}, \textit{message}) & Add a function invocation to the transaction \\
\textit{send\_on\_retryable}(\textit{type}, \textit{id}, \textit{message}) & Sends a message if the transaction aborted because of a deadlock \\
\textit{send\_after\_on\_retryable}(\textit{delay}, \textit{type}, \textit{id}, \textit{message}) & Sends a delayed message if the transaction aborted because of a deadlock \\
\textit{send\_egress\_on\_retryable}(\textit{type}, \textit{egress\_message}) & Sends a message to an egress if the transaction aborted because of a deadlock \\
\midrule
\multicolumn{2}{c}{Sagas function methods} \\
\midrule
\textit{saga\_invocation\_pair}(\textit{type}, \textit{id}, \textit{message}, \textit{compensating\_message}) & Add a pair of a message and a compensating message to the transaction \\
\midrule
\multicolumn{2}{c}{Ordinary functions} \\
\midrule
\textit{FunctionInvocationException} & Raised to fail the function invocation \\
\bottomrule
\end{tabular}
}
}
\caption{Coordinator functions' Python API.}
\label{ch2:tbl:API}
\end{table*}

In this section, we introduce the concept of stateful coordinator functions and provide an overview of our approach. Our approach is based on the simple observation that since an underlying SFaaS system provides exactly-once processing and message delivery guarantees, conceptually, it would be much simpler to implement a transaction coordinator as a regular, stateful function. With this in mind, we opted for implementing a transaction API on top of stateful functions, which we present in \Cref{ch2:tbl:API}. Notably, further work is required to raise the transaction abstractions at an even higher level \cite{Katsifodimos2019operational, SFaaS-in-action} as syntactic sugar.

A stateful coordinator function is a stateful function that preserves state about the execution of a given transaction. Coordinator functions have the ability to force other function instances to abort or compensate for the changes they applied. 

\para{API Overview} Our coordinator function implements two transaction coordination patterns: two-phase commit and Sagas \cite{sagas}. A complete example of a coordinator function for two-phase commit and Saga is shown in Listings \ref{ch2:lst:2pc-coordinator} and \ref{ch2:lst:saga-coordinator}, respectively. In short, to coordinate a two-phase commit transaction, the user needs to invoke function instances via \textit{tpc\_invocation}, while for a Saga, an invocation \textit{pair} is expected, which consists of the normal transaction invocation and the corresponding compensation invocation to be sent to the same function instance. A Saga invocation pair can be called with \textit{saga\_invocation\_pair}. An important difference between the behavior of the two schemes is that a failure in a Saga workflow will incur a compensating function call.

\begin{figure}[t]
\captionsetup{justification=centering,margin=0cm}
\begin{lstlisting}[style=pythonlang, xleftmargin=0in,label={ch2:lst:2pc-coordinator}, caption={Two-phase commit coordinator function.}]
def serializable_transfer(context, message: Transfer):
    subtract_credit = SubtractCreditMessage(amount=message.amount)
    context.tpc_invocation("account_function",
                          message.debtor,
                          subtract_credit)

    add_credit = AddCreditMessage(amount=message.amount)
    context.tpc_invocation("account_function",
                          message.creditor,
                          add_credit)
\end{lstlisting}
\vspace{-5mm}
\end{figure}

\para{Two-Phase Commit} The \textit{serializable\_transfer} function of \Cref{ch2:lst:2pc-coordinator} receives a context (the underlying context of StateFun as we have extended it to support transactions) and a message. The message is of type \textit{Transfer}, and it contains three fields: the amount of money transferred, the creditor, and the debtor. The amount mentioned in the message must be subtracted from the debtor and transferred to the creditor. To this end, assuming that there is a function type registered in the system as \textit{account\_function}, as per the original StateFun API, we need to construct an object containing the parameters for the \textit{account\_function} and push that message to the transaction coordinator. This is done in lines 5-7: we give the TPC coordinator the function type to invoke, alongside the ID of the debtor to form the address of the function instance, and the \textit{SubtractCreditMessage}, which is going to be given to that function as a parameter. Subsequently, we do the same for the creditor: we construct an \textit{AddCreditMessage}, and we pass it over to the function type \textit{account\_function}. In short, the transaction coordinator function instance will make sure that the two function instances are invoked with serializable guarantees. It does this by coordinating a two-phase commit protocol across the function instances with locking to ensure isolation. More details on these aspects are given in \Cref{ch2:sec:implementation}.

\begin{figure}[t]
\captionsetup{justification=centering,margin=0cm}
\begin{lstlisting}[style=pythonlang, xleftmargin=0in,label={ch2:lst:saga-coordinator}, caption={Saga coordinator function.}]
def sagas_transfer(context, message: Transfer):
    subtract_credit = SubtractCreditMessage(amount=message.amount)
    add_credit = AddCreditMessage(amount=message.amount)
    context.saga_invocation_pair("account_function",
                                 message.debtor,
                                 subtract_credit,
                                 add_credit)
    context.saga_invocation_pair("account_function",
                                 message.creditor,
                                 add_credit,
                                 subtract_credit)
\end{lstlisting}
\vspace{-5mm}
\end{figure}

\para{Sagas} Similarly to two-phase commit, our API offers the ability to specify Sagas: as seen in \Cref{ch2:lst:saga-coordinator}, the \textit{saga\_invocation\_pair} function in line 6 will receive the target function name, the ID of the debtor as well as two messages: the \textit{subtract\_credit} and its compensating action \textit{add\_credit}. If there is a failure during the execution of \textit{subtract\_credit}, our Sagas transaction coordinator will execute the compensating action  \textit{add\_credit}, which will put back the original credit to the debtor's account. The details on how Sagas are executed are given in \Cref{ch2:sec:implementation}.

\para{Extensions to Regular Functions}
To allow the execution of a transaction by the two types of coordinator functions across any arbitrary function instances, some extensions to regular functions are required. First, functions that can partake in a coordinated transaction need to be able to fail explicitly. After a failure is communicated to a coordinator function, it results in a transaction rollback. Currently, there is no notion of failing an invocation in Flink StateFun; the function invocation may simply perform no side effects. To allow explicit failure, a field containing these details is added to the protocol between StateFun and the remotely deployed functions. From the API perspective, a function failure can be triggered by throwing an exception. The failure of a function can be roughly compared to integrity constraint violations based on the state encapsulated in a function instance in traditional database terms. Second, any batching mechanism needs to be changed. TPC coordinator functions ensure isolated transactions. This means that any function invocation that is part of such a transaction may not be batched between other function invocations. Third, appropriate locking should be implemented on the level of function instances to ensure the isolation of serializable transactions based on two-phase commit coordinator functions. Finally, the function instances should transparently communicate with the coordinator functions so as not to burden developers with this task.

\section{Transactional Workflows} \label{ch2:sec:implementation}

In this section, we present our Python API in more detail, and we present the implementation for transactional workflows across stateful serverless functions on Apache Flink StateFun (T-Statefun). Our implementation consists of coordinator functions that enforce either a distributed serializable transaction with a two-phase commit or a Saga workflow as a transaction without isolation.

\subsection{Coordinator Functions}
Coordinator Functions orchestrate transactional workflows across ordinary Stateful functions. To achieve this, coordinator functions encapsulate the state of active transactional workflows that they are in charge of, but hold no state of the participating function executions or custom user-defined state. A coordinator function can be invoked simply by its name (uniquely identified by a type internally) and an ID generated randomly at initialization time. Then an input message will arrive at the coordinator's input queue. If the coordinator function is involved in an ongoing transaction, the message will be queued until the workflow that is executing completes. The coordinator functions' Python API is listed in \Cref{ch2:tbl:API}. 

\Cref{ch2:fig:transaction-communication-flow} shows the common communication flow between a coordinator function and regular function instances. Specializations of this communication for two-phase commit and Saga workflows are described in \Cref{ch2:sec:impl-sagas-func} and \Cref{ch2:sec:impl-tpc-func} respectively. Messages that are not always sent in both cases are annotated with a *. \Cref{ch2:fig:transaction-communication-flow} shows the enriched internal structure for regular function instances compared to \Cref{ch2:fig:original-communication-flow}. These are the extensions that we implement for regular functions so that they can participate in transactional workflows. 

\newcommand{\tripleitem}{%
  \begingroup
  \stepcounter{enumi}%
  \edef\tmp{(\theenumi, }%
  \stepcounter{enumi}
  \edef\tmp{\tmp\theenumi, }%
  \stepcounter{enumi}
  \edef\tmp{\endgroup\noexpand\item[\tmp\theenumi)]}%
  \tmp}
  
\newcommand{\doubleitem}{%
  \begingroup
  \stepcounter{enumi}%
  \edef\tmp{(\theenumi, }%
  \stepcounter{enumi}
  \edef\tmp{\endgroup\noexpand\item[\tmp\theenumi)]}%
  \tmp}

\subsection{Saga Coordination}
\label{ch2:sec:impl-sagas-func}
The programming model of the Saga coordinator function is shown in \Cref{ch2:lst:saga-coordinator} through an example. \Cref{ch2:tbl:API} presents the API. In Sagas, the developer is responsible for defining pairs of function invocations so that the invocation of the second function compensates for the one of the first function~\cite{sagas}. Additionally, the Saga coordinator function can define side effects (e.g., outgoing egress messages) based on the transaction's completion scenarios (success or failure). The function invocations composing a Saga are executed in parallel in the current implementation. In the following, we describe the messages specifically for Sagas seen in \Cref{ch2:fig:transaction-communication-flow}.

\para{Initialization \& Remote Coordinator Function Call} First, a message is sent to the coordinator function to initialize a transaction (step 1).
The message is taken from the queue to initialize the transaction (step 2).
Then, the remote Saga coordinator function is called with the incoming message (step 3).
The remote function returns the definition of the Saga workflow to its embedded counterpart (step 4). This includes the function invocations involved in the transaction and their compensating invocations, as well as the side effects to perform on success or failure.

\para{Processing the Remote Coordinator Function's Result} When the embedded function processes the result of the remote function (step 5), a random transaction ID is generated, and a map is created holding the addresses of function instances and the result of their execution (at this stage, those are initialized as \textit{null} values). It follows that only one invocation per function instance can be involved in a particular workflow. If multiple invocations of a single function instance are required, this can be solved at the application level by allowing a single message, which combines multiple function invocations, to be sent to the function instance.

\para{Invoking Regular Functions} In step 6, each of the participating regular (non-coordinator) function instances receives a function invocation in its input queue. All the invocations are sent simultaneously, and the function instances can do the work in parallel. These function invocations are distinguishable as function invocations that belong to a Saga workflow.
Each Saga function invocation is fetched from the queue, and it is either directly sent to the remote function or batched with other invocations for efficiency (step 7). Because Sagas do not require isolation, a function invocation can be batched with other invocations.
Then, it is sent to the regular remote function (step 8). After processing it, the function's response is added to the queue of its stateful embedded representation in StateFun (step 9).
When the response of the stateless remote function is processed in the embedded stateful function at step 10, the indices in the in-flight function invocation metadata and the new list added to the Protobuf interface, i.e., the regular function extensions, are used to identify the result status of the Saga function invocations and the corresponding coordinator's addresses. If the function invocation fails, no side effects of the function are performed. After this, this function can continue processing other function invocations.

\para{Saga Success vs Compensation} Based on the success status of the Saga function invocation, a success or failure message is sent to the coordinator function (step 11).
When the embedded coordinator function processes the success status of each function invocation, the map is updated with either a success or failure status (step 12). If a function instance fails, any function instances that successfully executed their function invocation are messaged with their respective compensating actions (step 13), and the side effects in case of a failure are performed (steps 14, 15, 16). The coordinator function has to wait until the result of all function invocations is received before it is done. In case any of the function invocations fails, the coordinator function sends the compensating messages to all function instances that successfully processed their invocation.
Note that there is no need to send compensating invocations to function instances that failed since those function instances have applied no side effects. The compensating messages are processed as regular messages and are only required when any of the function invocations fail. This means that the performance of a Saga workflow will be worse if it is likely to fail, as this will require extra messaging and processing, up to double.
As a matter of fact, this is the trade-off offered by optimistic transaction approaches like Sagas.

\begin{figure}
    \centering
    \captionsetup{justification=centering,margin=0cm}
    \includegraphics[width=\columnwidth]{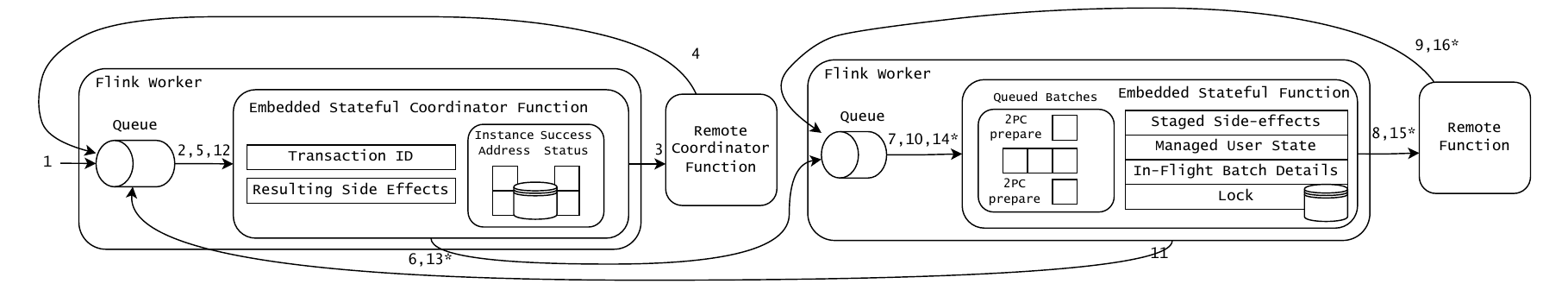}
    \caption{Communication flow for transactions in T-Statefun.}
    \label{ch2:fig:transaction-communication-flow}
\end{figure}

\subsection{Two-phase Commit Coordination}
\label{ch2:sec:impl-tpc-func}
 
In \Cref{ch2:lst:2pc-coordinator}, we presented the programming model for a two-phase commit coordinator function; \Cref{ch2:tbl:API} shows the available functions of the two-phase commit API. Similar to Saga coordinator functions, two-phase commit coordinator functions can also define side effects to execute for any completion scenario. Beyond successful and failed completion, two-phase commit transactions can also be completed as ``retryable''. This occurs when the transaction is aborted due to a deadlock. In the following, we describe the workflow of the two-phase commit as seen in \Cref{ch2:fig:transaction-communication-flow}. Note that the initialization of the workflow, i.e., steps 1-5, is the same as in Sagas. Thus, we do not detail it here. 
 
\para{\textit{PREPARE} \& Two-phase Locking Growing Phase} Each involved function instance is messaged with its respective function invocation in step 6. This message is identifiable as a \textit{PREPARE} message of the two-phase commit protocol.
When a two-phase commit function invocation arrives at the embedded stateful regular function, and a batch of invocations for this function is currently in-flight, this two-phase commit function invocation is not batched with other invocations. Instead, the two-phase commit function invocations split batches and send them to the remote function in isolation, as shown in \Cref{ch2:fig:transaction-communication-flow}. This practice increases the complexity of the batching mechanism, as it now requires a queue of batches rather than an append-only batch as shown in  \Cref{ch2:fig:original-communication-flow}.

\para{Invoking Regular Remote Functions} When the message (and current state) is processed and sent to the remote function in steps 7 and 8, the transaction ID and the address of the two-phase commit coordinator function are stored in the details of the in-flight batch of invocations. The lock on the function instance is also set at this point.
The response from the stateless remote function includes the function invocation status and any side effects (step 9). 
Suppose a \textit{FunctionInvocationException} is thrown at the stateless remote function. In that case, the response of the remote function is discarded, a response to the coordinator function instance is sent to notify it that the invocation failed, and the regular function instance's lock is removed, as it knows the transaction will be aborted. If the function invocation is successful, the lock is kept, and a success response is sent to the coordinator function instance. The state effects are then stored as staged side effects in the function instance (step 10). Any other messages that arrive while the function instance is locked are put in the queued batches.

\para{ABORT \& Two-phase Locking Shrinking Phase Upon Failure} The message at step 11 notifies the two-phase commit coordinator function instance whether the function invocation succeeded.
If the two-phase commit function instance receives the message that a function invocation failed (step 12), it immediately sends an \textit{ABORT} message to all other function instances and performs the appropriate side effects (step 13), and calls the two-phase lock shrinking phase. After this, the two-phase commit function is done. 

\para{COMMIT \& Two-phase Locking Shrinking Phase} If the two-phase commit function instance receives the message that a function invocation was successful, it updates the map it keeps of all involved function instances. If all function instances succeed, it sends \textit{COMMIT} messages to all involved function instances and publishes the appropriate side effects (i.e., applies the changes to the embedded function state).

\para{COMMIT/ABORT \& Two-phase locking Shrinking Phase} When a function instance receives a \textit{COMMIT} message (step 14), it executes its staged side effects, releases the lock and continues processing the next request. When a function instance receives an \textit{ABORT} message, it discards its staged changes, releases the lock, and continues processing. Note that a function could also receive the \textit{ABORT} message while the \textit{PREPARE} message is still in the queue or in-flight. In this case, the \textit{PREPARE} message is discarded.
Messages 15 and 16 are never sent for two-phase commit transactions.

\para{Deadlock Detection}
Due to the use of locks, the two-phase commit protocol is susceptible to deadlocks. A deadlock can happen when two or more different two-phase commit transactions wait on the locks on function instances that are held by other transactions.
To deal with deadlocks, we have implemented a deadlock detection mechanism, which we describe below. 
All participants in the two-phase commit transaction can be partitioned across different machines, and the state of active transactions is encapsulated in different coordinator function instances. Thus, we do not want transactions to rely on any centralized component for handling deadlocks.
We implemented the Chandy-Misra-Haas algorithm \cite{chandy-misra-haas} that provides a simple way to detect deadlocks in a distributed manner, without dependence on a single global coordinator.
Whenever a deadlock is detected in a transaction, it immediately completes as a retryable transaction and sends \textit{abort} messages to all involved function instances. Upon receiving a retryable result status, a two-phase commit regular function may send itself a delayed invocation with the same initial message (and possibly a counter attached) to perform a retry. This is left to the developer so that the system remains flexible across various use cases.

\section{The Transactional Guarantees of T-Statefun}
\label{ch2:sec:deterministic-db}

Our approach offers serializable transactions by virtue of using the two-phase locking protocol. 
Under certain transactional scenarios, which we discuss in this section, our approach can achieve \textit{strict} serializability, where the processing of transactions happens in the same order that the transactions have reached the system. 
In order to achieve strict serializability, our approach would require extensions. In the following, we explain various design decisions or changes that need to occur in our system to support different flavors of serializability.

\para{Single-partition Transactions} A single-threaded operator instance executes every operation on the state of a given partition. Thus,  single-partition transactions are guaranteed to be processed in a serial manner. This also follows that single-partition transactions will be guaranteed strict serializability even when executed in a distributed fashion. Moreover, transactions that operate on different partitions are going to scale horizontally.

\para{Multi-partition Transactions} In the general case, a transaction in our approach can access multiple functions, mutate multiple state partitions, or both. 
Since two-phase locking is used, the system can enforce serializability across multiple functions and data partitions of the same function.
In addition, our approach does not guard against changes in the order of transaction executions. For example, induced by transaction aborts due to a deadlock or system failures, transactions may be resubmitted for execution.

\para{Strict Serializability} Our approach features three core advantages that provide important foundations for achieving strict serializability. First, since we support one-shot transactions, the system is aware of the keys that will be touched from a transaction prior to its execution. Furthermore, these one-shot transactions can be arranged prior to their execution in a specific serial order -- that order can be set to be the order of arrival, thus guaranteeing strict serializability. Second, Apache Flink, which executes our transactions, recovers from a failure by falling back to the latest completed checkpoint and re-processes input requests following the checkpoint. This strategy allows us to reconstruct the exact same state as prior to the failure under the assumption of deterministic computations. 
Finally, data-parallel processing in disjoint state partitions allows us to execute concurrent transactions in a parallel manner and without the need for concurrency control.

\para{Relation to Deterministic Databases}
Interestingly, the three aforementioned characteristics of our approach resemble design choices opted by deterministic databases~\cite{abadi2018overview, abadithecase, Ren2014Evaluation}, which achieve strict serializability: the concurrent processing of a specific set of transactions across a distributed system is guaranteed to result in one, single runtime state.

Furthermore, one could draw inspiration from deterministic databases for advancing its transactional model in two ways. First, transactions on dataflow systems would benefit from an input transaction log for pre-determining the order of transactions in a way that would not introduce aborts during execution, essentially implementing a protocol like Calvin~\cite{abadithecase}. Second, one could leverage a determinism service \cite{silvestre2021clonos} to wrap nondeterministic computations, which would cause its state to diverge when recovering from a failure. Essentially, pre-ordering a batch of transactions and ensuring deterministic transaction processing would help dataflow-based transactional FaaS systems guarantee strict serializability.

\section{Experimental Evaluation} \label{ch2:sec:experimentalSetup}
\label{ch2:sec:experiments}

In this section, we describe in detail our experimental evaluation methodology. For the lack of a benchmark aimed at SFaaS, we opted for an extension of the \textit{Yahoo! Cloud Serving Benchmark} (YCSB)\cite{ycsb} benchmark. Furthermore, we go through the experimental evaluation of our system, which is split into six experiments with the following goals. 

\begin{description}
    \item i) Determine the overhead that function coordination introduced to StateFun (\Cref{ch2:sec:exp1}).
    
    \item ii) Compare between the two transaction protocols with/out rollback operations (\Cref{ch2:sec:exp2}).
    
    \item iii) Evaluate the system's scalability  (\Cref{ch2:sec:exp3}).
    
    \item iv) Perform a microbenchmark with a fixed number of machines and a variable number of keys and proportions of \textit{transfer} operations (\Cref{ch2:sec:exp4}).
    
    \item v) Compare against the CockroachDB with Kafka clients deployment (\Cref{ch2:sec:exp5}).
    
    \item vi) Compare against Beldi (\Cref{ch2:sec:exp6}).
    
\end{description}

\noindent Regarding resources used, for (i, ii, iv, v), we used three 4-CPU StateFun workers/CockroachDB nodes, and for (iii), each worker had 2 CPUs. In (v), we kept the default settings, meaning that CockroachDB replicates data three times for fault tolerance and high availability. For (vi), we allowed AWS and DynamoDB to autoscale while measuring the maximum concurrency reached by AWS Lambda.

\begin{figure}[t]
    \centering
    \captionsetup{justification=centering,margin=0cm}
    \includegraphics[width=0.7\columnwidth]{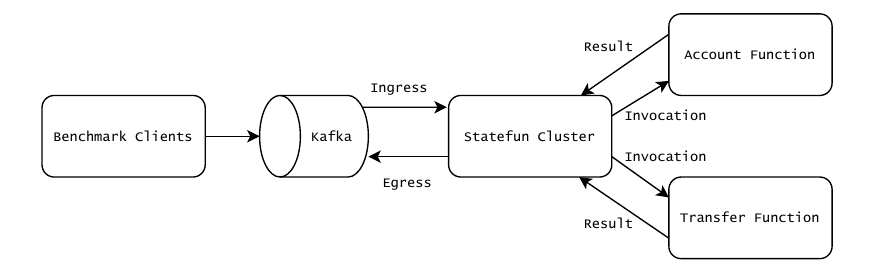}
    \caption{StateFun Benchmark Design}
    \label{ch2:fig:application_architecture}
\end{figure}

\subsection*{Benchmark Workload} \label{ch2:sec:workload}

In YCSB, the first step is to insert records into the system with a unique ID and several task-specific fields. After the data insertion stage, the benchmark performs operations on the initialized state. YCSB defines \textit{read} and \textit{write} operations as part of their core workloads. Because this work's main contribution is distributed transactions across stateful function instances, we added a new operation based on an extension introduced in \cite{dey2014ycsb+}. This operation is called a \textit{transfer}, and it atomically subtracts \textit{balance} from one account and adds this to another, meaning that records also include a numeric \textit{balance} field. These additions mean that the workloads can consist of the following three operations: 

\vspace{1.5mm}
\noindent\textbf{read} Reads the state associated with a single key and outputs it to the egress.

\vspace{1.5mm}
\noindent\textbf{write} Updates a field associated with a key and outputs a \textit{success} message to the egress.

\vspace{1.5mm}
\noindent\textbf{transfer} Requires two keys and a specified amount, subtracts the amount from the balance of one key, and adds it to the other. Depending on the transaction result, the output is either a success or failure message to the egress.
\vspace{1.5mm}

Across experiments, we vary the proportion of each operation in the resulting workloads. In YCSB, the user selects the probability distribution of the operations' record IDs. In this work, we assume uniform key distributions. The added benefit is that the number of requests for a single key can be increased transparently by decreasing the system's total number of records. Finally, YCSB allows variations in the number of fields and the size of the values associated with each field. In this evaluation process, all records have ten fields containing a random string of 128 bits and a single integer field. A StateFun application is implemented with the following two functions to support the operations defined in \Cref{ch2:sec:workload}:

\para{ -- Account Function} This is a regular function containing the record state for each key. It processes messages to read the state, updates the fields, and subtracts or adds balance as part of a transaction. It throws an exception and rolls back the transaction if the key does not exist or if there is an insufficient balance to subtract the transaction amount.

\para{ -- Transfer Function} The transfer function is a transactional/coordinator function that takes a message consisting of two different keys and an amount. That message represents a transaction consisting of two function invocations, one to each function key. This function is implemented with both the two-phase commit and the Saga API.
\vspace{1.5mm}

\Cref{ch2:fig:application_architecture} depicts the architecture of the system under test. The benchmark publishes the workload to a Kafka cluster. StateFun reads from Kafka as ingress, invokes the appropriate functions, and then publishes the result to a Kafka topic as an egress. For CockroachDB (v21.1.7), Kafka clients read from the relevant topics and submit queries to the non-geo-replicated database. 

Although CockroachDB and Kafka can provide exactly-once semantics individually, because the state (CockroachDB) and messaging (Kafka clients) are not managed by a single entity and do not share a single checkpointing mechanism, this deployment offers at-least-once semantics. More specifically, clients that consume Kafka queues that deliver transaction-initiating events need to pull an event from a Kafka topic, submit a query to CockroachDB, and acknowledge the transaction's execution back to Kafka. However, in the event of a client (or database) failure, the transaction may be executed, but the message to the queue may never be acknowledged. Not having returned the acknowledgment to Kafka, the client will re-execute the same transaction after recovery. In general, unless the transactions come with application-specific idempotence keys, the system by itself cannot enforce exactly-once processing guarantees, falling back to at-least-once guarantees.

Our StateFun-based implementation and the CockroachDB deployment are deployed on SurfSara\footnote{\url{https://userinfo.surfsara.nl/systems/hpc-cloud}}, an HPC cloud with instances with up to 80-vCPUs. For our experiments, we used a two-VM Kubernetes cluster to simplify the deployment and management of the system's separate components with enough vCPUs to support the system's configuration under test. Beldi was deployed on AWS. All components shown in \Cref{ch2:fig:application_architecture} can be horizontally scaled as necessary. Additionally, we give the Kafka cluster and the clients enough resources to ensure that they can handle the load:  when a bottleneck appears, it can be attributed to the system performing the application logic, i.e., the StateFun cluster, CockroachDB, or Beldi's API.

 \begin{figure}[t]
    \centering
    \captionsetup{justification=centering,margin=0cm}
    \includegraphics[width=0.48\columnwidth]{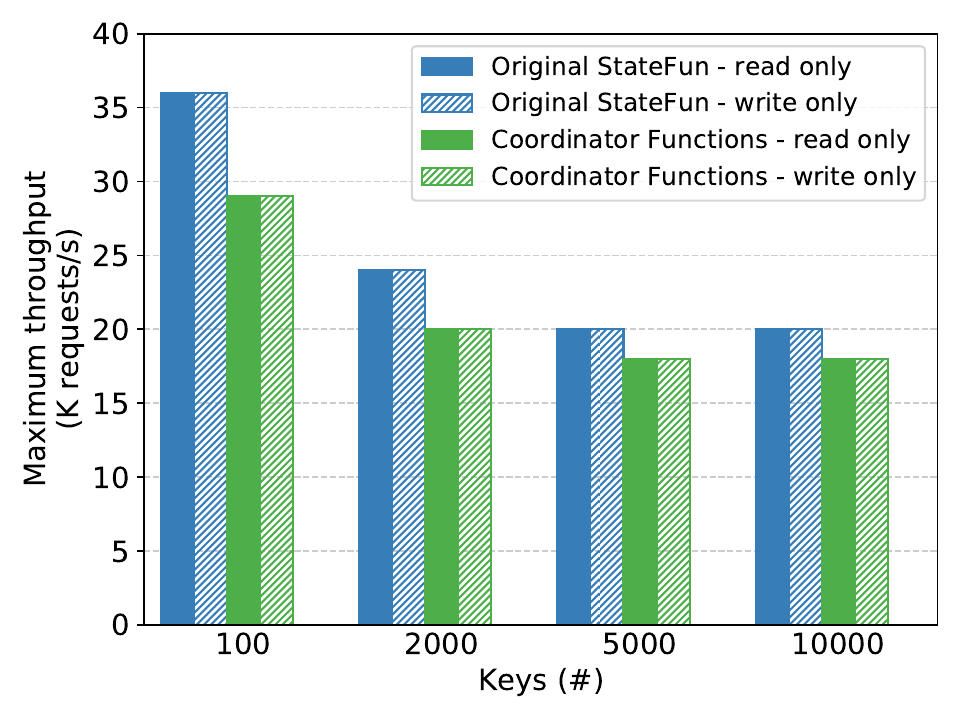}
    \caption{Maximum throughput for the original StateFun vs. StateFun with coordinator functions.}
    \label{ch2:fig:throughput-overhead}
\end{figure}

\begin{figure*}[t]
     \centering
     \begin{subfigure}[b]{0.44\textwidth}
         \centering
         \captionsetup{justification=centering,margin=0cm}
         \includegraphics[width=\textwidth]{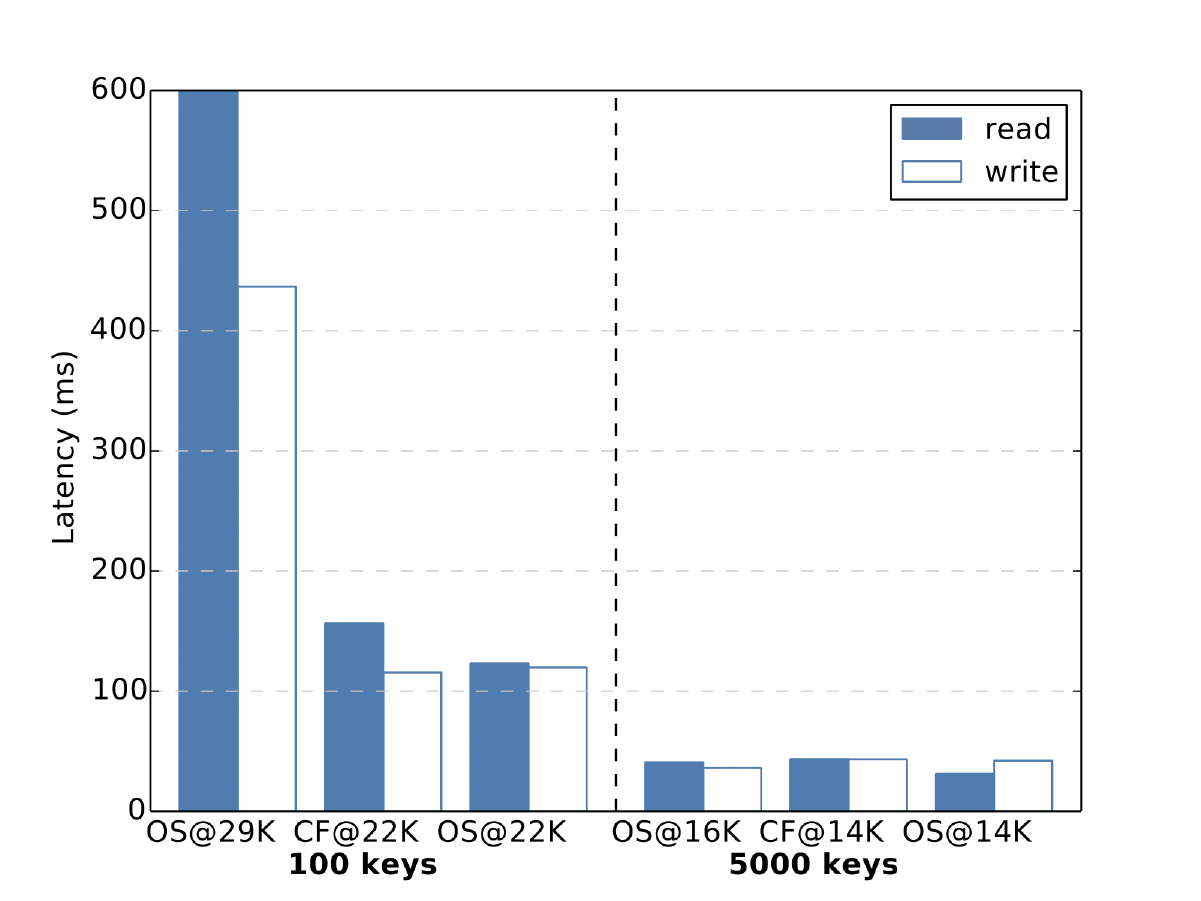}
         \caption{Mean}
         \label{ch2:fig:exp1-mean}
     \end{subfigure}
     \hfill
     \begin{subfigure}[b]{0.44\textwidth}
         \centering
         \captionsetup{justification=centering,margin=0cm}
         \includegraphics[width=\textwidth]{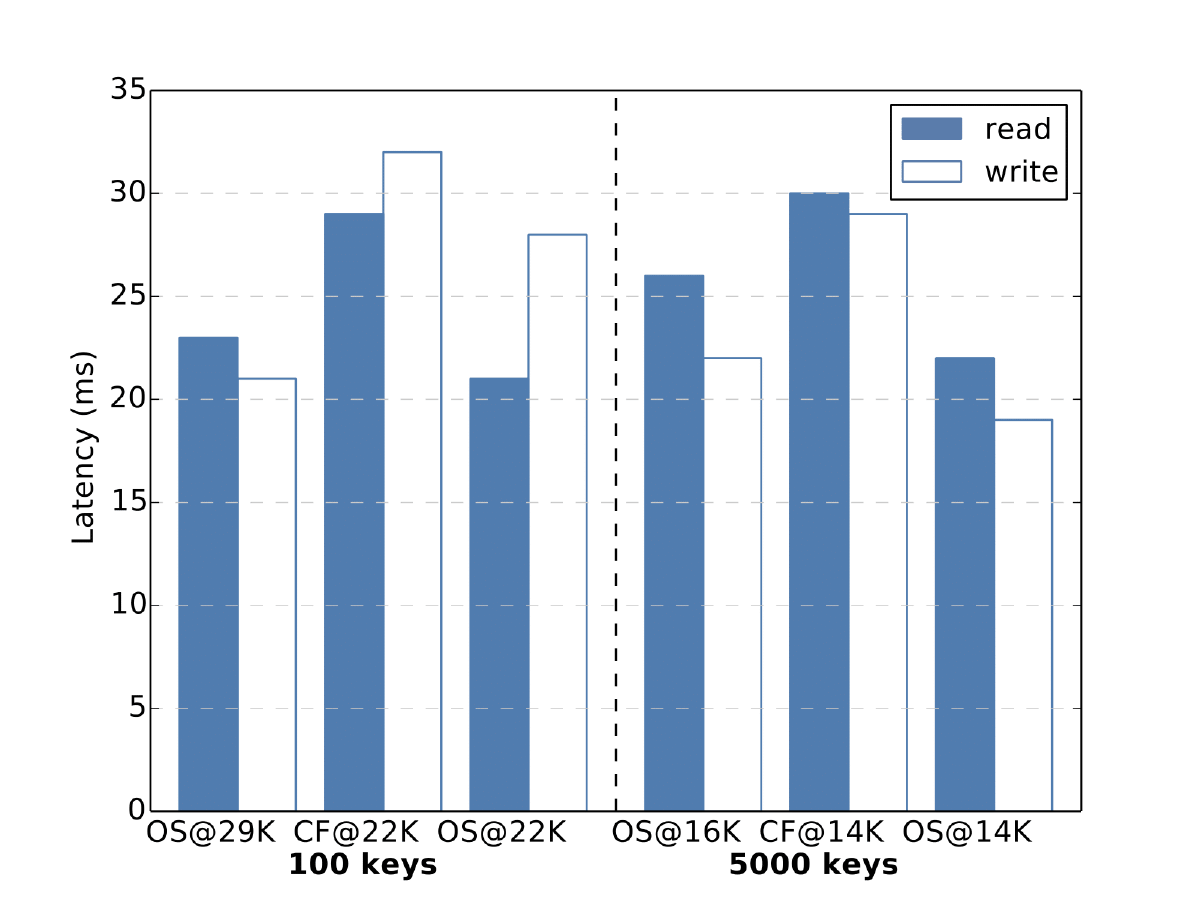}
         \caption{Median}
         \label{ch2:fig:exp1-median}
     \end{subfigure}
     \hfill
     \begin{subfigure}[b]{0.44\textwidth}
         \centering
         \captionsetup{justification=centering,margin=0cm}
         \includegraphics[width=\textwidth]{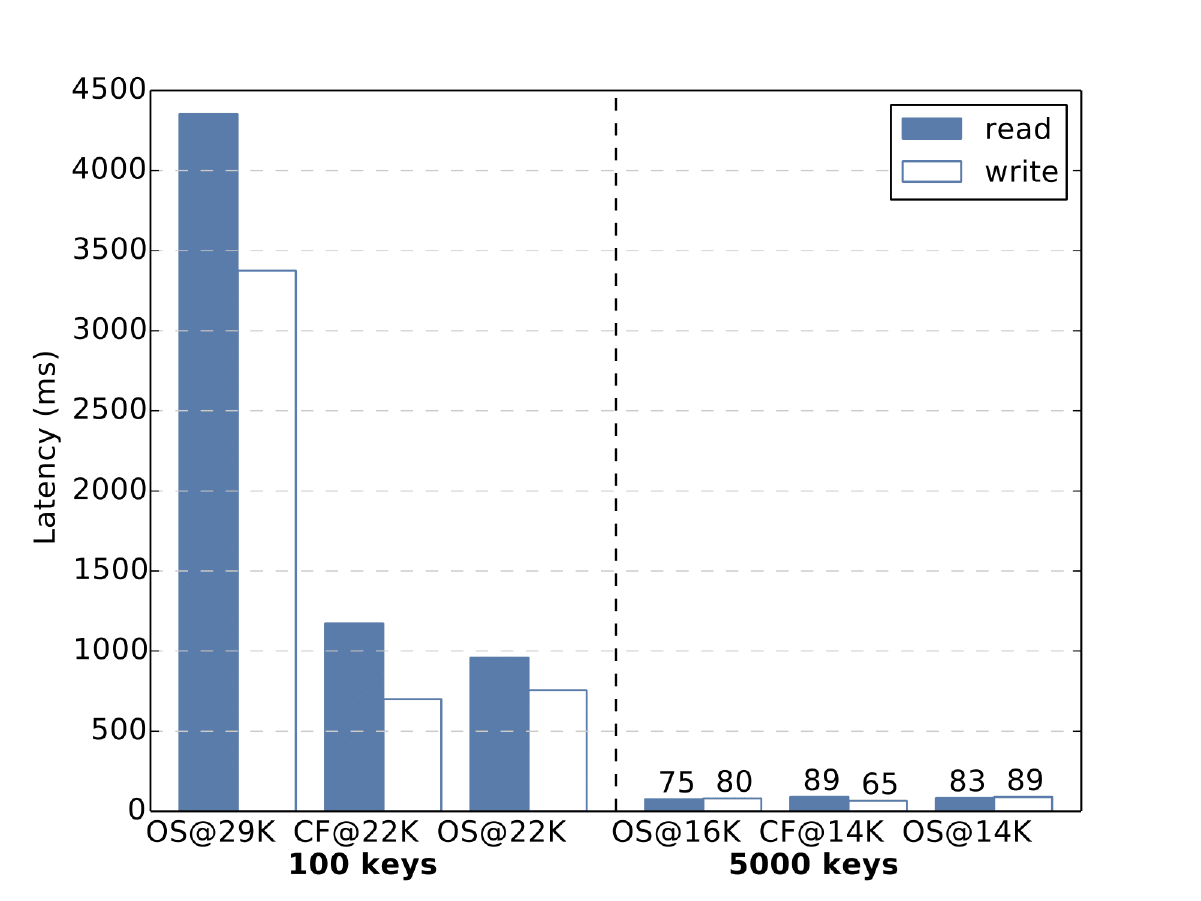}
         \caption{95th percentile}
         \label{ch2:fig:exp1-95}
     \end{subfigure}
        \caption{Graphs comparing latencies of original StateFun (OS) and StateFun with coordinator functions (CF) at different throughputs for read-only and write-only workloads.}
        \label{ch2:fig:exp1-latency}
\end{figure*}

\begin{figure*}[t]
     \centering
     \includegraphics[width=\textwidth]{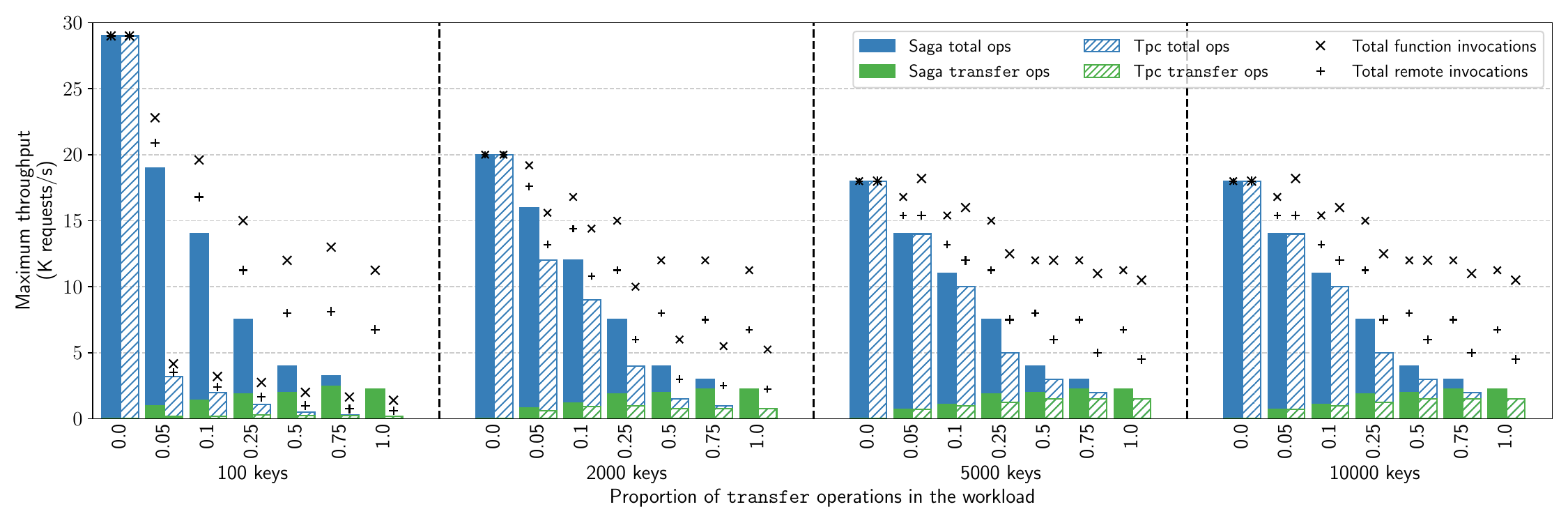}
     \vspace{-7mm}
     \caption{Maximum throughput for workloads with increasing proportions of \textit{transfer} operations in the workload}
     \label{ch2:fig:exp2-throughput}
\end{figure*}

\subsection*{Evaluation metrics}

We evaluate the systems based on two metrics. First, the throughput is either at max or stable (80\%), showing the number of workload operations the system can handle per second, and the latency, showing the time it takes to process an operation.

The maximum throughput of each workload and system configuration is found by steadily increasing the input throughput created by the benchmark clients in Kafka until the StateFun cluster/CockroachDB can no longer consistently handle the load, as measured by the system's output throughput in Kafka. At some point, the output throughput starts fluctuating, and we define this value as the maximum throughput for the configuration. In the comparison with Beldi, we could not measure it this way since it will always rescale to accommodate the new load. So our approach in this matter is to take the 80\% throughput of the StateFun configuration and run Beldi with the same input throughput.

We use the Kafka event time for the ingress and egress events of operations to measure their end-to-end latency. Since latency is always dependent on the throughput, in our experiments, we set the throughput to 80\% of the maximum throughput to allow consistent operation of the system under test and measure the latency accurately. When comparing latencies, the different throughput rates at which the latency is measured should be considered.

\subsection{Coordination Overhead}\label{ch2:sec:exp1}

In the first experiment, the performance of StateFun with coordinator functions is compared against the original on non-transactional workloads to see how much computational overhead the coordination logic has added. In \Cref{ch2:fig:throughput-overhead}, we show the maximum throughput achieved by the two systems for a varying number of keys. While in \Cref{ch2:fig:exp1-latency}, we show the different latencies for the systems across \textit{read} and \textit{write} workloads at different throughputs and numbers of keys. 

\para{Throughput} The first observation we can make is that there is a 20\% decrease in throughput in the case of 100 keys that plateaus to 10\% as the number of keys increases. The decreased performance is because the batching mechanism is more complex than the original append-only approach by enforcing isolated function invocations as part of a two-phase commit transaction. In addition, the coordinator functions keep track of transaction progress, which incurs some overhead. Another observation is that there is no noticeable throughput difference between workloads with only \textit{read} or \textit{write} operations. The reason behind this behavior is that, in StateFun, both operations need to access the remote function, making the communication layer the bottleneck. 

\para{Latency} The latencies in \Cref{ch2:fig:exp1-latency} are approximately 20\% higher for our version of StateFun for \textit{read} workloads. However, as the number of keys increases, the difference becomes smaller, towards 7\%. This decrease in performance is due to the additional logic required for function coordination. Another interesting observation is the indifference in performance for \textit{write} workloads. The reason is that StateFun batches every read operation before serialization, adding up over time for larger batches. In contrast, only the last version needs to be serialized for writes. Additionally, serialization happens at the remote function for both types of operations, explaining why it does not affect throughput, but it does affect latency. Finally, we consider the introduced overhead as a reasonably low price to pay for having full-fledged transaction execution primitives added to the system.

 \begin{figure*}[t]
     \centering
     \resizebox{\textwidth}{!}{
     \begin{subfigure}[b]{0.3\textwidth}
         \centering
         \captionsetup{justification=centering,margin=0cm}
         \includegraphics[width=\textwidth]{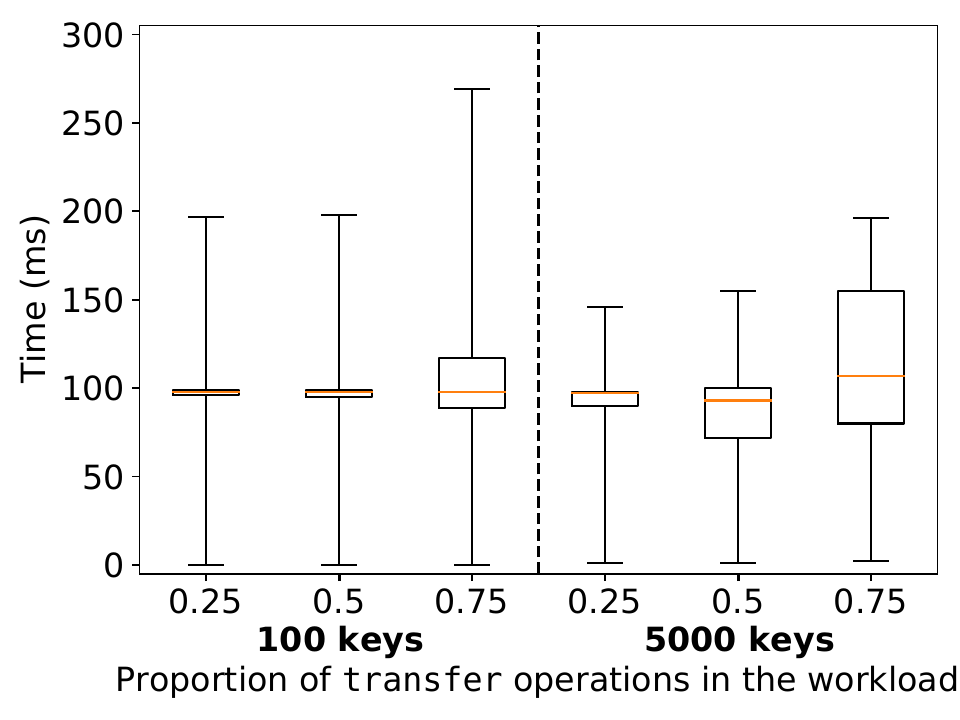}
         \caption{Time regular locks are held}
         \label{ch2:fig:time-regular-locks-are-held-transfer}
     \end{subfigure}
     \hfill
     \begin{subfigure}[b]{0.3\textwidth}
         \centering
         \captionsetup{justification=centering,margin=0cm}
         \includegraphics[width=\textwidth]{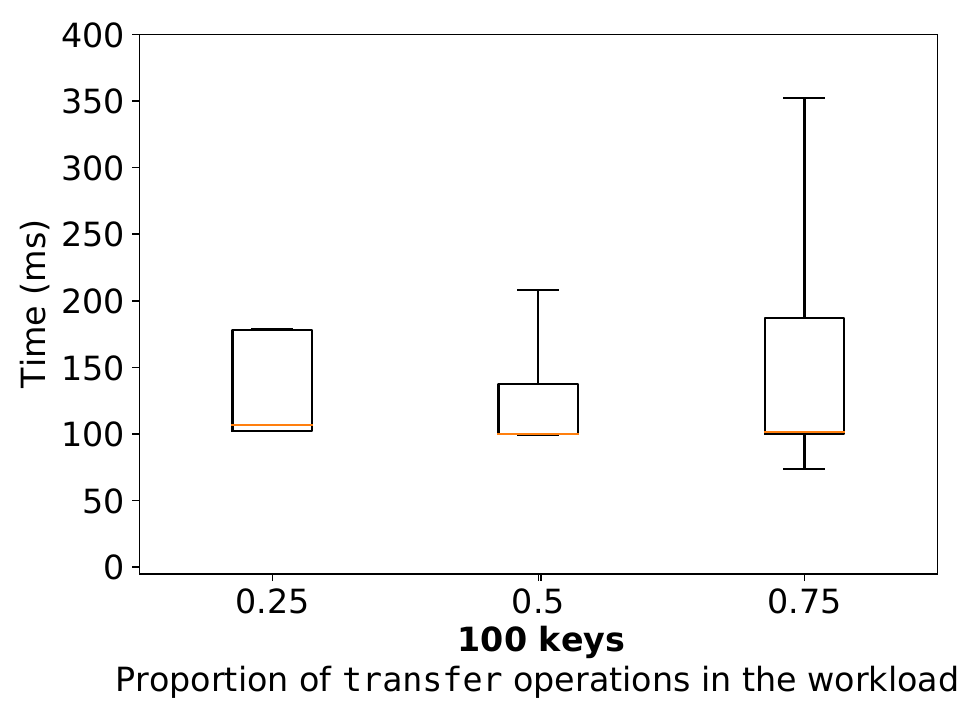}
         \caption{Time to detect deadlocks}
         \label{ch2:fig:deadlock-detect-time-transfer}
     \end{subfigure}
     \hfill
     \begin{subfigure}[b]{0.3\textwidth}
        \centering
        \captionsetup{justification=centering,margin=0cm}
        \resizebox{\textwidth}{!}{
        \begin{tabular}{|l|l|r|}
        \toprule
            \textbf{Keys} & \textbf{Transfer} & \textbf{Deadlocks /} \\
                          & \textbf{proportion} & \textbf{ transfer ops} \\
            \hline
            100 & 0.25 & 9/12014 (0.07\%) \\
                & 0.5 & 27/24107 (0.11\%) \\
                & 0.75& 82/35875 (0.22\%) \\
            \hline
            5000& 0.25 & 0/60121 \\
                & 0.5 & 0/120089 \\
                & 0.75& 0/179794 \\
        \bottomrule
        \end{tabular}}
        \\
        \caption{Frequency of deadlocks}
        \label{ch2:fig:deadlock-frequency-transfer}
     \end{subfigure}
     }
        \caption{Details of locking behavior for a workload for 100 and 5000 keys with various proportions of transfers without rollbacks. The boxplots show the 5th and 95th percentiles.}
        \label{ch2:fig:latency-lock-transfer}
\end{figure*}

 \begin{figure}[t]
    \centering
    \includegraphics[width=0.48\columnwidth]{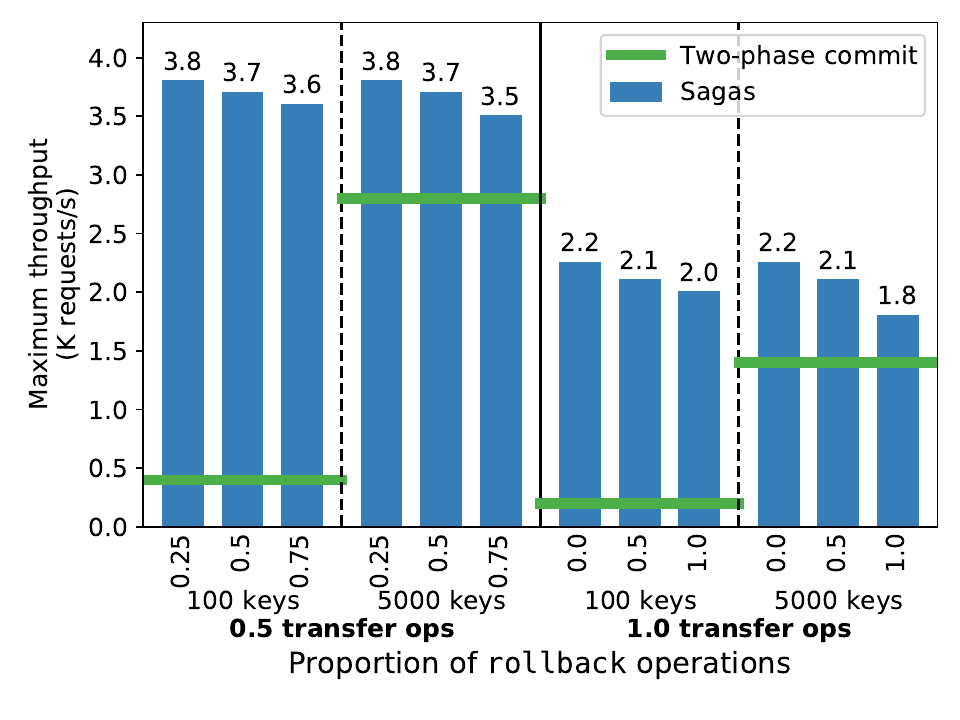}
    \caption{Throughput with different proportions of rolled back \textit{transfer} operations for workloads with 50\% and 100\% \textit{transfer} operations}
    \label{ch2:fig:rollback-throughput}
\end{figure}

\subsection{Sagas vs Two-Phase Commit} \label{ch2:sec:exp2}

The second experiment shows a performance comparison between the two implemented transaction protocols, their impact on the maximum throughput in perfect conditions (\Cref{ch2:fig:exp2-throughput}), and with failures, measuring the impact of locking for the two-phase commit (\Cref{ch2:fig:latency-lock-transfer}) and of rollbacks (\Cref{ch2:fig:rollback-throughput}) for the Saga protocols. In these experiments, we set a certain proportion of the workload to be \textit{transfer} operations, and the remaining proportion is equally shared between \textit{read} and \textit{write} operations. In our case, each \textit{transfer} operation causes three remote function invocations (coordinator function and one function per account holder taking part in the transfer). When evaluating two-phase commit functions, we do not include messages sent to detect deadlocks in the total number of invocations. Therefore, the indicator should be considered a lower bound on the actual number of messages. Finally, we used a uniform key access distribution for these experiments. At the same time, in some real-world scenarios, this can be skewed (e.g., lots of transactions on very active accounts vs. a long tail of inactive ones).

 \Cref{ch2:fig:exp2-throughput} plots the achieved throughput against the absolute number of \textit{transfer} operations in the workload with a varying number of keys, given that the benchmark provided the accounts enough balance to ensure all transactions succeeded. It also displays indicators for the absolute amount of total internal function invocations, considering additional internal invocations required for transactions, and the absolute amount of total remote function invocations. We observe that Sagas perform much better than two-phase commit for a few keys (100 and 2000). This happens for two reasons: i) Sagas can still benefit from the batching mechanism of StateFun since they do not require isolation, and ii) the locking in two-phase commit severely limits the throughput. However, it is also interesting that two-phase commit performs comparably to Sagas for a higher number of keys (5000-10000) even though it provides much stronger guarantees. This is because there is less contention on a single function, decreasing the effect of locking, while batching provides no benefits, as also shown in \Cref{ch2:fig:throughput-overhead}. A second observation from  \Cref{ch2:fig:exp2-throughput} is that the total function invocations still drop when the proportion of transactions increases. This is because the total function invocations account for the additional messaging required to coordinate transactions, leading to the overall throughput of workloads with a high proportion of \textit{transfer} operations being relatively low.
 
\begin{figure*}[t]
    \centering
    \includegraphics[width=\textwidth]{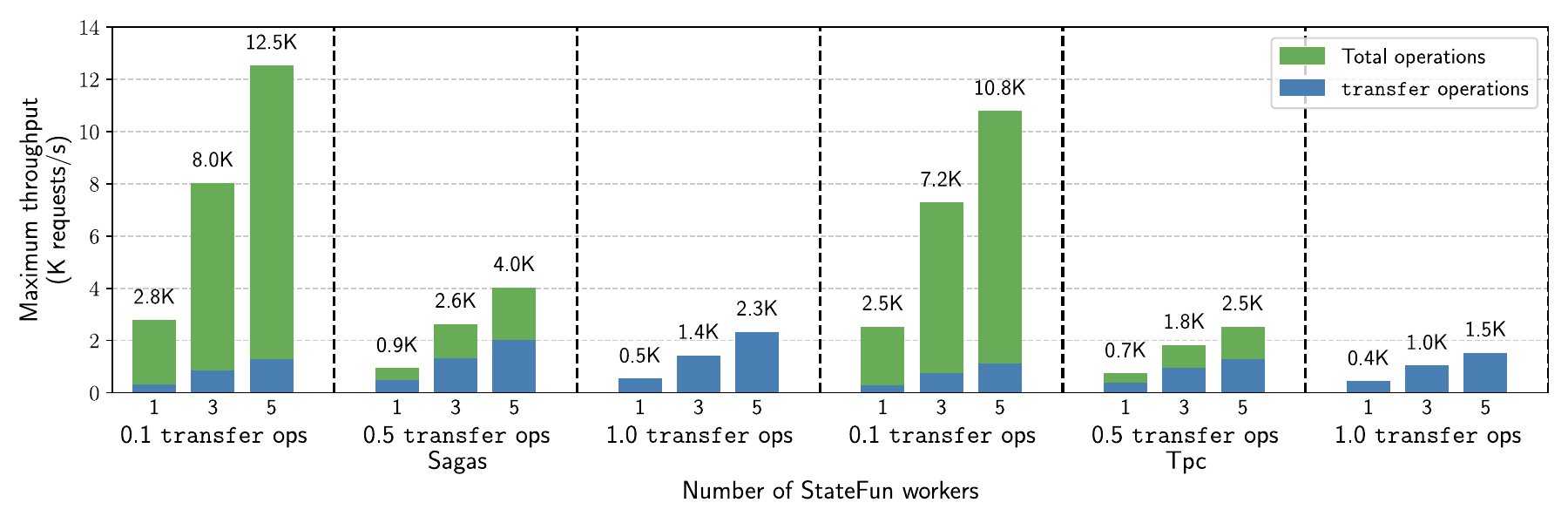}
    \caption{Maximum throughput for the system with 5000 keys for different numbers of StateFun workers for workloads with different proportions of \textit{transfer} operations}
    \label{ch2:fig:exp3-scalability}
\end{figure*}

\begin{figure*}[t]
     \centering
     \begin{subfigure}[b]{0.44\textwidth}
         \centering
         \captionsetup{justification=centering,margin=0cm}
         \includegraphics[width=\textwidth]{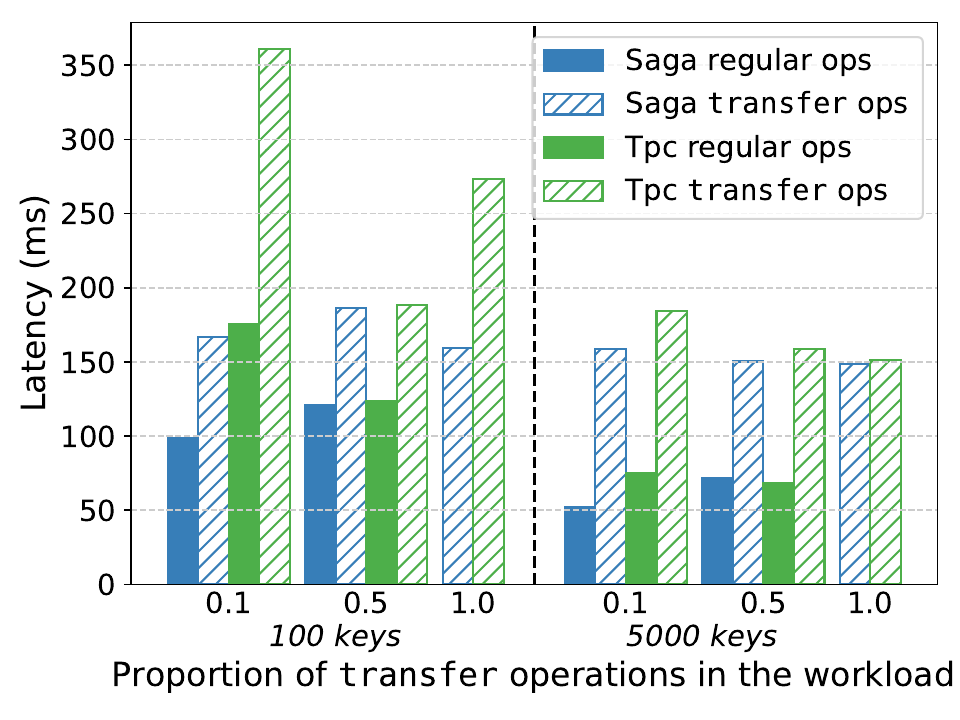}
         \caption{Mean}
         \label{ch2:fig:exp4-mean}
     \end{subfigure}
     \hfill
     \begin{subfigure}[b]{0.44\textwidth}
         \centering
         \captionsetup{justification=centering,margin=0cm}
         \includegraphics[width=\textwidth]{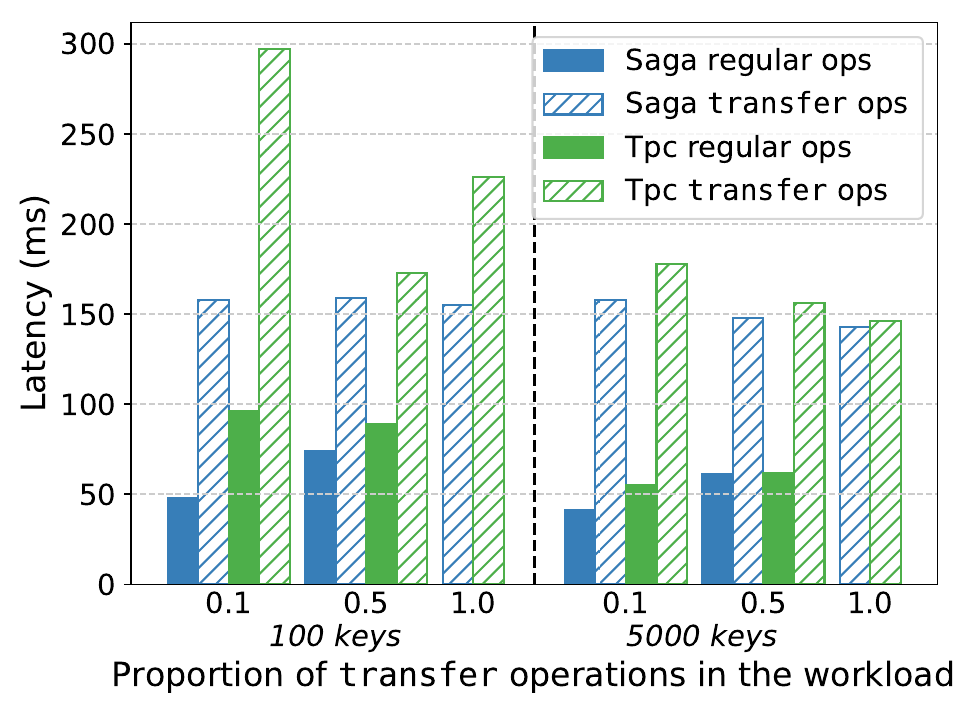}
         \caption{Median}
         \label{ch2:fig:exp4-median}
     \end{subfigure}
     \hfill
     \begin{subfigure}[b]{0.44\textwidth}
         \centering
         \captionsetup{justification=centering,margin=0cm}
         \includegraphics[width=\textwidth]{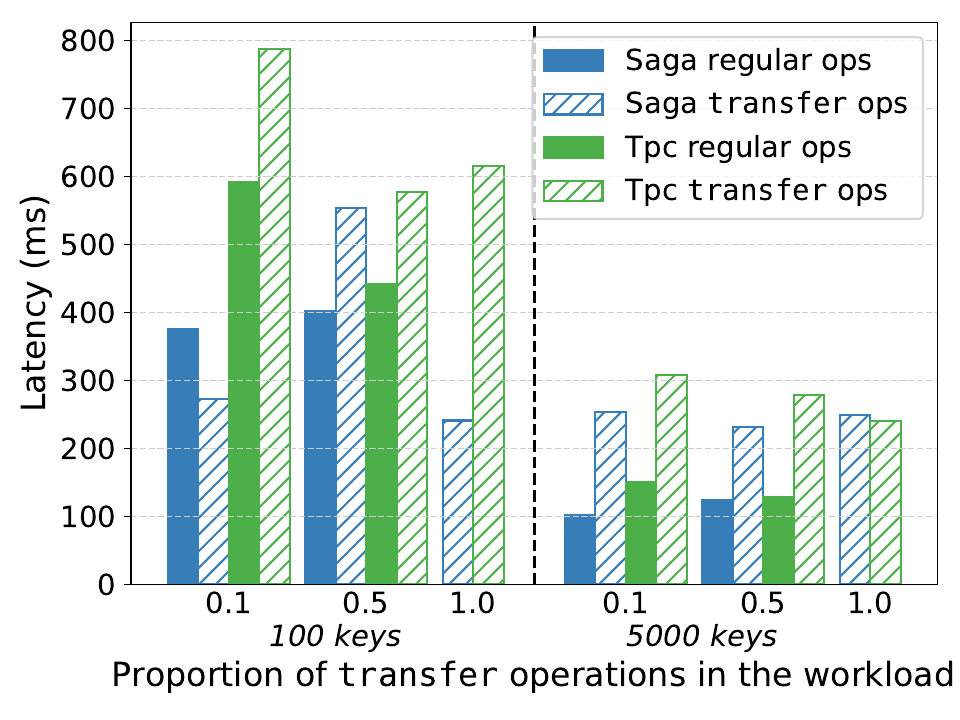}
         \caption{95th percentile}
         \label{ch2:fig:exp4-95}
     \end{subfigure}
     \par\bigskip
     \begin{subfigure}[b]{0.6\textwidth}
        \centering
        \captionsetup{justification=centering,margin=0cm}
        \resizebox{\columnwidth}{!}{%
        \begin{tabular}{|c c c | c c c | c c c | c c c|}
        \toprule
            \multicolumn{6}{|c|}{\textbf{100 keys}} & \multicolumn{6}{c|}{\textbf{5000 keys}} \\
            \hline
            \multicolumn{3}{|c|}{\textbf{Sagas}} & \multicolumn{3}{c|}{\textbf{Tpc}} &
            \multicolumn{3}{c|}{\textbf{Sagas}} &
            \multicolumn{3}{c|}{\textbf{Tpc}} \\
            0.1 & 0.5 & 1.0 & 0.1 & 0.5 & 1.0 & 0.1 & 0.5 & 1.0 & 0.1 & 0.5 & 1.0 \\
            \hline
             11K & 3K & 2K & 1.5K & 0.38K & 0.16K & 9K & 3K & 2K & 8K & 2K & 1.2K  \\
        \bottomrule
        \end{tabular}
        }
        \caption{Throughputs at which latency was measured}
        \label{ch2:fig:exp4-latency-throughputs}
     \end{subfigure}
        \caption{Graph comparing latencies for Sagas and two-phase commit coordinator function for different keys and transaction proportions in the workload at 80\% of the respective maximum throughputs}
        \label{ch2:fig:exp4}
\end{figure*}

\para{Locking Overhead} In \Cref{ch2:fig:latency-lock-transfer}, we measure the behavior of locking and deadlocks that accompany the two-phase commit protocol. The lock duration is measured between the point in time where the function instance sends the response to the \textit{prepare} message and when it either receives a \textit{commit} or \textit{abort} message,  sending the next batch to the remote function. In \Cref{ch2:fig:time-regular-locks-are-held-transfer}, we see little to no difference in the median across the different workloads, but when the proportion of \textit{transfer} operations is higher, the higher percentiles increase significantly. Next, we want to measure the deadlock frequency, and \Cref{ch2:fig:deadlock-frequency-transfer} shows the number of deadlocks against the total number of \textit{transfer} operations in the workload. As expected, there are no deadlocks in workloads with 5000 keys, since contention is low. For 100 keys, we observe an increasing number of deadlocks while the proportion of \textit{transfer} operations increases. However, the percentage of deadlocks across all \textit{transfer} operations is still small. Finally, \Cref{ch2:fig:deadlock-detect-time-transfer} shows the time it takes to detect a deadlock, i.e., perform the Chandy-Misra-Haas algorithm. We observe that the median of the time this takes is similar across all workloads, and it also shows that as the amount of \textit{transfer} operations increases, so do the higher percentile times.

\para{Rollback Overhead} \Cref{ch2:fig:rollback-throughput} shows the maximum throughput for workloads with 50\% and 100\% \textit{transfer} operations where different proportions of \textit{transfer} operations fail for Sagas and two-phase commit coordinator functions. As expected, when using two-phase commit, a rollback does not increase the load in the system because the coordinator function needs to send a second message either way. Again, nothing out of the ordinary happened as the proportion of \textit{transfer} operations to be rolled back increased. The throughput decreased as the protocol required additional compensating messages to be sent in the system. However, with 5000 keys, the difference is small at 50\% \textit{transfer} operations: 8\% when going from 25 to 75\% rollbacks and increasing to 18\% with 100\% \textit{transfer} operations. This is larger than the 100-key case that can still leverage the batching mechanism of StateFun and limit the performance drop to 10\% in the worst case. Still, no matter the decrease in performance due to the compensating actions of the Saga protocol, it remains 20\% faster than two-phase commit in the worst-case scenario of 5000 keys and 75\% rollbacks.

\subsection{Scalability Comparison} \label{ch2:sec:exp3}

In the last experiment, we evaluated the scalability of the proposed system with coordinator functions. In \Cref{ch2:fig:exp3-scalability} we display the maximum throughput for both two-phase commit and Sagas at different amounts of StateFun workers and different transaction proportions in the workload. For Sagas, the scalability from 1 to 5 workers is close to 90\% throughout for all workloads. For two-phase commit, the scalability from 1 to 5 workers starts at 87\% at 10\% \textit{transfer} operations and drops to 75\% at 100\% \textit{transfer} operations. 

The reason for the low decrease in scalability on both protocols is that as the number of workers increases, more traffic needs to go over the network. In the Sagas' case, the efficiency does not decrease across all workloads for the same reasons as expressed in \Cref{ch2:sec:exp2}. Namely, the system can still utilize batching, no locking is required, and the number of messages is two times lower than the two-phase commit protocol when all transactions succeed. On the other hand, the 8\% decrease in scalability in two-phase commit from 10\% to 100\% \textit{transfer} operations is due to the protocol's requirements for locks, more messages, and the inability to use batching. Considering all the impending factors, it still achieves decent efficiency with strong consistency guarantees in fully transactional workloads. 

\subsection{Microbenchmark} \label{ch2:sec:exp4}

As a final experiment, we conduct a microbenchmark on the system. At first, we keep the number of resources fixed, and then for every \textit{transfer} proportion and number of keys, we measure the throughput at 80\% load and the corresponding latency. By the results presented in \Cref{ch2:fig:exp4} we can see that for a use case with a low number of keys, the Sagas beat by a large margin the two-phase commit protocol in both throughput, with more than a 650\% increase in performance, and latency that is at least two times lower. The contention becomes less of a problem for a larger number of keys. We observe a smaller difference between the two protocols at around 40\% on average for throughput and a stable difference in latency around 20\%. To conclude, Sagas seems to be the obvious choice for a few keys or high contention, if the business logic permits it. In any other case, the choice is mainly about the consistency guarantee requirements since the difference is not that significant.

\begin{figure}[t]
    \centering
    \includegraphics[width=0.48\columnwidth]{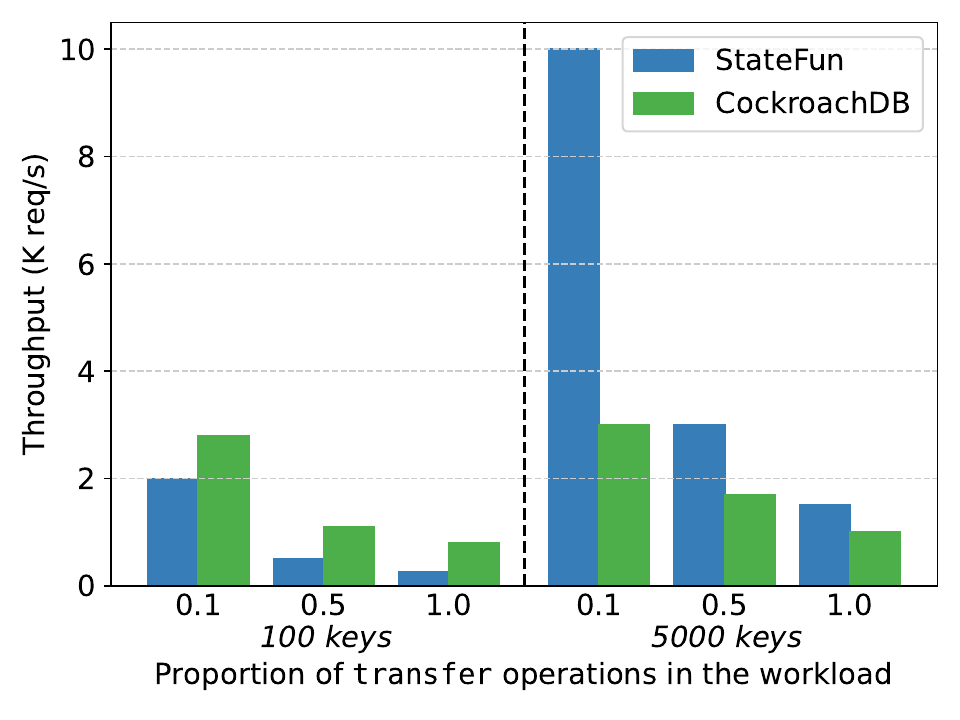}
    \caption{Comparing the maximum throughput of CockroachDB and Flink StateFun for workloads with different proportions of \textit{transfer} operations (the remaining operations are \textit{read} or \textit{update} operations with equal probability). Both systems are deployed with 3 instances, each with 4 CPUs.}
    \label{ch2:fig:cockroach-throughput}
\end{figure}

\begin{table}[t]
    \centering
    \resizebox{0.48\columnwidth}{!}{
    \centering
    \begin{tabular}{|l c c c c |}
    \toprule
        \multicolumn{5}{|c|}{\textbf{100 keys}} \\
        \hline
        \textbf{Operations} & \multicolumn{2}{c}{\textbf{StateFun}} & \multicolumn{2}{c|}{\textbf{CockroachDB}} \\
        & median & 95th \%tile & median & 95th \%tile \\
        \hline
        \textit{Transfer} (0.1) & 297 & 787 & 48 & 86 \\
        \textit{Non-transfer} & 96 & 591 & 47 & 79 \\
        \hline
        \textit{Transfer} (0.5) & 173 & 577 & 22 & 72 \\
        \textit{Non-transfer} & 89 & 441 & 17 & 68 \\
        \hline
        \textit{Transfer} (1.0) & 226 & 615 & 41 & 114 \\
        \bottomrule
        \multicolumn{5}{|c|}{\textbf{5000 keys}} \\
        \hline
        \textbf{Operations} & \multicolumn{2}{c}{\textbf{StateFun}} & \multicolumn{2}{c|}{\textbf{CockroachDB}} \\
                & median & 95th \%tile & median & 95th \%tile \\
        \hline
        \textit{Transfer} (0.1) & 178 & 308 & 36 & 113 \\
        \textit{Non-transfer} & 55 & 150 & 21 & 107 \\
        \hline
        \textit{Transfer} (0.5) & 156 & 278 & 21 & 64 \\
        \textit{Non-transfer} & 62 & 129 & 11 & 62 \\
        \hline
        \textit{Transfer} (1.0) & 146 & 240 & 43 & 78 \\
    \bottomrule
    \end{tabular}
    }
    \caption{Latency compared for StateFun and CockroachDB, each system was run at 80\% of the maximum throughput measured as shown in figure \ref{ch2:fig:cockroach-throughput}}
    \label{ch2:tbl:cockroach-latency}
\end{table}

\subsection{Comparison Against the State-of-the-Art} \label{ch2:sec:exp5} \label{ch2:sec:exp6}

\para{CockroachDB} We compare the performance of StateFun against a production-grade distributed database, CockroachDB, in terms of throughput and latency. Due to the fundamental differences between the two systems, this is merely a reference comparison. In this experimental setting, the input requests consist of a varying proportion of transactional and non-transactional requests. We signify transactional requests as transfer operations and non-transactional requests as non-transfer operations.

As \Cref{ch2:fig:cockroach-throughput} shows, CockroachDB outperforms StateFun in terms of throughput by a constant factor when transactions are evoked on a small number of unique keys. In addition, this experiment configuration examines the performance of the two systems when there is high lock contention since subsequent transactions on the same key have to wait for previous ones to complete.
Notably, the performance difference in terms of throughput remains the same while the proportion of transactions in the input request set increases from 0.1 to 0.5 to 1.
When there are many keys (e.g., 5000), StateFun outperforms CockroachDB. In fact, for a small proportion of transactions (0.1), StateFun achieves four times more throughput. As the number of transactions grows, the performance difference shrinks. These results can be explained by a more sophisticated or aggressive batching mechanism that enables StateFun to efficiently batch non-transactional requests. When there are many non-transactional requests, the effect of batching provides a significant performance advantage, which is shrinking as the number of non-transactional requests becomes smaller.

On the other hand, CockroachDB is superior in terms of latency performance as \Cref{ch2:tbl:cockroach-latency} depicts. Both median latency and latency at the 95th percentile are roughly six times better on average than StateFun's in all configurations. This result can be explained because CockroachDB is run with default settings, and there is no batching implemented at the application level. This contributes to lower throughput, but it also favors lower latency.
On the other hand, StateFun can inherently apply batching at several points in the system, such as when i) it sends a request to a remote function, ii) fetches requests from Kafka, and iii) produces responses to Kafka.

In summary, CockroachDB seems more suitable for handling skewed transactional workloads, although the performance improvement against StateFun is constant in terms of throughput. Thus, a potential superiority based on the locking mechanism of CockroachDB is capped and does not result in a scalable advantage. Furthermore, CockroachDB replicates data three times, leading to additional overhead but providing the capacity to serve requests even in the case of node failures. On the other hand, StateFun provides no replication and needs to recover from a checkpoint following a node failure. On the other hand, StateFun can leverage its sophisticated batching mechanism to drive significantly better throughput for workloads containing a modest number of transactions. Notably, while CockroachDB supports full transactional SQL and StateFun supports only one-shot functions, due to the simplicity of the workload, the feature set should not have a significant impact on performance. In addition, the executed workloads allow for less locking and more batching. Finally, CockroachDB demonstrates reliably low latency in all configurations, roughly six times lower than StateFun.

 \begin{table*}[t]
    \centering
    \resizebox{0.8\textwidth}{!}{
    \centering
    \begin{tabular}{|l c | c c c | c c c |}
    \toprule
        \multicolumn{8}{|c|}{\textbf{100 keys}} \\
        \hline
        \textbf{Operations} & \textbf{Throughput} & \multicolumn{3}{c|}{\textbf{StateFun}} & \multicolumn{3}{c|}{\textbf{Beldi}} \\
        & & CPU & median & 95th \%tile & Max. concurrency & median & 95th \%tile \\
        \hline
        \textit{Transfer} (0.1) & 1.5K & 80 & 298 & 723 & 128 & 49 & 739 \\
        \textit{Non-transfer} & &  & 99 & 572 & & 83 & 693 \\
        \hline
        \textit{Transfer} (1.0) & 0.16K & 80 & 223 & 532 & 182 & 91 & 724 \\
        \bottomrule
        \multicolumn{8}{|c|}{\textbf{5000 keys}} \\
        \hline
        \textbf{Operations} & \textbf{Throughput} & \multicolumn{3}{c|}{\textbf{StateFun}} & \multicolumn{3}{c|}{\textbf{Beldi}} \\
        & & CPU & median & 95th \%tile & Max. concurrency & median & 95th \%tile \\
        \hline
        \textit{Transfer} (0.1) & 8K & 80 & 184 & 287 & 1000* & 123 & 174 \\
        \textit{Non-transfer} & & & 52 & 184 & & 50 & 77 \\
        \hline
        \textit{Transfer} (1.0) & 1.2K & 80 & 146 & 273 & 902 & 114 & 847 \\
    \bottomrule
    \end{tabular}
    }
    \caption{Comparison between latencies of Beldi and StateFun (*experiment is throttled and runs at a lower throughput $\approx 4$K, i.e., experiment lasted longer)}
    \label{ch2:tbl:beldi-comparison}
 \end{table*}

\para{Beldi} We also compare StateFun with a stateful function as a service library and runtime, Beldi, which runs on AWS Lambda and uses DynamoDB as backend storage for transactions. Because of the intricacies of the serverless environment and the restricted way it can be configured, we limit our comparison to latency performance, given a fixed amount of throughput requests, since we have limited visibility to the number of resources used by Beldi. AWS only exposes the concurrency level of the Lambda functions and allows for the restriction of that to a maximum number. In Figure~\ref{ch2:tbl:beldi-comparison}, max concurrency refers to the max concurrency utilized by AWS Lambdas. Max concurrency was fairly stable throughout each experiment. Notably, there is no information regarding the specification of the underlying hardware that is used.

Furthermore, even latency performance does not provide a fair comparison because Beldi only measures latency from when a request's execution starts until the time it completes, without considering the amount of time spent for routing and waiting in an input queue before the request's execution begins. Consequently, a performance throttle in Beldi due to excess load will not show in the measured latency. We try to compensate for this by measuring the experiment's completion time and estimating Beldi's actual throughput. On the other hand, we run StateFun in an IaaS cloud infrastructure where we provide it with a specific amount of computational resources and measure latency end-to-end. The disparity between the two infrastructures and experimental settings limits the insights that can be extracted.

Figure~\ref{ch2:tbl:beldi-comparison} shows the experimental results, from which we draw two notable observations regarding latency.
For non-transfer operations, regardless of the number of keys, StateFun and Beldi achieve the same level of low-latency performance.
Beldi demonstrates 2-3 times superior performance in terms of median latency for transfer operations, while tail latency at the 95th percentile suggests no important differences between the two systems. In Beldi, latency only captures delays that are internal to the system, which may be owed to lock contention inside Beldi, communication stalls between Lambda functions and DynamoDB, as well as queuing in DynamoDB. Unfortunately, it is impossible to pinpoint the exact factors and their merit in the observed tail latency.

Lastly, an important factor in the experiments is that Beldi is let free to auto-scale up to 1000 concurrently executing functions. This aggressive availability of resources far exceeds the 80 CPUs given to the remote functions executing on StateFun. Interestingly, when the number of unique keys is large, meaning that lock contention is low, this level of concurrency is not adequate to accommodate the input throughput of 8K requests per second. In this case, AWS Lambdas used all the available concurrency, and the execution of requests was throttled, waiting for CPUs to become available. Given the experiment's duration, we approximated the level of throughput achieved by Beldi at 4K requests per second. Note that Beldi's latency remains unaffected since it does not account for external delays, such as queuing. On the other hand, we observe that when the number of unique keys is small, meaning that lock contention is high, Beldi is quite efficient. It used more concurrency than what was available to StateFun, but at the same overall level of magnitude. Beldi's efficiency is probably owed to juggling between requests that can be executed immediately and others that should be put to sleep until they can get hold of the lock they require to proceed.

Finally, the observed performance of Beldi does not account for garbage collection. Beldi features a garbage collector to shrink its transaction log periodically, but the garbage collector does not need to run during the presented experiments because their duration is too short. In general, however, the garbage collector is expected to add overhead not represented in our set of experiments.

\section{Related Work} \label{ch2:sec:relatedWork}

\para{SFaaS Systems} SFaaS has been a very active area in both research and the open-source community. From the research community, the most relevant work is Beldi \cite{beldi}, which, like AFT\cite{aft}, builds on top of Amazon's AWS Lambda to add fault tolerance and transaction support, allowing for more complex state management. Their principal difference is that Beldi's execution environment is completely serverless, while AFT relies on external servers for transaction support. To make that happen, Beldi uses atomic logging, extending  Olive~\cite{setty2016olive}, to ensure fault tolerance for read and write operations, with garbage collection to manage the logs' growth. Regarding transactions, Beldi supports a variant of the two-phase commit protocol, enforcing strong consistency guarantees with wait-die deadlock prevention. Cloudburst with Hydrocache \cite{causal-consistency-cloudburst} provides causal consistency guarantees within the same DAG workflow backed by Anna~\cite{anna-kvs}, a key-value state backend. Another promising SFaaS system,  FAASM\cite{shillaker2020faasm}, supports direct memory access between functions while maintaining isolation and speeds up initialization times compared to containers. At the time of writing, FAASM does not provide transactional support.
Finally, the two most prominent open-source SFaaS projects are Cloudstate\footnote{\url{https://cloudstate.io/}}, based on stateful actors, and Apache Flink StateFun, which is presented in detail in \Cref{ch2:sec:flinkstatefun}. In Cloudstate, communication is allowed between different actors within the same cluster and between user-defined functions over gRPC with at-least-once processing guarantees.

\para{Transactional Programming Model} The most notable difference among these systems in terms of programming model is state access. Both StateFun and Cloudstate encapsulate state within a specific function instance. In contrast, Cloudburst and Beldi allow any function access to any state stored in Anna or DynamoDB, respectively. Regarding transactions, only Beldi offers a programming model where the developer writes two markers (begin/end\_tx), and every function invocation in between will execute as part of a transaction. Our contribution is a programming model that supports transactions on StateFun with the choice of strong or relaxed consistency guarantees.

\para{Stream Processing Transactions} Furthermore, transactions on top of stream processing systems have received some attention in the literature. In \cite{botan2012transactional}, the authors introduce a transactional model over both data streams and traditional tabular data. Following a similar model, in  \cite{gotze2019snapshot}, the authors add guarantees for snapshot isolation and consistency across partitioned state. Then TSpoon~\cite{affetti2020tspoon}, an extension of FlowDB~\cite{affetti2017flowdb}), proposes a data management system built on top of a stream processor that supports transactions, giving the option of both strong and weak transactional guarantees and queryable state.  
Our work focuses on transactional workflows between generic stateful functions executed on a serverless dataflow system.

\para{Distributed Databases} The mentioned stream processing systems share the same main goal as distributed databases \cite{mohan1986transaction, kallman2008h, spanner, yan2018carousel, ren2019slog, cockroachdb}, that is, how to scale to multiple machines while providing serializable transactional guarantees. This is an old problem in database research. The R* system \cite{mohan1986transaction} was one of the first to try the two-phase commit protocol with distributed deadlock detection. Then, more recent approaches like H-store~\cite{kallman2008h} showed that distributed database solutions could provide both very high performance and transactional guarantees when transactions touch a single partition. Currently, research in distributed databases revolves around globally distributed databases with Spanner~\cite{spanner} introducing serializable transactions using a timestamp mechanism across all locations/machines based on atomic clocks. Furthermore, approaches like Carousel~\cite{yan2018carousel} and SLOG~\cite{ren2019slog} improve globally distributed database transactions. Carousel enhances transaction execution by minimizing network usage, while SLOG offers a fine-grained transaction protocol based on the proximity between the data and the client. Finally, CockroachDB~\cite{cockroachdb} provides serializable globally distributed transactions without a complicated time mechanism.

\para{Benchmarks} The large variety of use cases and systems makes them difficult to compare using a standardized benchmark. The related benchmarks that could be used to evaluate SFaaS systems are the \textit{Yahoo! Cloud Serving Benchmark} (YCSB)\cite{ycsb} and the DeathStarBench \cite{deathstar}. Given that StateFun is based on Flink, which is a stream processing system, a stream processing benchmark~\cite{karimov2018benchmarking} would be another alternative. However, its workloads are not representative of those executed by an SFaaS system. In addition, we did not consider TPC-C \cite{tpcc} because it was created to test relational database management systems, including transactions, and requires many additional features not present in SFaaS. We ultimately chose to develop and use an extension of YCSB \cite{dey2014ycsb+} that introduced explainable transactional workloads, allowing for an easier interpretation of the results.

\section{Discussion \&\ Open Problems}

\para{Programming Models for the Cloud}
Although the stateful dataflow model has been very successful as an execution model, it has not been leveraged thus far as an intermediate representation. Historically, MapReduce/Hadoop~\cite{dean2008mapreduce} and Dryad \cite{dryadlinq} were first proposed as a means of authoring and executing distributed data-parallel applications using high-level language constructs, such as Java functions and LINQ~\cite{linq} respectively. Many systems have followed that execution model subsequently, including streaming dataflow systems such as Apache Storm~\cite{storm}, Flink~\cite{flink}, Naiad~\cite{murray2013naiad}. However, none of these systems could execute general-purpose cloud applications; their programming model focuses on distributed collection processing and adopts a functional programming API.`

\para{Dataflows for Cloud Applications} We believe that abstractions such as stateful functions can play the role of a high-level programming model for dataflow engines and have a high impact on cloud programming. The current approach to programming in the cloud is to either use domain-specific languages (DSLs)  such as Bloom \cite{alvaro2010boom}, Hilda\cite{hilda}, and Erlang \cite{armstrong2013programming}, or as libraries within mainstream languages like Akka \cite{wyatt2013akka}, Spring Boot (\url{spring.io}). The main observation here is that the developer either has to learn a new domain-specific language or use libraries that leak implementation details into the business logic. Very close to the spirit of this work are virtual actors, and Orleans \cite{bykov2011orleans,orleans} from which Apache Flink's StateFun drew inspiration. However, Orleans requires a specialized runtime and does not offer exactly-once function execution. As we show in this paper, implementing very complex protocols (with lots of corner cases) can be simpler since we benefit from the state management and exactly-once guarantees of modern dataflow systems. Since dataflow systems are well understood, scalable, and consistent nowadays, we believe they will play a critical role in the future of cloud execution engines.

\para{Future Dataflow Systems} However promising they can be, dataflow engines still suffer from several issues. Stream processors such as Apache Flink \cite{flink}, or Jet \cite{jet} have been designed for continuous operation on high-throughput streams. However, stateful functions have very different workload characteristics. For instance, lots of cloud applications may have to call external services  -- a source of non-determinism \cite{silvestre2021clonos}, and functions calling other functions, expecting return values, introduce cycles in the dataflow graph. Current dataflow systems either do not support cycles or support a few special cases of cycles. This is because cycles can cause deadlocks and various other issues \cite{faucet, carbone2017state} that need to be dealt with. Finally, in this paper, we introduced transactions at the function level without having to touch the core of Apache Flink's dataflow engine. However, proper implementation of transactions would require the dataflow system itself to be aware of transaction boundaries (e.g., commit, prepare)  and incorporate transaction processing into its fault-tolerance protocol. We think that more research needs to be performed to get dataflow systems fully capable of leveraging their potential.

\section{Conclusions} \label{ch2:sec:conclusion}

In this chapter, we tackle the problem of supporting transactional workflows across cloud applications on a serverless platform. This problem is notorious in the microservices and cloud applications landscape.
In addition to that, we introduced a programming model and corresponding implementation for authoring workflows across stateful serverless functions with configurable transactional guarantees.
Developers can opt for a distributed transaction across functions with strict atomicity and consistency guarantees or a Saga workflow that provides eventual atomicity and consistency.
These complementary alternatives faithfully represent the requirements of real-world use cases.
We described our implementation on top of Apache Flink StateFun, and evaluated our implementation on an extended version of the YCSB benchmark that we developed in terms of a) throughput and latency overhead against the original StateFun, b) performance efficiency between distributed transactions and Saga workflows, and c) scalability.
We found that our transactional workflows add affordable overhead to the system around 10\%, Sagas significantly outperform distributed transactions on a scale of 15\% -- 34\% depending on the amount of ongoing transactional workflows in the system, and scalability manifests a factor of 90\% for Sagas compared to 75\% -- 87\% for two-phase commit. Furthermore, our comparison against a serverless SFaaS runtime showed that our work could achieve higher throughput, but it also incurs higher latency. Finally, we compared against a popular distributed database, CockroachDB, which achieved better performance in high contention scenarios and in terms of latency. Notably, our work achieved better results in sparse key distributions, while it provides exactly-once processing semantics compared to our deployment of Kafka with a CockroachDB backend at-least-once semantics.

\chapter{Stateflow: a Domain-Specific Language for General-Purpose Cloud Applications}
\label{chapter3}

\vfill

\begin{abstract}

\Cref{chapter2} examined how transactional guarantees can be integrated into existing stateful function-as-a-service platforms, focusing on extending Apache Flink Statefun. While this approach demonstrated the feasibility of augmenting runtime environments with strong consistency guarantees, it did not address the equally important challenge of programmability. Cloud developers continue to face significant complexity when translating high-level application logic into distributed execution semantics. In this chapter, we turn our attention to this challenge.

We introduce \emph{Stateflow}, a domain-specific language and compiler pipeline that enables developers to author general-purpose cloud applications using familiar object-oriented constructs. Stateflow compiles such programs into a dataflow intermediate representation that is portable across multiple execution backends. This chapter details the design of the programming model, its compilation strategy, and the execution guarantees it preserves, thereby advancing the thesis's objective of democratizing cloud applications.

\end{abstract}

\vfill

\blfootnote{This chapter is based on the following work: \\\faFileTextO~\hangindent=15pt\emph{K. Psarakis, W. Zorgdrager,  M. Fragkoulis, G. Salvaneschi, and A. Katsifodimos. Stateful Entities: Object-oriented Cloud Applications as Distributed Dataflows, EDBT'24 (vision) and CIDR'23 (abstract)}~\cite{stateflow}.}

\newpage

\dropcap{O}{}rganizations nowadays enjoy reduced costs and higher reliability, but cloud developers still struggle to manage infrastructure abstractions that leak through in the application layer. As a result, managing application components, such as service invocation, messaging, and state management, requires much more effort than developing the application's business logic \cite{blanastransactions}. Worse, moving a cloud application between cloud providers is prohibitive due to significant differences in the underlying systems.

While there are multiple approaches for distributed application programming (e.g.,  Bloom \cite{alvaro2010boom}, Hilda \cite{hilda}, Cloudburst \cite{cloudburst}, AWS Lambda, Azure Durable Functions, and Orleans \cite{bykov2011orleans,orleans}), in practice developers mainly use libraries of popular general purpose languages such as Spring Boot in Java, and Flask in Python.

None of these approaches offers processing guarantees, failing to support \emph{exactly-once processing}: the ability of a system to reflect the changes of a message to the state exactly once. Instead, they offer at-most- or at-least-once processing semantics. Programmers then have to  ``pollute'' their business logic with consistency checks, state rollbacks, timeouts, retries, and idempotency\cite{rodrigosurvey, microservices-drawbacks}.

We argue that no matter how we approach cloud programming, unless an execution engine can offer exactly-once processing guarantees so that it can be assumed at the level of the programming model, we will never remove the burden of distributed systems aspects from programmers. To the best of our knowledge, the only systems able to guarantee exactly-once message processing~\cite{silvestre2021clonos,carbone2017state} at the time of writing are batch~\cite{dean2008mapreduce, dryadlinq, CUSTOM:web/Spark} and streaming~\cite{storm, flink, murray2013naiad} dataflow systems. However, their programming model follows the paradigm of functional dataflow APIs, which are cumbersome to use and require training and heavy rewrites of the typical imperative code that developers prefer to use for expressing application logic.

For these reasons, we argue that the dataflow model should be used as a low-level intermediate representation (IR) for the modeling and executing distributed applications, but not as a programmer-facing model. Technically, one of the main challenges in adopting a dataflow-based IR is that the dataflow model is functional, with immutable values propagating across operators that typically do not share a global state. Hence, adopting a dataflow-based IR entails translating (arbitrary) imperative code into a functional style.
Compiler research has systematically explored only the opposite direction: to compile code in functional programming languages into a representation that is executable on imperative architectures -- like modern microprocessors. Yet, the translation from imperative to functional or dataflow programming remains largely unexplored.

\begin{figure*}[t]
    \centering
    \includegraphics[width=1\textwidth]{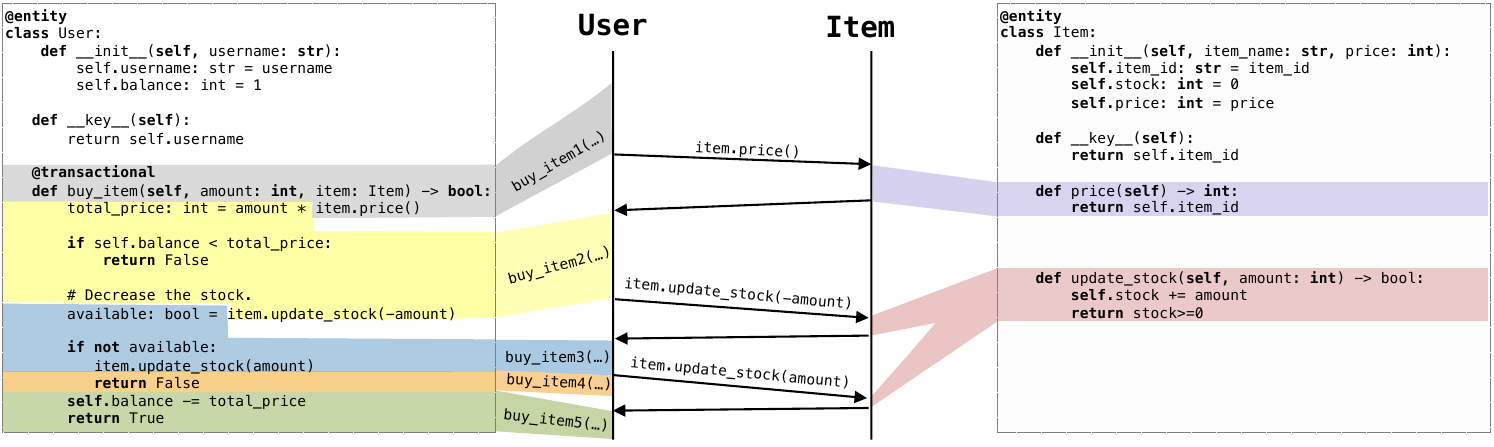}
    \caption{Two stateful entities: \textit{User} and \textit{Item}. The content of imperative functions is split into multiple functions that access the common state of a given entity. Those functions are then encoded into a stateful dataflow that can be executed in a distributed streaming dataflow engine. As a result, $i)$ imperative code is executed in an event-based manner without the need to block, and $ii)$ the code retains exactly-once processing guarantees without the need for programmers to write failure-handling code such as state management, call retries or idempotency.}
    \label{ch3:fig:calls}
\end{figure*}

This chapter presents a prototypical programming model and an IR that compiles imperative, transactional object-oriented applications into distributed dataflow graphs and executes them on existing dataflow systems. 
Instead of designing an external Domain-Specific Language (DSL) for our needs, we opted for an internal DSL embedded in Python - a popular language for cloud programming.
Specifically, a given Python program is first compiled into an IR, an enriched stateful dataflow graph independent of the target execution engine. That dataflow graph can then be compiled and deployed to various distributed systems. The current set of supported systems includes Apache Flink Statefun~\cite{statefun} and Styx (\Cref{chapter4}). The choice of a runtime system is entirely independent of the application layer, which allows switching to different runtime systems with no changes to the application code.

\vspace{2mm}
\noindent The contributions of this chapter go as follows:
\begin{itemize}
    \item To the best of our knowledge, this is the first work to propose compiling and executing imperative programs into distributed, stateful streaming dataflows.
    \item We present a compiler pipeline that analyzes an object-oriented application and transforms it into an IR tailored to stateful dataflow systems.
    \item We describe an IR for cloud applications and how that IR translates to a dataflow execution graph, targeting various distributed systems, thereby making cloud applications portable across different systems and infrastructures.
    \item We compare Stateflow, a novel transactional dataflow system, against Apache Flink Statefun and demonstrate the limitations of existing dataflow systems, motivating further research. Our experimental evaluation shows that Stateflow incurs low latency in the YCSB+T \cite{dey2014ycsb+} workload.
\end{itemize}

\noindent The proposed programming model presented in this chapter can be found at:\\ \url{https://github.com/delftdata/stateflow}.

\section{From Imperative Code to Dataflows}

Historically, imperative programming and functional programming have evolved in parallel:  imperative as a direct codification of (operational) computational models (e.g., Von Neumann architecture, Turing machines) and functional inspired by mathematical abstractions (e.g., lambda calculus, program denotation).
While functional programming has been embraced by several languages (e.g., Haskell \cite{hudak1992report}, ML \cite{harper1986standard}), imperative programming has taken the scene, with most mainstream languages featuring object-oriented (mutable) abstractions.
Over the last few years, imperative languages like Java and Python, which support various domain-specific packages, e.g., networking, statistics, numeric computation, etc., have become extremely popular among non-expert programmers.

Yet, the benefits of functional programming have been known for a while. Most notably, functional code is often {\it embarrassingly parallelizable} because of the lack of side effects and mutability. Developers working with imperative languages -- let alone non-expert developers -- can hardly access this feature.

\subsection{Approach Overview}

The main principle behind our compiler pipeline is that developers simply annotate Python classes with \textit{@stateflow}, and the system automatically analyzes and transforms these classes into an intermediate representation, which is then transformed into stateful dataflow graphs, ready to be deployed on a dataflow system.  Similar to (Virtual) Actors \cite{bykov2011orleans,wyatt2013akka}, \emph{entities} can make calls to methods of other entities. \Cref{ch3:fig:calls} depicts two sample entities: \textit{User} and \textit{Item}. Details of the programming model are provided in \Cref{ch3:sec:assumptions}.

In the first pass of an Abstract Syntax Tree (AST) static analysis, we extract the class's variables (i.e., instance attributes referenced with \textit{self}), the names of each method, and all respective types indicated by the programmer (\Cref{ch3:sec:assumptions}). In the second round of analysis, classes that interact with each other are identified to create a function call graph (\Cref{ch3:sec:classes-dataflows}). Then, the call graph is analyzed to identify calls to other functions (possibly residing in a remote machine), at which point functions have to be split, composing the final dataflow (\Cref{ch3:sec:splitting}). 

This dataflow graph, enriched with the compiled classes, execution plans, and all metadata obtained from static analysis, comprises the intermediate representation (\Cref{ch3:sec:ir}). Finally, that intermediate representation can be translated, deployed, and executed in different target systems (\Cref{ch3:sec:runtimes}). 

\begin{figure}[t!]
    \centering
    \captionsetup{justification=centering}
    \includegraphics[width=0.6\columnwidth]{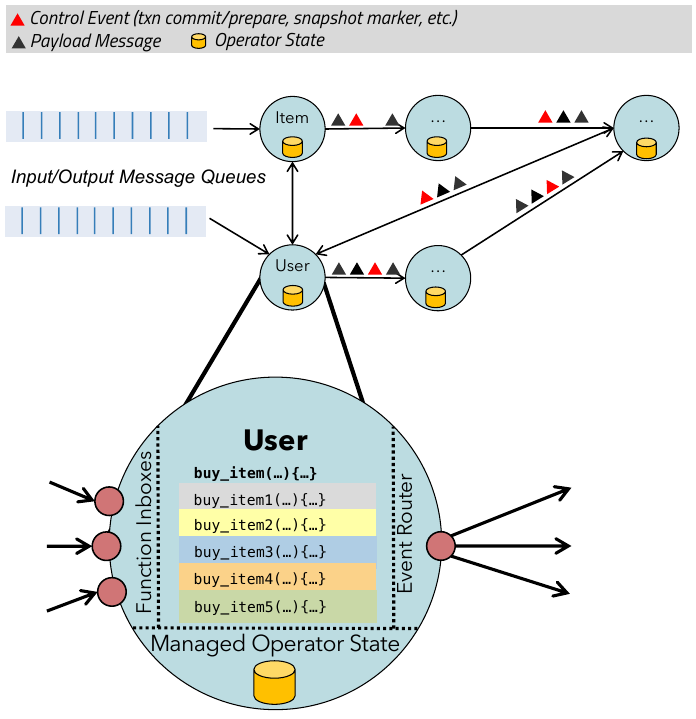}
    \vspace{-2mm}
    \caption{Logical dataflow graph of five entities, focusing on the \textit{User} entity found in \Cref{ch3:fig:calls}.}
    \label{ch3:fig:dataflow}
    \vspace{-4mm}
\end{figure}

\subsection{Programming Model \& Limitations}
\label{ch3:sec:assumptions}

\para{Expressiveness} Our programming model allows programmers to specify simple, object-oriented Python programs. Classes can have references to other classes and call their functions. We term an instance of such a class as a \textit{stateful entity}. The Stateflow compiler currently can analyze conditionals, \textit{for-}loops that iterate through Python lists, as well as general \textit{while} loops. 

\para{Limitations} Stateflow requires static type hints for the input/output of stateful entity functions and ensures the existence of those hints via a static pass over the analyzed classes. Moreover, the functions cannot be recursive. Another assumption that Stateflow makes is that each entity contains a \textit{key()} function. This \textit{key()} function is used by a routing and translation mechanism to partition and distribute the load among parallel instances of that entity within a cluster. Furthermore, the key of a stateful entity cannot change throughout that entity's lifetime. Finally, the entities' state needs to be serializable, i.e., connections to databases, local pipes, and other non-serializable constructs are not allowed and will eventually generate a runtime error.

\para{Running Example}
\Cref{ch3:fig:calls} contains the code for a User and an \textit{Item} entities. Note that since \textit{Item} is a stateful entity, a call to \textit{item.update\_stock(...)} is a remote function call. Both \textit{User} and the \textit{Item} entities are partitioned across the cluster nodes, using the given entity's \textit{key} function.

\subsection{From Entities to Dataflow Operators}
\label{ch3:sec:classes-dataflows}

Each Python class translates to an operator (also called a vertex) in the dataflow graph. In a dataflow graph, an operator cannot be "called" directly, like a function of an object. Instead, an \textit{event} has to enter the dataflow and reach the operator holding the \textit{code} of that entity (e.g., the \textit{User} class) as well as the actual \textit{state} of the entities that instantiate the class (e.g., the \textit{balance} and \textit{username} of the \textit{User} in \Cref{ch3:fig:calls}).

Specifically, each dataflow operator can execute all functions of a given entity, and the triggering function depends on the incoming event. Since operators can be partitioned across multiple cluster nodes, each partition stores a set of stateful entities indexed by their unique key. When an entity's function is invoked, the entity's state is retrieved from the local operator state. Then, the function is executed using the arguments found in the incoming event that triggered the call, as well as the state of the entity at the moment that the function is called.

\para{Example} A \textit{User} operator as seen in \Cref{ch3:fig:dataflow}, is partitioned on \textit{username}. Upon invocation of a function of the \textit{User} entity, an event is sent to the dataflow graph's input queues. The incoming event is partitioned on \textit{username} by an ingress router. Via the dataflow graph, the event ends up at the operator storing the state for that specific \textit{User}. The system then reconstructs the \textit{User} object using the operator's code and the function's state and executes the function. Finally, the function return value is encoded in an outgoing event forwarded to the egress router. This egress router determines if the event can be sent back to the client (caller outside the system, such as an HTTP endpoint) or if it needs to loop back into the dataflow to call another function.

\para{The Need for Function Splitting} For simple functions that do not call other remote functions, both the translation to dataflows and the execution are straightforward. However, if the function \textit{User.buy\_item} calls the (remote) function \textit{item.update\_stock} whose state lies on a different partition, the situation becomes more complicated. Note that a streaming dataflow should never stop and wait for a remote function to complete and return before moving on with processing the next event. Instead, it must ``suspend'' the execution of, e.g., \textit{buy\_item} of \Cref{ch3:fig:calls}, right at the spot that the remote function \textit{item.price()} is called until the remote function is executed. An event comes back to the \textit{User} operator with a return value.

To do this, we adopt a technique to transform the imperative functions into the continuation passing style (CPS)  \cite{cps}. More specifically, we propose an approach to split a function definition into multiple ones (\Cref{ch3:sec:splitting}) at the AST level as depicted (approximately) in \Cref{ch3:fig:calls}.

\subsection{From Imperative Functions to Dataflows}
\label{ch3:sec:splitting}

\para{References to Remote Functions} After the first round of static analysis, the compiler identifies if a function definition has references to a remote stateful entity using Python type annotations. These functions may require \emph{function splitting}. The algorithm traverses the statements of a function definition, and the function is split either when a remote call occurs or on a control-flow structure. For example, the following \textit{buy\_item} calls the remote function \textit{item.update\_stock}:
\begin{lstlisting}[style=pythonlang, numbers=left, xleftmargin=0in, label={ch3:lst:wordcount1}]
def buy_item(self, amount: int, item: Item):
    total_price: int = amount * item.price
    is_removed: bool = item.update_stock(amount)
    return total_price 
\end{lstlisting}

\noindent This function is split at the assign statement on line 3 and results in two new function definitions:
\begin{lstlisting}[style=pythonlang, numbers=left, xleftmargin=0in, label={ch3:lst:wordcount2}]
def buy_item_0(self, amount: int, item: Item):
    total_price: int = amount * item.price
    update_stock_arg = amount
    return total_price, {"_type": "InvokeMethod",
                         "args": [update_stock_arg], ..}
    
def buy_item_1(self, total_price, update_stock_return):
    is_removed: bool = update_stock_return
    return total_price   
\end{lstlisting}
The \textit{buy\_item\_0} function defines the first part of the original function \textit{and} evaluates the arguments for the remote call.
The \textit{buy\_item\_1} function assumes the remote call \textit{item.update\_stock} has been executed, and its return variable is passed as an argument. In general, each function that was split takes as arguments the variables it references in its body and returns the variables it defines. For example, since \textit{buy\_item\_0}  defines the variable \textit{total\_price}, its value is returned from the function. Next, since \textit{buy\_item\_1}  uses \textit{total\_price}, it is defined as a parameter. 

\para{Control Flow} The compiler also needs to split functions when encountering remote function calls within control flow constructs like \textit{if}-statements or \textit{for}-loops. In short, an \textit{if}-statement is split into three new definitions: one that evaluates its condition, one that evaluates the `true' path, and one that evaluates the `false' path. Similarly, a \textit{for}-loop is also split into three new definitions: one that evaluates the iterable, one that evaluates the \textit{for}-body path, and one that evaluates the code path after the loop. The function splitting algorithm is recursively applied to the statements inside the \emph{for} path and inside the true and false path of the \textit{if}-statement.  

\subsection{Intermediate Representation}
\label{ch3:sec:ir}

Our intermediate representation is a stateful dataflow graph enriched with a number of aspects. After the static analysis, each dataflow operator is enriched with the entity/method names that it can run, their input/return types, as well as their method body. After splitting functions, we also need to build what we term a state machine. For every split function (Section \ref{ch3:sec:splitting}), we maintain an execution graph that tracks the execution stage of a given stateful entity's function invocation. 

Essentially, the process of deriving the state machine consists of unrolling the control flow graph of the program. Conceptually, the translation to a state machine is possible by deriving a finite program representation. To this end, we $i)$ do not allow unbounded recursion, and we $ii)$ keep track of the current iteration for loop control structures by enriching the state machine with the additional state. When invoking a function that was split, the state machine is inserted into the function-calling event. As the event flows through the system, the execution graph is traversed, and the proper functions are called. The execution graph stores intermediate results -- the return values of the invoked functions.

\section{Supported Dataflow System Runtimes}
\label{ch3:sec:runtimes}

Stateful entities can be deployed as dataflow graphs to streaming dataflow systems, offering exactly-once fault-tolerance guarantees.

\para{Flink's Statefun}
The IR is translated to a streaming dataflow graph that, for example, Apache Flink can execute. In that case, a Kafka source pushes events to the ingress router, which is a map operator performing a \textit{keyBy} operation to route an event to the correct stateful map operator instance where function execution will take place. Each execution's output is forwarded to the egress router, which forwards outputs to a Kafka sink.

We use Kafka to re-insert an event into the streaming dataflow, thereby avoiding cyclic dataflows, which are not supported by most streaming systems. Notably, our system implements all the logic required for routing and execution in this process.
On the downside, when an event reenters a dataflow to reach the next function block of a split function, race conditions attributed to events coming from non-split functions could lead to state inconsistencies due to other events changing the same function's state in the meantime. Time tracking with watermarks, support for cyclic dataflows, and locking could solve these problems. Since the IR is well-aligned with Statefun's dataflow, only simple translation and mapping are required when using the Statefun runtime.

\para{Styx: a Transactional Dataflow System} Existing dataflow systems cannot execute multi-partition transactions. To this end, we built Styx, a prototype dataflow system in Python. Styx treats each function -- and the state effects it creates via calls to other functions -- as a transaction with ACID guarantees. We achieve consistency by implementing an extension of Aria \cite{aria}, a deterministic transaction protocol. The dataflow system is built to allow for dataflow cycles used in function-to-function communication and leverages co-routines for optimal resource utilization. For fault-tolerance, Styx implements the consistent snapshots protocol~\cite{chandy1985distributed, carbone2017state}, which has been adopted by many streaming dataflow systems~\cite{flink, SilvaZD16, ArmbrustDT18} alongside a replayable source as an ingress, allowing Styx to rollback messages and restore the snapshot upon failure. Although still a prototype, Styx is already able to execute transactional workloads (YCSB-T \cite{dey2014ycsb+} and partly TPC-C) with promising performance (\Cref{ch3:sec:experiments}).

\para{Local} A Styx dataflow graph can execute all its components in a local environment. The only difference is that the state is kept in a local HashMap data structure instead of a state management backend. Local execution allows developers to debug, unit test, and validate a Stateflow program as they would do for an arbitrary application. Afterward, they can deploy the program to one of the supported runtime systems.

\section{Preliminary Experiments}
\label{ch3:sec:experiments}

For the experiments of this section, we opted for running Apache Flink Statefun against Styx (\Cref{ch3:sec:runtimes}). 

\para{Workload} We are using workloads A and B from the original YCSB benchmark \cite{ycsb}. A is update-heavy -- 50\% reads 50\% updates, and B is ready-heavy -- 95\% reads 5\% updates. In addition, we use the transactional workload T from YCSB+T  \cite{dey2014ycsb+}, which atomically transfers an amount from one entity's bank account to another (2 reads and 2 writes). For the throughput test, we defined a mixed workload M (45\% reads 45\% updates 10\% transfers). For the latency tests, we use Zipfian and uniform key distributions. 

\para{Setup} We conducted all the experiments on 14 CPUs: 4 for the Kafka cluster, 6 for the systems, and 4 for the benchmark clients. For Statefun, we gave half of the resources to the Flink cluster and the other to the remote functions. Styx requires a single core coordinator, and the rest are used for its workers.

\para{Baseline} In Styx, we execute complex business logic resulting in state operations. YCSB is a benchmark that supports simple inserts, deletes, and updates, not complete executions of transactions across multiple function calls. It is, therefore, expected that Stateflow, since it executes function calls and application logic, would have a larger overhead than key-value stores. Styx is not a key-value store; instead, it is a stateful function-as-a-service compiler (Stateflow) and runtime that allows programmers to author object-oriented Python code.

\para{Latency} In the first experiment, we measured the end-to-end latency of all the YCSB workloads against the integrated backend systems with both Zipfian and uniform key distributions at a low rate of 100RPS.
As seen in \Cref{ch3:fig:endpoint} both systems perform well with low latencies across all workloads and distributions. Some interesting observations go as follows. First, Statefun performs the same in both the A and B workloads and in both Zipfian and uniform distributions. This happens because Statefun does not use locking, allowing for concurrent access (but also inconsistency). Additionally, since all functions must run in an external Python runtime, the cost of reads and writes is the same due to network costs. We also observe that Styx outperforms Statefun because it allows for internal function-to-function communication and does not require roundtrips to Kafka. Note that Styx additionally supports transactional workloads with higher latency than the rest. Still, if we consider that a transfer operation is 2 read and 2 write operations, the transactional overhead of the system is minimal. Finally, we did not run Statefun against transactional workloads since it offers no transaction support. 

\begingroup
\setlength{\abovecaptionskip}{0pt}
\begin{figure}[t]
    \centering
    \includegraphics[width=0.6\columnwidth]{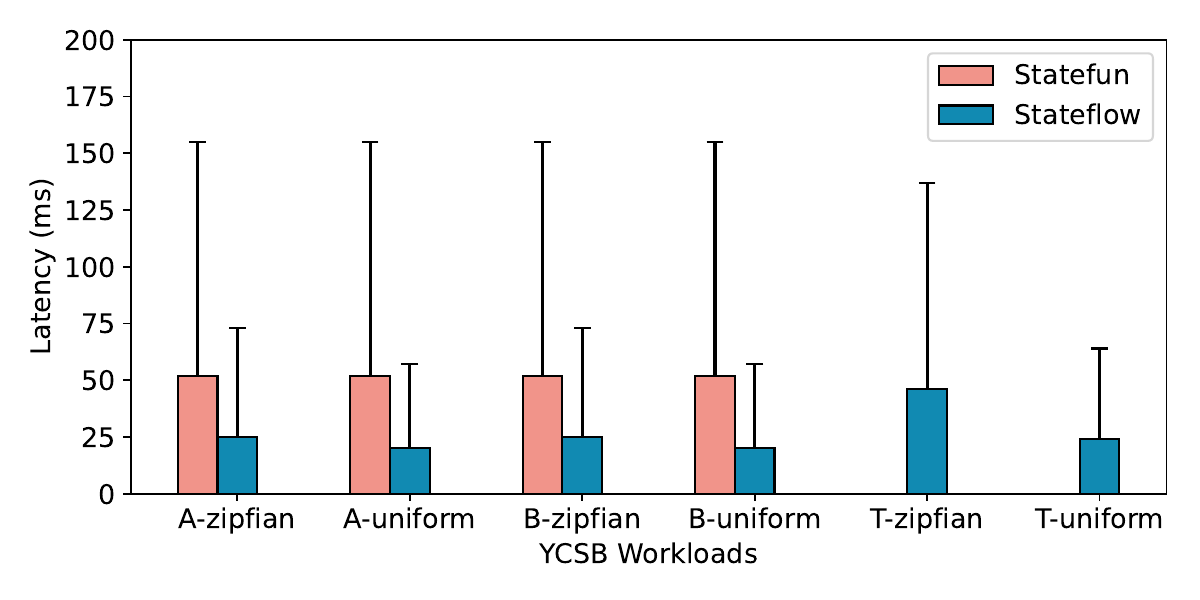}
        \caption{Average latency at the 99th percentile, in YCSB (100~RPS) with both Zipfian and uniform key distributions.}
    \label{ch3:fig:endpoint}
\end{figure}
\endgroup

\para{Throughput} In the second experiment, we gradually increase the input throughput and measure the end-to-end latency. This time, we use the mixed workload that we defined, M (45\% reads, 45\% updates, 10\% transfers). In \Cref{ch3:fig:throughput}, we observe consistent results with the latency experiment up until the point where the difference in efficiency appears. The reason for this is that Styx is using more execution cores since it bundles execution, state, and messaging. In contrast, the Statefun deployment uses half its CPUs for messaging and state within the Apache Flink cluster and the other half for execution in a remote stateless function runtime. In the current experiments, this balanced deployment was the optimal one in terms of resource utilization. 

\para{System Overhead} Finally, we also measured the overhead that program translation (function splits, instrumentation, etc.) incurs as part of the complete runtime (not depicted for the sake of space preservation). We created a synthetic workload that varied different state sizes from 50 to 200 KB. For each event, we measured the duration of different runtime components. Some components, like object construction, are attributed to program transformation overhead, whereas others, like state storage, are attributed to the runtime. In short, function splitting/instrumentation is only responsible for less than 1\% of the total overhead.

\para{Conclusion} The experimental evaluation demonstrates the potential of dataflows as an intermediate representation and execution target for scalable cloud applications. In short, these preliminary experiments show that we can translate imperative programs that hide all the aspects of distributed systems and error management from programmers and still achieve high performance. That said, the experiments also uncover the limitations of dataflow systems and implementation issues that we address in the following section.

\begingroup
 \setlength{\abovecaptionskip}{0pt}
 \setlength{\belowcaptionskip}{0pt}
\begin{figure}[t]
    \centering
    \captionsetup{justification=centering}
    \includegraphics[width=0.6\columnwidth]{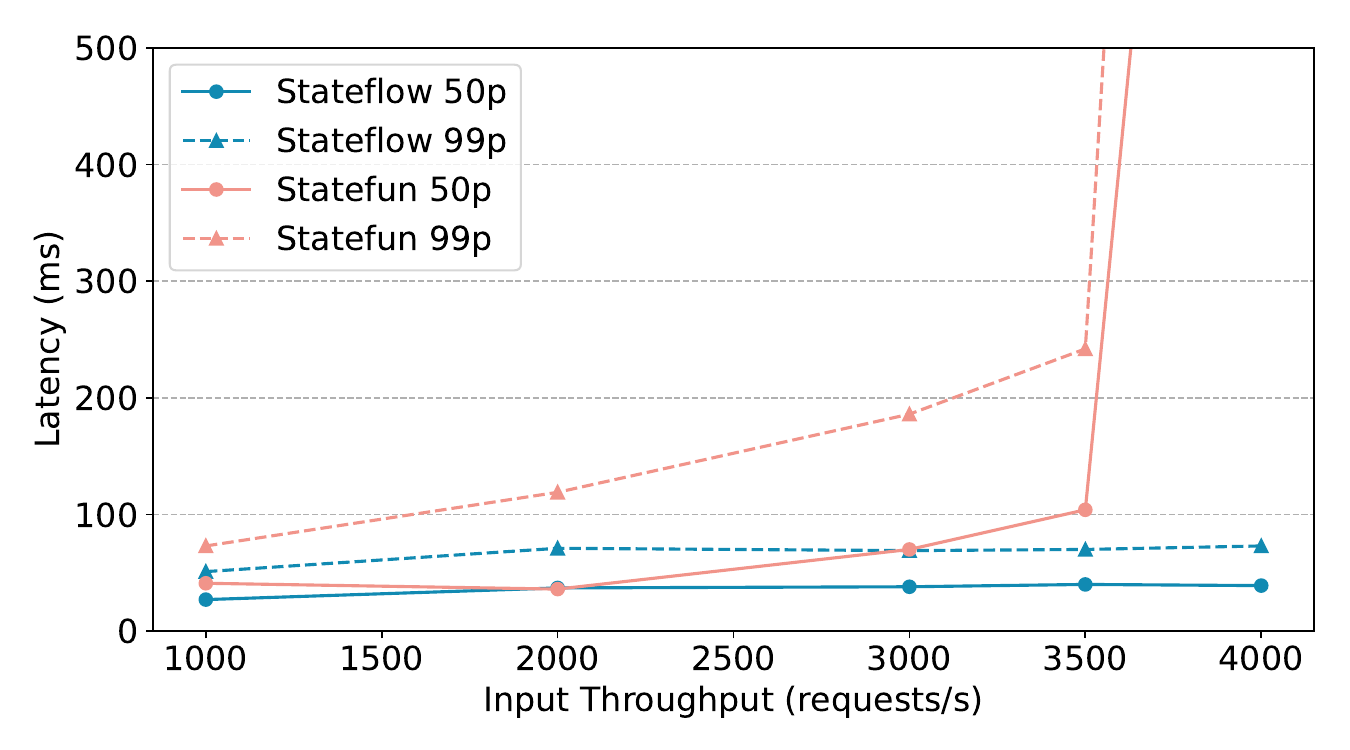}
    \caption{Average and 99th percentile latency for the M workload, with increasing input throughput.}
    \label{ch3:fig:throughput}
\end{figure}
\endgroup

\section{Open Problems \&\ Opportunities}
\label{ch3:sec:roadmap}

The ability to query the global state of a dataflow processor, as well as perform transactional state updates on its state, can transform a dataflow processor into a full-fledged, distributed database system. The envisioned system will be capable of executing Turing-complete ``stored procedures''  (such as the entity functions in the case of this chapter) that are distributed and partitioned and can perform function-to-function calls with exactly-once guarantees. This is the ultimate goal of this work.

In this section, we discuss several opportunities emerging mainly from transactional workloads with low-latency requirements and outline future research directions to enable the adoption of dataflow systems for executing general cloud applications.

\para{Program Analysis}
The dataflow model is essentially a {\it finite state machine} where nodes are the functions from the original ({\it Turing-complete}) program and arcs indicate event flow. 
In the case of loops, events also carry information about the previous iterations of the loop (e.g., the variables that are read and written in the loop body and the loop condition clause). This information handles loops correctly (\Cref{ch3:sec:ir}).
For method calls, if a method is mapped to a single state, it would be problematic to determine where to return after a call if, in the codebase, there are multiple calls that have different return points. We map each method {\it call} into a transition to a state that is specific for that call. This means that calls to the same method may result in a different state in the automata, ensuring that each state has the correct return point as the next. This approach requires to {\it unroll} the program, expanding each potential method call that may occur at runtime into a different state.

Following this approach, recursive functions would result in a state for each recursive step. Since unbounded recursion would result in infinite automata, we prohibit recursion.
Yet, from a compiler perspective, since a program can be CPS-transformed, recursion can be translated into loops via tail-call elimination~\cite{10.1145/214448.214454}, which could potentially affect the dataflow engine's performance.

In addition, in what is traditionally referred to as {\it dataflow languages} (e.g., Esterel~\cite{BERRY199287}, Lucid~\cite{lucid85}), the computation is driven by data propagation -- just like in streaming dataflows. However, the expressivity of such languages has been intentionally limited to enable efficient execution (automatic) verification techniques. While in this work, we aim to target Turing-complete Python programs, the trade-off between expressivity, efficiency, and automatic verification is yet to be researched in the future.

\para{Transactions} Current dataflow systems guarantee the consistency of single-event effects on a given state key.  To support transactional executions across stateful entities, we could employ \textit{single-shot} transactions\cite{kallman2008h} or, like in our prototypical dataflow system (\Cref{ch3:sec:runtimes}), borrow ideas from deterministic databases\cite{abadithecase,abadi2018overview,aria} for minimizing the coordination of transactions. In practice, a large percentage of transactions can be expressed as single-shot transactions \cite{single-shot-aws}; very popular databases such as Amazon's DynamoDB \cite{sivasubramanian2012amazon} and VoltDB \cite{stonebraker2013voltdb} support single-shot transactions. These ideas can define how a programming model can support patterns adopted by practitioners in recent years, starting with SAGAs \cite{sagas} and Try-Confirm-Cancel \cite{DBLP:conf/cidr/Helland07}.

\para{Exactly-once, Latency \& External Systems} Exactly-once guarantees can incur high latency: the outputs of a dataflow only become visible after an epoch terminates successfully\cite{carbone2017state}. Epoch intervals cannot be too small because they would incur a high overhead. However, one can leverage causal recovery \cite{wang2019lineage} and determinants \cite{silvestre2021clonos} alongside replayable sinks to minimize the latency within each epoch.
The replayable sinks are required to be able to retrieve determinants. However, at the border of a system, i.e., when a message leaves the dataflow graph and is sent to an external system, replayable sinks may be hard to assume. In that case, one should use more traditional techniques for deduplication (e.g., the standard idempotence keys used in the HTTP protocol). Under certain assumptions (deterministic computations, persistent/replayable request queues, etc.), such idempotence keys can be generated automatically. However, this will not be the case for a generic distributed application, which will have to generate, keep track of, check, and recycle unique identifiers to enforce the delivery of its output exactly-once. These issues have not been studied enough in the context of distributed databases or models for cloud programming.

\para{Querying Stateful Entities}
In previous work \cite{squery}, we have shown that querying the global state of a dataflow processor can be not only efficient but can also come with certain correctness guarantees. Some work on querying actors has already been done in the context of Orleans \cite{bernstein2017indexing}. However, querying (e.g., with SQL) a set of entities still poses a number of challenges, especially with respect to the tradeoff between the freshness and consistency of query results. To this end, we could borrow ideas from RAMP (read-atomic) transactions \cite{bailis2016scalable} that match well the execution model of transactions and read operations in stateful entities.

\section{Related Work}

The idea of democratizing distributed systems programming is not new. For instance, in~\cite{hydro},  the authors mention that a combination of dataflows and reactivity would provide a good execution model for cloud applications. In this work, we share the same belief and build a prototype towards that direction.

\para{Programming models} In the past, approaches like Distributed ML~\cite{krumvieda1993distributed}, Smalltalk~\cite{deutsch1984efficient}, and Erlang~\cite{armstrong2013programming} aimed at simplifying the programming and deployment of distributed applications. Many of those ideas, including the Actor model, can be reused and extended today.
Erlang implemented a flavor of the actor model. 
Akka~\cite{wyatt2013akka} offers a low-level programming model for actors. Closest to our work is the Virtual Actors model introduced by Orleans \cite{bykov2011orleans,orleans}, which aims at simplifying Cloud programming and even supports some form of transactions~\cite{eldeeb2016transactions}.
However, Orleans requires a specialized runtime system for virtual actors, which does not support exactly-once messaging and does not compile its actors into stateful dataflows.

\para{Imperative programming to Dataflows} The idea of translating imperative code to dataflow is not new. In the database community, there has been work on detecting imperative parts of general applications that can be converted into SQL queries (e.g.,\cite{emani2017dbridge}) but also for automatic parallelization of imperative code in multi-core systems. For instance, the work by Gupta and Sohi~\cite{gupta2011dataflow} compiles sequential imperative code to dataflow programs and executes them in parallel. Our work draws inspiration from both these lines of work and extends them by considering the partitioning of state and other considerations that we outline in~\Cref{ch3:sec:roadmap}.

\para{Stateful Functions} A new breed of systems marketed as stateful functions, such as Cloudburst \cite{cloudburst}, Lightbend's \url{Cloudstate.io}, and Apache Flink's \url{Statefun.io}~\cite{tstatefun}, as well as our early prototype in Scala~\cite{SFaaS-in-action}, also aim at abstracting away the details of deployment and scalability. However, none of those compiles general-purpose object-oriented code into dataflows.

\section{Conclusions}

In this chapter, we argue that if we want to hide failures from the top-level programming models of Cloud applications, exactly-once guarantees should become first-class citizens. While dataflow systems can provide such guarantees, their programming model makes the development of general Cloud applications cumbersome. To this end, we have developed a compiler pipeline that statically analyzes an object-oriented Python application to create an intermediate representation in the form of a dataflow graph and then submits that dataflow graph to existing dataflow systems. Leveraging dataflow systems' exactly-once guarantees can essentially hide all Cloud failures from programmers with low overhead: our preliminary experimental evaluation demonstrates that function splitting and program transformation incur less than 1\% overhead and the YCSB+T benchmark, with low-latency execution. 

\para{Current Status} Despite the encouraging results, lots of problems remain open, specifically in the area of transaction execution, programming models, program analysis, and dataflow engines for general cloud applications. Our work currently focuses primarily on $i)$ strengthening the formal underpinnings of program transformation to dataflows, $ii)$ extending the programming model with different transactional paradigms, and $iii)$ further developing Styx, our novel transactional dataflow system.

\chapter{Styx: a Transactional Dataflow-Based Runtime for Stateful Functions as a Service}
\label{chapter4}

\vfill

\begin{abstract}

The previous chapter (\Cref{chapter3}) introduced Stateflow, a high-level programming model for cloud applications that compiles object-oriented Python code into distributed dataflows. While Stateflow simplifies application development and abstracts away failures and transactions, it surfaces a set of core runtime requirements to support its execution model—namely, the need for efficient, fault-tolerant, and transactional orchestration of stateful functions.

This chapter presents Styx, a distributed runtime system purpose-built to meet these requirements. Styx implements a novel transactional execution protocol over a streaming dataflow engine, enabling exactly-once semantics and serializable transactions across arbitrary function calls. By integrating deterministic execution, co-located state and compute, and an efficient acknowledgment mechanism, Styx addresses the key limitations identified in existing serverless platforms and transactional SFaaS systems, as discovered in \Cref{chapter2}.

The chapter introduces the architecture, programming model, and execution protocol of Styx and evaluates its performance across a range of benchmarks. We conclude the chapter with a demonstration that showcases the system in action, highlighting its ease of use, scalability, and fault tolerance.

\end{abstract}

\vfill

\blfootnote{

\faFileTextO~\hangindent=15pt\emph{K. Psarakis, G. Christodoulou, G. Siachamis, M. Fragkoulis, and A. Katsifodimos. Styx: Transactional Stateful Functions on Streaming Dataflows, SIGMOD'25}~\cite{styx}.

\faFileTextO~\hangindent=15pt\emph{K. Psarakis, O. Mraz, G. Christodoulou, G. Siachamis, M. Fragkoulis, and A. Katsifodimos. Styx in Action: Transactional Cloud Applications Made Easy (Demo), VLDB'25}~\cite{styxdemo}.
}

\newpage

\dropcap{D}{}espite the commercial offerings of the Functions-as-a-Service (FaaS) cloud service model, its suitability for low-latency stateful applications with strict consistency requirements, such as payment processing, reservation systems, inventory keeping, and low-latency business workflows, is quite limited. The reason behind this unsuitability is that current FaaS solutions are stateless, relying on external, fault-tolerant data stores (blob stores or databases) for state management. In addition, while multiple frameworks can perform workflow execution (e.g., AWS Step Functions \cite{stepfunctions}, Azure Logic Apps \cite{logicapps}), they do not provide primitives for \textit{transactional} execution of such applications.
As a result, distributed applications (e.g., microservice architectures) suffer from serious consistency issues when the responsibility of transaction execution is left to developers \cite{blanastransactions,rodrigosurvey}.

In line with recent research \cite{cloudburst, boki, beldi, tstatefun, nightcore, spenger2022portals}, we agree that for FaaS offerings to become mainstream, they should include state management support for stateful functions according to the Stateful Functions-as-a-Service (SFaaS) paradigm. In addition, we argue that a suitable runtime for executing workflows of stateful functions should also provide $i)$ end-to-end serializable transactional guarantees across multiple functions, $ii)$ low-latency and high-throughput execution, and $iii)$ a high-level programming model, devoid of low-level primitives for locking and transaction coordination. To the best of our knowledge, no existing approach addresses all these requirements together.

\begin{figure}[t]
    \centering
    \includegraphics[width=0.6\columnwidth]{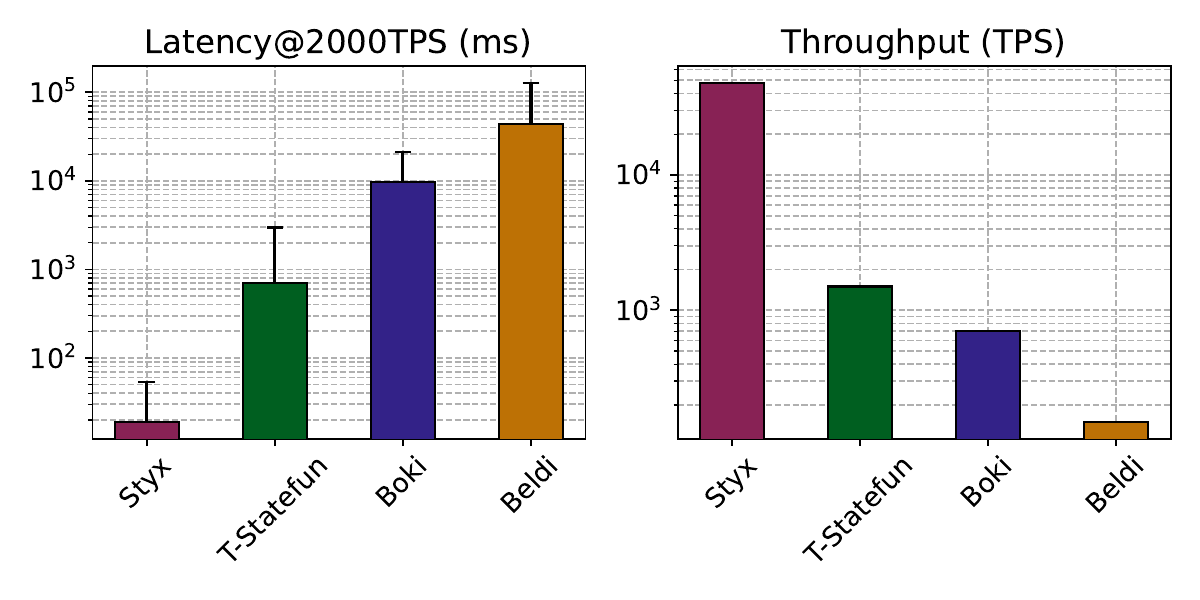}
    \caption{Styx outperforms the SotA by at least one order of magnitude in transactional workloads (\Cref{ch4:sec:exp}). The figure shows median (bar)/99p (whisker) latency and throughput. For the latency plot, the input throughput is 2000 transactions per second (TPS), and for the throughput plot, we report the throughput that the systems achieve at subsecond latency.}
    \label{ch4:fig:summary_comparison}
\end{figure}

The state-of-the-art transactional SFaaS with serializable guarantees, Boki~\cite{boki}, Beldi~\cite{beldi}, and T-Statefun~\cite{tstatefun} do support transactional end-to-end workflows but induce high commit latency and low throughput. The main reason behind their inefficiency is the separation of state storage and function logic, as well as the use of locking and Two-Phase Commit (2PC)~\cite{Gray1978} to coordinate and ensure the atomicity of cross-function transactions.

This paper proposes Styx, a novel dataflow-based runtime for SFaaS. Styx ensures that each transaction's state mutations will be reflected once in the system's state, even under failures, retries, or other potential disruptions (known as exactly-once processing).
Additionally, Styx can execute arbitrary function orchestrations with end-to-end serializability guarantees, leveraging concepts from deterministic databases to avoid costly 2PCs. 

Our work stems from two important observations. First, modern streaming dataflow systems such as Apache Flink \cite{flink} guarantee exactly-once processing\cite{flink,carbone2017state,silvestre2021clonos} by transparently handling failures. A limitation of those streaming systems is that they cannot execute general cloud applications such as microservices or guarantee transactional SFaaS orchestrations. Second, deterministic database protocols\cite{calvin,aria} that can avoid expensive 2PC invocations have not been designed for complex function orchestrations and arbitrary call-graphs. For the needs of transactional SFaaS, Styx leverages a deterministic transactional protocol, enabling early commit replies to clients (i.e., before a snapshot is committed to persistent storage).

Our work is in line with recent proposals in the area, such as DBOS~\cite{dbos}, Hydro~\cite{hydro}, and SSMSs~\cite{li2024serverless}. Contrary to these systems, our work adopts the streaming dataflow execution model and guarantees serializability \textit{across} functions. As shown in \Cref{ch4:fig:summary_comparison}, Styx achieves one order of magnitude lower median latency, two orders of magnitude lower 99p latency at 2000 transactions/sec, and one order of magnitude higher throughput compared to state-of-the-art (SotA) serializable SFaaS systems \cite{boki, beldi, tstatefun}.

\vspace{2mm}

\noindent In short, this paper makes the following contributions:

\noindent \textbf{--} Styx combines deterministic transactions with dataflows and overcomes the challenges that arise from this design choice (\Cref{ch4:sec:background}).

\noindent \textbf{--} Styx enables high-level SFaaS programming models that abstract away transaction and failure management code (\Cref{ch4:sec:programming-model}). Styx does so, by guaranteeing exactly-once processing (\Cref{ch4:sec:dataflow-system}) and transactional serializability across arbitrary function calls (\Cref{ch4:sec:seq_f_ex} and \Cref{ch4:sec:cmt_t}).

\noindent \textbf{--} Styx extends the concept of deterministic databases to support arbitrary workflows of stateful functions, contributing a novel acknowledgment scheme (\Cref{ch4:sec:ack_shares}) to track function completion efficiently, as well as a function-execution caching mechanism (\Cref{ch4:sec:caching}) to speed up function re-executions.

\noindent \textbf{--} Styx's deterministic execution enables early commit replies: transactions can be reported as committed, even before a snapshot of executed transactions is committed to durable storage (\Cref{ch4:sub:early}).

\noindent \textbf{--} Styx outperforms the state-of-the-art \cite{beldi,boki,tstatefun} by at least one order of magnitude higher throughput in all tested workloads while achieving lower latency and near-linear scalability (\Cref{ch4:sec:exp}).

\vspace{2mm}
\noindent Styx is available at: \url{https://github.com/delftdata/styx}

\section{Motivation} \label{ch4:sec:background}
In this section, we analyze the specifics of streaming dataflow systems design and argue that they can be extended to encapsulate the primitives required for consistently and efficiently executing workflows of stateful functions.
Our work is based on a key observation: the architecture of high-performance cloud services closely resembles a parallel dataflow graph, where the state is partitioned and co-located with the application logic~\cite{styxcidr}.
Additionally, as we detail in \Cref{ch4:sec:determinism-transactions}, there is a synergy between deterministic transactions and dataflow systems. Such a combination can offer state consistency and ease of programming as monolithic solutions did in the past, while improving scalability and eliminating developer involvement. Finally, we show how deterministic transactions can be extended for SFaaS, where transaction boundaries are unknown, unlike online transaction processing (OLTP).

\subsection{Dataflows for Stateful Functions}

Stateful dataflows are the execution model implemented by virtually all modern stream processors \cite{flink,murray2013naiad,NoghabiPP17}. Besides being a great fit for parallel, data-intensive computations, stateful dataflows are the primary abstraction supporting workflow managers such as Apache Airflow \cite{airflow}, AWS Step Functions \cite{stepfunctions}, and Azure's Durable Functions\cite{durable_functions}. In the following, we present the primary motivation behind using stateful dataflows to build a suitable runtime for orchestrating general-purpose cloud applications.

\para{Exactly-once Processing} Message-delivery guarantees are fundamentally hard to deal with in the general case, with the root of the problem being the well-known Byzantine Generals problem \cite{lamport1982byzantine}. However, in the closed world of dataflow systems, exactly-once processing is possible \cite{flink,carbone2017state,silvestre2021clonos}. As a matter of fact, the APIs of popular streaming dataflow systems, such as Apache Flink, require no error management code (e.g., message retries or duplicate elimination with idempotency IDs).

\para{Co-Location of State and Function} The primary reason streaming dataflow systems can sustain millions of events per second \cite{flink,jet} is that their state is partitioned across operators that operate on local state. While the structure of current Cloud offerings favors the disaggregation of storage and computation, we argue that co-locating state and computation is the primary vehicle for high performance and can also be adopted by modern SFaaS runtimes, as opposed to using external databases for state storage.

\para{Coarse-Grained Fault Tolerance} To ensure atomicity at the level of workflow execution, existing SFaaS systems perform fine-grained fault tolerance \cite{beldi,boki}; each function execution is logged and persisted in a shared log before the next function is called. This requires a round-trip to the logging mechanism for each function call, which adds significant latency to function execution. Instead of logging each function execution, streaming dataflow systems \cite{checkmate,carbone2017state,chandy1985distributed} opt for a coarse-grained fault tolerance mechanism based on asynchronous snapshots, reducing this overhead.

\subsection{Determinism \&\ Transactions}
\label{ch4:sec:determinism-transactions}

Given a set of database partitions and a set of transactions, a deterministic database\cite{abadi2018overview, abadithecase} will end up in the same final state despite node failures and possible concurrency issues. Traditional database systems offer \textit{serializable} guarantees, allowing multiple transactions to execute concurrently, ensuring that the database state will be equivalent to the state of one serial transaction execution. Deterministic databases guarantee not only serializability but also that a given set of transactions will have exactly the same effect on the database state despite transaction re-execution.
This guarantee has important implications \cite{abadi2018overview} that have not been leveraged by SFaaS systems thus far.

\para{Deterministic Transactions on Streaming Dataflows} Unlike 2PC, which requires rollbacks in case of failures, deterministic database protocols \cite{aria,calvin} are "forward-only": once the locking order \cite{calvin} or read/write set \cite{aria} of a batch of transactions has been determined, the transactions are going to be executed and reflected on the database state, without the need to rollback changes. This notion is in line with how dataflow systems operate: events flow through the dataflow graph, from sources to sinks, without stalls for coordination. This match between deterministic databases and the dataflow execution model is the primary motivation behind Styx's design choice to implement a deterministic transaction protocol on top of a dataflow system.

\subsection{Challenges} 
Despite their success and widespread applicability, dataflow systems need to undergo multiple changes before they can be used for transactional stateful functions. In the following, we list challenges and open problems tackled in this work.

\para{Programming Models} Dataflow systems at the moment are only programmable through functional programming-style dataflow APIs: a given cloud application has to be rewritten by programmers to match the event-driven dataflow paradigm. Although it is possible to rewrite many applications in this paradigm, it takes a considerable amount of programmer training and effort. We argue that dataflow systems would benefit from object-oriented or actor-like programming abstractions in order to be adopted for general cloud applications, such as microservices.

\para{Support for Transactions} Transactions in the context of streaming dataflow systems typically refer to processing a set of input elements and their state updates with ACID guarantees \cite{zhang2024survey}. Despite progress, critical challenges remain open, such as the performance overhead incurred by multi-partition transactions, as well as the need to block flows of data for locking and message re-ordering. In this work, we argue that in order to implement transactions in a streaming dataflow system, we need to "keep the data moving" \cite{stonebraker20058} by avoiding disruptions in the natural flow of data while tightly integrating transaction processing into the system's state management and fault tolerance protocols.

\para{Deterministic OLTP and SFaaS} OLTP databases that use deterministic protocols like Calvin~\cite{calvin, zhou2022lotus, aria} either require each transaction's read/write set a priori or are extended to discover the read-write sets of a transaction by first executing it. Additionally, in both scenarios, deterministic protocols assume that a transaction is executed as a single-threaded function that can perform remote reads and writes from other partitions.
In the case of SFaaS, arbitrary function calls enable programmers to take advantage of both the separation of concerns principle, which is widely applied in microservice architectures \cite{rodrigosurvey}, as well as code modularity. Although deterministic database systems have been proven to perform exceptionally well~\cite{abadi2018overview}, designing and implementing a deterministic transactional protocol for arbitrary workflows of stateful functions is non-trivial. Specifically, arbitrary function calls create complex call-graphs that need to be tracked in order to establish a transaction's boundaries before committing.

\para{Dataflows for Arbitrary-Workflow Execution} The prime use case for dataflow systems nowadays is streaming analytics. However, general-purpose cloud applications have different workload requirements. Functions calling other functions and receiving responses introduce cycles in the dataflow graph. Such cycles can cause deadlocks and need to be dealt with \cite{faucet}.


\noindent In this work, we tackle these challenges and propose a dataflow system tailored to the needs of stateful functions with built-in support for deterministic transactions and a high-level programming model. 

\section{Programming Model} \label{ch4:sec:programming-model}

The programming model of Styx is based on Python and comprises operators that encapsulate partitioned mutable state and functions that operate on that. An example of the programming model of Styx is depicted in \Cref{ch4:fig:programming-model}.


\subsection{Programming Model Notions}

\para{Stateful Entities} Similar to objects in object-oriented programming, \entities{} in Styx are responsible for maintaining and mutating their own \sstate{}. Moreover, when a given entity needs to update the state of another entity, it can do so via a function call. Each entity bears a unique and immutable \key{}, similar to Actor references in Akka~\cite{akka}, with the difference that entity keys are application-dependent and contain no information related to their physical location. The dataflow runtime engine (\Cref{ch4:sec:dataflow-system}) uses that key to route function calls to the right operator that accommodates that specific \entity{}. 

\para{Functions} \functions{} can mutate the state of an entity. By convention, the \context{} is the first parameter of each function call. Functions are allowed to call other functions directly, and Styx supports both synchronous and asynchronous function calls. For instance, in lines 9-11 of \Cref{ch4:fig:programming-model}, the instantiated reservation entity will call asynchronously the function \textit{'reserve\_hotel'} of an entity with key \textit{'hotel\_id'} attached to the Hotel operator. Similarly, one can make a synchronous call that blocks waiting for results. In this case, Styx will block execution until the call returns. Depending on the use case, a mix of synchronous and asynchronous calls can be used. Asynchronous function calls, however, allow for further optimizations that Styx applies whenever possible, as we describe in \Cref{ch4:sec:seq_f_ex} and \Cref{ch4:sec:cmt_t}.

\begin{figure}[t]
\begin{lstlisting}[style=pythonlang]
from styx import Operator
from deathstar.operators import Hotel, Flight

reservation_operator = Operator('reservation', n_partitions=4)

@reservation_operator.register
async def make_reservation(context, flight_id, htl_id, usr_id):

    context.call_async(operator=Hotel,
                       function_name='reserve_hotel',
                       key=htl_id)
    context.call_async(operator=Flight,
                       function_name='reserve_flight',
                       key=flight_id)

    reservation = {"fid":flight_id, "hid":htl_id, "uid":usr_id}
    await context.state.put(reservation)|\label{ch4:code:context}|
    
    return "Reservation Successful"
\end{lstlisting}
\caption{Deathstar's~\cite{deathstar} Hotel/Flight reservation in Styx. From lines 9-14, the $reserve\_hotel$ and $reserve\_flight$ functions are invoked asynchronously. Finally, in lines 16-17, the reservation information is stored. In Styx, the transactional and fault tolerance logic are handled internally.}
\label{ch4:fig:programming-model}
\end{figure}

\para{Operators} Each \entity{} directly maps to a dataflow operator (also called a vertex) in the dataflow graph. When an \textit{event} enters the dataflow graph, it reaches the operator holding the \textit{function code} of the given entity as well as the \textit{state} of that entity. In short, a dataflow operator can execute all functions of a given entity and store the state of that entity. Since operators can be partitioned across multiple cluster nodes, each partition stores a set of stateful entities indexed by their unique \key{}. When an entity's function is invoked (via an incoming event), the entity's state is retrieved from the local operator state. Then, the function is executed using the arguments found in the incoming event that triggered the call.

\para{State \& Namespacing} As mentioned before, each entity has access only to its own state. In Styx, the state is \emph{namespaced} with respect to the entity it belongs to. For instance, a given key "\textit{hotel53}" within the operator \textit{Hotel} is represented as: \textit{entities://Hotel/hotel53}. This way, a reference to a given key of a state object is unique and can be determined at runtime when operators are partitioned across workers. Programmers can store or retrieve \sstate{} through the \textit{context} object by invoking \textit{context.put()} or \textit{get()} (e.g., in \Cref{ch4:code:context} of \Cref{ch4:fig:programming-model}). Styx's \textit{context} is similar to the context object used in other systems such as Flink Statefun, AWS Lambda, and Azure Durable Functions.

\para{Transactions} A transaction in Styx begins with a client request. The functions that are part of the transaction form a workflow that executes with serializable guarantees. Styx's programming model allows transaction aborts by raising an uncaught exception. In the example of \Cref{ch4:fig:programming-model}, if a hotel entity does not have enough availability when calling the \textit{'reserve\_hotel'} function, the \textit{'make\_reservation'} transaction should be aborted, alongside potential state mutations that the \textit{'reserve\_flight'} has made to a flight entity. In that case, the programmer has to raise an exception as follows:

\begin{lstlisting}[style=pythonlang]
...
# Check if there are enough rooms available in the hotel
if available_rooms <= 0: 
    raise NotEnoughSpace(f'No rooms in hotel: {context.key}')
...
\end{lstlisting}

The exception is caught by Styx, which automatically triggers the abort/rollback sequence of the transaction where the exception occurred and sends the user-defined exception message as a reply.

\para{Exactly-once Function Calling}

Styx offers \emph{exactly-once processing} guarantees: it reflects the state changes of a function call execution exactly-once. Thus, programmers do not need to ``pollute'' their application logic with consistency checks, state rollbacks, timeouts, retries, and idempotency \cite{rodrigosurvey, microservices-drawbacks}. We detail this capability in \Cref{ch4:sec:fault-tolerance}.

\begin{figure}[t]
    \centering
    \includegraphics[width=0.6\columnwidth]{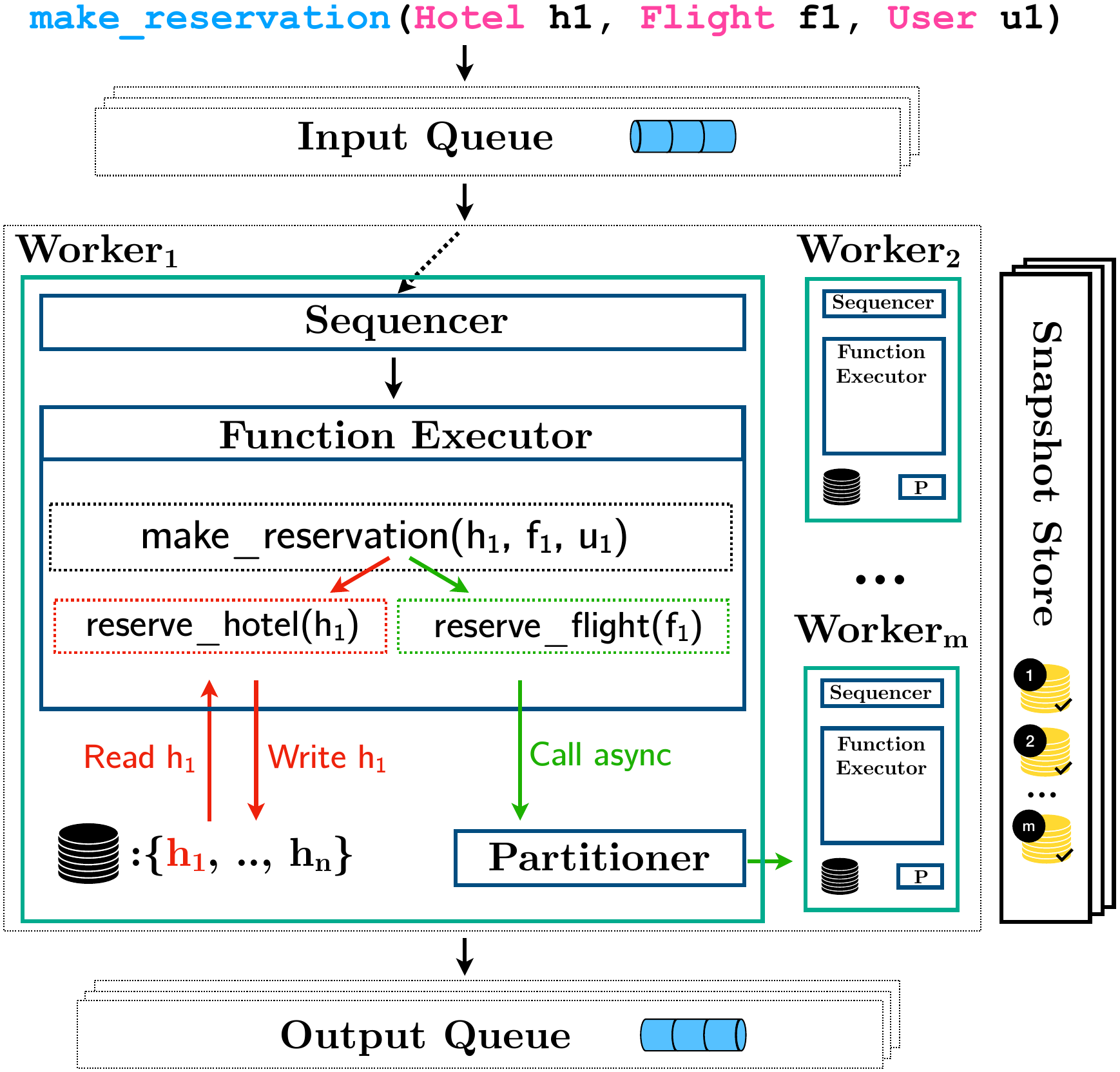}
        \caption{Stateful-Function execution in Styx. In each worker, one coroutine manages the sequencing of incoming transactions, while another coroutine handles their processing. In this example, transaction (\textit{make\_reservation}) consists of two functions: \textit{reserve\_hotel} and \textit{reserve\_flight}. A function can access local state (\textit{reserve\_hotel}) but also perform remote calls to different partitions (\textit{reserve\_flight}). This remote call uses the partitioner to locate the correct worker storing that partition.}
    \label{ch4:fig:styxfunction}
\end{figure}

\section{Styx's Architecture} 
\label{ch4:sec:dataflow-system}

In this section, we describe the components (\Cref{ch4:fig:styxfunction}) and the main design decisions of Styx.

\subsection{Components}

\para{Coordinator} The coordinator manages and monitors Styx's workers, as well as the runtime state of the cluster (transactional metadata, dataflow state, partition locations, etc.). It also performs scheduling and health monitoring. Styx monitors the cluster's health using a heartbeat mechanism and initiates the fault-tolerance mechanism (\Cref{ch4:sec:fault-tolerance}) once a worker fails.

\para{Worker} As depicted in \Cref{ch4:fig:styxfunction}, the worker is the primary component of Styx, processing transactions, receiving or sending remote function calls, and managing state.

The worker consists of two primary coroutines. The first coroutine ingests messages for its assigned partitions from a durable queue and sequences them. The second coroutine receives a set of sequenced transactions and initiates the transaction processing. By utilizing the coroutine execution model, Styx increases its efficiency since the most significant latency factor is waiting for network or state-access calls. Coroutines allow for single-threaded concurrent execution, switching between coroutines when one gets suspended during a network call, allowing others to make progress. Once the network call is completed, the suspended coroutine resumes processing.

\begin{figure*}[t]
    \centering
    \includegraphics[width=\columnwidth]{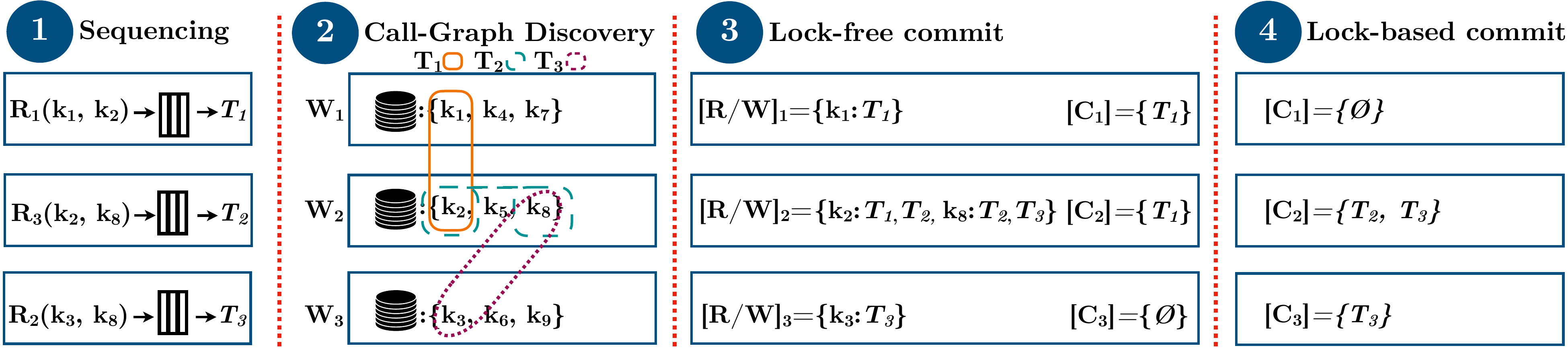}
    \caption{The transaction execution pipeline in Styx is divided into $4$ parts. First, each external request ($R_i$) is sequenced as a transaction and is assigned a unique ID. Afterward, the transactions execute their application logic, accessing local keys and performing remote function calls. While a transaction executes, Styx tracks its accessed keys ($[R/W]_i$) and incrementally constructs its call-graph. Subsequently, Styx commits the transactions that do not participate in unresolved conflicts without having to perform locking. For example, we observe that workers $W_{1}$ and $W_{2}$ are capable to commit $C_1 = C_2 = \bigl\{T_1\bigr\}$ while \textit{$T_{1}$} interacts with the same keys as \textit{$T_{2}$}; although it has the lowest id. In the final part, we commit all the transactions by resolving the conflicts with a lock-based mechanism ($C_2 = \bigl\{T_2, T_3\bigr\}$), $C_3 = \bigl\{T_3\bigr\}$).}
    \label{ch4:fig:styxlifecycle}
\end{figure*}

\para{Partitioning Stateful Entities Across Workers} Styx makes use of the entities' \textit{key} to distribute those entities and their state across a number of workers. By default, each worker is assigned a set of keys using hash partitioning.

\para{Input/Output Queue} For fault tolerance, Styx assumes a persistent input queue from which it receives requests from external systems (e.g., from a REST gateway API). Styx requires the input queue to be able to deterministically replay messages based on an offset when a failure occurs. As we detail in \Cref{ch4:sec:fault-tolerance}, the replayable input queue is necessary for Styx to produce the same sequence of transactions after the recovery is complete and to enable early commit-replies (\Cref{ch4:sub:early}). In the same way, Styx sends the result of a given transaction to an output queue from which an external system (e.g., the same REST gateway API) can receive it. Currently, Styx leverages Apache Kafka \cite{kreps2011kafka}.

\para{Durable Snapshot Store} Alongside the replayable queue, durable storage is necessary for storing the workers' snapshots. Currently, Styx uses Minio, an open-source S3 clone, to store the incremental snapshots as binary data files.

\subsection{Transaction Execution Pipeline}
Styx employs an epoch-based transactional protocol that concurrently executes a batch of transactions in each epoch. A transaction may include multiple functions that, during runtime, form a call-graph of function invocations. Each function may mutate its entity's state, and the effects of function invocations are committed to the system state in a transactional manner. In \Cref{ch4:fig:styxfunction}, once \textit{make\_reservation} enters the system, it is persisted and replicated by the input queue. Then, a worker ingests the call into its local sequencer that assigns a Transaction ID (TID) and processes all the encapsulated function calls as a single transaction. In the \textit{make\_reservation} case, the transaction consists of two functions: \textit{reserve\_hotel} and \textit{reserve\_flight}. For this example, let us assume that \textit{reserve\_hotel} is a local function call and \textit{reserve\_flight} runs on a remote worker. \textit{reserve\_hotel} will execute locally in an asynchronous fashion using coroutines and apply state changes. In contrast, \textit{reserve\_flight} 
will execute asynchronously on a remote worker, applying changes to the remote state.

\begin{figure}[t]
\centering
    \includegraphics[width=0.6\columnwidth]{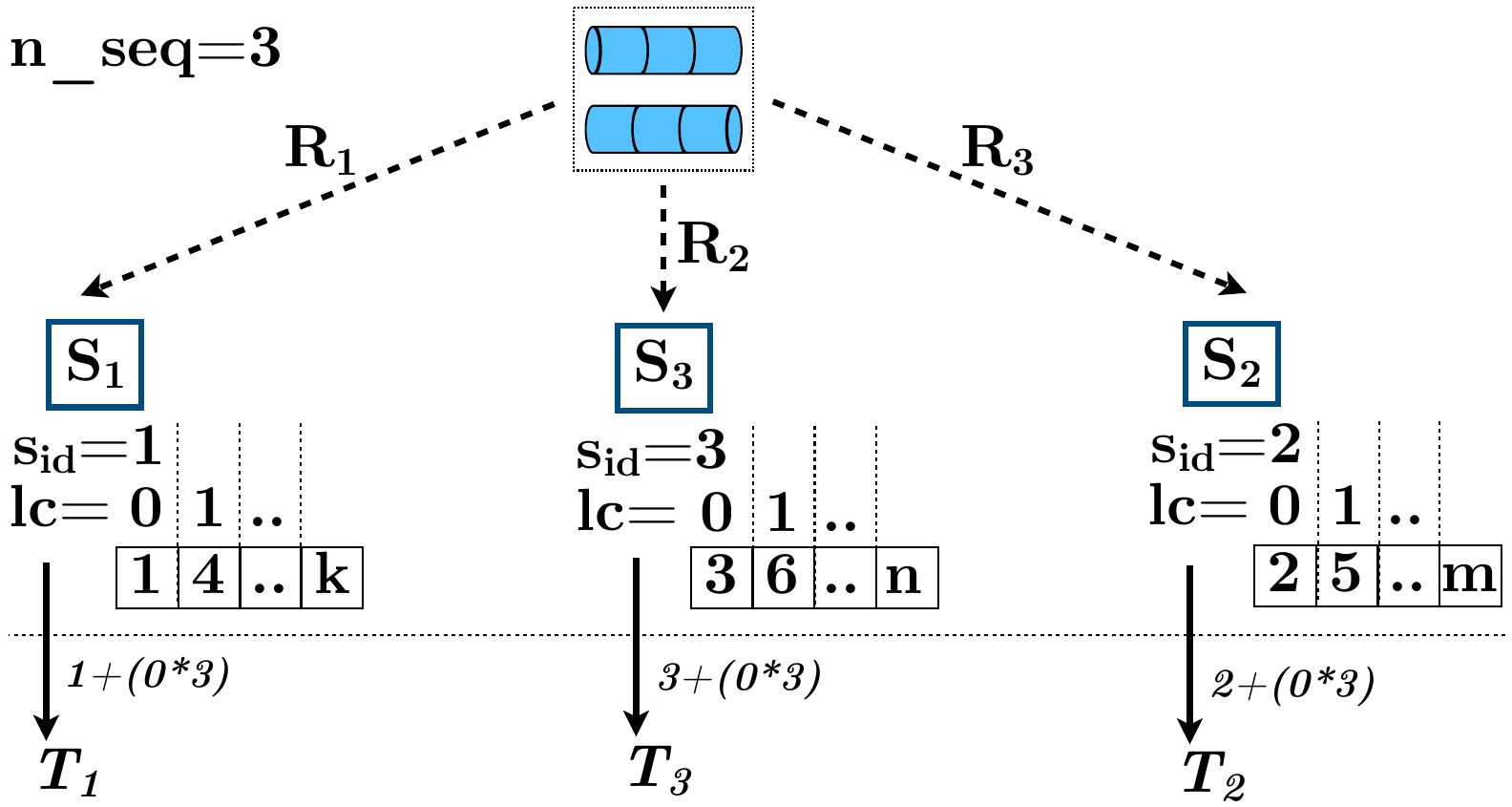}
    \caption{Example of TID assignment in Styx with three sequencers. Their identifiers $\{1, 2, 3\}$ lead to the following sequences: $S_1 = \{1, 4, ..., k\}$,  $S_2 = \{2, 5, ..., m\}$,  $S_3 = \{3, 6, ..., n\}$ following the formula expressed in \Cref{ch4:eq:seq}.}
    \label{ch4:fig:sequencer}
\end{figure}

\section{Sequencing \&\ Function Execution} \label{ch4:sec:seq_f_ex}

The deterministic execution of functions with serializable guarantees requires a sequencing step that assigns a transaction ID (TID), which, in combination with the read/write (RW) sets, can be used for conflict resolution (\Cref{ch4:sec:cmt_t}). The challenge we tackle in this section is determining the boundaries of transactions (i.e., when a transaction's execution starts and finishes), which emerges from the execution of arbitrary function call-graphs \Cref{ch4:sec:ack_shares}.

\subsection{Transaction Sequencing} \label{ch4:sec:sequencing}
In this section, we discuss the sequencing mechanism (\textbf{\circleb{1})} of Styx. Deterministic databases ensure the serializable execution of transactions by forming a global sequence. In Calvin~\cite{calvin}, the authors propose a partitioned sequencer that retrieves the global sequence by communicating across all partitions, performing a deterministic round-robin.

\para{Eliminating Sequencer Synchronization}
Instead of the original sequencer of Calvin that sends $\mathcal{O}(n^2)$ messages for the deterministic round-robin, Styx adopts a method similar to the one followed by Mencius~\cite{barcelona2008mencius}, allowing Styx to acquire a global sequence without any communication between the sequencers ($\mathcal{O}(1)$). This is achieved by having each sequencer assign unique transaction identifiers (TIDs) as follows:

\begin{equation} \label{ch4:eq:seq}
TID_{sid, lc} = sid + (lc * n\_seq)
\end{equation}
\vspace{3mm}

\noindent where $sid \in \mathbb{N}_1$ is the sequencer id assigned by the Styx coordinator in the registration phase, $lc \in \mathbb{N}_0$ is a local counter of each sequencer specifying how many TIDs it has assigned thus far and $n\_seq \in \mathbb{N}_1$ is the total number of sequencers in the Styx cluster. In the example of \Cref{ch4:fig:styxlifecycle}, the sequencers of the three workers will sequence $R_1$, $R_2$ and $R_3$ to $T_1$, $T_3$ and $T_2$ respectively. \Cref{ch4:fig:sequencer} illustrates how those TIDs are generated in parallel. Note that, conceptually, Styx implements a partitioned sequencer where the global sequence $S = \{S_1 \cup S_2 \cup \dots \cup S_n\}$ is the union of all partitioned sequences.

\para{Mitigating Sequence Imbalance} In case a single sequencer $S_1$ receives more traffic than the other sequencers, its local counter ($lc_1$) will increase more than the local counter of the rest of the sequencers. As a result, in the next epoch, sequencer $S_1$ would produce larger TIDs than the rest of the sequencers. This means that new transactions arriving at a less busy sequencer will receive higher priority for execution: transactions with higher TID receive less priority in our transactional protocol. In case of high contention in the workload, this would increase latencies for the busy ($S_1$) worker node. To avoid this, at the end of an epoch, the coordinator calculates the maximum $lc$ ($max(lc_1, lc_2, \ldots, lc_n)$) and communicates it to all workers so that they can adjust their local counter re-balancing sequences in every epoch. Balancing the workers' transaction priorities reduces the 99th percentile latency.

\para{Replication and Logging} There is no need to replicate and log the sequence within Styx since the input is logged and replicated within the replayable queue. In case of failure, after transaction replay, the sequencers will produce the exact same sequence (\cref{ch4:sec:seq_recovery}).


\subsection{Call-Graph Discovery}

After sequencing, Styx needs to execute the sequenced transactions and determine their call-graphs and RW sets (\circleb{2}). To this end, the function execution runtime ingests a given sequence of transactions to process in a given epoch. The number of transactions per epoch is either set by a polling interval or by a configurable maximum number of transactions that can run per epoch (by default, 1000 transactions per epoch). We have chosen an epoch-based approach since processing the incoming transactions in batches increases throughput.

Styx's runtime executes all the sequenced transactions on a snapshot of the data to discover the read/write sets. Transactions that span multiple workers will implicitly change the read/write sets of the remote workers via function calls. There is an additional issue related to discovering the RW set of a transaction: before the functions execute, the call-graph of the transaction is unknown. This is an issue because the protocol requires all transactions to be completed before proceeding to the next phase. To tackle this problem, Styx proposes a function acknowledgment scheme, which is explained in more detail in \Cref{ch4:sec:ack_shares}.

After this phase, all the stateful functions that comprise transactions will have finished execution, and the RW sets will be known. In \Cref{ch4:fig:styxlifecycle}, transactions $T_1$, $T_2$, and $T_3$ will execute and create the following RW sets: $Worker_1 \rightarrow \{k_1: T_1\}$, $Worker_2 \rightarrow \{k_2: T_1,T_2$ and $k_8: T_2, T_3\}$ and $Worker_3 \rightarrow \{k_3: T_3\}$.

\begin{figure}[t]
\centering
    \includegraphics[width=0.65\columnwidth]{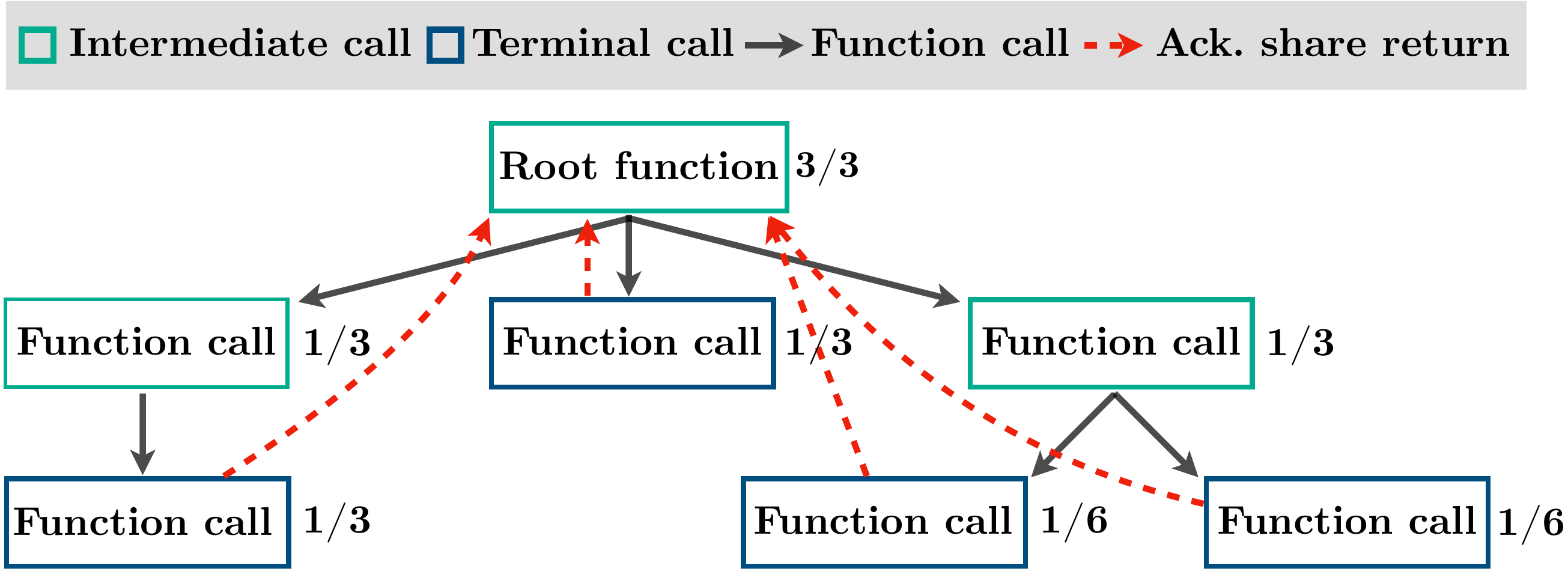}
        \caption{Asynchronous function call chains. A given root function call may invoke other functions throughout its execution. The original acknowledgment $(3/3)$ splits into parts as the function execution proceeds, and each function receives its own ack-share. For instance, in this function execution, the root function calls three other functions, thus splitting the ack-share into three equal parts. The same applies to subsequent calls, where the caller further splits their ack-share. The sum of ack-shares of terminal (blue) calls (i.e., function calls that do not perform further calls) adds to exactly $3/3$, which allows the root function to report the completion of execution.}
    \label{ch4:fig:chains}
\end{figure}

\subsection{Function Execution Acknowledgment} \label{ch4:sec:ack_shares}
In the SFaaS paradigm, the call-graph formed by a transaction is unknown; functions could be coded by different developer teams and can form complex call-graphs. This uncertainty complicates determining when a transaction has completed processing, which is essential because phase \circleb{3} can only start after all transactions have finished processing. To that end, each asynchronous function call of a given transaction is assigned an \textit{ack\_share}. A given function knows how many shares to create by counting the number of asynchronous function calls during its runtime. The caller function then sends the respective acknowledgment shares to the downstream functions. For instance, in \Cref{ch4:fig:chains}, the transaction entry-point (root of the tree) calls three remote functions, splitting the ack\_share into three parts (3 x $1\//3$). The left-most function invokes only one other function and passes to it its complete ack\_share ($1\//3$). The middle function does not call any functions, so it returns the share to the root function when it completes execution, and the right-most function calls two other functions, splitting its share ($1\//3$) to 2 x $1\//6$. After all the function calls are complete, the root function should have collected all the shares. When the sum of the received shares adds to 1, the root/entry-point function can safely deduce that the execution of the complete transaction is complete.

This design is devised for two reasons: i)~if every participating function just sent an ack when it is done, the root would not know how many acks to expect to decide whether the entire execution has finished, and ii) if we used floats instead of fractions, we could stumble upon a challenge related to adding floating point numbers.
For instance, if we consider floating-point numbers in the example mentioned above of the three function calls, the sum of all shares would not equal 1, but 0.99, since each share contributes 0.33. Subsequently, we cannot accurately round inexact division numbers; therefore, Styx uses fraction mathematics instead.

A solution close to the \textit{ack\_share} is distributed futures~\cite{distributedfutures}. However, it would not work in the SFaaS context as it either requires information about the entire call-graph for it to work asynchronously, or it would need to create a chain of futures that would make it synchronous. Hence, it would introduce high latency for our use case.

\section{Committing Transactions} \label{ch4:sec:cmt_t}
After completing an epoch's call-graph discovery, Styx needs to determine which transactions will commit and which will abort based on the transactions' Read/Write (RW) sets and TIDs. To this end, this section presents two different commit phases: $i)$ an optimistic lock-free phase that commits only the non-conflicting transactions, and $ii)$ a lock-based phase that only commits the transactions that were not able to commit in the first phase. The lock-based commit phase commits all conflicting transactions by acquiring locks in a TID-ordered sequence. To make the second phase faster, we have devised a caching scheme that can reuse the already discovered call-graph to avoid re-executing long function chains whenever possible (\cref{ch4:sec:caching}).

\subsection{Lock-free Commit Phase}

In case of conflict (i.e., a transaction $t$ writes a key that another transaction $t'$ also reads or writes on), similarly to~\cite{aria}, only the transaction with the lowest transaction ID will succeed to commit (\circleb{3}). The transactions that have not been committed are put in a queue to be executed in the next phase \circleb{4}  (maintaining their previously assigned ID).

In addition, workers ($W$) send their local conflicts to every other worker through the coordinator ($2*|W|$ messages): this way, every worker retains a global view of all the aborted/rescheduled transactions and can decide, locally, which transactions can be committed. Finally, note that transactions can also abort, not because of conflicts, but due to application logic causes (e.g., by throwing an exception). In that case, Styx removes the related entries from the read/write sets to reduce possible conflicts further.

In this phase, all the transactions that have not been part of a conflict apply their writes to the state, commit, and reply to the clients. In the example shown in \Cref{ch4:fig:styxlifecycle}, only $T_1$ can commit in $W_1$ and $W_2$ due to conflicts in the RW sets of $W_2$ regarding $T_2$ and $T_3$; more specifically, at keys $k_2$ and $k_8$.

\subsection{Lock-based Commit Phase}

In the previous phase, \circleb{3}, only transactions without conflicts can be committed. We now explain how Styx deals with transactions that have not been committed in a given epoch due to conflicts (\circleb{4}). First, Styx acquires locks in a given sequence ordered by transaction ID. Then, it reruns all transactions concurrently since all the read/write sets are known and commits them. However, if a transaction's read/write set changes in this phase, Styx aborts the transaction and recomputes its read/write set in the next epoch. Now, in \Cref{ch4:fig:styxlifecycle}, $W_2$ can sequentially acquire locks for $T_2$ and $T_3$, leading to their commits in $W_2$ and $W_3$.

\begin{figure}[t]
\centering
    \includegraphics[width=0.65\columnwidth]{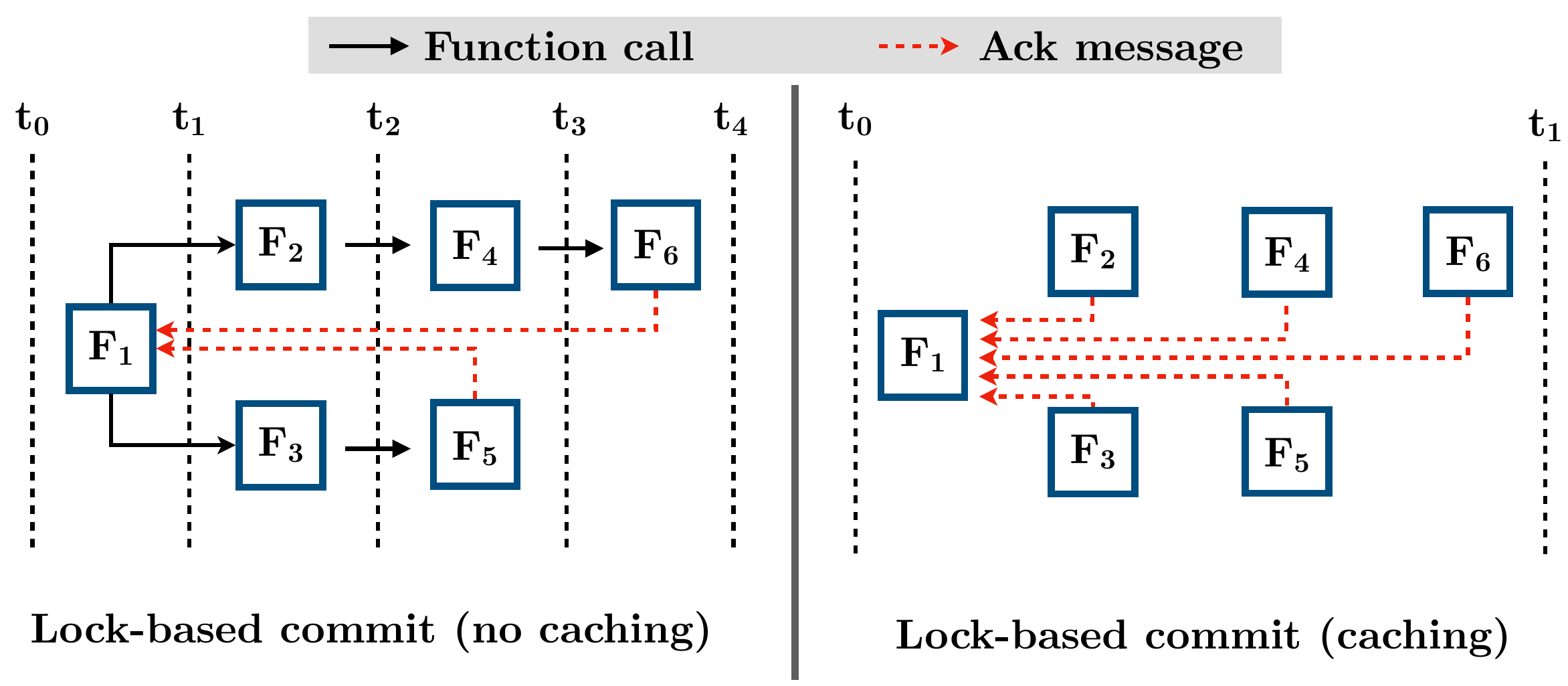}
    \caption{If no function caching is performed (left), the transaction execution will execute a deep call-graph; the messages will be sent sequentially and be equal to the number of function calls (5) in addition to the acks (2). Styx's function caching optimization (right) will lead to a concurrent function execution in the lock-based commit phase, between $t_0$ and $t_1$, and send only five acks asynchronously.}
    \label{ch4:fig:function_caching}
\end{figure}

\subsection{Call-Graph Caching} \label{ch4:sec:caching}
As depicted in \Cref{ch4:fig:styxlifecycle}, the lock-based commit phase \circleb{4} is used to execute any transactions that did not commit during the lock-free commit phase \circleb{3}. By the time the lock-based commit phase starts, the state of the database may have changed since the lock-free commit. As a result, function invocations need to be re-executed to account for the data updates. 

On the left part of \Cref{ch4:fig:function_caching}, we depict such a function invocation. At time $t_0$, F$_1$ is invoked, which in turn invokes two function chains: $F_1 \rightarrow F_2 \rightarrow F_4 \rightarrow F_6$ and $F_1 \rightarrow F_3 \rightarrow F_5$. Once the two function chains finish their execution (on time $t_4$ and $t_3$ respectively), they can acknowledge their termination to the root call $F_1$. 

\para{Potential for Caching} During our early experiments, we noticed cases where $F_1$ is invoked and the parameters with which it calls $F_2$ (and in turn the invocations across the $F_1 \rightarrow ... \rightarrow F_6$ call chain) do not change. The same applied to the RW set of those function invocations; the RW sets remained unchanged. Since Styx tracks those call parameters as well as the functions' RW sets, it can cache input parameters during the lock-free commit phase and reuse them during the lock-based commit, avoiding long sequential re-executions along the call chains. This case is depicted on the right part of \Cref{ch4:fig:function_caching}: the function-call chain does not need to be invoked in a sequential manner from $F_1$ all the way to $F_6$, leading to high latency. Instead, the individual workers can re-invoke those function calls locally and concurrently. As a result, all functions can execute in parallel and save on latency and network overhead ($t_4 - t_1$ in \Cref{ch4:fig:function_caching}).
Furthermore, caching does not require user input, is transparent to the API, and does not depend on the synchronous or asynchronous specification. Nonetheless, synchronous calls can be automatically transformed into asynchronous ones under certain conditions~\cite{beillahi2022automated, stateflow}.

\para{Conditions for Parallel Function Re-invocation} Intuitively, if the parameters with which, e.g., $F_2$ is called, and the RW set of $F_2$ remains the same, we can safely assume that function $F_2$ can be invoked concurrently without having to be invoked sequentially by $F_1$. If those functions are successfully completed and acknowledge their completion to the root function $F_1$, it means that the transaction can be committed. To the contrary, if the RW set of any of the functions $F_1 - F_6$ changes, or the parameters of any of the functions along the call chains change, the transaction must be fully re-executed. In that case, Styx will have to reschedule that transaction in the next epoch.

\subsection{Early Commit Replies via Determinism} \label{ch4:sub:early}

Implementing Styx as a fully deterministic dataflow system offers a set of advantages involving the ability to communicate transaction commits to external systems (e.g., the client), even before the state snapshots are persisted to durable storage. A traditional transactional system can respond to the client only when $i)$ the requested transaction has been committed to a persistent, durable state or $ii)$~the write-ahead log is flushed and replicated. In Styx's case, that would mean when an asynchronous snapshot completes (i.e., is persisted to durable storage such as S3), leading to high latency.

Since Styx implements a deterministic transactional protocol that executes an agreed-upon sequence of transactions among the workers, after a failure, the system would run the same transactions with exactly the same effects. This determinism enables Styx to give early commit replies: \emph{the client can receive the reply even before a persistent snapshot is stored.}
The assumption here is that the input queue, persisting the client requests, will provide Styx's sequencers with the requests in the same order after replay, a guarantee that is typically provided by most modern message brokers. 
Performing state mutations and message passing before persistence has also been explored in DARQ's speculative execution~\cite{darq}.

\section{Fault Tolerance}
\label{ch4:sec:fault-tolerance}

Styx implements a coarse-grained fault tolerance mechanism. Instead of logging each function execution, it adopts a variant of existing checkpointing mechanisms used in streaming dataflow systems \cite{silvestre2021clonos,carbone2017state,chandy1985distributed}. Styx asynchronously snapshots state and stores it in a replicated fault-tolerant blob store (e.g., Minio / S3), enabling low-latency function execution. We describe Styx's fault tolerance mechanism below.

\begin{algorithm}[t]
    \footnotesize
    \DontPrintSemicolon
    \SetAlgoLined
    
    \KwResult{Compacted Snapshot stored in durable storage}
    \SetKwInOut{Input}{Input}\SetKwInOut{Output}{Output}
    \SetKwComment{comm}{\hfill$\triangleright$\ }{}
    \Input{$\delta$: Delta changes, $O_{input}$: Input offset, $O_{output}$: Output offset, $E_{count}$: Epoch count, $SEQ_{count}$: Sequence count}
    \Output{$\mathcal{S}$: Compacted snapshot}
    \BlankLine
    
    \If{snapshotInterval}{
        state $\leftarrow \delta$ \comm*[r]{Prepare data and metadata for snapshot}
        metadata $\leftarrow \{O_{input}, O_{output}, E_{count}, SEQ_{count} \}$\;
        $\mathcal{S^{\delta}} \leftarrow$ serialize(state, metadata)\;
        store $\mathcal{S^{\delta}}$\; 
        inform coordinator\;
    }
    \If{compactionInterval}{
        $\mathcal{S} \leftarrow \emptyset$\;
        \ForEach{$\mathcal{S^{\delta}}$}{
            $\mathcal{S} \leftarrow \text{compact}( \mathcal{S}, \mathcal{S^{\delta}})$ \comm*[r]{Compact delta snapshots}
        }
    }
    
    \caption{Snapshotting Mechanism}
    \label{ch4:algo:snapshot}
\end{algorithm}

\begin{algorithm}[t]
    \footnotesize
    \DontPrintSemicolon
    \SetAlgoLined
    
    \KwResult{Recovered state from durable storage, possible duplicate messages}
    \SetKwInOut{Input}{Input}\SetKwInOut{Output}{Output}
    \SetKwComment{comm}{\hfill$\triangleright$\ }{}
    \Input{
        $\mathcal{S}$: Latest compacted snapshot,\\
        $\mathcal{S^{\delta}}$: Incremental (delta) snapshots,\\
        $O^{last}_{output}$: Offset of last output,\\
    }
    \Output{
        $\mathcal{R}$: Set of possible duplicate messages, $state^s$: Snapshotted state, $O^s_{input}$: Snapshotted input offset, $O^s_{output}$: Snapshotted output offset, $ E^s_{count}$: Snapshotted epoch count, $SEQ^s_{count}$: Snapshotted sequencer count
    }
    \BlankLine
    \If{$\mathcal{S^{\delta}} \neq \emptyset$}{
        $\mathcal{S} \leftarrow \text{compact}(\mathcal{S},\mathcal{S^{\delta}})$ \comm*[r]{Compact delta snapshots, if any}
    }
    $state^s, O^s_{input}, O^s_{output}, E^s_{count}, SEQ^s_{count} \leftarrow$ deserialize $\mathcal{S'}$ \; \comm*[r]{Extract persisted state}
    $R \leftarrow \{ m \mid O^s_{\text{output}} \leq m \leq O^{last}_{\text{output}} \}$ \comm*[r]{Possible duplicates (\Cref{ch4:sec:exactly-once-output})}

    \caption{Recovery Mechanism}
    \label{ch4:algo:recovery}
\end{algorithm}

\subsection{Incremental Snapshots \&\ Recovery} \label{ch4:sec:incr_sn_rec}

The snapshotting mechanism of Styx resembles the approach of many streaming systems\cite{carbone2017state, jet, SilvaZD16, ArmbrustDT18}, that extend the seminal Chandy-Lamport snapshots \cite{chandy1985distributed}. Modern stream processing systems checkpoint their state by receiving snapshot barriers at regular time intervals (epochs) decided by the coordinator. In contrast, Styx leverages an important observation: workers do not need to wait for a barrier to enter the system to take a snapshot, since the natural barrier in a transactional epoch-based system like Styx is at the end of a transaction epoch.

\para{Snapshotting} To this end, instead of taking snapshots periodically by propagating markers across the system's operators, Styx aligns snapshots with the completion of transaction epochs to take a consistent cut of the system's distributed state, including the state of the latest committed transactions, the offsets of the message broker, and the sequencer counters ($lc$).
The minimal information included in the snapshot is $O(N + c)$ where $N$ is the number of updates affecting the delta map, and $c$ is the fixed number of integers stored for the Kafka offsets and the sequencer variables.

When the snapshot interval triggers, Styx makes a copy of the current state changes to a parallel thread and persists incremental snapshots asynchronously, allowing Styx to continue processing incoming transactions while the snapshot operation is performed in the background.
The snapshotting procedure is described in Algorithm \ref{ch4:algo:snapshot}.

\para{Recovery} In case of a system failure, Styx $i)$ rolls back to the epoch of the latest completed snapshot, $ii)$ loads the snapshotted state, $iii)$ rolls back the replayable source's topic partitions (that are aligned with the Styx operator partitions) to the offsets at the time of the snapshot, $iv)$ loads the sequencer counters, and finally, $v)$ verifies that the cluster is healthy before executing a new epoch. The recovery procedure is described in Algorithm \ref{ch4:algo:recovery}.

\para{Incremental Snapshots \& Compaction} Each snapshot stores a collection of state changes in the form of \textit{delta maps}. A delta map is a hash table that tracks the changes in a worker's state in a given snapshot interval. When a snapshot is taken, only the delta map containing the state changes of the current interval is snapshotted. To avoid tracking changes across delta maps, Styx periodically performs compactions where the deltas are merged in the background, as shown in \Cref{ch4:fig:snapshot}. The cost of compacting is equivalent to the cost of merging two hashmaps with the same key-spaces $(O(N))$. The total cost will be $O(M*N)$, with $M$ denoting the number of deltamaps we need to compact.

\begin{figure}[t]
\centering
\captionsetup{justification=centering}
    \includegraphics[width=0.6\columnwidth]{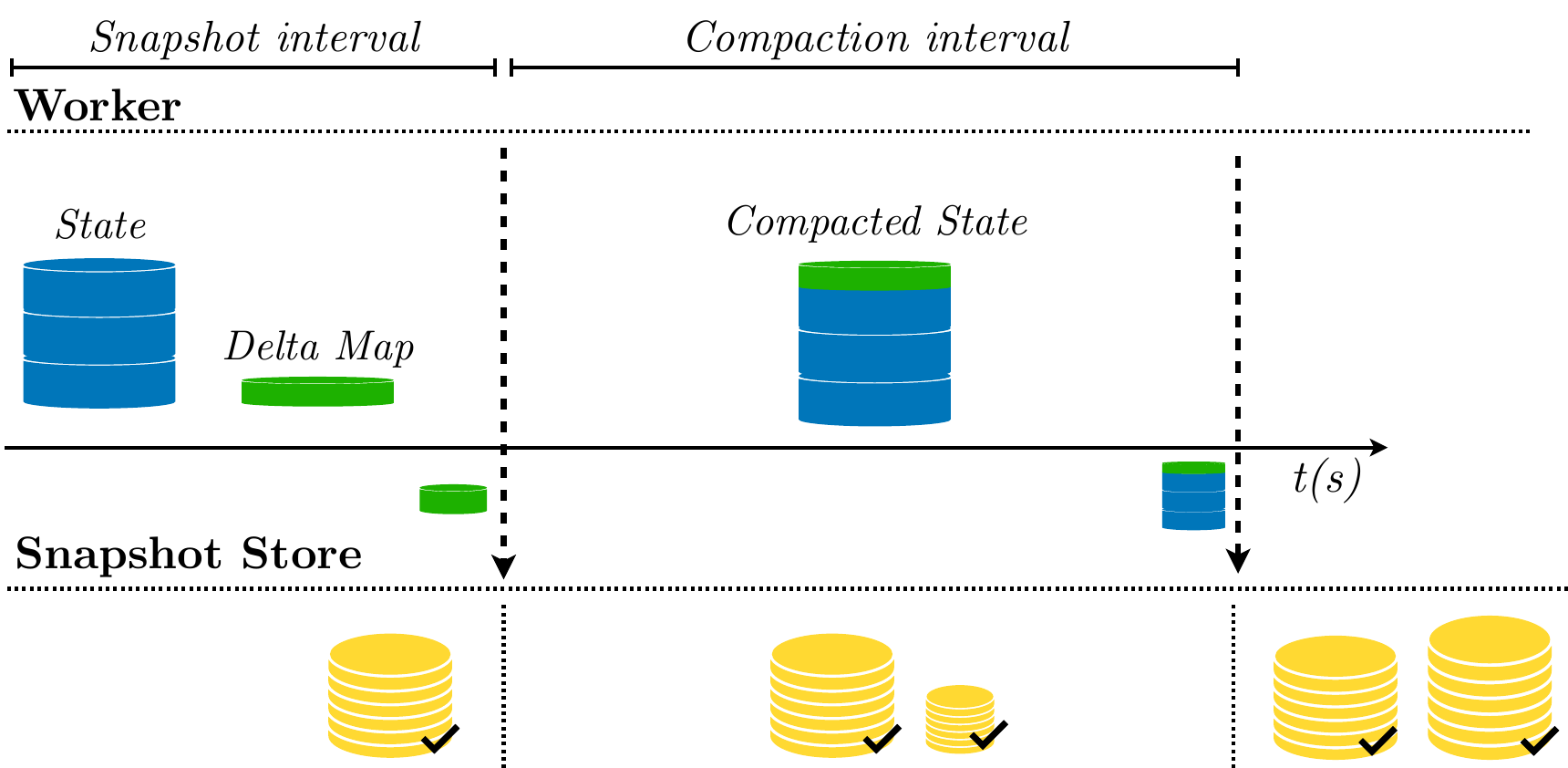}
    \caption{Incremental snapshots with Delta Maps in Styx.}
    \label{ch4:fig:snapshot}
\end{figure}

\subsection{Sequencer Recovery} \label{ch4:sec:seq_recovery}
To guarantee determinism, upon recovery, Styx 's sequencer needs to generate identical sequences as the ones generated between the latest snapshot and failure. The recovery protocol of the sequencer operates as follows: At first, during the snapshot, we store the local counter of each sequencer partition ($lc$) with its ID ($sid$) and the epoch counter. Additionally, at the start of each epoch, Styx  logs the number of transactions contained in that epoch, denoted as epoch size. Logging the epoch sizes is needed due to Styx 's varying epoch sizes and the sequencer rebalancing scheme (\cref{ch4:sec:sequencing}). After failure, the recovered sequencer partitions are initialized with the snapshot's $lc$ and $sid$. Afterward, each partition retrieves from its log all the sizes of all epochs executed since the last snapshot. Finally, after recovery, the sequencer matches the epoch sizes to the ones recorded by the log, leading to the same global sequence observed before failure.

\subsection{Exactly-Once Processing}

At first, the durable input queue, which acts as a replayable source, allows Styx to replay requests after failures. By rolling back the queue partitions (aligned with system operator partitions) to the respective offsets as recorded in the latest snapshot, Styx can reprocess only the transactions whose state changes are not reflected yet in the snapshot. Transactions committed and early-commit replies stored in the egress can be deduplicated (\Cref{ch4:sec:exactly-once-output}).
 
Styx runs each transaction to its completion in a single epoch. A given transaction can execute a large call-graph of functions that can affect the state. If a failure takes place, a transaction's state effects are restored to the latest snapshot, and the complete transaction is re-executed. As a result, no special attention is required to ensure that remote function calls are executed exactly-once, except for resetting all TCP channels between Styx's workers after recovery.

\begin{lemma}\label{ch4:lem:exactlyonce}
The state mutations of committed transactions in Styx are reflected exactly-once, even upon failure.
\end{lemma}

\begin{proof}

Let $S_t$ denote the state of the system at time $t$. $Q_t = \{r_1, \dots, r_n\}$ denotes the durable input queue at time $t$ that holds all requests $r_i$ to be processed. We assume that the input queue operates as FIFO and requests $r_i$ are deterministic.
Each $r_i$ will be sequenced as a transaction $T_i = \{upd_l, func_m\} $ by a deterministic sequencer, where $upd_l$ are the state updates and $func_m$ are the function calls of the transaction. We assume that $upd_l$  happens atomically and $func_m$ are also reflected once, given the use of a reliable communication protocol. Given the same initial state $S$ and input from $Q$, it always produces the same state transition $S \to S'$,which means $S'_{t+1} = mutation(S_t, Q_t)$. The execution of a transaction $T_i$ is deterministic.

At any time $t$, the state of the system $S_t$ reflects all transactions in $Q_t$ that have been fully executed and committed. Accordingly, the state $S_t$ ignores partially executed or in-progress transactions in $Q_t$. We denote the latest durable snapshot taken up to time $t$, as $\text{Snapshot}(S_t, i, n)$ where $n$ corresponds to the offsets of the first request $r_i$, and last request $r_n$ of the input queue to be processed up to time $t$. Upon failure, a subset of $Q_t$, $Q^{success}_t = \{r_1, \dots, r_k\}$ will contain successfully committed transactions and a subset $Q^{fail}_t = \{r_{k+1}, \dots, r_n\}$ will contain aborted transactions such that $Q_t = Q^{success}_t + Q^{fail}_t$. In order to recover from a failure, $Q_t$ is rolled back to $S_t$ from
$\text{Snapshot}(S_t, i, n)$ as we persist the offsets of our input queue. Transactions in $Q_t$ are replayed in the original order from offset $i$ to offset $n$ of our input queue. This is ensured by the FIFO queue and the deterministic sequencer. After processing the input transactions, $Q^{success}_t$ includes requests already reflected in $\text{Snapshot}(S_t)$, and $Q^{fail}_t$ includes pending requests. Since $\text{Snapshot}(S_t)$ reflects $Q^{success}_t$ and $Q_t = Q^{success}_t + Q^{fail}_t$, the replay and processing ensure: $S_{t+1}'' = mutation(S_t, Q^{fail}_t) = S_{t+1}'$. Thus, the effects of all transactions will be reflected in the state exactly-once, even after failure.
\end{proof}

\subsection{Exactly-Once Output} 
\label{ch4:sec:exactly-once-output}
A common challenge in the fault tolerance of streaming systems is that of the exactly once output~\cite{ElnozahyAW02, fragkoulis2024survey} in the presence of failures, which is hard to solve for low-latency use cases. For example, in Apache Flink's~\cite{flink} exactly-once output configuration, clients can only retrieve responses after those are persisted in a snapshot or a transactional sink. This arrangement is sufficient for streaming analytics but not for low-latency transactional workloads, as discussed previously in \Cref{ch4:sub:early}. 

To solve that, during recovery, Styx: $i)$ reads the last offset of the egress topic, $ii)$ compares it with the output offset persisted in the snapshot, determining for which transactions the clients have already received replies, $iii)$ retrieves the TIDs attached in those replies, and $iv)$ does not send a reply again to the egress topic for those transactions. Note that this deduplication strategy is based on the fact that TIDs have been assigned deterministically.

\subsection{Addressing Non-Deterministic Functions} \label{ch4:sec:addr_non_det}
As discussed in \Cref{ch4:sec:incr_sn_rec}, Styx's recovery mechanism is based on deterministic replay. To this end, Styx requires that the functions authored by developers are also deterministic, i.e., replaying the same function multiple times, using the same inputs and database state, should yield the same results. However, one can achieve determinism even in the presence of non-deterministic logic inside functions, such as randomness (e.g., random numbers/sampling) or calls to external systems (e.g., calling an external database or API). Styx can follow the approach of existing systems (e.g., Temporal~\cite{temporal}, Clonos~\cite{silvestre2021clonos}). In the following, we explain how this can be achieved.

\para{Randomness} To retain determinism in the case of randomness, Styx can use an external fault tolerant write-ahead log (WAL) to log the random number along with the TID. Thus, in the case of failure and replay, Styx can use the logged random number, essentially making the function call deterministic during replay.

\para{Calls to External Systems} As illustrated in \Cref{ch4:fig:externalcalls}, an interaction with an external system needs to consider three critical points to maintain determinism. Styx assumes that the external system supports idempotency~\cite{ietf}, meaning that if a call is made twice with the same idempotency key, the effects on the external system's state and its return value will remain the same. In \circler{1} Styx needs to log the idempotency key and the TID in the WAL before calling the external system. If the external system produces a response (\circler{2}), Styx can store it in the WAL and retrieve it from there in case of replay. Finally, when Styx completes a snapshot (\circler{3}), it can also clear the WAL for garbage collection since the prior entries are not needed.

Finally, Styx could mask those operations behind an API that exposes the following functionality, such as \textit{system\_x.random} for random number generation and \textit{system\_x.call\_external} for external system calls.

\begin{figure}[t]
    \centering
    \captionsetup{justification=centering}
    \includegraphics[width=0.6\columnwidth]{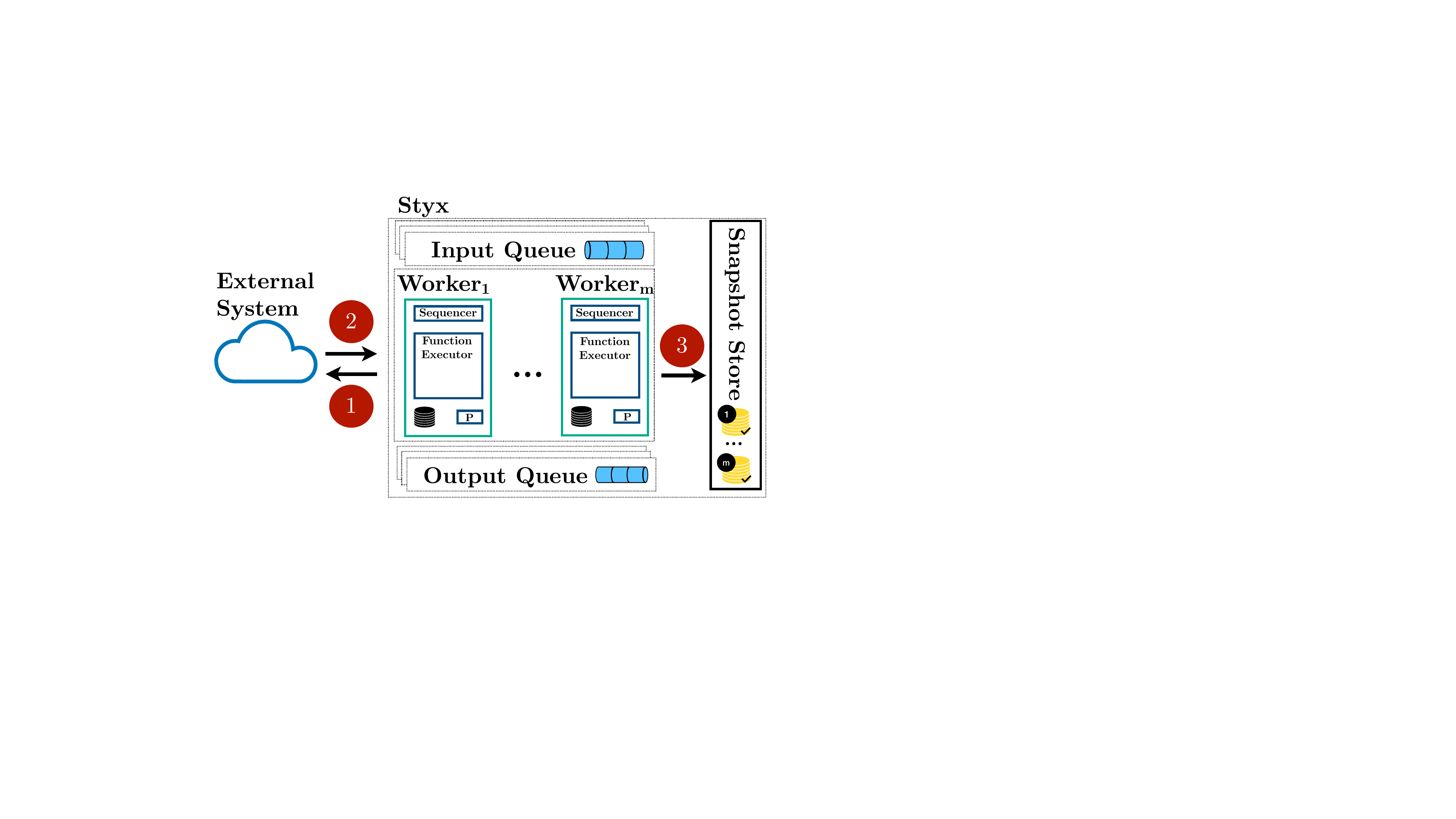}
    \caption{External system call critical points and Styx.}
    \label{ch4:fig:externalcalls}
\end{figure}

\section{Evaluation} \label{ch4:sec:exp}

We evaluate Styx by answering the following questions:
\begin{itemize}
  \renewcommand\labelitemi{--}
  \setlength\itemsep{0mm}
  \item (\Cref{ch4:sec:exp_lat_through}) How does Styx compare to State-of-the-Art serializable transactional SFaaS systems? 
  \item (\Cref{ch4:sec:exp_skew}) How does Styx perform under skewed workload? 
  \item (\Cref{ch4:sec:exp_scalability}) How well does Styx scale? 
  \item (\Cref{ch4:sec:exp:snapshots}) Does the snapshotting mechanism affect performance? 
\end{itemize}

\subsection{Setup}
\label{ch4:sec:exp:setup}

\para{Systems Under Test} In the evaluation, we include SFaaS systems that provide serializable transactional guarantees. Those are:

\parait{Beldi~\cite{beldi}/Boki~\cite{boki}} Both systems use a variant of two-phase commit and Nightcore~\cite{nightcore} as their function runtime and store their data in DynamoDB. Additionally, Boki is deployed with the latest improvements of Halfmoon \cite{qi2023halfmoon}.

\parait{T-Statefun~\cite{tstatefun}} T-Statefun maintains the state and the coordination of the two-phase commit protocol within an Apache Flink cluster and ships the relevant state to remote stateless functions for execution. For fault tolerance, it relies on a RocksDB state backend that performs incremental snapshots. 

\parait{Styx} Styx is implemented in Python 3.12 and uses coroutines to enable asynchronous concurrent execution. Apache Kafka is used as an ingress/egress and Minio/S3 \cite{MinIO} as a remote persistent store for Styx's incremental snapshots. Finally, Styx is a standalone containerized system that works on top of Docker and Kubernetes for ease of deployment.

\para{Workloads/Benchmarks} \Cref{ch4:tbl:scenaria} summarizes the three workloads used in the experiments. 

\parait{YSCB-T~\cite{dey2014ycsb+}} We use a variant of YCSB-T~\cite{dey2014ycsb+}  where each transaction consists of two reads and two writes. The concrete scenario is as follows: First, we create 10.000 bank accounts (keys) and perform transactions in which a debtor attempts to transfer credit to a creditor. This transfer is subject to a check on whether the debtor has sufficient credit to fulfill the payment. If not, a rollback needs to be performed. The selection of a relatively small number of keys is deliberate: we want to assess the systems' ability to sustain transactions under high contention. In addition, for the experiment depicted in \Cref{ch4:fig:zipf_styx_ts} (skewed distribution), we select the debtor key based on a uniform distribution and the creditor based on a Zipfian distribution, where we can vary the level of contention by modifying the Zipfian coefficient.

\parait{Deathstar~\cite{deathstar}} We employ Deathstar~\cite{deathstar}, as adapted to SFaaS workloads by the authors of Beldi~\cite{beldi}. It consists of two workloads: $i)$~the Movie workload implements a movie review service where users write reviews about movies. $ii)$~the Travel workload implements a travel reservation service where users search for hotels and flights, sort them by price/distance/rate, find recommendations, and transactionally reserve hotel rooms and flights. Both Deathstar workloads follow a uniform distribution. Note that T-Statefun could not run in this set of experiments since it does not support range queries.

\parait{TPC-C~\cite{tpcc}} The prime transactional benchmark targeting OLTP systems is TPC-C~\cite{tpcc}. In our evaluation, we used the NewOrder and Payment transactions and had to rewrite them in the SFaaS paradigm, splitting the NewOrder transaction into 20-50 function calls (one call per item) and the Payment transaction into 8 function calls. TPC-C scales in size/partitions by increasing the number of warehouses represented in the benchmark. While a single warehouse represents a skewed workload (all transactions will hit the same warehouse), increasing the number of warehouses decreases the contention, allowing for higher throughput and lower latency. Note that the TPC-C experiments do not include Beldi, Boki, or T-Statefun because they do not support them.

\begin{table}[t]
\centering
\resizebox{0.65\columnwidth}{!}{
    \begin{tabular}{l||c|c|c}
    \textbf{Scenario} & \textbf{\#keys} & \textbf{Function Calls} & \textbf{Transactions} \% \\ \hline \hline
    \textbf{YCSB-T} & 10k & 2 & 100\% \\ \hline
    \textbf{Deathstar Movie} & 2k & 9-10 & 0\% \\ \hline
    \textbf{Deathstar Travel} & 2k & 3 & 0.5\% \\ \hline
    \textbf{TPC-C} & 1m-100m & 8 / 20-50 & 100\%
    \end{tabular}
    }
    \vspace{2mm}
    \caption{Workload characteristics.}
    \label{ch4:tbl:scenaria}
\end{table}

 \para{Resources} For Beldi/Boki, T-Statefun and Styx, we assigned a total of 112 CPUs with 2GBs of RAM per CPU, matching what is presented in the original Boki paper \cite{boki}.  Additionally, throughout all the evaluation scenarios, the data fit in memory across all systems. Unless stated otherwise, Styx and T-Statefun are configured to perform incremental snapshots every 10 seconds. All external systems, i.e., DynamoDB (Beldi, Boki),  Minio, and Kafka (Styx, T-Statefun), are configured with three replicas for fault tolerance.

\parait{External Systems} Boki and Beldi use a fully managed DynamoDB instance at AWS, which does not state the amount of resources it occupies and is in addition to the 112 CPUs assigned to Boki and Beldi. Similarly, the resources assigned to Minio/S3 (Styx and T-Statefun) are not accounted for.

\para{Metrics} Our goal is to observe systems' behavior, measured by their latency, while varying the input throughput.

\noindent\underline{\textit{Input throughput}} represents the number of transactions submitted per second to the system under test. As the input throughput increases during an experiment, we expect the latency of individual transactions to increase until aborts start to manifest due to contention or high load.

\noindent\underline{\textit{Latency}} represents the time interval between submitting a transaction and the reported time when the transaction is committed/aborted. In Styx and T-Statefun, the latency timer starts when a transaction is submitted in the input queue (Kafka) and stops when the system reports the transaction as committed/aborted in the output queue. Similarly, in Beldi and Boki, the latency is the time since the input gateway has received a transaction and the time that the gateway reports that the transaction has been committed/aborted.

\begin{figure*}[t]
    \centering
    \begin{subfigure}[t]{0.49\columnwidth}
        \centering
        \captionsetup{justification=centering}
        \includegraphics[width=\columnwidth]{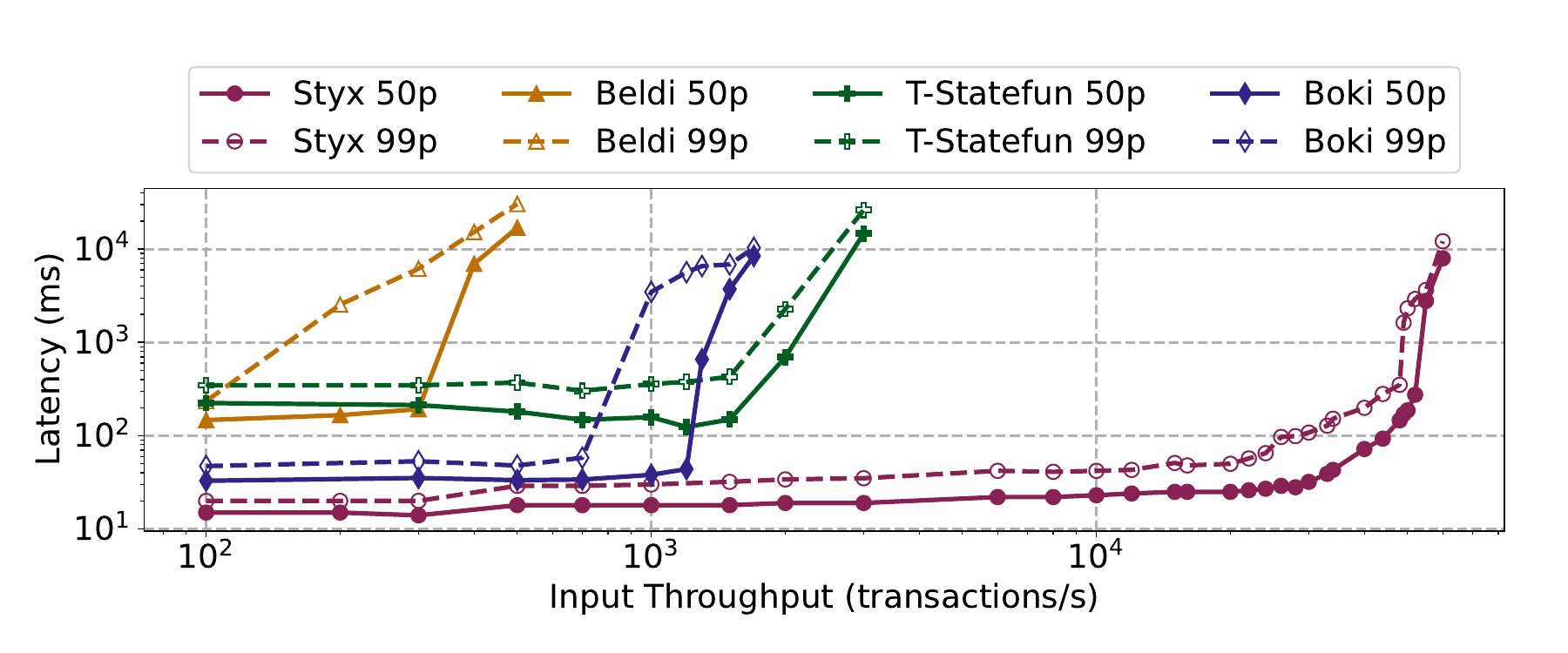}
        \caption[]%
        {{\footnotesize YCSB-T (uniform).}}
        \label{ch4:fig:tl_ycsbt}
    \end{subfigure}
    \hfill
    \begin{subfigure}[t]{0.49\columnwidth}  
        \centering 
        \captionsetup{justification=centering}
        \includegraphics[width=\columnwidth]{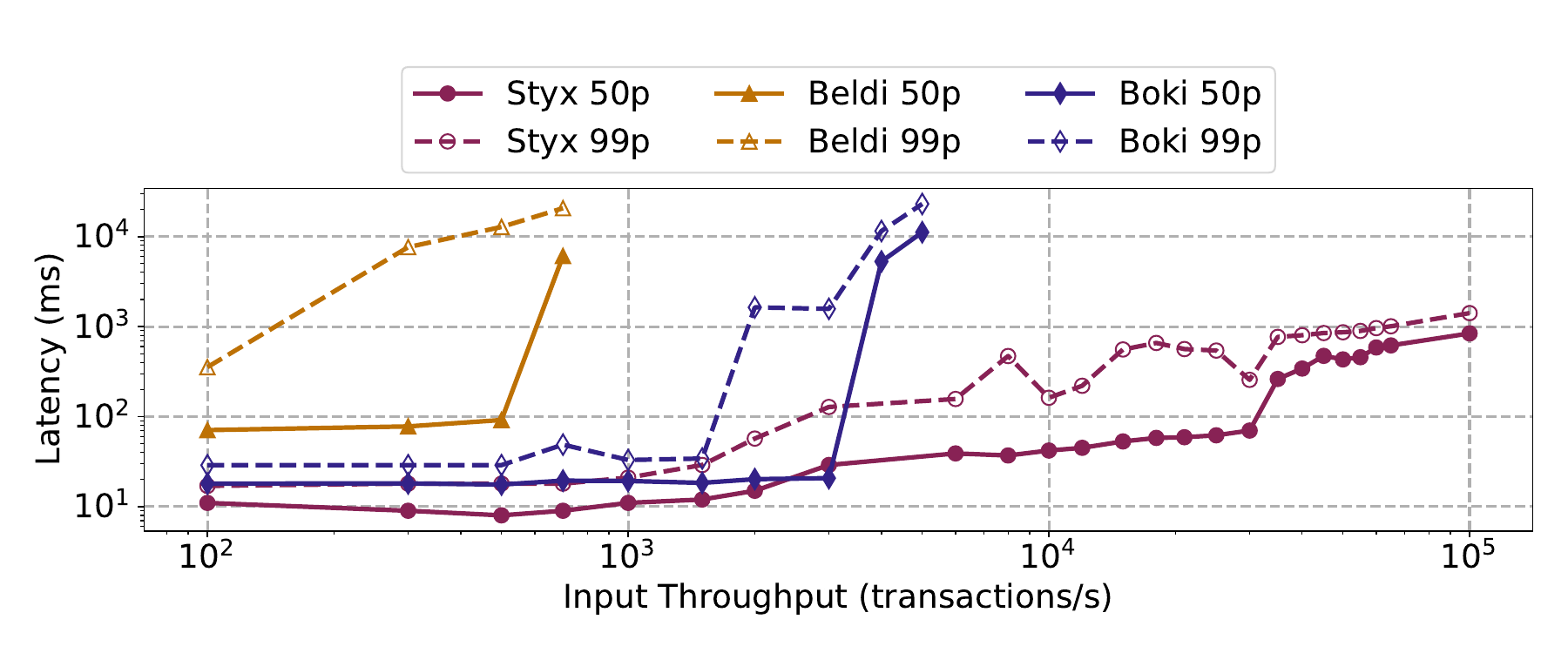}
        \caption[]%
        {{\footnotesize Deathstar Travel Reservation.}}    
        \label{ch4:fig:death_trav}
    \end{subfigure}
    \vskip\baselineskip
    \begin{subfigure}[t]{0.49\columnwidth}   
        \centering 
        \captionsetup{justification=centering}
        \includegraphics[width=\columnwidth]{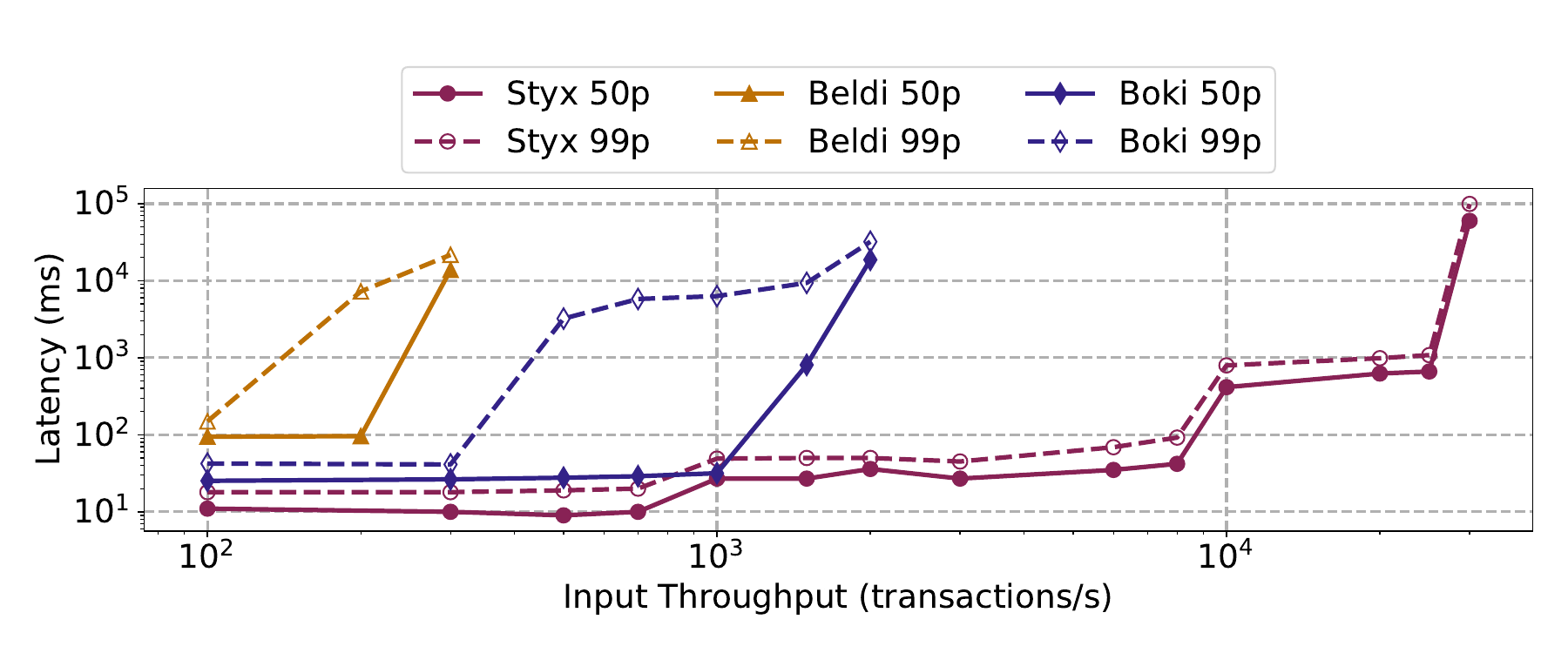}
        \caption[]%
        {{\footnotesize Deathstar Movie Review.}}    
        \label{ch4:fig:death_mov}
    \end{subfigure}
    \hfill
    \begin{subfigure}[t]{0.49\columnwidth}   
        \centering 
        \captionsetup{justification=centering}
        \includegraphics[width=\columnwidth]{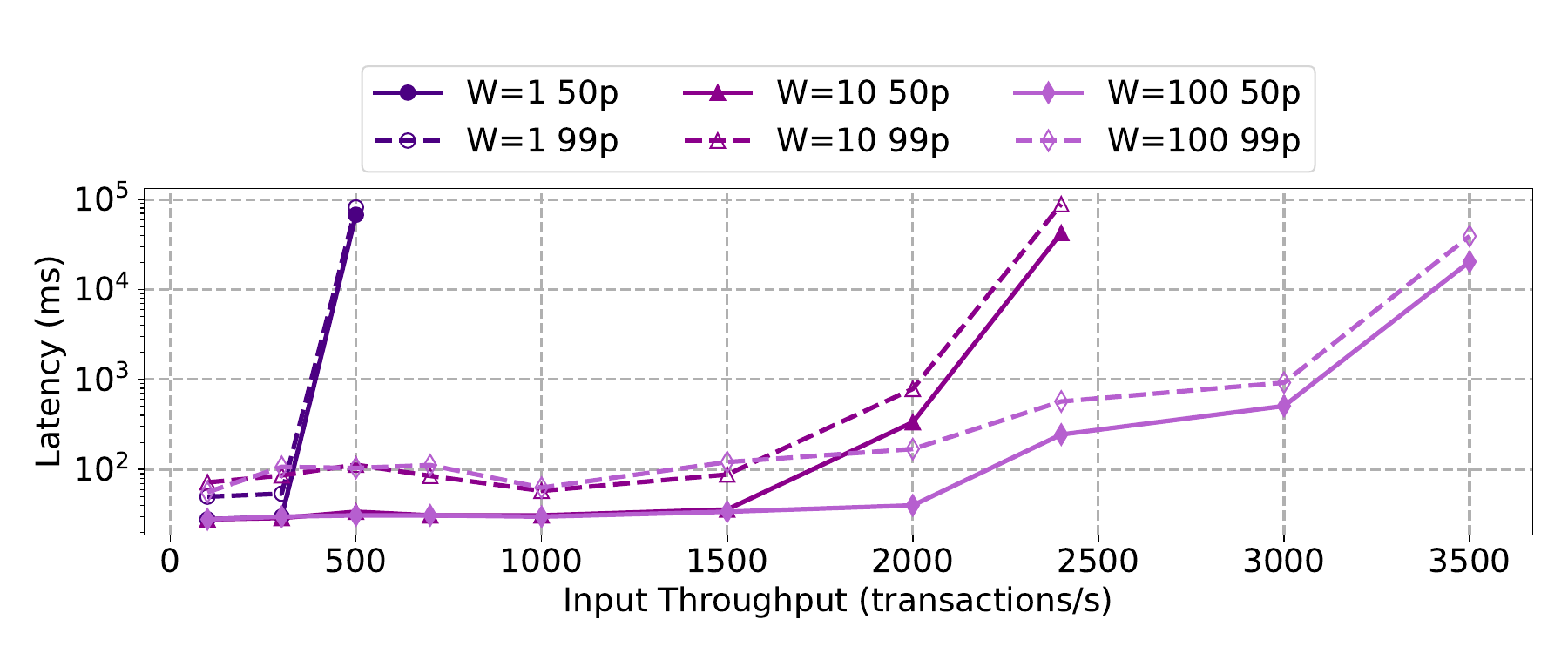}
        \caption[]%
        {{\footnotesize TPC-C on Styx with 1, 10, and 100 warehouses.}}    
        \label{ch4:fig:styx_tpcc}
    \end{subfigure}
    \caption{Evaluation in different scenarios. T-Statefun does not support range queries required by the Deathstar workloads. TPC-C is only supported by Styx.}
    \label{ch4:fig:throughput_latency}
\end{figure*}

\subsection{Latency vs. Throughput} \label{ch4:sec:exp_lat_through} \label{ch4:sec:exp_skew}

We first study the latency-throughput tradeoff of all systems. We retain the resources given to the systems constant (112 CPUs) while progressively increasing the input throughput. We measure the transaction latency. As depicted in \Cref{ch4:fig:throughput_latency}, Styx outperforms its baseline systems by at least an order of magnitude. Specifically, in YCSB-T (\Cref{ch4:fig:tl_ycsbt}), Styx achieves a performance improvement of \textasciitilde20x in terms of throughput against T-Statefun, which ranks second. In addition, Styx outperforms Boki by \textasciitilde30x in Deathstar's travel reservation workload (\Cref{ch4:fig:death_trav}) and by \textasciitilde35x in Deathstar's movie review \Cref{ch4:fig:death_mov}) workload. Finally, in the TPC-C benchmark (\Cref{ch4:fig:styx_tpcc}), which requires a large number of function calls per transaction (20-50), we observe that Styx's performance improves as we increase the input throughput for different numbers of warehouses, reaching up to 3K TPS with sub-second 99$^{\text{th}}$ percentile latency (100 warehouses).

\para{Aborts \& Throughput} Beldi and Boki follow a no-wait-die concurrency control approach, which leads to a significant amount of aborts as the throughput increases. Styx and T-Statefun do not use such a transaction abort mechanism. Instead, they execute all transactions to completion. This difference in handling transactions under high load makes the latencies across systems hard to compare. For this reason, in \Cref{ch4:fig:zipf_styx_ts}, we plot the results of Styx and T-Statefun and present the performance of Beldi and Boki in a separate table (\Cref{ch4:fig:bbabort}), alongside their abort rates. 

\begin{figure}[t]
\centering
  \includegraphics[width=0.65\columnwidth]{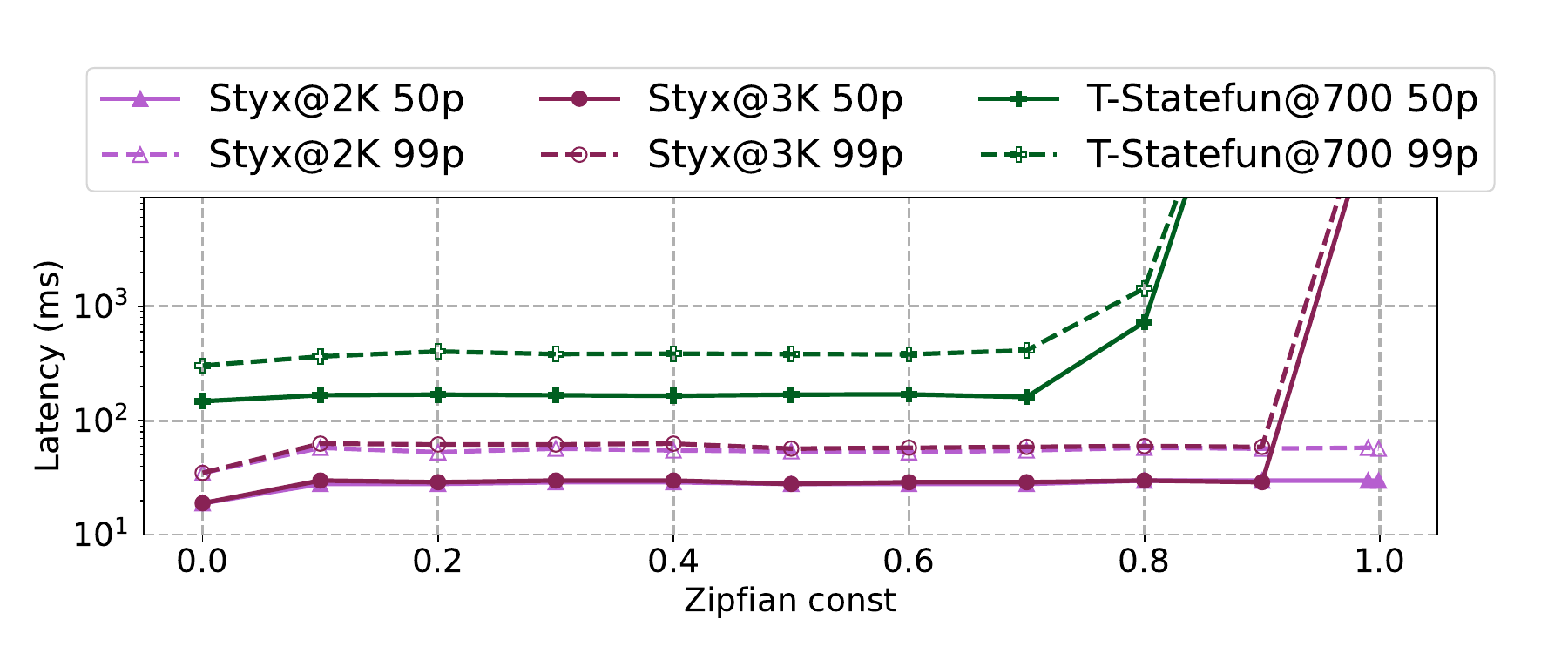} 
\caption{Latency evaluation for varying levels of contention (0.0 - 0.999) with YCSB-T (skewed). We ran Styx with two different input throughput variations to show clearly its behavior under contention. Note that Styx and T-Statefun execute all transactions to completion (abort\%=0).}
  \label{ch4:fig:zipf_styx_ts}
\end{figure}

\begin{table}[t]
\centering
    \resizebox{0.65\columnwidth}{!}{
    \begin{tabular}{rccccccccc}
    \multicolumn{1}{l}{} & \multicolumn{1}{l}{} & 0.0 & 0.2 & 0.4 & 0.6 & 0.8 & 0.9 & 0.99 & 0.999 \\ \hline \hline
    \multicolumn{1}{r|}{Beldi} & \multicolumn{1}{c||}{Abort \%} & \multicolumn{1}{c|}{47.93} & \multicolumn{1}{c|}{45.54} & \multicolumn{1}{c|}{44.31} & \multicolumn{1}{c|}{47.28} & \multicolumn{1}{c|}{52.40} & \multicolumn{1}{c|}{56.06} & \multicolumn{1}{c|}{61.62} & \multicolumn{1}{c|}{60.70} \\
    \multicolumn{1}{r|}{} & \multicolumn{1}{c||}{CMT TPS} & \multicolumn{1}{c|}{104} & \multicolumn{1}{c|}{108} & \multicolumn{1}{c|}{111} & \multicolumn{1}{c|}{105} & \multicolumn{1}{c|}{95} & \multicolumn{1}{c|}{76} & \multicolumn{1}{c|}{76} & \multicolumn{1}{c|}{78} \\ \hline
    \multicolumn{1}{r|}{Boki} & \multicolumn{1}{c||}{Abort \%} & \multicolumn{1}{c|}{48.77} & \multicolumn{1}{c|}{48.23} & \multicolumn{1}{c|}{49.54} & \multicolumn{1}{c|}{51.82} & \multicolumn{1}{c|}{61.29} & \multicolumn{1}{c|}{68.50} & \multicolumn{1}{c|}{74.47} & \multicolumn{1}{c|}{70.71} \\
    \multicolumn{1}{r|}{} & \multicolumn{1}{c||}{CMT TPS} & \multicolumn{1}{c|}{359} & \multicolumn{1}{c|}{362} & \multicolumn{1}{c|}{353} & \multicolumn{1}{c|}{337} & \multicolumn{1}{c|}{271} & \multicolumn{1}{c|}{220} & \multicolumn{1}{c|}{179} & \multicolumn{1}{c|}{205} \\ \hline
    \end{tabular}
    }
    \vspace{2mm}
    \caption{Evaluation of Boki and Beldi for varying levels of contention with YCSB-T. We report the abort ratio and committed transactions rate, and omit latency since the systems do not execute all transactions to completion. Both run at their maximum sustainable throughput.}
    \label{ch4:fig:bbabort}
\end{table}

We observe the following: $i)$ at the highest level of contention ($Zipfian$ at $0.999$) Styx achieves at least 2000 TPS, outperforming the rest by \textasciitilde 5-10x in terms of effective throughput, $ii)$ both Beldi and Boki (that run at their maximum sustainable throughput) abort more transactions as the level of contention increases (\textasciitilde40-70\%), which significantly impacts their effectiveness as shown in \Cref{ch4:fig:bbabort}, and $iii)$ Styx shows an increase in latency only in high levels of contention ($Zipfian>0.99$) while executing at \textasciitilde4x higher throughput than the rest.

\para{Runtime Breakdown} In \Cref{ch4:fig:breakdown}, we show where the systems under test spend their processing time. We use YCSB-T for this purpose since it is the only benchmark supported by all the systems (\Cref{ch4:sec:exp:setup}). We measured the median latency while all the systems were running at 100 TPS for 60 seconds and averaged the proportions of function execution, networking, and state access across all committed transactions. The key observations are: $i)$ Styx's co-location of processing and state led to minimal state access latency, and $ii)$ Styx's asynchronous networking allows for lower network latency.

\begin{table}[t]
\centering
\captionsetup{justification=centering}
    \resizebox{0.65\columnwidth}{!}{
    \begin{tabular}{c||c|c|c}
    \textbf{System} & \textbf{Function Execution} & \textbf{Networking} & \textbf{State Access} \\ \hline \hline
    \textbf{Styx} & \textbf{0.34ms} - 2.2\% & \textbf{14.33ms} - 95.6\% & \textbf{0.32ms} - 2.2\% \\ \hline
    \textbf{Boki} & 1.1ms - 3.3\% & 16.1ms - 49\% & 15.68ms - 47.7\% \\ \hline
    \textbf{T-Statefun} & 2.76ms - 2.2\% & 92.12ms - 74.3\% & 29.11ms - 23.5\% \\ \hline
    \textbf{Beldi} & 1.01ms - 0.7\% & 56.58ms - 38.4\%& 89.57ms - 60.9\% \\ \hline
    \end{tabular}
    }
    \vspace{2mm}
    \caption{Performance breakdown of all systems. (median latency - percentage from the total)}
    \label{ch4:fig:breakdown}
\end{table}

\para{Takeaway} The rather large performance advantages of Styx across all experiments are enabled by the following three properties and design choices: $i)$ the co-location of processing and state with efficient networking as shown in \Cref{ch4:fig:breakdown}, contrary to the other systems that have to transfer the state to their function execution engines; $ii)$ the asynchronous snapshots with delta maps for fault tolerance compared to the replication of Beldi/Boki and the LSM-tree-based incremental snapshots of T-Statefun; $iii)$ the efficient transaction execution protocol employed in Styx compared to the two-phase commit used by Styx's competition.

\subsection{Scalability} 
\label{ch4:sec:exp_scalability}

In this experiment, we test the scalability of Styx by increasing the number of Styx workers. Each worker is assigned 1 CPU and a state of 1 million keys. We measure the maximum throughput on YCSB-T. The goal is to calculate the speedup of operations as the input throughput and number of workers scale together. In addition, we control the percentage of multi-partition transactions in the workload, i.e., transactions that span across workers. In \Cref{ch4:fig:scalability}, we observe that in all settings, Styx retains near-linear scalability. Finally, Styx displays the expected behavior as the number of multi-partition transactions increases.

\begin{figure}[t]
    \centering
    \captionsetup{justification=centering}
    \includegraphics[width=0.65\columnwidth]{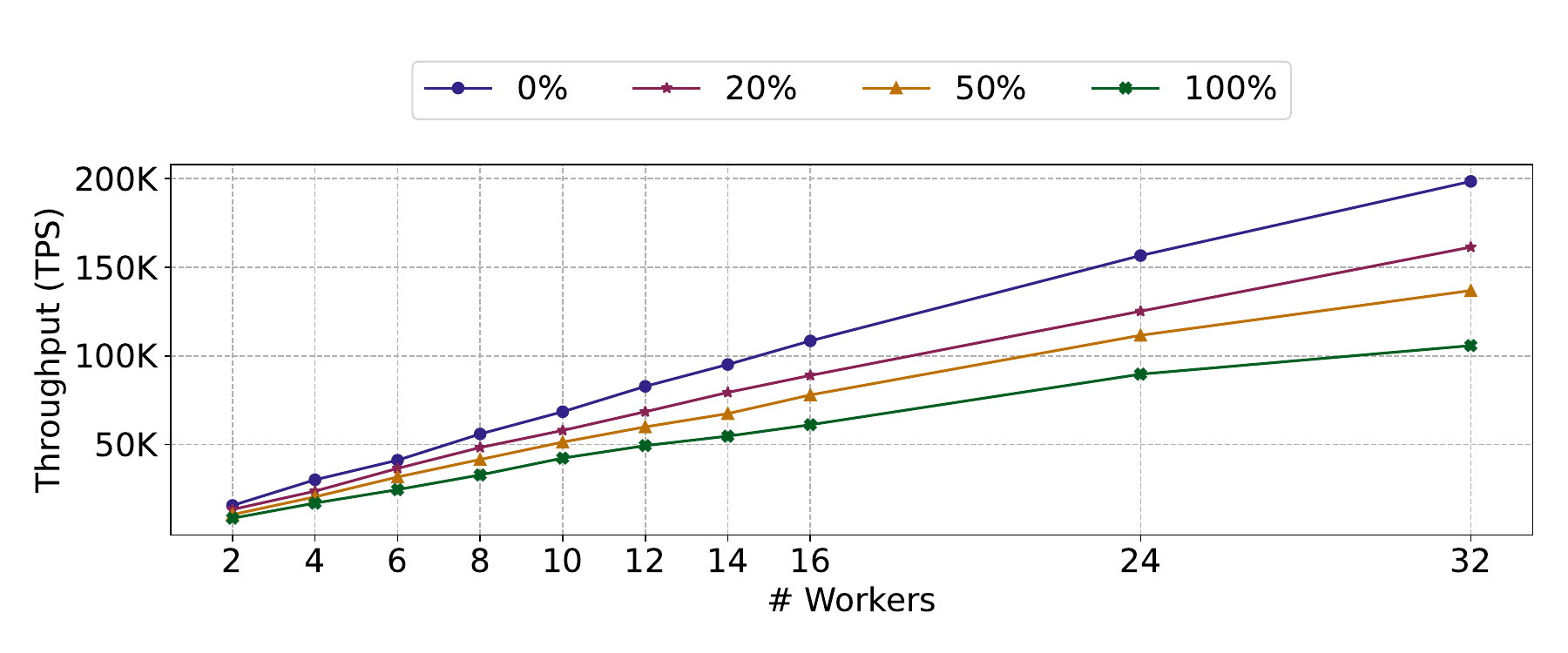}
    \caption{Scalability of Styx on YCSB-T with varying percentages of multi-partition transactions.}
    \label{ch4:fig:scalability}
\end{figure}

\subsection{Fault-Tolerance Evaluation}
\label{ch4:sec:exp:snapshots}

\para{Effect of Snapshots} In \Cref{ch4:fig:exp-snapshots}, we depict the impact of the asynchronous incremental snapshots on Styx's performance. In both figures, we mark when a snapshot starts and ends. The state includes 1 million keys, and we use a 1-second snapshot interval. Styx is deployed with four 1-CPU workers, and the input transaction arrival rate is fixed to 3K YCSB-T TPS. In \Cref{ch4:fig:exp-t-snap}, we observe that during a snapshot operation, Styx shows virtually no performance degradation in throughput. In \Cref{ch4:fig:exp-l-snap}, we observe a minor increase in the end-to-end latency in some snapshots. The reason for that is the concurrent snapshotting thread, which competes with the transaction execution thread during snapshotting. At the same time, it also has to block the transaction execution thread momentarily to copy the corresponding operator's state delta.

\begin{figure}[t]
     \centering
     \captionsetup{justification=centering}
     \begin{subfigure}[b]{0.49\columnwidth}
        \centering
        \captionsetup{justification=centering}
        \includegraphics[width=0.7\columnwidth]{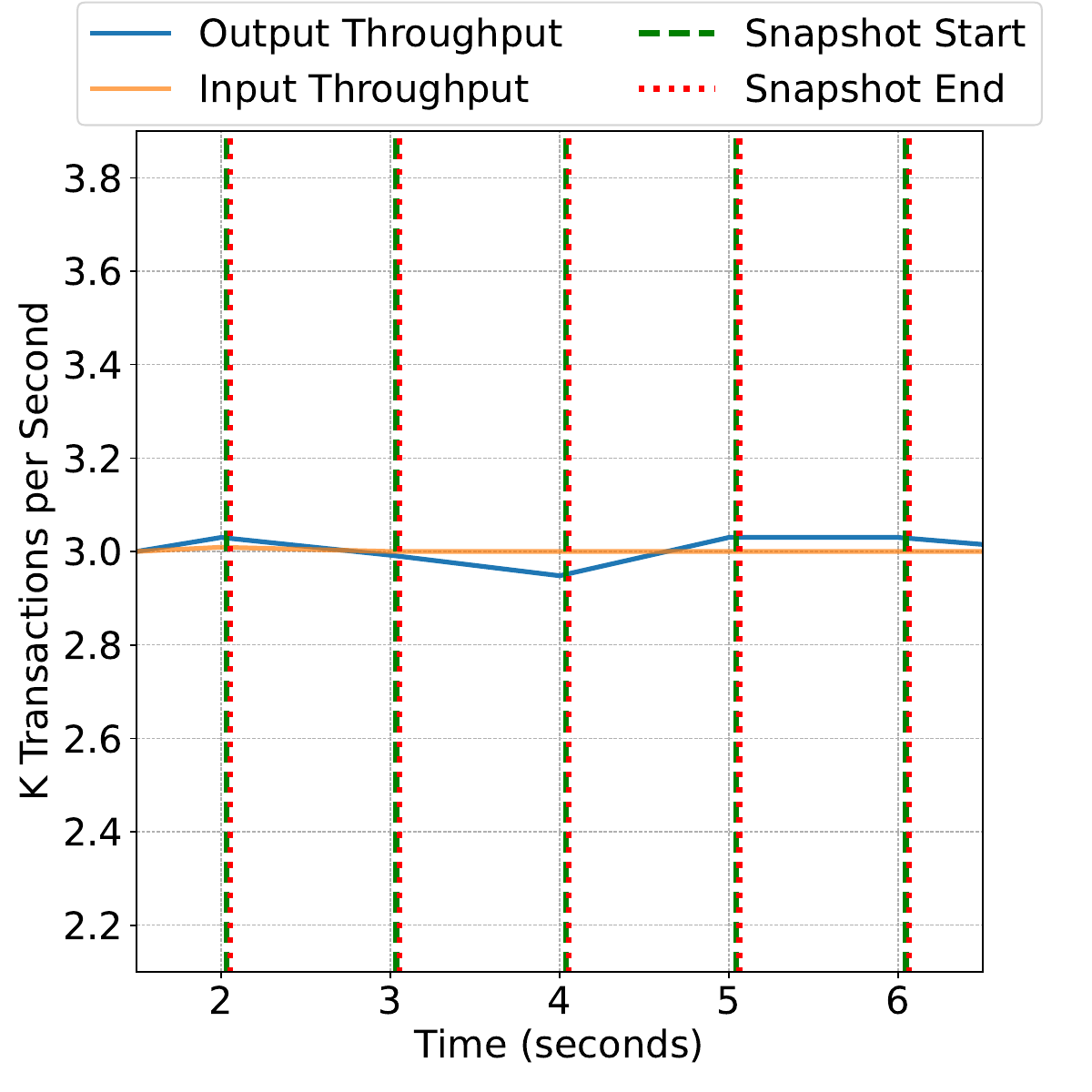}
            \caption{Throughput}
        \label{ch4:fig:exp-t-snap}
     \end{subfigure}
     \begin{subfigure}[b]{0.49\columnwidth}
        \centering
        \captionsetup{justification=centering}
        \includegraphics[width=0.7\columnwidth]{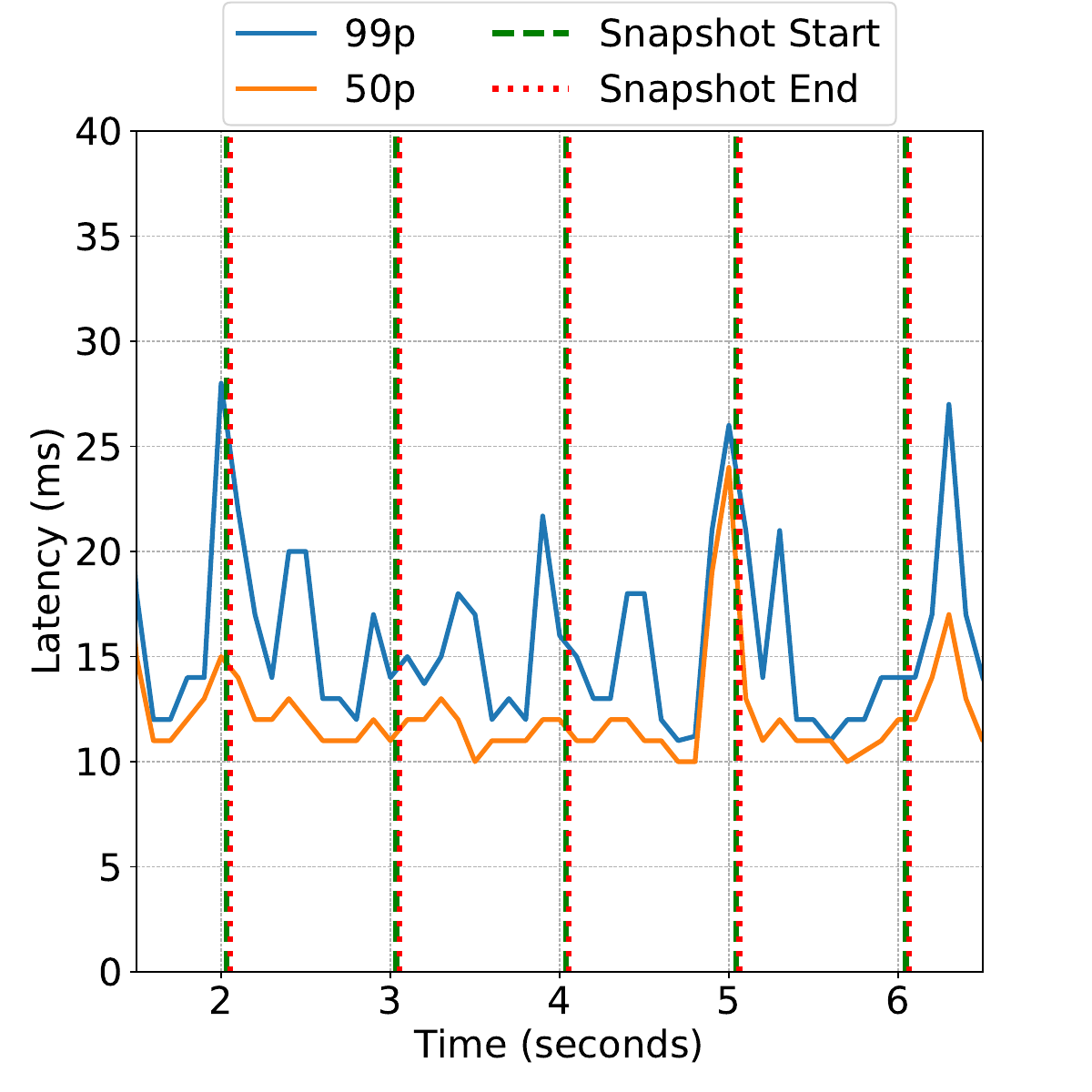}
            \caption{Latency}
        \label{ch4:fig:exp-l-snap}
     \end{subfigure}
        \caption{Impact of Styx's snapshotting on performance}
        \label{ch4:fig:exp-snapshots}
\end{figure}

\begin{figure}[t]
     \centering
     \captionsetup{justification=centering}
     \begin{subfigure}[b]{0.49\columnwidth}
        \centering
        \captionsetup{justification=centering}
        \includegraphics[width=0.7\columnwidth]{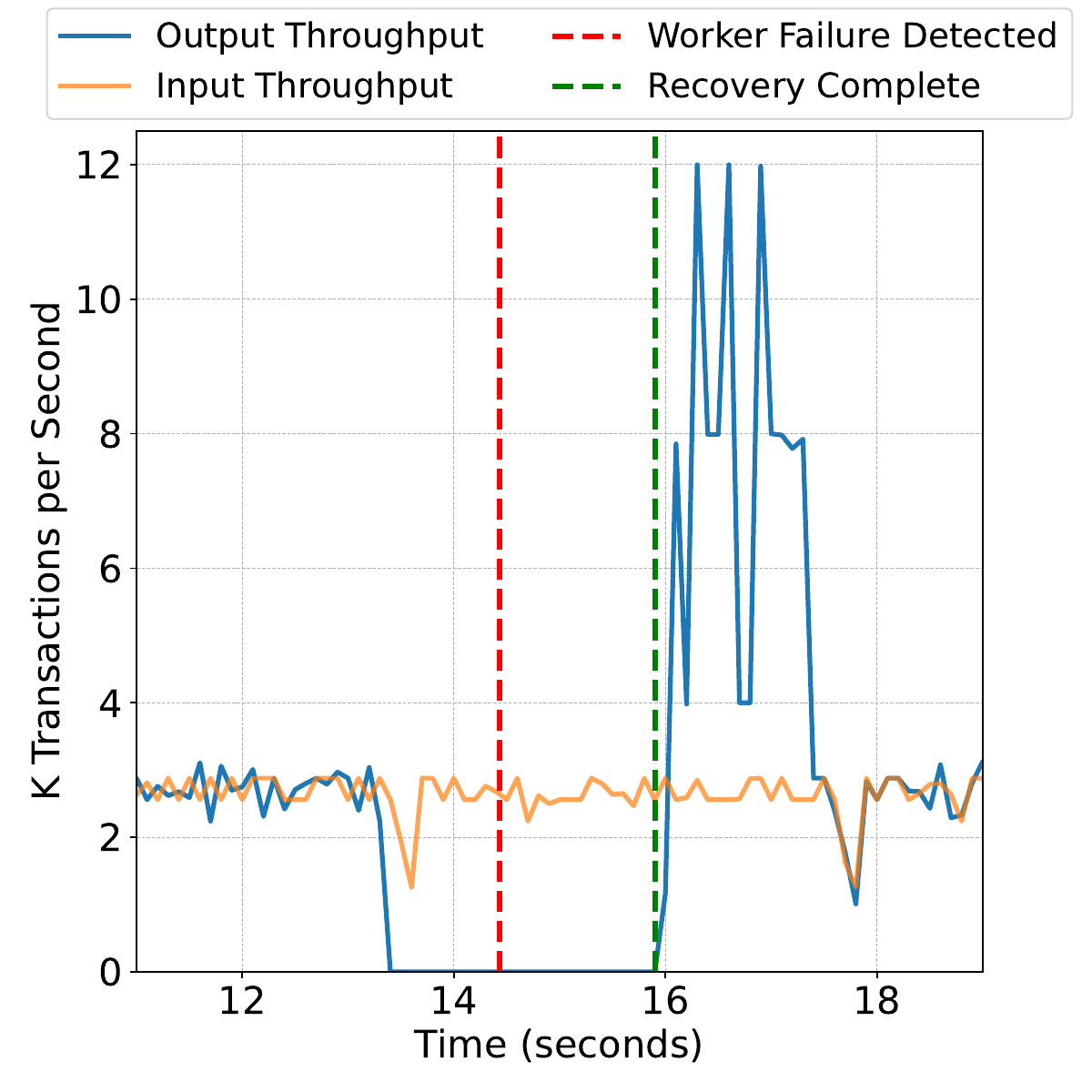}
            \caption{Throughput}
        \label{ch4:fig:exp-t-rec}
     \end{subfigure}
          \begin{subfigure}[b]{0.49\columnwidth}
        \centering
        \captionsetup{justification=centering}
        \includegraphics[width=0.7\columnwidth]{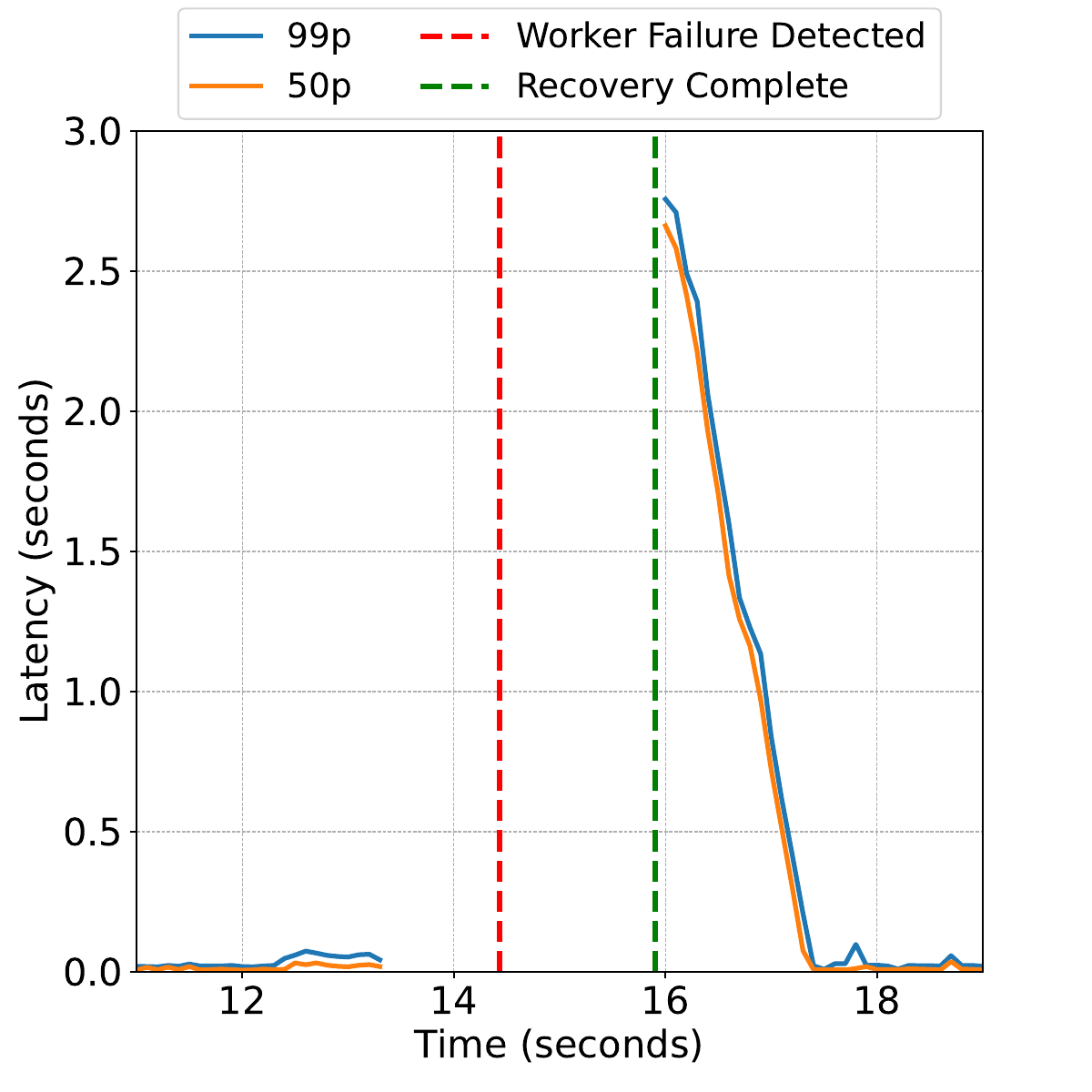}
            \caption{Latency}
        \label{ch4:fig:exp-l-rec}
     \end{subfigure}
        \caption{Styx's behavior during recovery.}
        \label{ch4:fig:exp-recovery}
\end{figure}

\para{Recovery Time} In \Cref{ch4:fig:exp-recovery}, we evaluate the recovery process of Styx with the same parameters as in \Cref{ch4:fig:exp-snapshots}. We reboot a Styx worker at \textasciitilde 13.5 seconds. It takes Styx's coordinator roughly a second to detect the failure. Then, after the reboot, the coordinator re-registers the worker and notifies all workers to load the last complete snapshot, merge any uncompacted deltas, and use the message broker offsets of that snapshot. The recovery time is also observed in the latency (\Cref{ch4:fig:exp-l-rec}) that is \textasciitilde 2.5 seconds (time to detect the failure in addition to the time to complete recovery). In terms of throughput (\Cref{ch4:fig:exp-t-rec}), we observe Styx working on its maximum throughput after recovery completes to keep up with the backlog and the input throughput.

\para{Effect of Large State Snapshots} In \Cref{ch4:fig:incr_snap}, we test the incremental snapshotting mechanism against a larger state of 20 GB from TPC-C using a bigger Styx deployment of 100 1-CPU workers at 10-second checkpoint intervals. From 0 to the 750-second mark, Styx is importing the dataset. Since there are no small deltas (importing is an append-only operation), snapshotting is more expensive than the normal workload execution, where only the deltas are stored in the snapshots. The increase in latency at \textasciitilde550 seconds corresponds to the loading of the largest tables (Stock and Order-Line) in the system. After loading the data and starting the transactional workload at 1000 TPS, we observe a drop in latency due to fewer state changes within the delta maps.

\begin{figure}[t]
    \centering
    \captionsetup{justification=centering}
    \includegraphics[width=0.65\columnwidth]{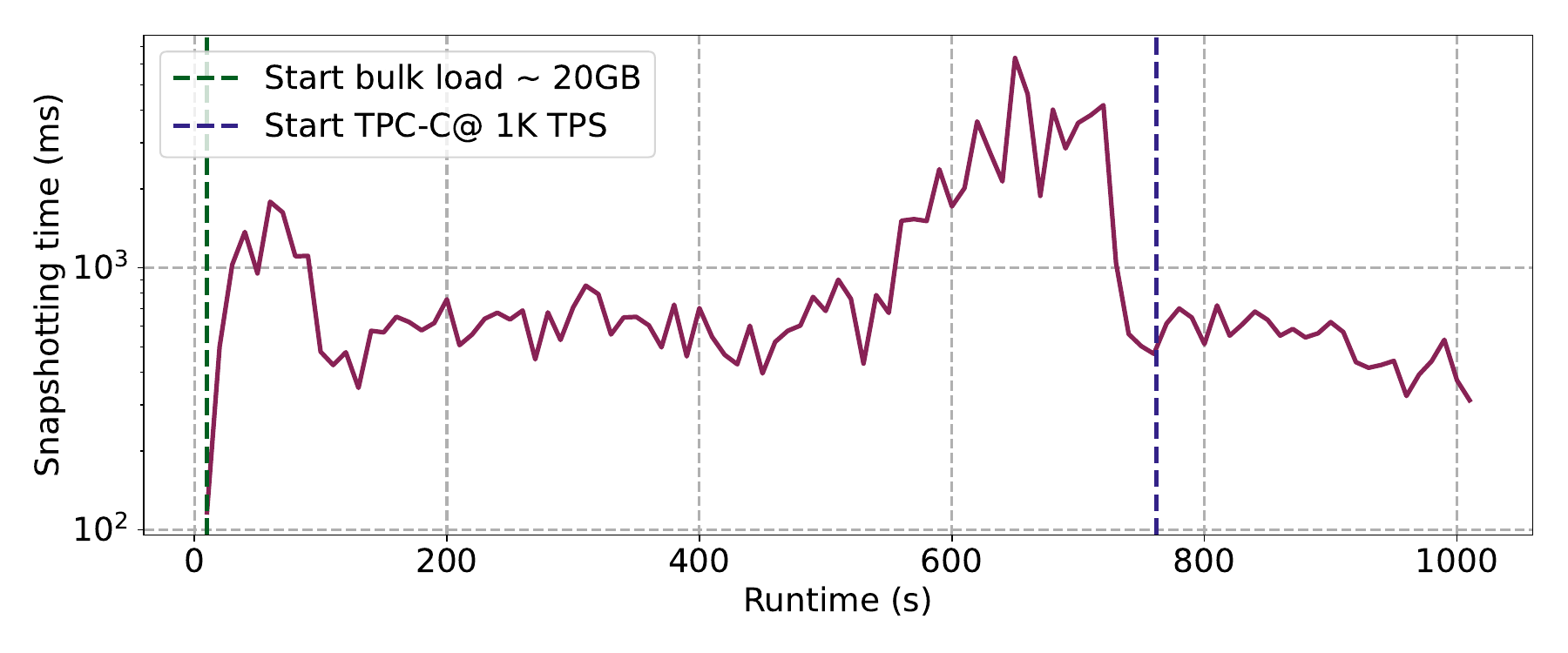}
    \caption{Behaviour of incremental snapshots on Styx with \textasciitilde20GB TPC-C state.}
    \label{ch4:fig:incr_snap}
\end{figure}

\section{Related Work} \label{ch4:sec:rel_work}
\para{Transactional SFaaS} SFaaS has received considerable research attention and open-source work. Transactional support with fault tolerance guarantees (that popularized DBMS systems) is necessary to widen the adoption of SFaaS. Existing systems fall into two categories: i) those that focus on transactional serializability and ii) those that provide eventual consistency. The first category includes Beldi~\cite{beldi}, Boki~\cite{boki}, and T-Statefun~\cite{tstatefun}. Beldi implements linked distributed atomic affinity logging on DynamoDB to guarantee serializable transactions among AWS Lambda functions with a variant of the two-phase commit protocol. Boki extends Beldi by adding transaction pipeline improvements regarding the locking mechanism and workflow re-execution. In turn, Halfmoon~\cite{qi2023halfmoon} extends Boki with an optimal logging implementation. T-Statefun~\cite{tstatefun} also uses two-phase commit with coordinator functions to support serializability on top of Apache Flink's Statefun. For eventually consistent transactions, T-Statefun implements the Sagas pattern. Cloudburst~\cite{cloudburst} also provides causal consistency guarantees within a DAG workflow. Proposed more recently, Netherite\cite{netherite}  offers exactly-once guarantees and a high-level programming model for Microsoft's Durable Functions\cite{durable_functions}, but it does not guarantee transactional serializability across functions. Unum \cite{LiuLNB23} needs to be paired with Beldi or Boki to ensure end-to-end exactly-once and transactional guarantees.

\para{Dataflow Systems} Support for fault-tolerant execution in the cloud with exactly-once guarantees~\cite{fernandez2014making, carbone2017state} is one of the main drivers behind the wide adoption of modern dataflow systems. However, they lack a general and developer-friendly programming model with support for transactions and a natural way to program function-to-function calls. Closer to the spirit of Styx are Ciel~\cite{murray2011ciel} and Noria~\cite{noria}. Ciel proposes a language and runtime for distributed fault-tolerant computations that can execute control flow. Noria solves the view maintenance problem via a dataflow architecture that can propagate updates to clients quickly, targeting web-based, read-heavy computations. However, neither of the two provides a transactional model for workflows of functions like Styx.

\para{Transactional Protocols} Besides Aria~\cite{aria} that inspired the protocol we created for Styx \Cref{ch4:sec:dataflow-system}, two other protocols fit the requirement of no a priori read/write set knowledge: Starry~\cite{zhang2022starry} and Lotus~\cite{zhou2022lotus}. Starry targets replicated databases with a semi-leader protocol for multi-master transaction processing. At the same time, Lotus~\cite{zhou2022lotus} focuses on improving the performance of multi-partition workloads using a new methodology called run-to-completion-single-thread (RCST). Styx makes orthogonal contributions to these works and could adopt multiple ideas from them in the future.

\section{Future Work}

\para{Elasticity in Dataflow Systems} Extensive work has been carried out in dynamic reconfiguration \cite{noria, ds2, dhalion} and state migration \cite{megaphone, meces, rhino} of streaming dataflow systems over the last few years. These advancements are necessary for providing serverless elasticity in the case of state and compute collocation to leverage dataflows as an execution model for serverless stateful cloud applications, which is a future goal of Styx.

\para{Replication for High Availability} 
In the Styx architecture, replication is only applied in the snapshot store and the Input/Output queues to ensure fault tolerance. For high-availability, Styx could adopt replication mechanisms from deterministic databases. Specifically, the design of deterministic transaction protocols, such as Calvin \cite{calvin}, features state replicas that require no explicit synchronization. First, the sequencer replicas need to agree on the order of execution. After that, the deterministic sequencing algorithm guarantees that the resulting state will be the same across partition/worker replicas by all replicas executing state updates in the same order.

\para{Non-Deterministic Functions on Streaming Dataflows} In its current version, Styx requires application logic to be deterministic, similar to OLTP~\cite{stonebraker2013voltdb, kallman2008h}, where stored procedures are required to be deterministic since they run independently on different replicas. The same determinism requirement applies to SFaaS \cite{tstatefun, boki} systems. However, real-world applications may encapsulate logic that makes the outcome of their execution non-deterministic. Examples of non-deterministic operations are calls to external systems and using random number generators or time-related activities. That said, we have a plan for supporting non-deterministic functions in Styx, as discussed in \Cref{ch4:sec:addr_non_det}.

\section{Conclusion} \label{ch4:sec:conclusion}
This paper presented Styx, a distributed streaming dataflow system that supports multi-partition transactions with serializable isolation guarantees through a high-level, standard Python programming model that obviates transaction failure management, such as retries and rollbacks. Styx follows the deterministic database paradigm while implementing a streaming dataflow execution model with exactly-once processing guarantees. Styx outperforms the state-of-the-art by at least one order of magnitude in all tested workloads regarding throughput.

\section{Styx in Action}

In this chapter, we introduced Styx~\cite{styx}, a dataflow-based runtime designed for transactional cloud applications built on the aforementioned principles. Styx ensures that each transaction's state mutations are reflected in the system's state exactly-once, even under failures, retries, or other potential disruptions. Additionally, it supports arbitrary function orchestrations with end-to-end serializability by leveraging a deterministic database protocol, eliminating the need for expensive two-phase commits. Our approach is inspired by two key observations~\cite{styxcidr}. First, modern streaming dataflow systems such as Apache Flink~\cite{flink} guarantee exactly-once processing by transparently handling failures. However, these systems lack the capability to execute general cloud applications and do not support transactional function orchestrations. Second, efficient transaction execution on top of dataflow systems can be enabled through deterministic database protocols like Calvin~\cite{calvin} or Aria~\cite{aria} without the overhead of two-phase commits. Styx bridges this gap by integrating a deterministic transactional protocol that allows early commit replies to clients, improving responsiveness while maintaining consistency.

\para{Demonstration Scenarios} To illustrate the capabilities of Styx, in our demonstration, we focus on three scenarios. \textbf{(1)}~We demonstrate the developer experience by showing how application logic can be free of transaction management and failure handling code. To this end, we have integrated a compiler for transforming object-oriented programs into dataflows optimized for our runtime~\cite{stateflow}. \textbf{(2)}~We highlight the system's deployment and rescaling capabilities, demonstrating how these processes can be performed with minimal overhead. \textbf{(3)}~We showcase how Styx seamlessly recovers from worker failures without affecting application performance. Additionally, the Styx UI provides live system metrics, offering attendees real-time visibility into system operations.

\begin{figure*}[t]
    \centering
    \captionsetup{justification=centering}
    \begin{subfigure}{0.48\textwidth}
        \centering
        \includegraphics[width=\linewidth]{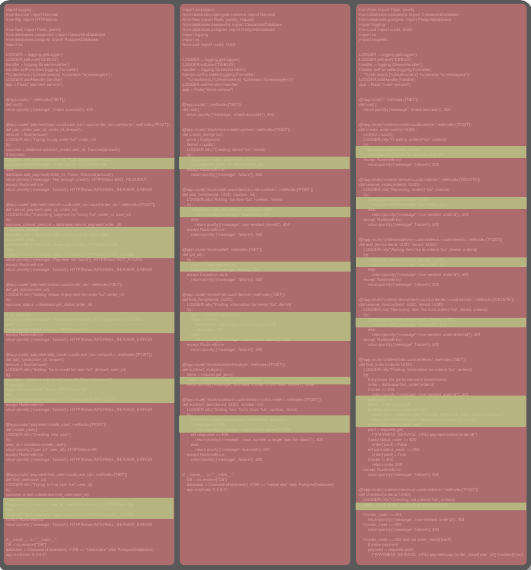}
        \caption{Microservice implementation using the saga pattern (Red: code to ensure atomicity and fault tolerance, Green: business logic).}
        \label{ch4:fig:demo:easy_use-a}
    \end{subfigure}
    \hfill
    \begin{subfigure}{0.48\textwidth}
        \centering
        \captionsetup{justification=centering}
        \begin{lstlisting}[style=pythonlang,basicstyle=\tiny]
from styx import Operator, StatefulFunction
from shopping_cart.operators import stock, payment, cart
from shopping_cart.exceptions import NotEnoughCredit, NotEnoughStock

@stock.register
async def decrement_stock(ctx: StatefulFunction, amount: int):
    item_stock = ctx.get()
    item_stock -= amount
    if item_stock < 0:
        raise NotEnoughStock(f"Item: {ctx.key} does not have enough stock")
    ctx.put(item_stock)

@payment.register
def pay(ctx: StatefulFunction, amount: int):
    credit = ctx.get()
    credit -= amount
    if credit < 0:
        raise NotEnoughCredit(f"User: {ctx.key} does not have enough credit")
    ctx.put(credit)

@cart.register
def checkout(ctx: StatefulFunction):
    items, user_id, total_price, paid = ctx.get()
    for item_id, qty in items:
        ctx.call_async(operator=stock,
                       function_name='decrement_stock',
                       key=item_id,
                       params=(qty, ))
    ctx.call_async(operator=payment,
                   function_name='pay',
                   key=user_id,
                   params=(total_price, ))
    paid = True
    ctx.put((items, user_id, total_price, paid))
    return "Checkout Successful" \end{lstlisting}
    
        \caption{Checkout workflow in Styx.}
        \label{ch4:fig:demo:easy_use-b}
    \end{subfigure}
    \caption{Comparison between the microservice paradigm (\cref{ch4:fig:demo:easy_use-a}) and Styx (\cref{ch4:fig:demo:easy_use-b}).}
    \label{ch4:fig:demo:easy_use}
    \vspace{-2mm}
\end{figure*}

\subsection{Demonstration Overview}


\begin{figure*}[t]
    \centering
    \captionsetup{justification=centering}
    \includegraphics[width=\textwidth]{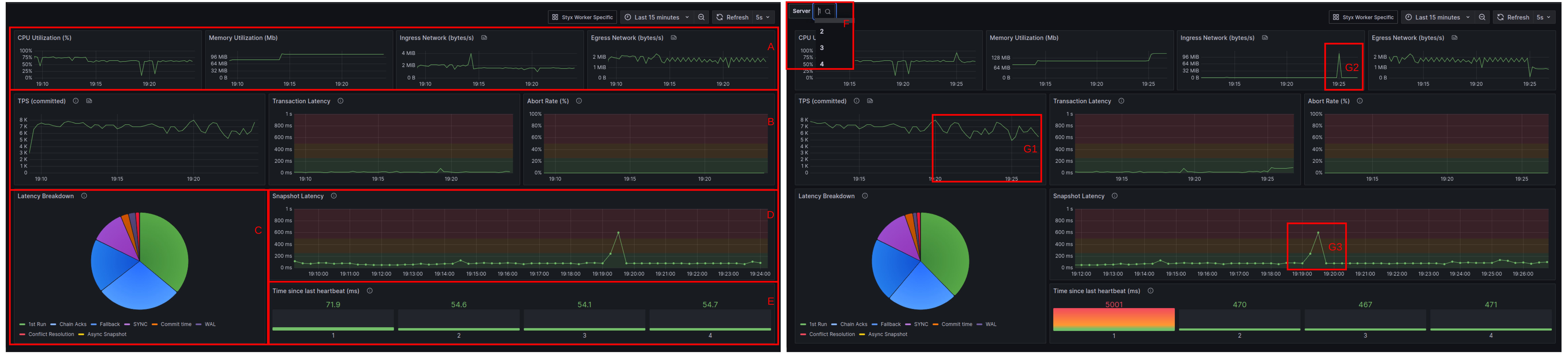}
    \caption{Styx monitoring dashboards.}
    \label{ch4:fig:demo:dashboard}
    \vspace{-3mm}
\end{figure*}

\subsubsection*{Scenario 1: Application Development}

Figure~\ref{ch4:fig:demo:easy_use} showcases the difference in developing a simplified shopping cart application with three services (stock, payment, cart) between a traditional microservice implementation and Styx. Styx eliminates the boilerplate code needed to ensure ACID guarantees in the microservice implementation and allows developers to focus solely on the core application logic. Beyond accelerating development, Styx enhances maintainability and reduces the likelihood of bugs by providing serializable transactions as a service. This results in cleaner, more reliable code, ultimately achieving long-term software quality. Attendees can change parts of the application code and submit applications to the Styx runtime.

\subsubsection*{Scenario 2: Deployment and Rescaling}

Styx is designed for seamless deployment and rescaling. To facilitate real-time monitoring, we have developed two dashboards. In the depicted scenario, we execute the YCSB-T workload across four Styx workers at 10.000 transactions per second (TPS) following a uniform distribution within one million keys.

\para{Part 1: System Overview} Attendees will be able to assess system performance across all Styx workers. The \textit{System Overview} dashboard~\Cref{ch4:fig:demo:dashboard} provides a high-level summary of key system metrics:
\begin{itemize}
    \item \textbf{Resource Metrics (A)}: Displays the average CPU and memory utilization, as well as ingress and egress network traffic across Styx workers.
    \item \textbf{Performance Metrics (B)}: Visualizes transaction throughput per second, average transaction latency, and abort rate. Under normal conditions, epoch latency (Styx uses a deterministic epoch-based commit protocol) for the YCSB-T workload remains below 250 ms (green-shaded region). At the same time, the abort rate fluctuates based on the level of contention from 0\% (no contention) to 100\% (all transactions within an epoch contain the same key at least once).
    \item \textbf{Latency Breakdown (C)}: A pie chart categorizes transaction latency into distinct components. Typically, the primary contributors to latency are the first optimistic transaction execution (1st Run) with the call-graph discovery (Chain Acks), the lock-based fallback commit mechanism (Fallback), and others like cross-worker synchronization, write-ahead-logging (WAL), conflict resolution + commit, and the asynchronous snapshots.
    \item \textbf{Snapshot Latency (D)}: Time taken for a complete delta snapshot throughout the deployment.
    \item \textbf{Worker Health (E)}: The final panel tracks time since the last heartbeat. If this value remains below 1000 ms (green-shaded), it confirms that all workers are healthy and operational.
    \item \textbf{Reconfiguration (G1-3)}:  At "19:19:30", we downscale the deployment from four partitions to three, and we observe an increase in snapshot latency and ingress network (data transferred across workers through S3), and finally a decrease in TPS since we decreased the parallelism.
\end{itemize}

The demonstration attendees will be able to change the skew factors of the supported workloads and perform updates to observe changes in the performance (latency, throughput) of Styx applications in real time.

\para{Part 2: Worker Specific} With a global view of the system's health in mind, attendees can drill down into the performance of individual workers using the \textit{Worker Specific} dashboard. This dashboard mirrors the system overview but focuses on a selected Styx worker. A~\textbf{drop-down menu (F)} allows attendees to choose a specific worker, enabling direct comparison with the overall system metrics. By analyzing the worker-specific metrics, attendees can quickly pinpoint anomalies. If a single worker exhibits significantly higher transactional latency or reduced throughput compared to the others, it may indicate an overloaded or unhealthy state.

\subsubsection*{Scenario 3: Fault Tolerance}

The final scenario showcases Styx’s ability to handle failures efficiently. During the demonstration, attendees can manually terminate a Styx worker to observe how the system detects the failure and triggers its recovery process.

\para{Part 1: System Overview} The system overview dashboard provides a real-time visual indicator of worker failures. When a Styx worker stops responding, the \textit{Time Since Last Heartbeat} metric (panel E) spikes, signaling the loss of communication. This event is accompanied by a sharp increase in transactional latency and a temporary dip in throughput until the system fully recovers. Once the operators assigned to the dead worker are rescheduled to a new or existing worker, Styx begins handling the delayed transactions, and these metrics gradually return to their normal ranges.

\para{Part 2: Worker Specific} Using the worker-specific dashboard, attendees can further investigate the failed worker’s behavior. The impact of the failure is more pronounced here. Transaction \textit{latency} and throughput fluctuations become more drastic, and for a brief period, the failed worker will stop reporting metrics entirely. Once the recovery process is completed, these values stabilize, confirming that the system has successfully recovered. This scenario demonstrates Styx’s resilience and self-healing mechanisms, ensuring system reliability even in the event of failures.
\chapter{State Migration in Styx: Towards Serverless Transactional Functions}
\label{chapter5}

\vfill

\begin{abstract}

 \Cref{chapter4} laid the foundation for Styx as a transactional dataflow engine for Serverless Functions-as-a-Service (SFaaS), achieving strong guarantees and performance with minimal developer effort. However, the vision of serverless is not fulfilled by programmability and transactionality alone. True serverless systems must also offer operational transparency: the ability to scale up and down dynamically, adapt to varying loads, and reassign resources autonomously, while preserving fault tolerance and correctness. To achieve this, Styx must evolve beyond its current static deployment model and embrace elastic state management. This transition demands a robust state migration mechanism that preserves Styx’s transactional semantics and exactly-once guarantees even under dynamic reconfiguration. This chapter takes on this challenge by introducing the design and implementation of state migration in Styx, marking a critical step toward fully serverless transactional functions.

\end{abstract}

\vfill

\blfootnote{Parts of this chapter are under review:\\ \faFileTextO~\hangindent=15pt\emph{K. Psarakis, G. Christodoulou, G. Siachamis, M. Fragkoulis, and A. Katsifodimos. State Migration in Styx: Towards Serverless Transactional Functions (Under Review)}.}

\newpage

\dropcap{D}{}evelopment-wise, Styx internals are already transparent to the developer, making it, in that sense, `serverless' since they do not require any transactional or fault-tolerance code to be written by the developer. The next challenge is to make its operational aspects, such as resource utilization, transparent. The first step toward this is implementing a state migration mechanism. In the meantime, Styx continues to evolve towards a serverless system. In this chapter, we propose an extension to Styx, which adds state migration capabilities. For Styx, which collocates state and processing, state migration is the cornerstone of its elasticity mechanism that enables a serverless offering. Now, Styx's state can be assigned to and moved between workers at key-set granularity. Key movement can occur in two ways: either on demand for transactions that need direct access to keys or asynchronously for non-accessed keys. 

\vspace{2mm}

\noindent This chapter makes the following contributions:

\noindent \textbf{--} Styx is elastic and can migrate state with near-zero downtime while maintaining high-throughput and low-latency (Section~\ref{ch5:sec:styx-state-migration}).

\noindent \textbf{--} Styx's tailored state migration approach outperforms the stop and restart baseline in scale-up and scale-down scenarios by having 4x less downtime and keeps the transactions to sub-second latencies mid-migration.(Section~\ref{ch5:sec:exp_state_migration}).

\section{State Migration in Styx}\label{ch5:sec:styx-state-migration}

 Implementing transactions on top of dataflows, the architectural core of stream processing engines (SPEs), adds additional challenges to state migration support in Styx. Methods that strictly target SPEs for state migration ~\cite{rhino,meces,megaphone} do not apply to our use case since they do not manage transactional semantics. On the other hand, state migration methods for transactional databases~\cite{squall,clay,abebe2020morphosys} are tightly coupled to the traditional OLTP database architecture. They cannot be directly applied to a dataflow engine such as Styx. Therefore, Styx requires a new tailor-made approach that maintains transactionality and adapts well to its exactly-once execution and snapshotting mechanisms.

The most straightforward approach for migrating state in any system is Stop-and-Restart (S\&R), where the system will stop processing incoming requests, shuffle the data to their new assignments, and restart processing. More sophisticated approaches often adopt some of the following mechanisms to migrate state: $i)$ maintain state replicas across workers to minimize the amount of data in need of migration~\cite{megaphone,rhino,abebe2020morphosys}, $ii)$ on-demand migration that only sends the data once a worker requires them and~\cite{meces,squall}, $iii)$ async migration to transfer data during idle time to progress a migration asynchronously~\cite{squall,abebe2020morphosys}.

In Styx, we implement two state migration approaches: a version of S\&R tailored to Styx that serves as our baseline method and an approach denoted as Online Migration (OM), which combines elements of migration approaches $(i)$ and $(ii)$, matching the current state-of-the-art.

In this section, we will first present an overview of state migration in Styx, specify the migration stages (triggering, handling, and resumption of processing), and the changes we made to Styx to support them. Then, we will elaborate on the S\&R and OM methods and discuss how we maintain fault-tolerance, determinism, end-to-end exactly-once, and serializability during state migration. In essence, in this work, we aim to extend Styx by adding elasticity, which is the first and most crucial step for Styx to be serverless.

\subsection*{Overview of Changes for State Migration}\label{ch5:sec:migration:overview}

To support transactional state migration, we extend Styx's base core runtime. This section outlines the architectural modifications required to enable this capability. We assume that migration is initiated by an external client or Styx's coordinator based on metrics such as load imbalance or resource utilization. The migration triggering policy, including autoscaling heuristics and monitoring, is considered orthogonal and left as future work. To enable correct and efficient state movement, Styx introduces the following system-level changes:

\para{Partition-level State and Metadata} To make state movement more flexible, operator state and Kafka offsets are tracked at a finer granularity, specifically, on a per-partition basis. Without considering state movement, offsets, and state were maintained per operator, which made it impossible to distinguish between multiple partitions of the same operator on a single worker. The finer-grained tracking enables selective state migration without relying on hashing to determine a key’s partition or requiring complete operator-level checkpoints.

\para{Shadow Partitions} Message arrival from Kafka is not guaranteed to be aligned with Styx’s internal reconfiguration. To handle this potential misalignment, Styx keeps shadow partitions temporarily, forwarding out-of-partition transactions to the correct partition.  In that way, Styx ensures correctness without requiring global input suspension. This issue is particularly evident when down-scaling, where partitions are removed. For instance, when a client issues a down-scale migration action, triggering repartition, other clients can potentially keep sending transactions to Styx based on the previous partitioning scheme. This leads to transactions arriving in an outdated partition that no longer exists. Keeping old partitions as shadow partitions, responsible only for forwarding such transactions without any state mutation responsibilities, is essential for Styx to preserve exactly-once processing guarantees.

\para{Global Offset Restoration} A rerouting mechanism is required, not only during downscaling involving the shadow partitions but as a part of the general migration solution. Its role is twofold: $i)$ to detect incoming transactions from the input queue, routed to an outdated partition due to client partitioning misalignment, and $ii)$ reroute them to the correct partition. Following our previous example, a transaction can be routed to an outdated partition until all of Styx's clients update their routing table and align with the new partitioning scheme. To maintain exactly-once processing and output, it is essential to restore the input/output queue offsets, as they might get updated in at most two places (previous and new partitioning). This ensures that no records are skipped or reprocessed during migration in case of failure.

\para{Blocking Actions Minimization} To ensure low latency even in the presence of large data transfers during migration, we had to make the two following adjustments to Styx $i)$ add compression to large messages and $ii)$ streaming asynchronous snapshots. First, Styx enforces compression using the Zstandard compression algorithm~\cite{zstd} for messages larger than a configurable size (by default set to 1MB). Second, regarding the snapshot mechanism in Styx, Styx now spawns a background thread that receives state deltas in a streaming fashion to prevent blocking if the delta becomes large (i.e., under heavy load, the involved subset of keys is significantly large). In the previous version, Styx would accumulate the entire delta and then send it to the background thread, leading to significant latency spikes that are now resolved.

\para{Composite Key Partitioning} In its current version, Styx also supports composite key partitioning to enhance data locality. Keys can be grouped by logical attributes (e.g., \textit{warehouse\_id} in TPC-C that is a prefix in all table primary keys other than the \textit{item}), allowing transactions to access colocated partitions. During migration, Styx leverages this structure to colocate groups of related keys to the same worker. This optimization reduces cross-worker communication and improves transaction commit latency.

\vspace{2mm}

\noindent These design extensions allow Styx to support both synchronous (stop-and-restart) and asynchronous (online) migration strategies without violating transactional or fault tolerance guarantees while maintaining low latency during reconfiguration.

\begin{algorithm}[t]
\footnotesize
\DontPrintSemicolon
\SetAlgoLined
\SetKwInOut{Input}{Input}\SetKwInOut{Output}{Output}
\SetKwComment{comm}{\hfill$\triangleright$\ }{}
\Input{$P_{w}$: Current partitions assigned to a worker, $H^{new}_{w_i}$: New hash function per partition, $K$: Keys, $w_{id}$: Worker id, $O_{in}$: Input offset, $O_{out}$: Output offset, $E_{count}$: Epoch count, $SEQ_{count}$: Sequence count}
\Output{$G$: new Dataflow Graph}
\BlankLine
\ForEach{$w_i \in workers$}{
    $P^{new} \leftarrow \emptyset$ \comm*[r]{Snapshotted partitions}
    \ForEach{$P_{w_i}$ $\in$ $P_w$}{

        \ForEach{$(key, value) \in P_{w_i}$}{
            $new_p  \leftarrow H^{new}_{w_i}(key)$\;
            $P^{new}_{new_p} \leftarrow$ $P^{new}_{new_p} \cup (key, value)$\;
        }
    }
    Store $P^{new}$ to persistent storage\;
    $meta_{new} \leftarrow \{O_{in}, O_{out}, E_{count}, SEQ_{count} \}$\;
    Send $meta_{new}$ to Coordinator\;
}
$G \leftarrow$ \Call{NewDataflowGraph}{()}

\ForEach{$w_i \in workers$}{
    $subG_{w_i} \leftarrow$ \Call{AssignSubgraph}{$(G, w_i)$}\;
    Receive $meta_{new}$ from Coordinator\;
    $O_{in}, O_{out}, E_{count}, SEQ_{count} \leftarrow meta_{new}$
}
\ForEach{$w_i \in workers$}{
    \ForEach{$P_i^{new} \in P^{new}$}{
        $P_i \leftarrow P_i^{new}$ \comm*[r]{Restore partitions}
    }
}
\caption{Stop and Restart}
\label{ch5:algo:stopandrestart}
\end{algorithm}

\section{State Migration Methods}

\subsection{Stop \&\ Restart}\label{ch5:sec:migration:snr}

Stop-and-Restart (S\&R) is the most straightforward migration strategy that we could implement on top of Styx. It suspends execution, performs state migration, and resumes computation with updated routing. Styx implements an optimal variant of S\&R tailored to its transactional runtime. In Algorithm~\ref{ch5:algo:stopandrestart}, we detail S\&R where, at first, for each worker ($w_i$) and all the partitions assigned to them ($P_{w_i}$) Styx rehashes all the keys of that partition based on the new partitioning using its hash function ($H^{new}_{w_i}$) and adds them alongside their values to the new partitions ($P_{new}$). Once a worker finishes the hashing step, the new partitions are stored as a snapshot of the persistent storage. Then, each worker sends their metadata (input/output offsets, sequencer count, and epoch count) to the coordinator, concluding the `Stop' step. To `Restart' Styx based on the new partitioning, each worker is assigned its part from the graph and receives the updated metadata from the coordinator. Finally, the worker loads the new partitions from the previously rehashed snapshot stored in persistent storage. 

While simple, robust, and independent from the transactional protocol,  S\&R incurs downtime due to the rehashing, storing, and loading of the data. This violates availability, making it more suitable for planned migration settings than Styx's serverless requirements.

\begin{figure*}[t]
    \centering
    \includegraphics[width=1\columnwidth]{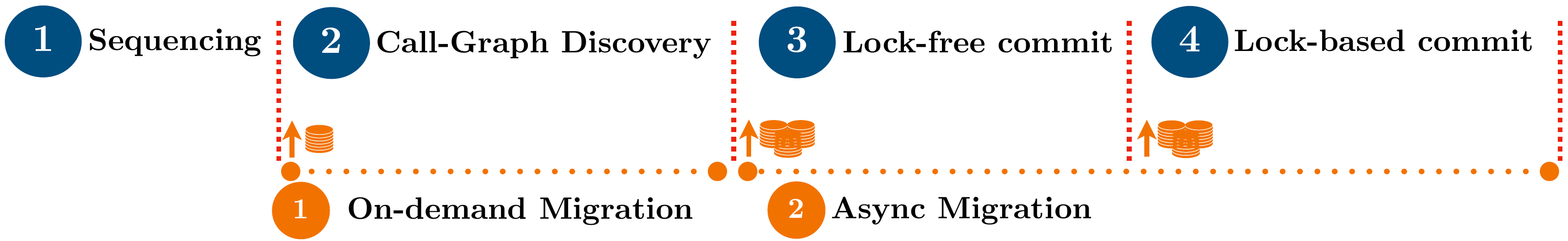}
    \vspace{-3mm}
        \caption{Alignment of epoch phases with online state migration. In phase \protect\circleb{2} transactions determine their call-graphs, and therefore, the associated keys need to be migrated immediately. Contrarily, in \protect\circleb{3} and \protect\circleb{4} keys can be migrated asynchronously without interfering with the transactional protocol.}
    \label{ch5:fig:migration_sync}
\end{figure*}

\begin{algorithm}[t]
\footnotesize
\DontPrintSemicolon
\SetAlgoLined
\SetKwInOut{Input}{Input}\SetKwInOut{Output}{Output}
\SetKwComment{comm}{\hfill$\triangleright$\ }{}

\Input{$P_{w}$: Current partitions assigned to a worker, $H^{new}_{i}$: New hash function per partition, $K$: Keys, $w_{id}$: Worker id, $O_{in}$: Input offset, $O_{out}$: Output offset, $E_{count}$: Epoch count, $SEQ_{count}$: Sequence count}
\Output{$G$: new Dataflow Graph}

\ForEach{$w_i \in workers$}{
    \ForEach{$P_i$ $\in$ $P_w$}{

        \ForEach{$key \in P_i$}{
            $new_p  \leftarrow H^{new}_{i}(key)$\;
            $P^{new}_{new_p} \leftarrow$ $P^{new}_{new_p} \cup (key)$\;
        }
    }

    $meta_{cur} \leftarrow \{O_{in}, O_{out}, E_{count}, SEQ_{count}, P^{new}_{i} \}$\;
    Send $meta_{cur}$ to Coordinator\;

}

$G \leftarrow$ \Call{NewDataflowGraph}{()}

\ForEach{$w_i \in workers$}{
    Receive $meta_{new}$ from Coordinator\;
    $O_{in}, O_{out}, E_{count}, SEQ_{count}, P^{new}_{i} \leftarrow meta_{new}$\;
    $subG_{w_i} \leftarrow$ \Call{AssignSubgraph}{$(G, w_i)$} 
}

\caption{Online Migration}
\label{ch5:algo:onlinemigration}
\end{algorithm}

\subsection{Online}\label{ch5:sec:migration:squall}

To support online, near-zero-downtime migration, Styx introduces an online method adapted to its transactional model. In contrast to S\&R, the online method performs migration in an on-demand and asynchronous manner, allowing the system to remain available throughout. Styx leverages its transactional epoch protocol to piggyback migration steps on normal processing cycles. As shown in Figure~\ref{ch5:fig:migration_sync}, keys accessed during Phase~\circleb{2} are migrated synchronously and on-demand to ensure consistency. Other keys, if idle, are migrated asynchronously during subsequent phases using worker idle time and batching mechanisms in phases \circleb{3} and \circleb{4}. This asynchronous migration is essential for the migration to complete, which is necessary for the fault tolerance mechanism to be reactivated  (snapshots are switched off mid-migration as per the SotA approaches to maintain consistent snapshots). 

In Algorithm~\ref{ch5:algo:onlinemigration}, we detail the Online Migration method where, similar to S\&R, for each worker ($w_i$) and all the partitions assigned to them ($P_{w_i}$) Styx rehashes all the keys of that partition based on the new partitioning using its hash function ($H^{new}_i$). The core difference is that it does not add the values and creates a routing table of where the keys are located and their destination ($P^{new}$). Once a worker finishes the hashing step, it
sends this information alongside its metadata (input/output offsets, sequencer count, and epoch count) to the coordinator. Finally, the worker loads the new routing tables and metadata and performs migration alongside the transactional protocol using the on-demand (\circlem{1}) and asynchronous (\circlem{2}) mechanisms as displayed in Figures~\ref{ch5:fig:migration_sync} and~\ref{ch5:fig:migration_rehash}. In Figure~\ref{ch5:fig:migration_rehash}, we display how a down-scaling action with repartitioning is performed in Styx while going from $4$ to $3$ workers and visualize the on-demand and asynchronous migration.

\begin{figure}[t]
    \centering
    \includegraphics[width=0.4\columnwidth]{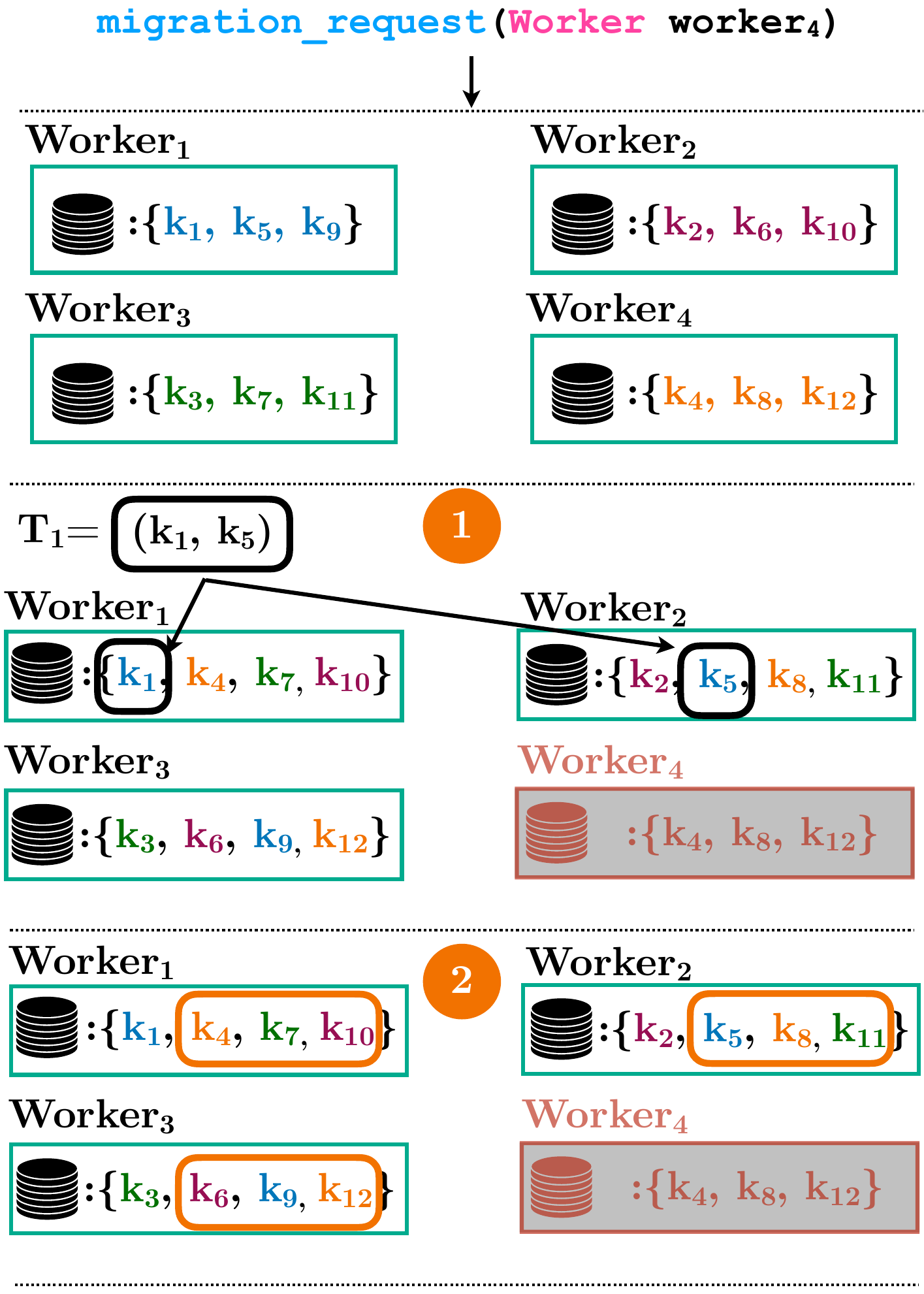}
    \vspace{-3mm}
        \caption{State distribution before and after migration request, which removes $worker_4$. After migration, we see the example of a transaction ($T_1 = \bigl\{k_1, k_5\bigr\}$). $T_1$ interacting with keys $k_1$ and $k_5$, meaning that while $k_1$ is already in-place, $k_5$ needs to be migrated on-demand as migration phase \protect\circlem{1} suggests in Figure~\ref{ch5:fig:migration_sync}. The rest of the keys, assuming that there is no other transaction interacting with them, can be migrated asynchronously in the migration phase \protect\circlem{2}.}
    \label{ch5:fig:migration_rehash}
\end{figure}

\section{Maintaining Guarantees}\label{ch5:sec:migration:correctness}

Both state migration approaches need to preserve Styx correctness guarantees, namely: $i)$ exactly-once processing, $ii)$ determinism, $iii)$ serializable transactional guarantees and, $iv)$ exactly-once output. For the S\&R method, maintaining correctness is straightforward since it stops execution, shuffles the data, and restarts. It does not affect any of the already-in-place mechanisms of Styx detailed in Section~\ref{ch4:sec:fault-tolerance}. Thus, in this subsection, we will primarily explain how the Online Migration method operates while preserving Styx's correctness guarantees.

Styx maintains deterministic execution and guarantees serializability throughout online migration. When a transaction requires access to a key located on another worker (triggering On-Demand Migration), the worker blocks execution until that key is received. This procedure is safe, as Styx's single-process coroutine approach ensures that no other transaction on the same worker can simultaneously request the same key. Transactions can only operate on fully available and up-to-date keys, and migrations are aligned with epochs to ensure consistency. Additionally, the asynchronous phase of Online Migration is only performed after the call-graph of all transactions within the epoch has been discovered. At that point, all the requested key transfers of the on-demand migration phase are guaranteed to have been completed. Moreover, fault tolerance remains unaffected; if a failure occurs, Styx will recover from the latest snapshot and restart the migration without compromising correctness. For the same reason, exactly-once processing and output remain unaffected by the migration mechanism. 

Finally, the only critical point to be addressed in both the S\&R and Online migration methods is out-of-partition events due to client-server partitioning misalignment. In Section~\ref{ch5:sec:migration:overview}, we explained the two new mechanisms of Styx that address this issue, namely Shadow Partitions and centralized offset restoration. Shadow partitions are used temporarily to reroute out-of-partition transactions from Kafka, ensuring that the correct worker and partition process the incoming transaction. To fully address this issue, the Kafka offset progress that might be affected by two different workers is restored by the coordinator before being stored in a snapshot. This coordination ensures exactly-once processing in the case of failure during state migration.

\section{State Migration Experiments on Styx} \label{ch5:sec:exp_state_migration}

In this section, we evaluate the state migration mechanism of Styx in sustainable throughput for scaling up and down while repartitioning the entire state. The repartitioning operation involves considerable data movement since Styx's partitioner is hash-based.

\subsection*{Setup} 

Styx executes in Python 3.13 and contains the optimizations mentioned in Section~\ref{ch5:sec:migration:overview}.

\para{Workload} The workloads used in our experiments follows the SotA transactional approaches~\cite{squall,abebe2020morphosys}, which are the YCSB and TPC-C benchmarks. In this experiment section, all tables are partitioned into 16 parts.

\parait{YSCB~\cite{ycsb}} The Yahoo! Cloud Serving Benchmark is a suite of workloads designed to represent large-scale, commonly developed web services. In our experiments, we use two YCSB datasets: a smaller dataset with 1 million records (1GB) for small-state experiments and a larger one with 10 million records (10GB). Each record consists of a primary key and 10 columns containing 100-byte randomly generated strings. We follow the workload configuration from~\cite{squall}, which includes two transaction types: 15\% of operations perform a single-record update, while the remaining 85\% perform a single-record read. It is important to note that YCSB differs from the YCSB-T variant used in Section~\ref{ch4:sec:exp}.

\parait{TPC-C} We follow the exact spec of TPC-C as in Section~\ref{ch4:sec:exp:setup} and generate two datasets, a small one (10 warehouses, 1GB) and a larger one (100 warehouses, 10GB) for our experiments in this section.


\begin{figure*}[t]
    \centering

    \begin{subfigure}[t]{0.24\textwidth}
        \centering
        \includegraphics[width=\linewidth]{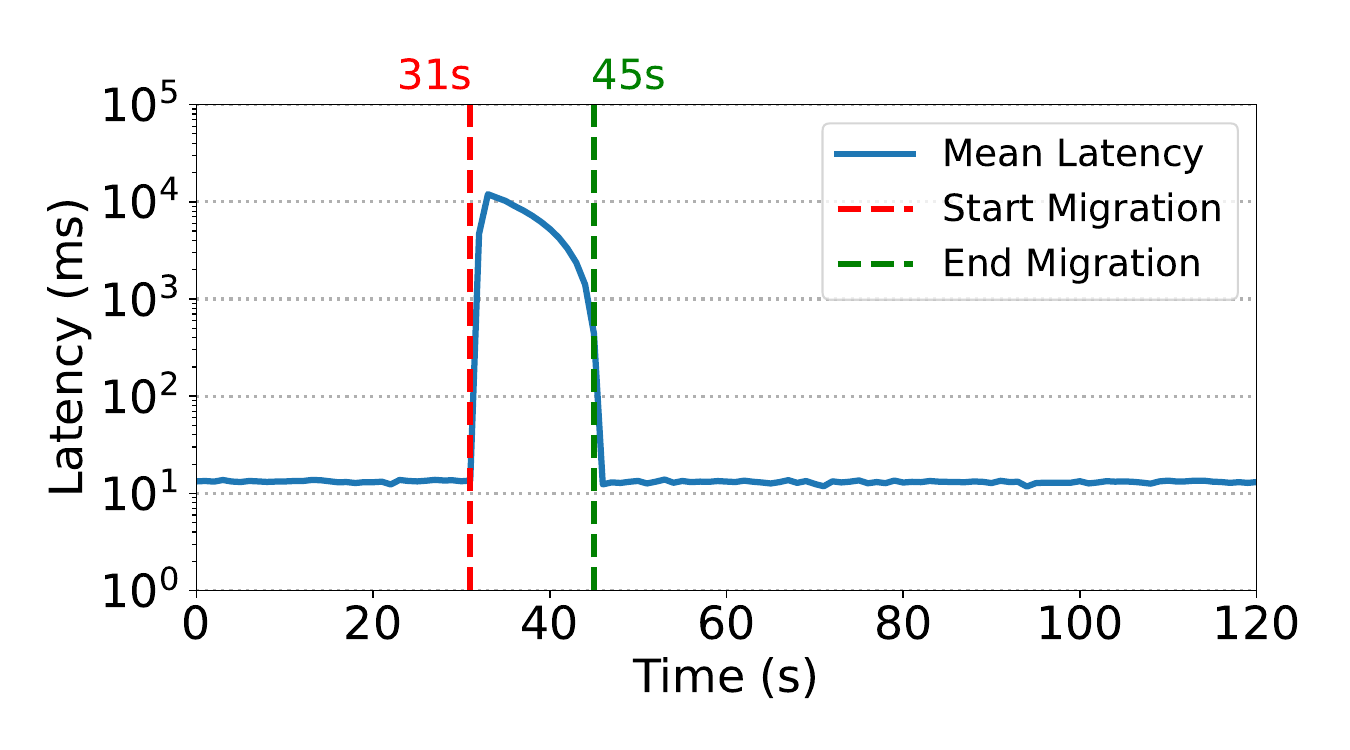}
    \end{subfigure}
    \hfill
    \begin{subfigure}[t]{0.24\textwidth}  
        \centering 
        \includegraphics[width=\linewidth]{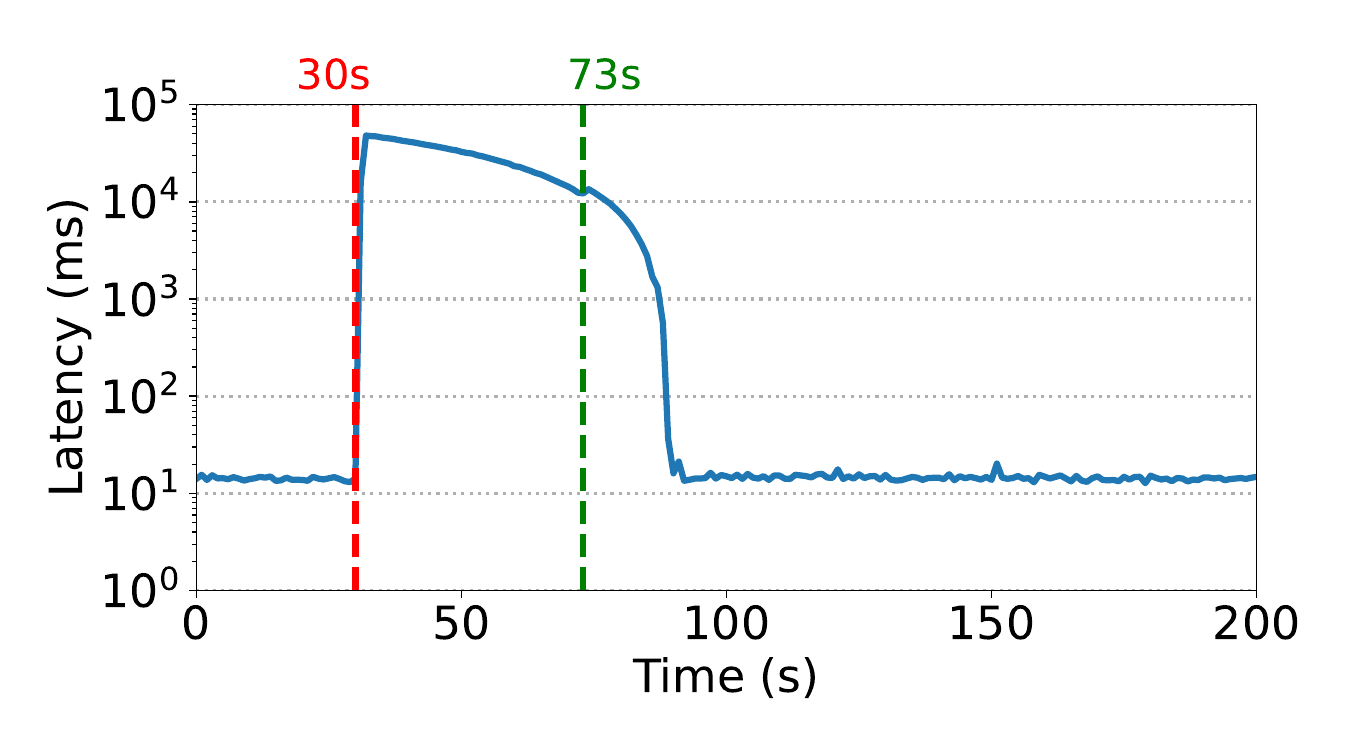}
    \end{subfigure}
    \hfill
    \begin{subfigure}[t]{0.24\textwidth}   
        \centering 
        \includegraphics[width=\linewidth]{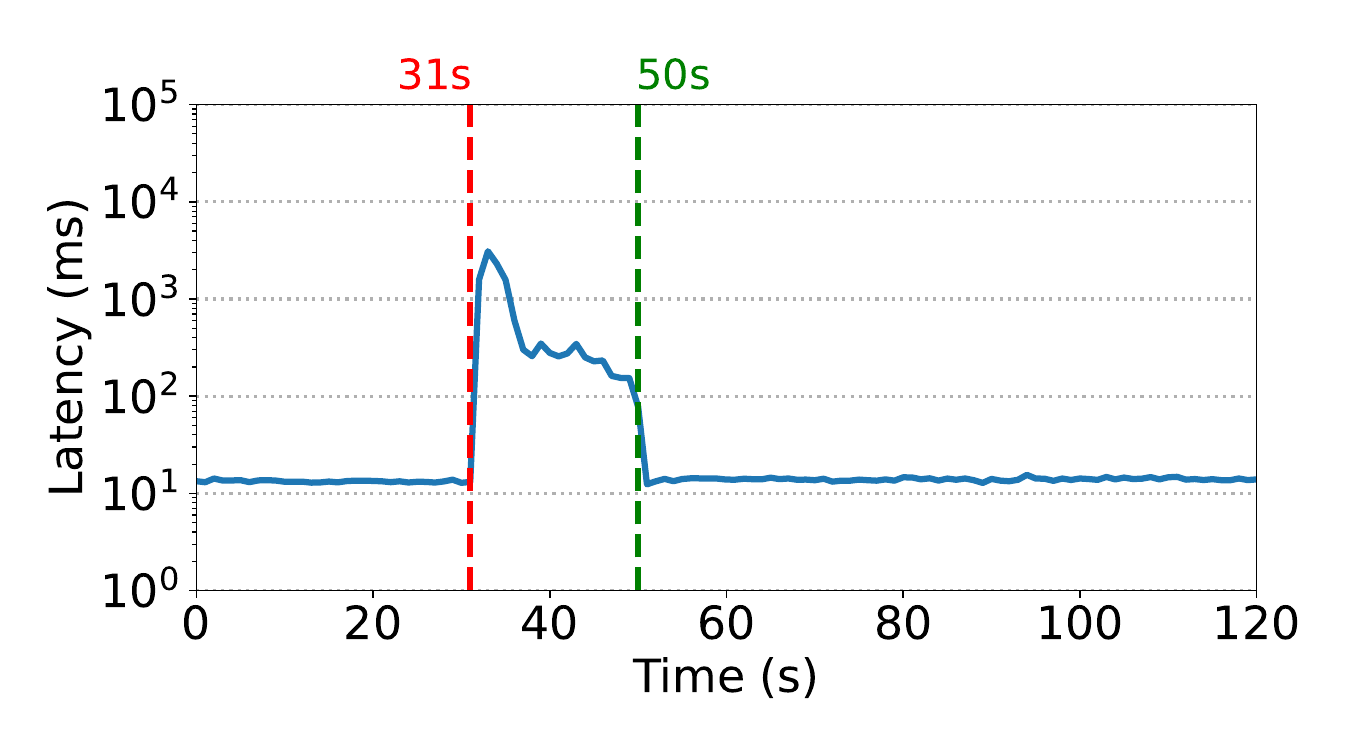}
    \end{subfigure}
    \hfill
    \begin{subfigure}[t]{0.24\textwidth}   
        \centering 
        \includegraphics[width=\linewidth]{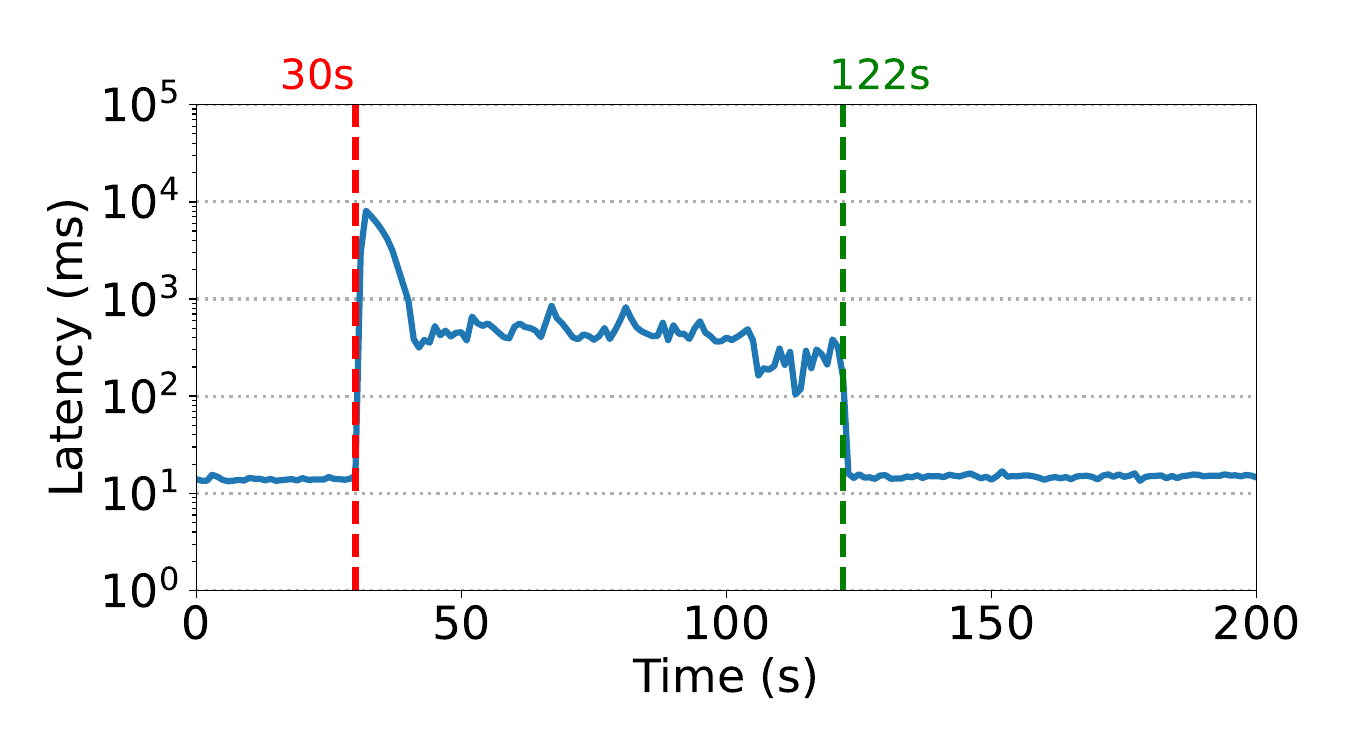}
    \end{subfigure}


    \begin{subfigure}[t]{0.24\textwidth}
        \centering
        \captionsetup{justification=centering}
        \includegraphics[width=\linewidth]{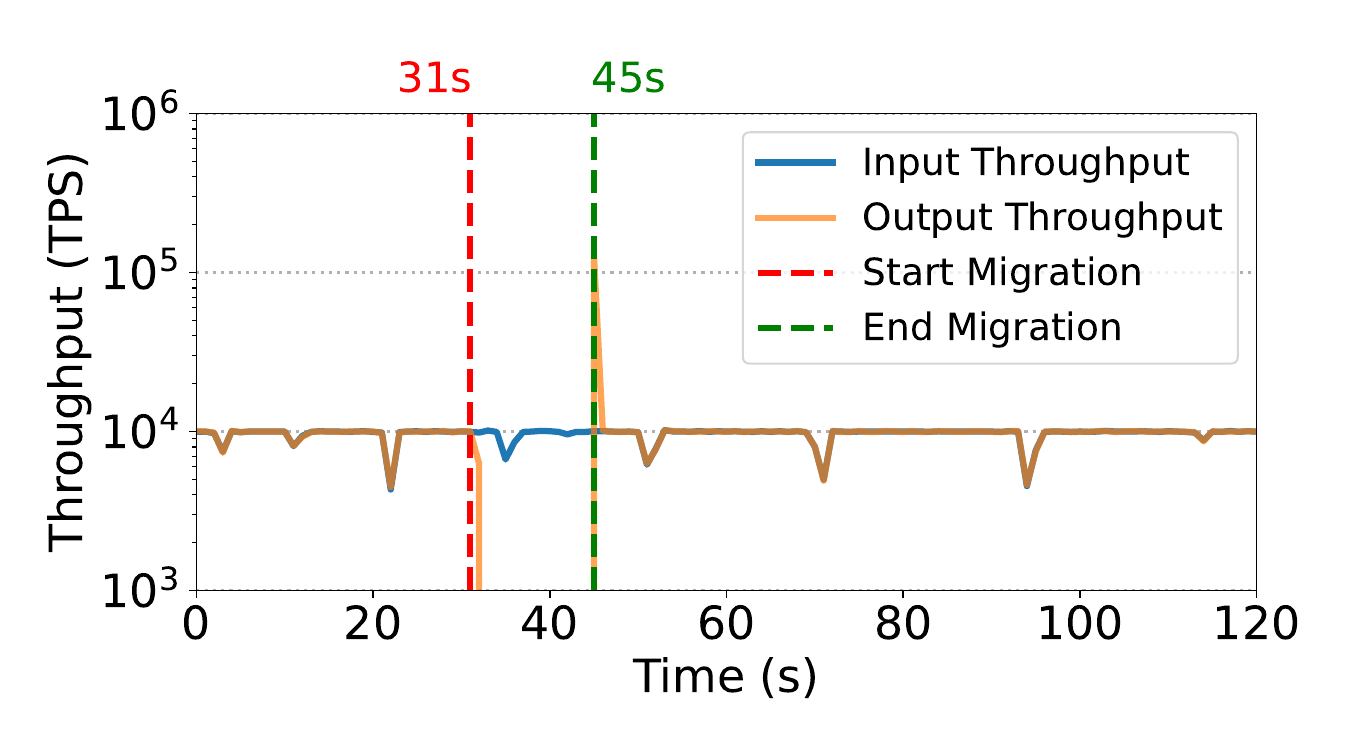}
        \vspace{-5mm}
        \caption*{{\footnotesize S\&R YCSB}}
    \end{subfigure}
    \hfill
    \begin{subfigure}[t]{0.24\textwidth}  
        \centering 
        \captionsetup{justification=centering}
        \includegraphics[width=\linewidth]{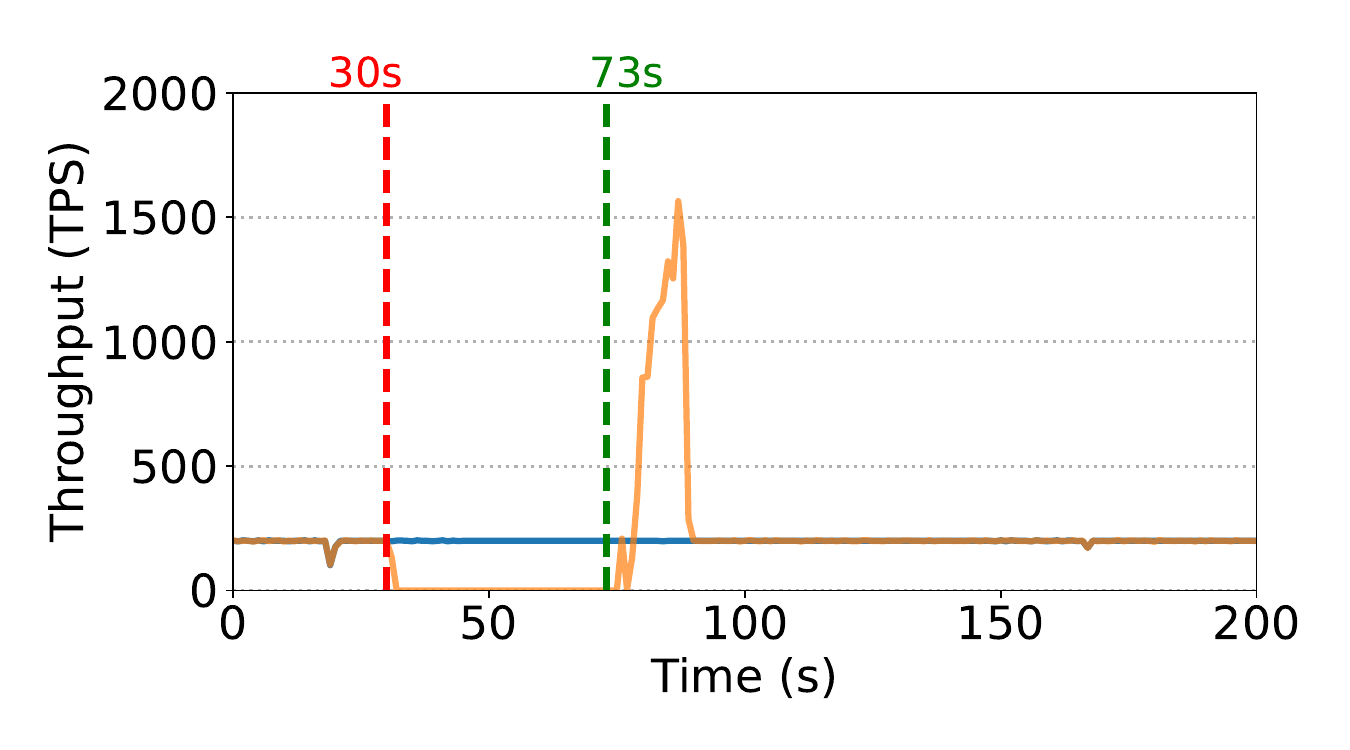}
        \vspace{-5mm}
        \caption*{{\footnotesize S\&R TPC-C}}
    \end{subfigure}
    \hfill
    \begin{subfigure}[t]{0.24\textwidth}   
        \centering 
        \captionsetup{justification=centering}
        \includegraphics[width=\linewidth]{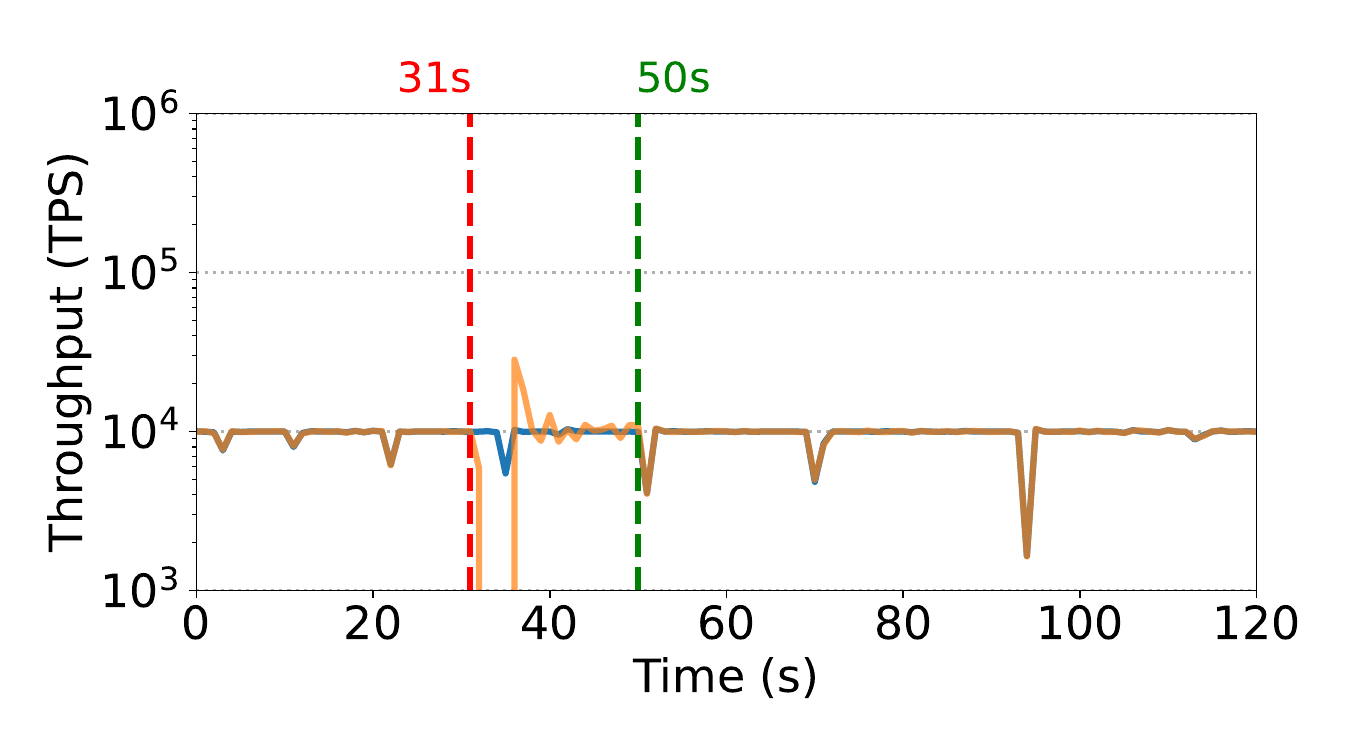}
        \vspace{-5mm}
        \caption*{{\footnotesize Online YCSB}}
    \end{subfigure}
    \hfill
    \begin{subfigure}[t]{0.24\textwidth}   
        \centering 
        \captionsetup{justification=centering}
        \includegraphics[width=\linewidth]{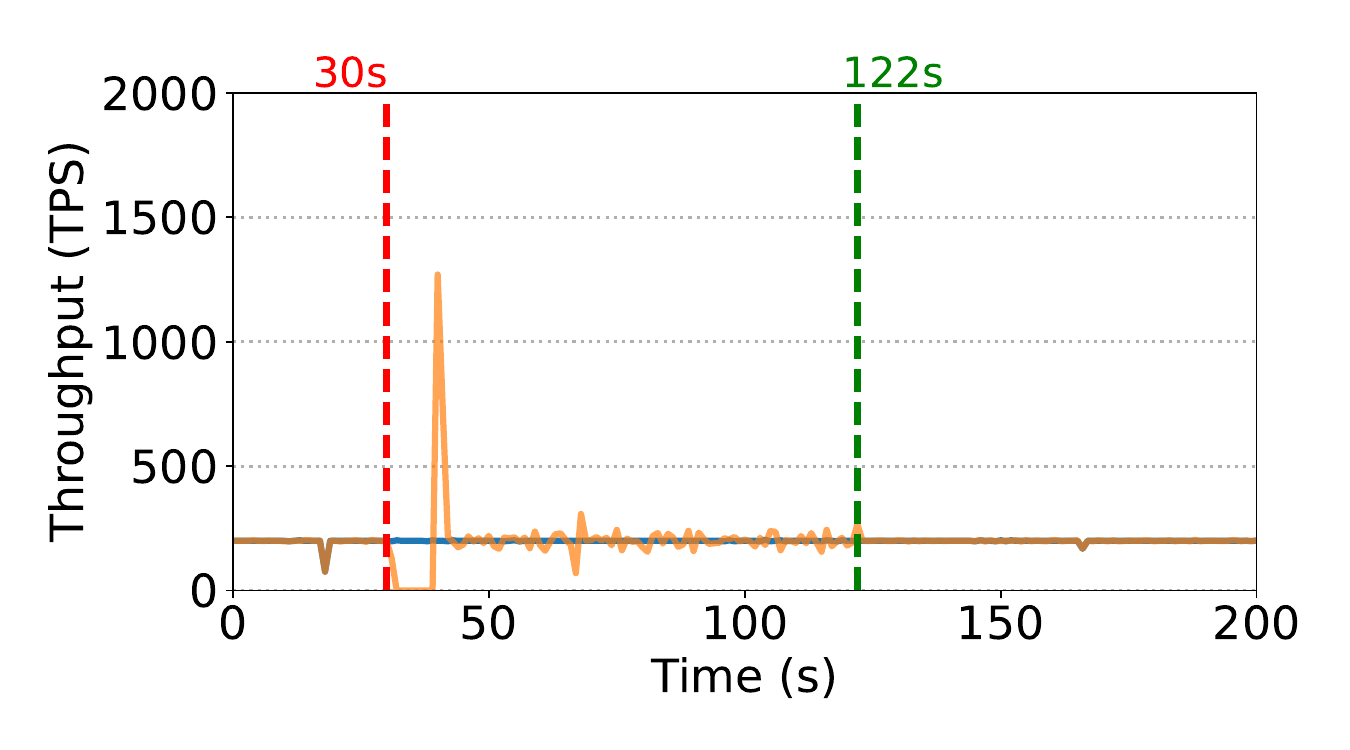}
        \vspace{-5mm}
        \caption*{{\footnotesize Online TPC-C}}
    \end{subfigure}
        \vspace{-2mm}
    \caption{Stop and Restart (S\&R) and Online Migration for big State (10GB) Scale Down in both TPC-C and YCSB: Latency (top row) and Throughput (bottom row).}
    \label{ch5:fig:scale_down_big_combined}
\end{figure*}


\begin{figure*}[t]
    \centering
    \captionsetup{justification=centering}

    \begin{subfigure}[t]{0.24\textwidth}
        \centering
        \includegraphics[width=\linewidth]{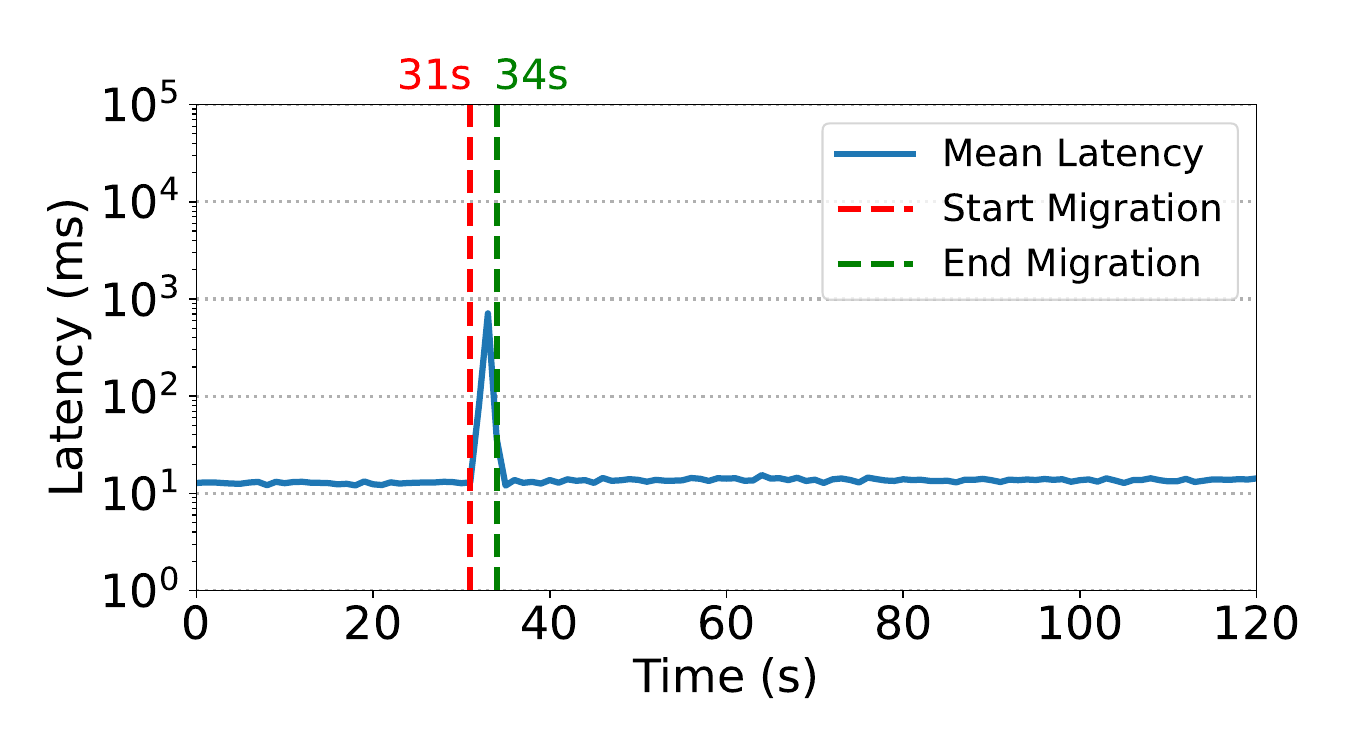}
    \end{subfigure}
    \hfill
    \begin{subfigure}[t]{0.24\textwidth}  
        \centering 
        \includegraphics[width=\linewidth]{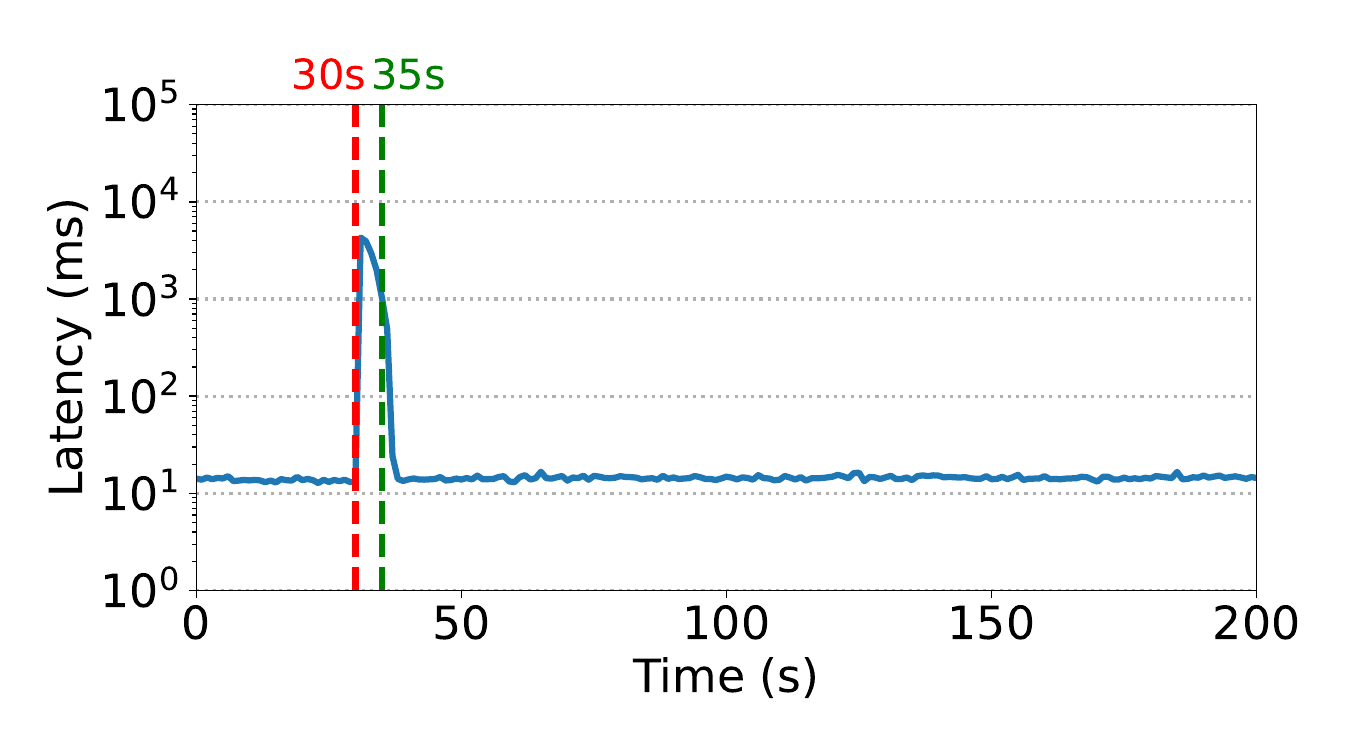}
    \end{subfigure}
    \hfill
    \begin{subfigure}[t]{0.24\textwidth}   
        \centering 
        \includegraphics[width=\linewidth]{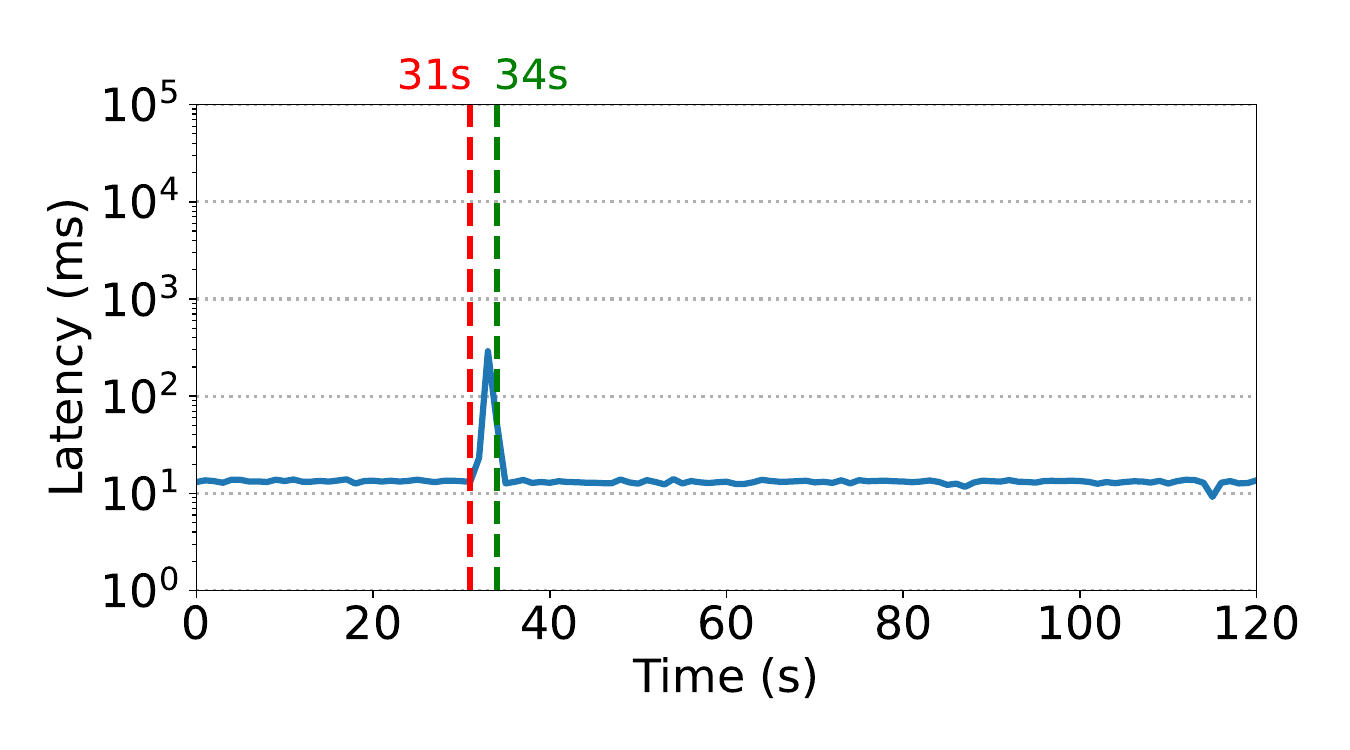}
    \end{subfigure}
    \hfill
    \begin{subfigure}[t]{0.24\textwidth}   
        \centering 
        \includegraphics[width=\linewidth]{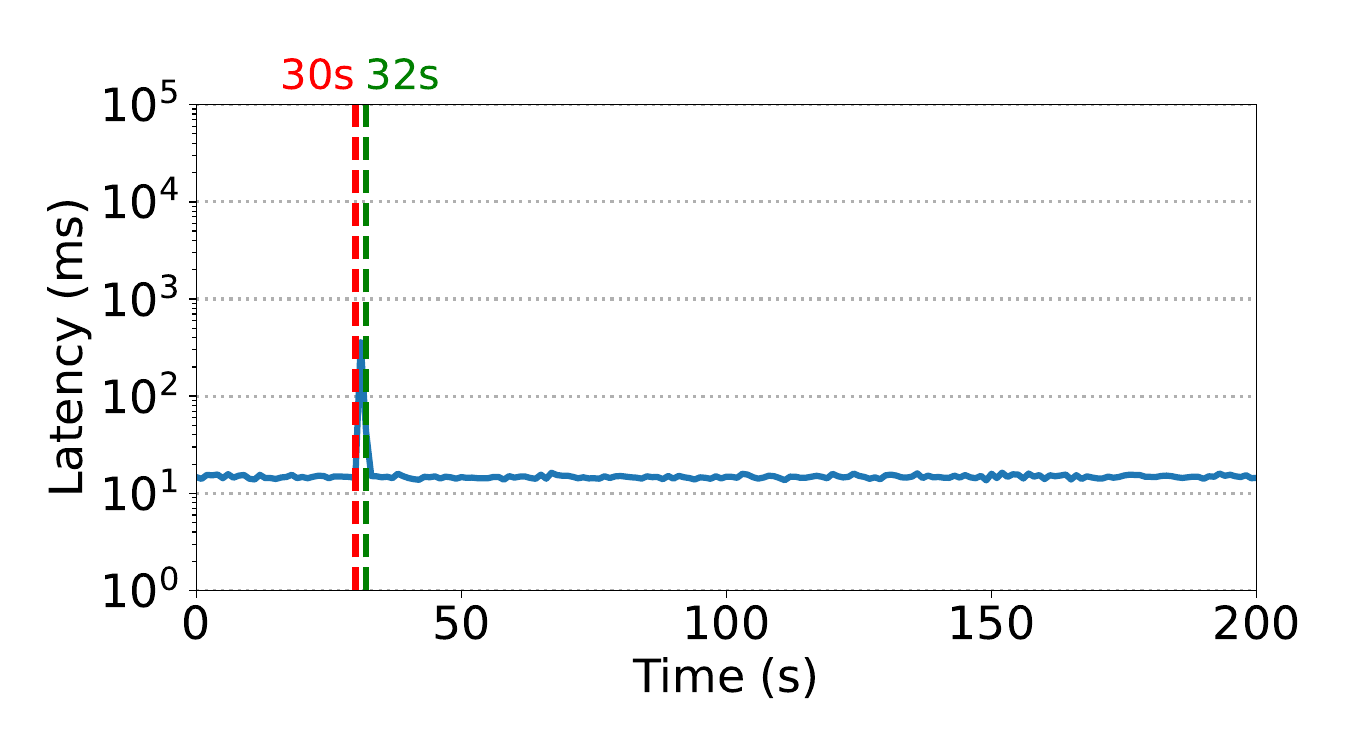}
    \end{subfigure}


    \begin{subfigure}[t]{0.24\textwidth}
        \centering
        \captionsetup{justification=centering}
        \includegraphics[width=\linewidth]{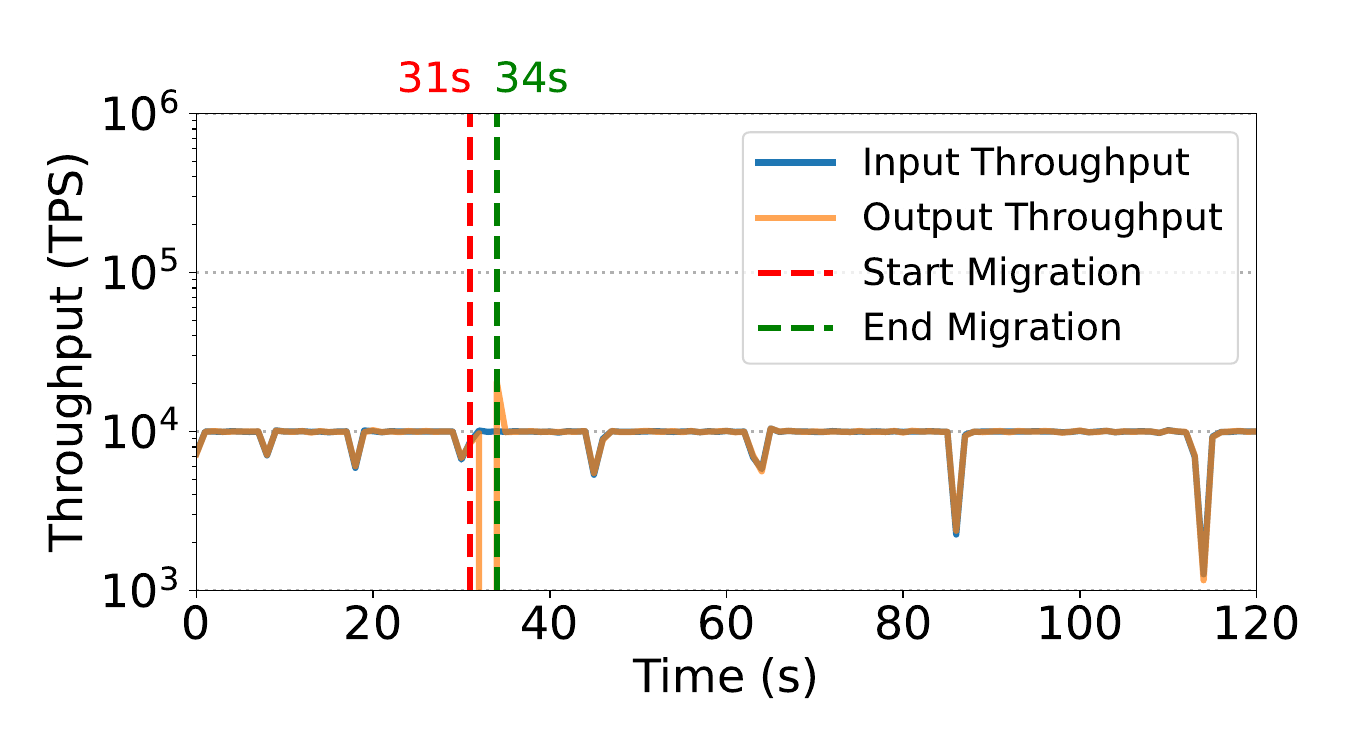}
        \vspace{-5mm}
        \caption*{{\footnotesize S\&R YCSB}}
    \end{subfigure}
    \hfill
    \begin{subfigure}[t]{0.24\textwidth}  
        \centering 
        \captionsetup{justification=centering}
        \includegraphics[width=\linewidth]{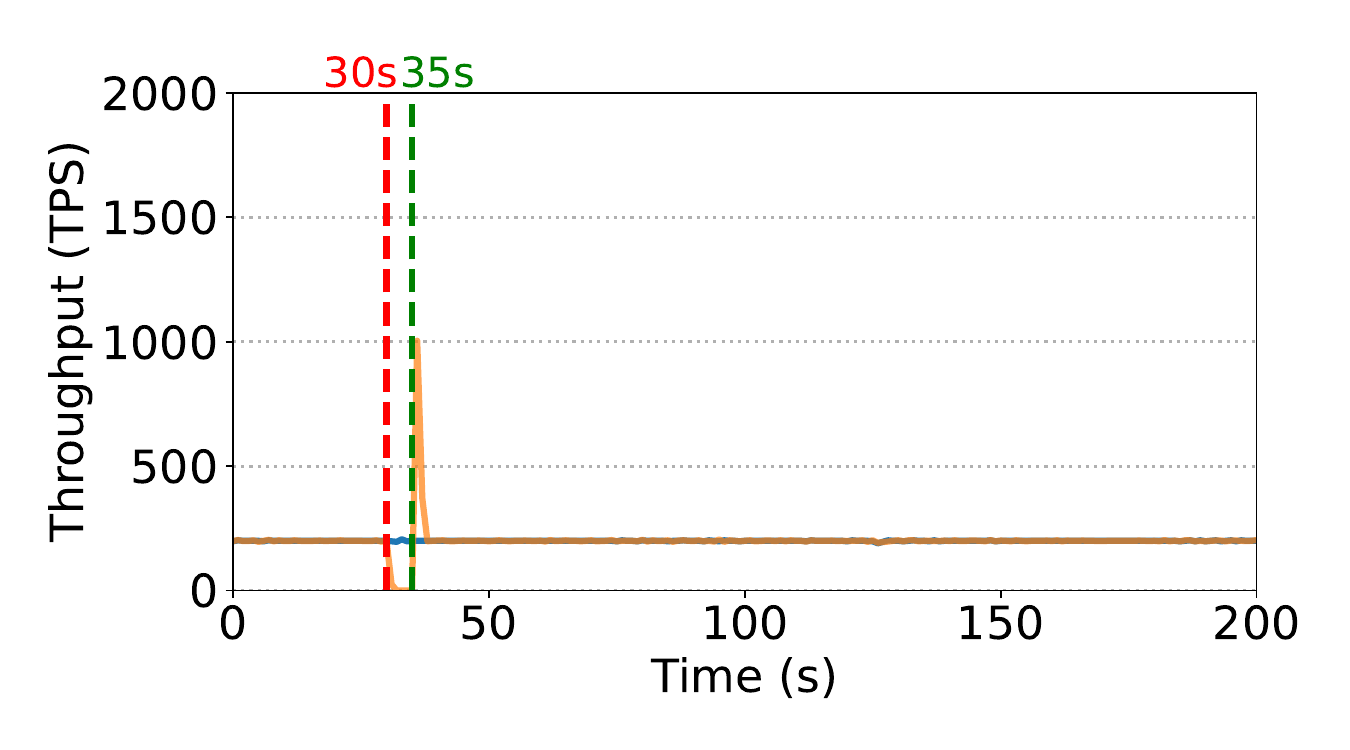}
        \vspace{-5mm}
        \caption*{{\footnotesize S\&R TPC-C}}
    \end{subfigure}
    \hfill
    \begin{subfigure}[t]{0.24\textwidth}   
        \centering 
        \captionsetup{justification=centering}
        \includegraphics[width=\linewidth]{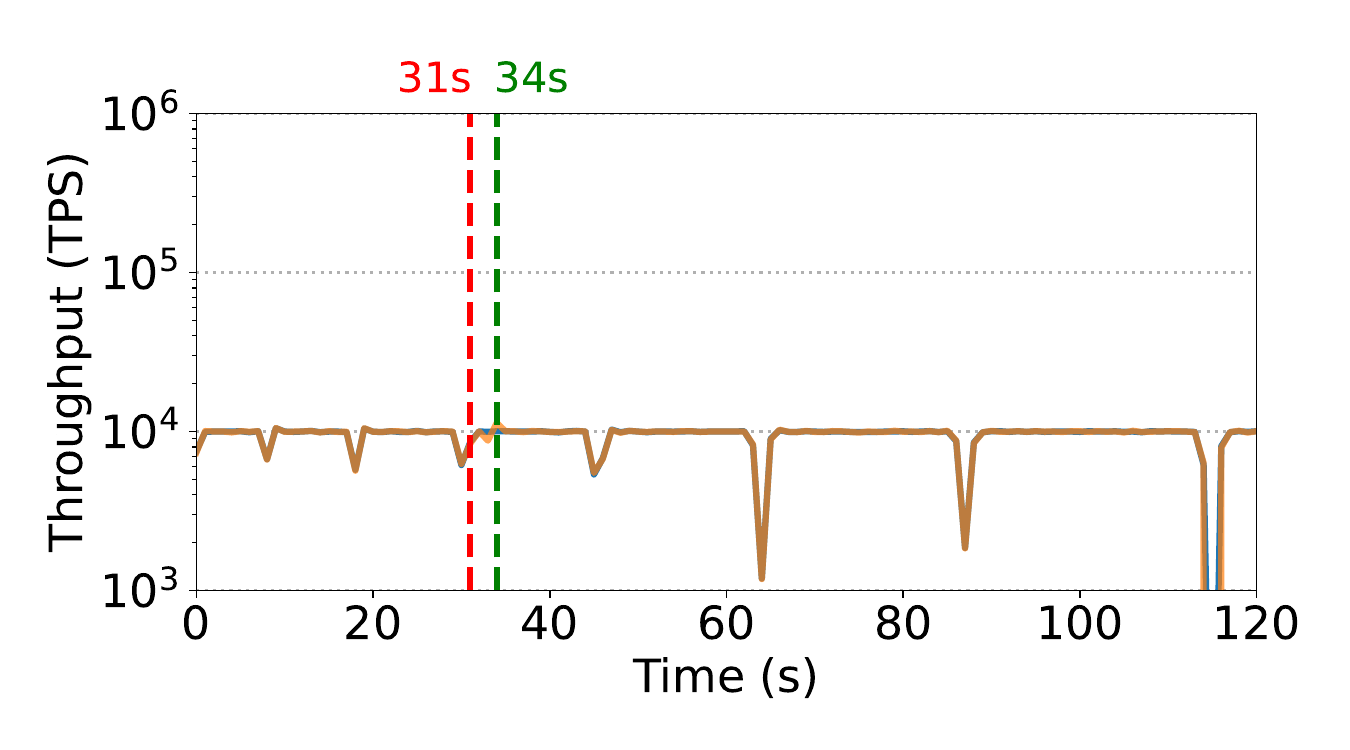}
        \vspace{-5mm}
        \caption*{{\footnotesize Online YCSB}}
    \end{subfigure}
    \hfill
    \begin{subfigure}[t]{0.24\textwidth}   
        \centering 
        \captionsetup{justification=centering}
        \includegraphics[width=\linewidth]{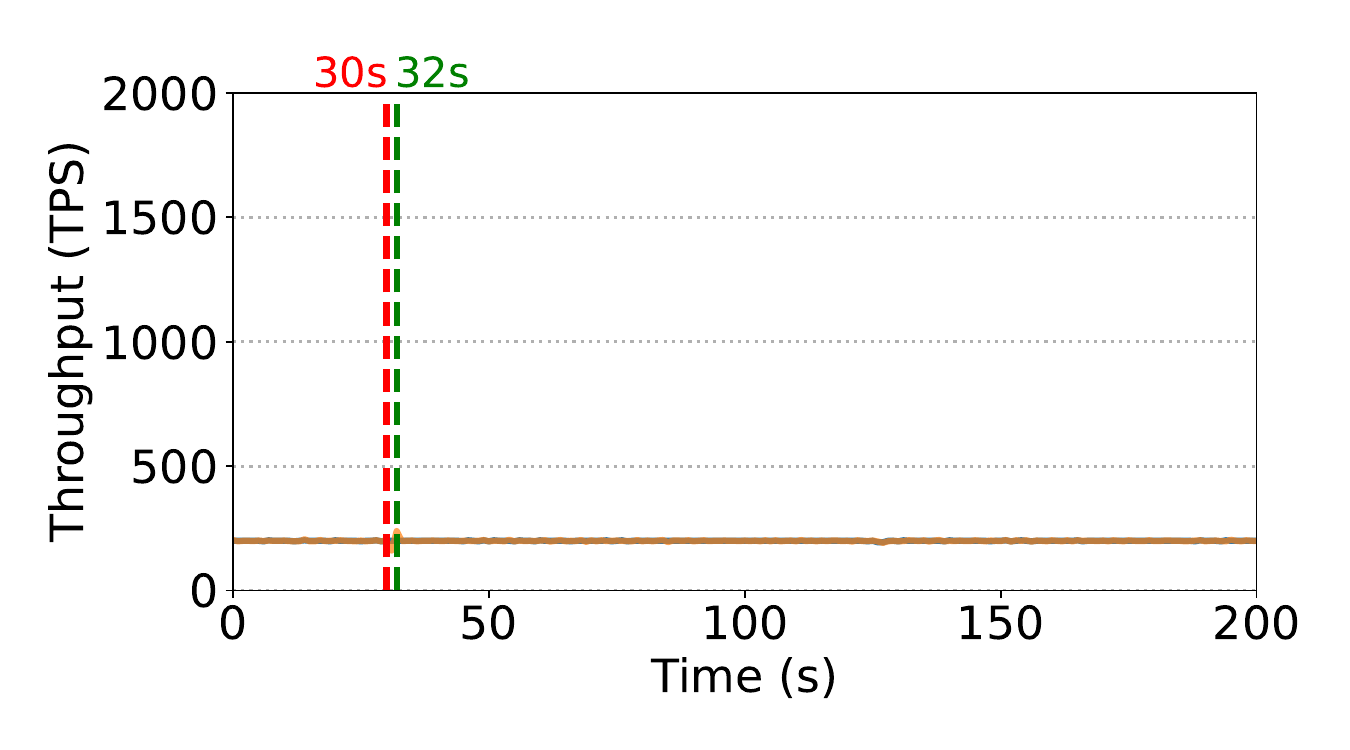}
        \vspace{-5mm}
        \caption*{{\footnotesize Online TPC-C}}
    \end{subfigure}
    \vspace{-2mm}
    \caption{Small State 1GB Scale Down: Latency (top row) and Throughput (bottom row).}
    \label{ch5:fig:scale_down_small_combined}
\end{figure*}

\para{Metrics} In all the migration experiments, we measure input/output throughput, mean latency, and the migration interval. 

\parait{Input/Output Throughput} In addition to \textit{input throughput} as mentioned in Section~\ref{ch4:sec:exp:setup}, we also display the output throughput, which is the number of transaction responses Styx produces per second. During migration, we expect $i)$ the input throughput to remain stable since we do not pause the clients and $ii)$ the output throughput to drop.

\parait{Mean Latency} For the state migration experiments, latency is defined in the same way as in Section~\ref{ch4:sec:exp:setup}. The only deviation from our previous latency reporting is that in line with prior work, \cite{squall,delmonte2020rhino,meces}, we report mean latency instead of the 99th percentile (P99). This decision is motivated by the observation that, during migration, the system enters a transitional phase in which latency spikes are the norm. While P99 latency metrics effectively capture worst-case performance under steady-state conditions, they tend to produce inflated values during migration, which can obscure a clear understanding of system behavior. In this context, mean latency provides a more informative measure of the overall impact of migration on transaction performance.

\parait{Migration Interval} An important migration-only metric is the migration interval, or how long the migration process takes. To show this, we plot the beginning and the end migration timestamps in all experiments.


\begin{figure*}[t]
    \centering
    \captionsetup{justification=centering}
    \begin{subfigure}[t]{0.24\textwidth}
        \centering
        \includegraphics[width=\linewidth]{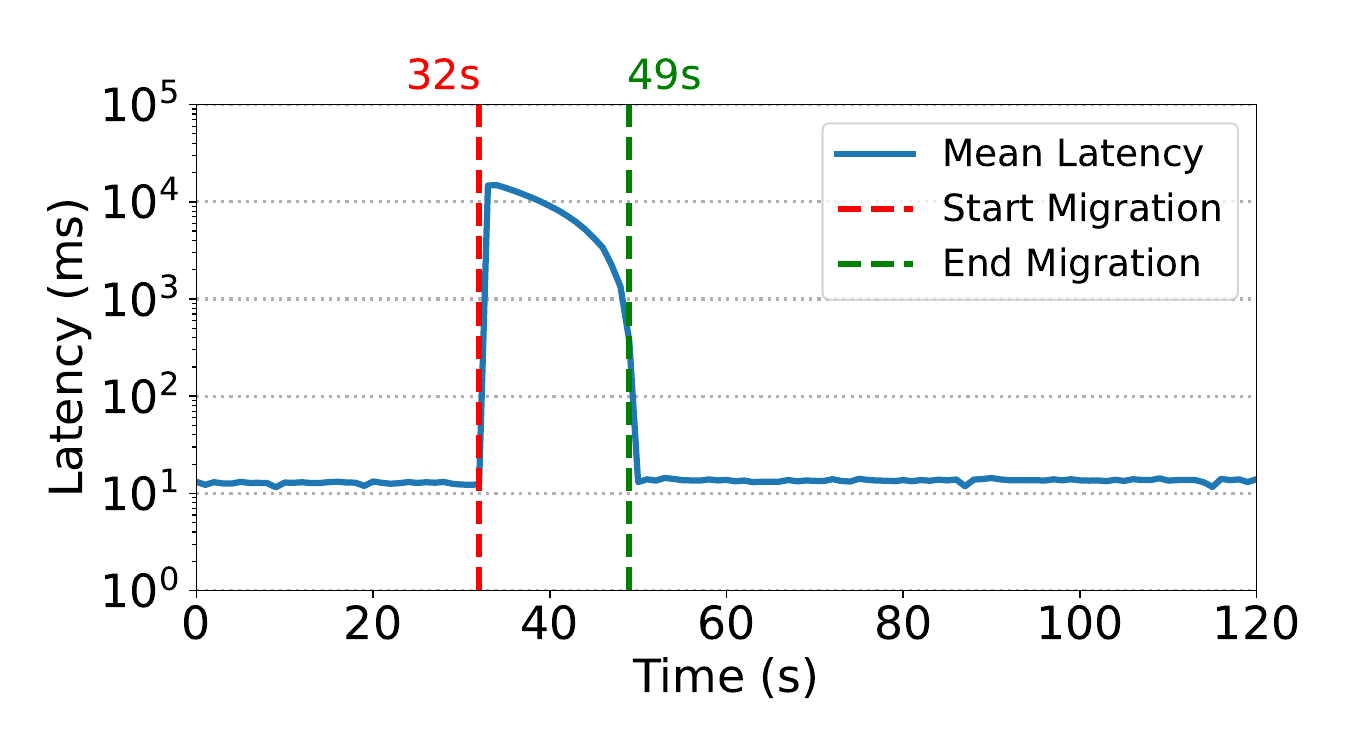}
    \end{subfigure}
    \hfill
    \begin{subfigure}[t]{0.24\textwidth}  
        \centering 
        \includegraphics[width=\linewidth]{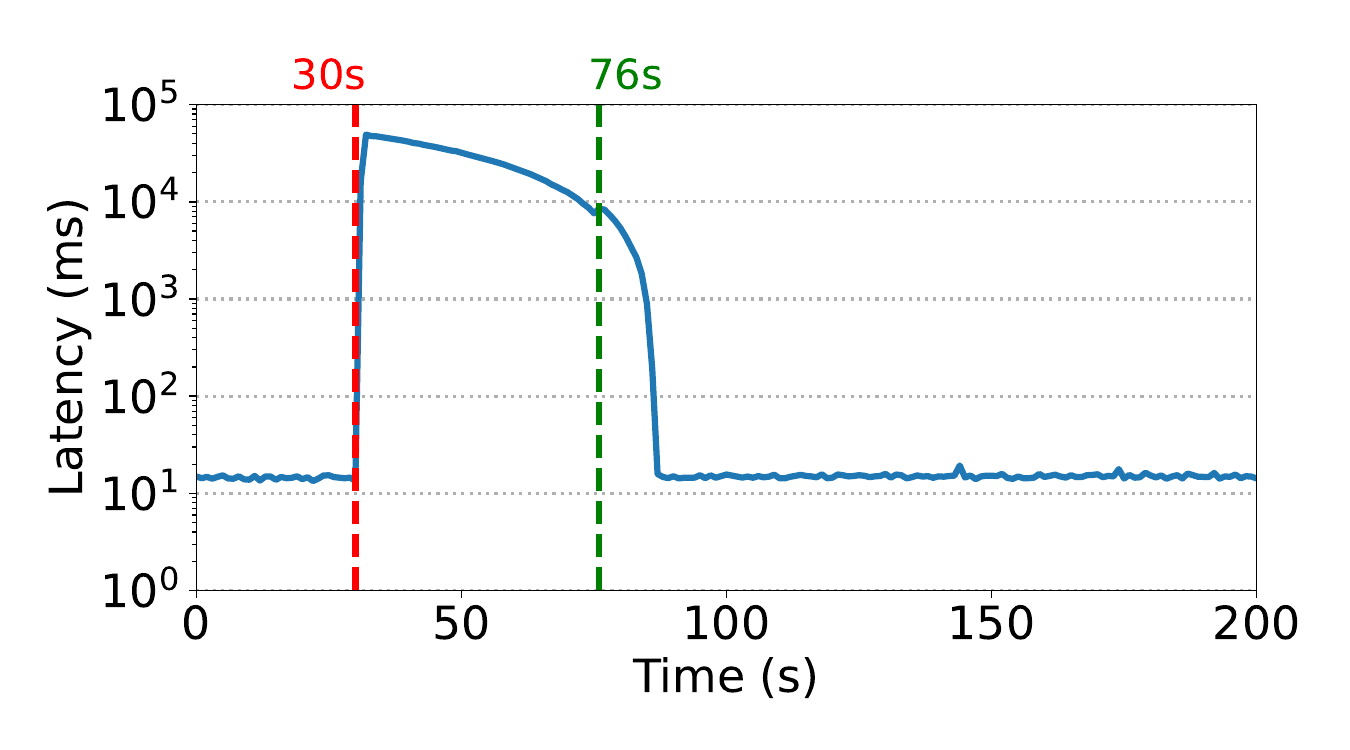}
    \end{subfigure}
    \hfill
    \begin{subfigure}[t]{0.24\textwidth}   
        \centering 
        \includegraphics[width=\linewidth]{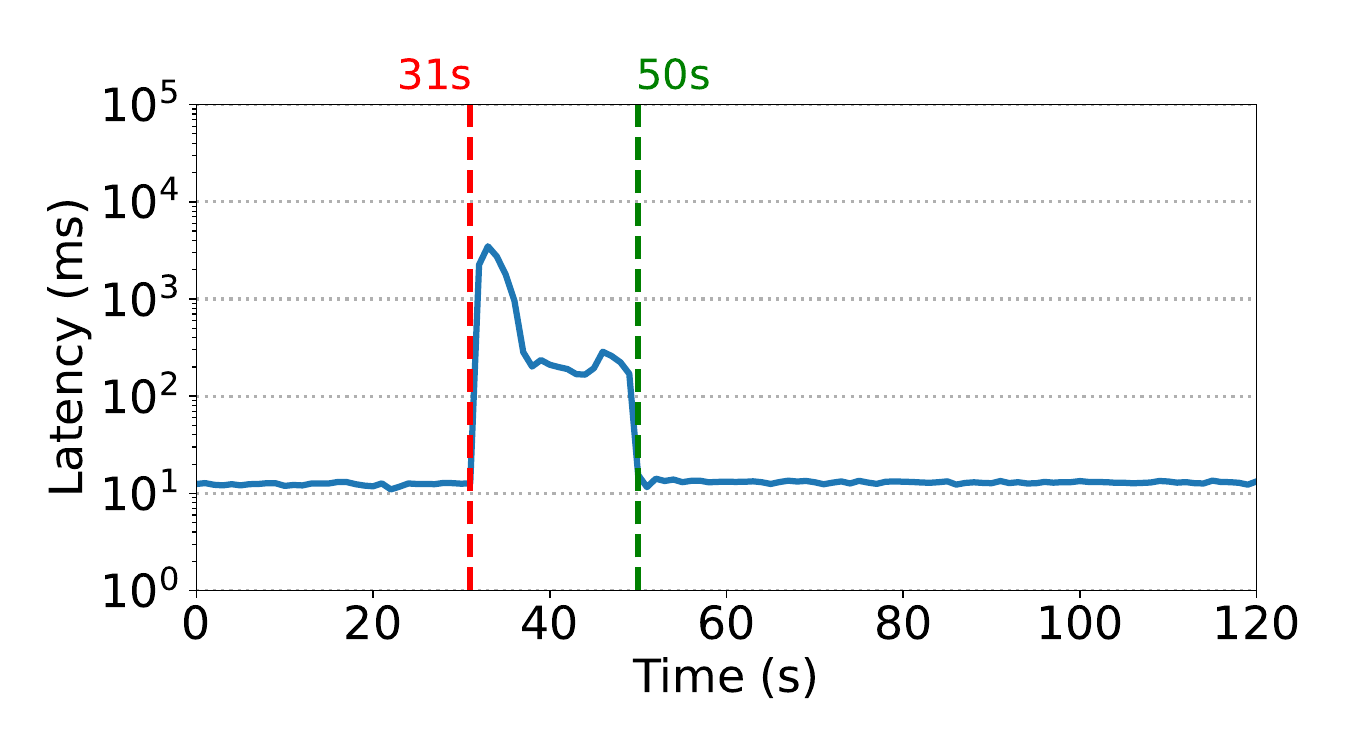}
    \end{subfigure}
    \hfill
    \begin{subfigure}[t]{0.24\textwidth}   
        \centering 
        \includegraphics[width=\linewidth]{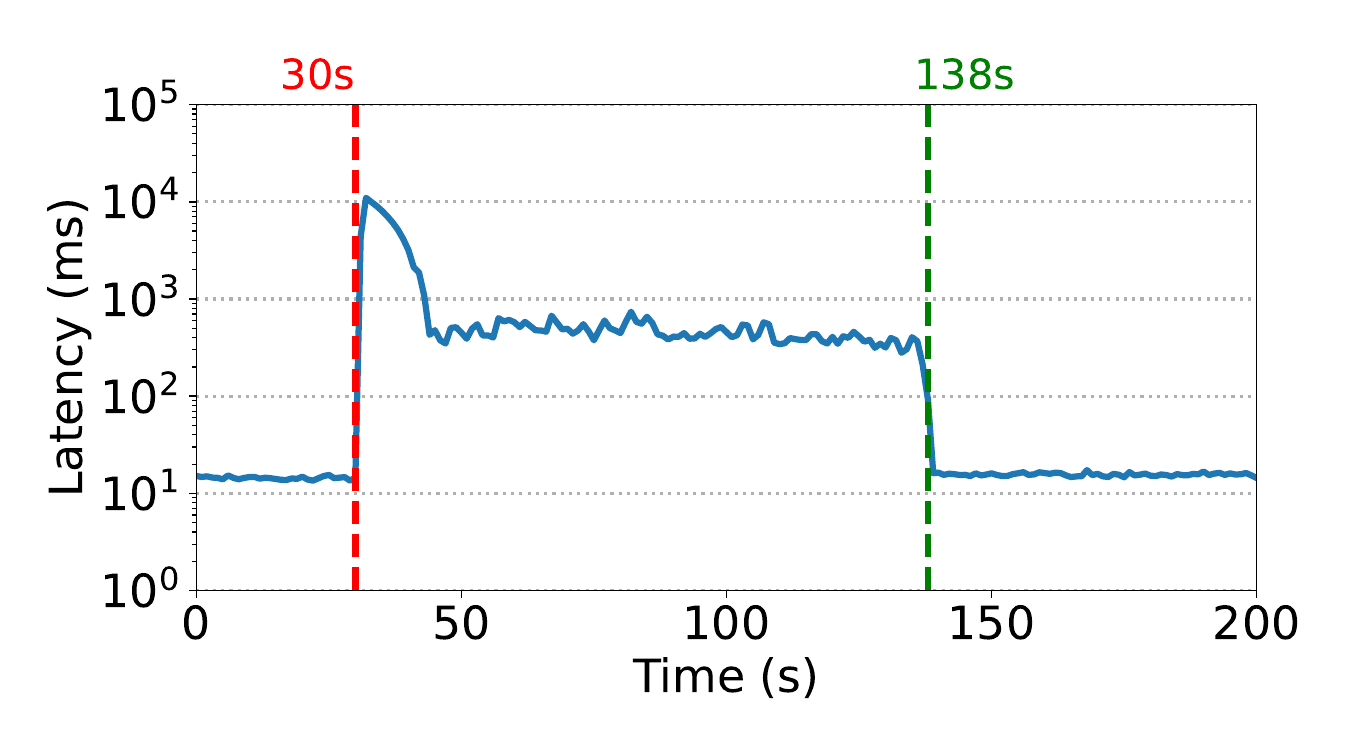}
    \end{subfigure}


    \begin{subfigure}[t]{0.24\textwidth}
        \centering
        \captionsetup{justification=centering}
        \includegraphics[width=\linewidth]{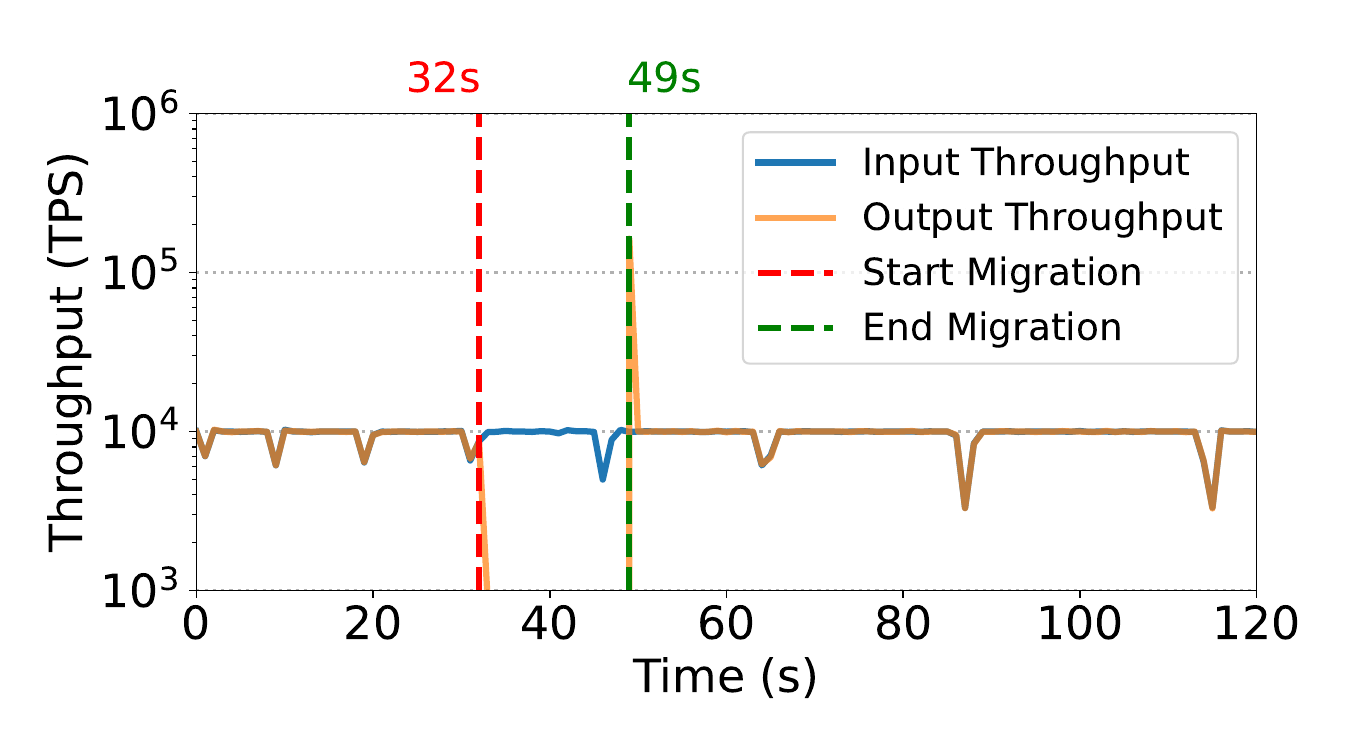}
        \vspace{-5mm}
        \caption*{{\footnotesize S\&R YCSB}}
    \end{subfigure}
    \hfill
    \begin{subfigure}[t]{0.24\textwidth}  
        \centering 
        \captionsetup{justification=centering}
        \includegraphics[width=\linewidth]{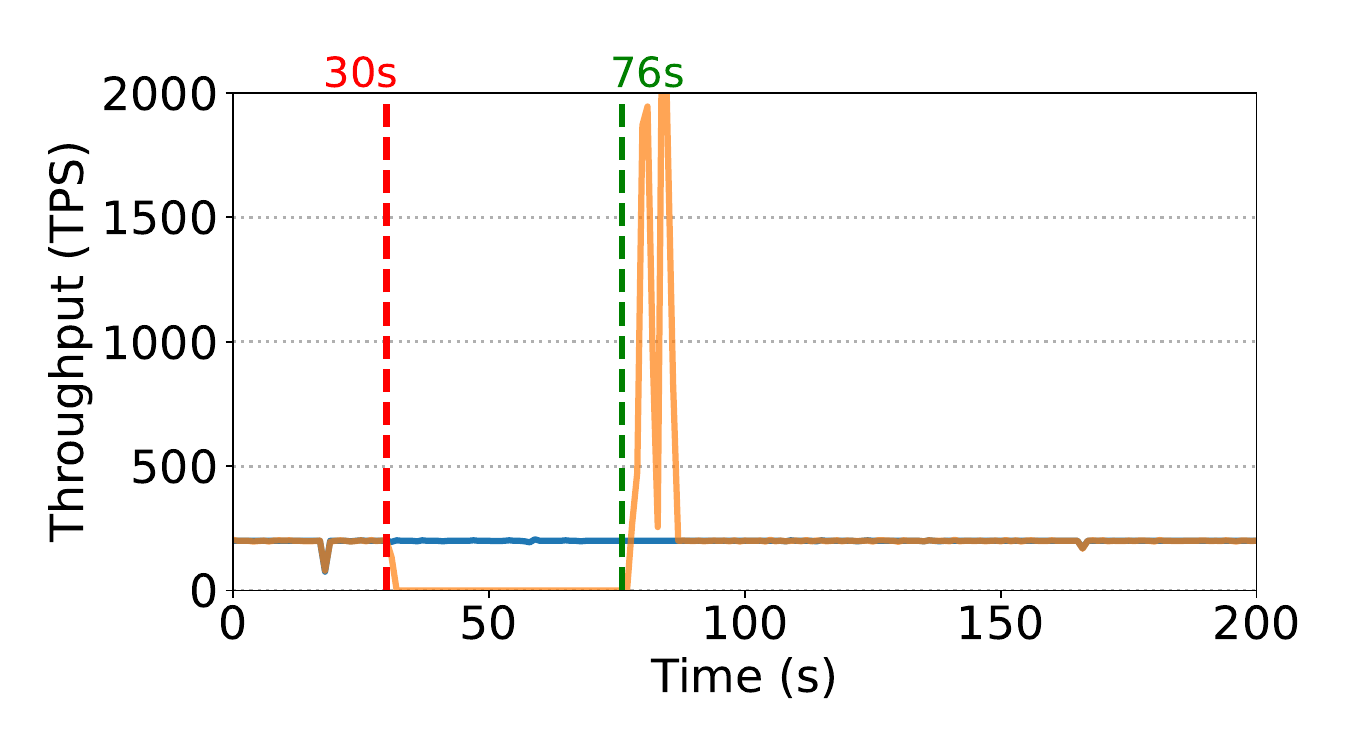}
        \vspace{-5mm}
        \caption*{{\footnotesize S\&R TPC-C}}
    \end{subfigure}
    \hfill
    \begin{subfigure}[t]{0.24\textwidth}   
        \centering 
        \captionsetup{justification=centering}
        \includegraphics[width=\linewidth]{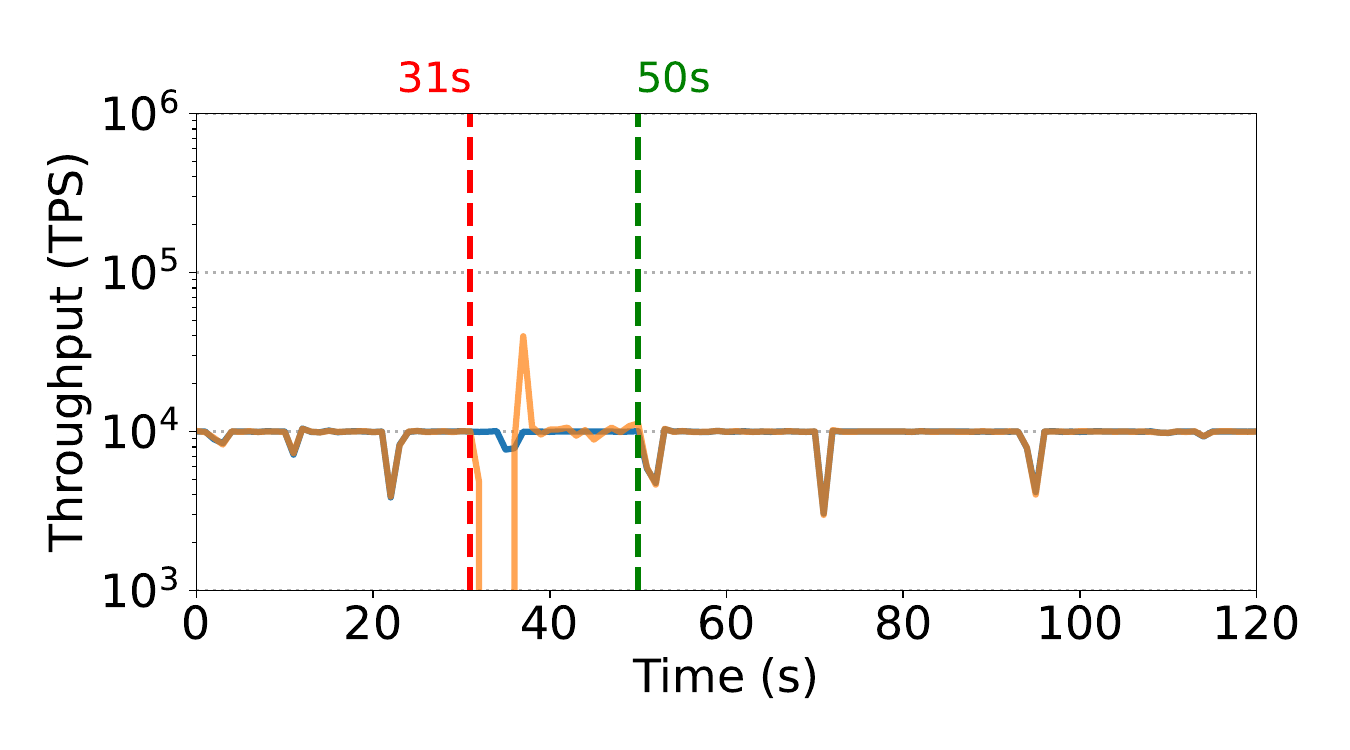}
        \vspace{-5mm}
        \caption*{{\footnotesize Online YCSB}}
    \end{subfigure}
    \hfill
    \begin{subfigure}[t]{0.24\textwidth}   
        \centering 
        \captionsetup{justification=centering}
        \includegraphics[width=\linewidth]{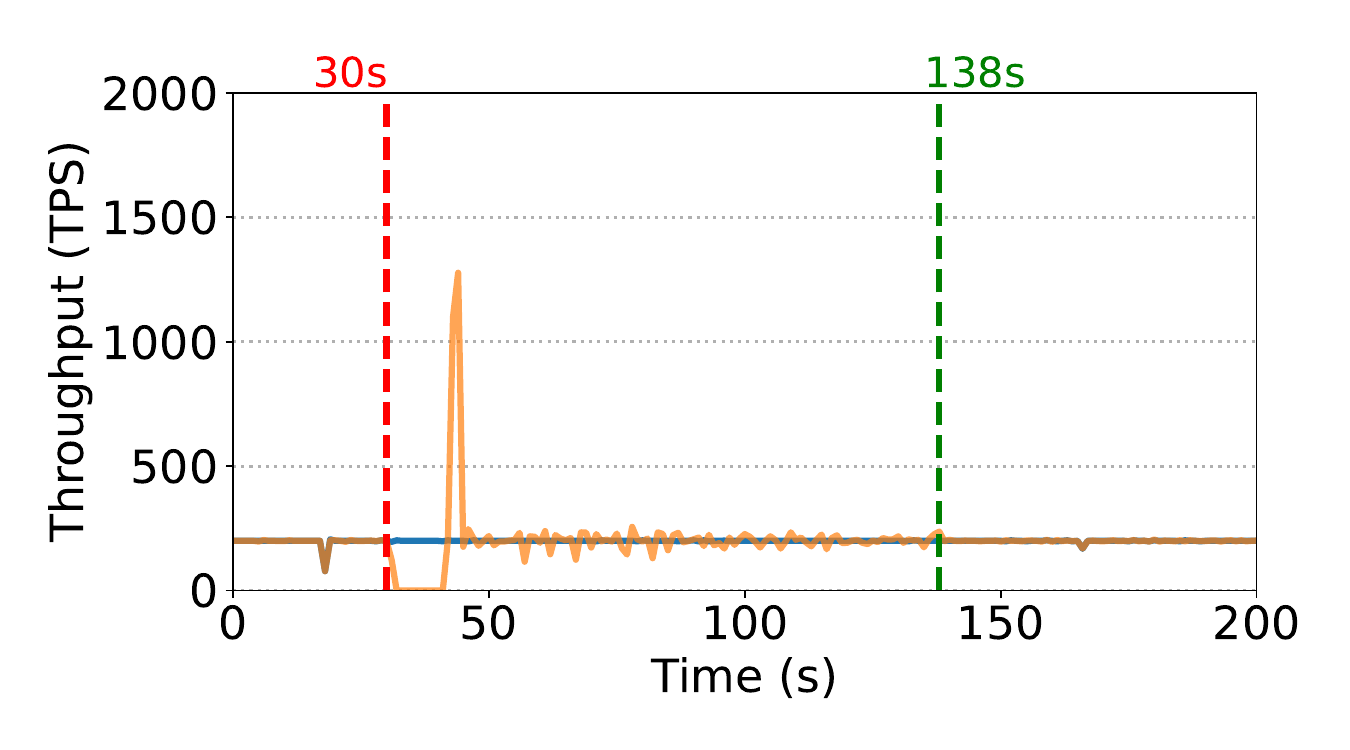}
        \vspace{-5mm}
        \caption*{{\footnotesize Online TPC-C}}
    \end{subfigure}
    \vspace{-2mm}
    \caption{Big State 10 GB Scale Up: Latency (top row) and Throughput (bottom row).}
    \label{ch5:fig:scale_up_big_combined}
\end{figure*}


\begin{figure*}[t]
    \centering
    \captionsetup{justification=centering}
    \begin{subfigure}[t]{0.24\textwidth}
        \centering
        \includegraphics[width=\linewidth]{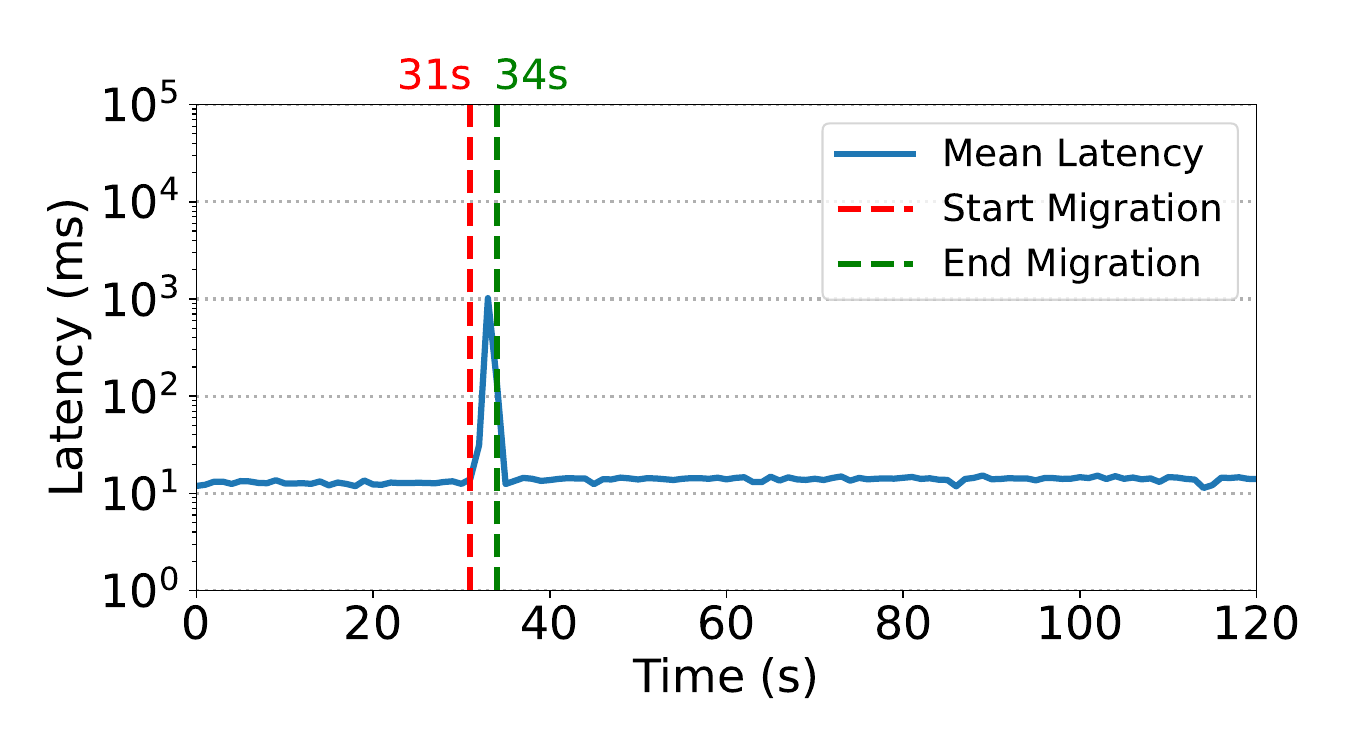}
    \end{subfigure}
    \hfill
    \begin{subfigure}[t]{0.24\textwidth}  
        \centering 
        \includegraphics[width=\linewidth]{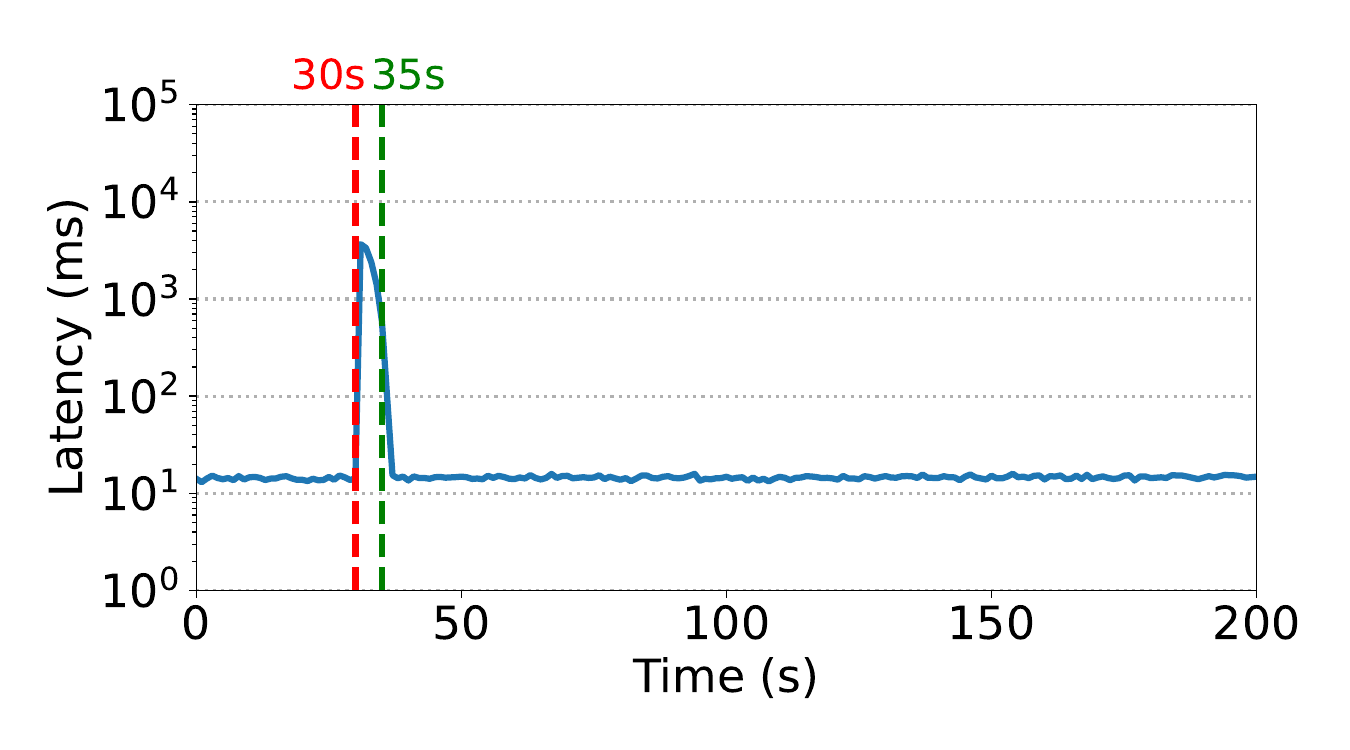}
    \end{subfigure}
    \hfill
    \begin{subfigure}[t]{0.24\textwidth}   
        \centering 
        \includegraphics[width=\linewidth]{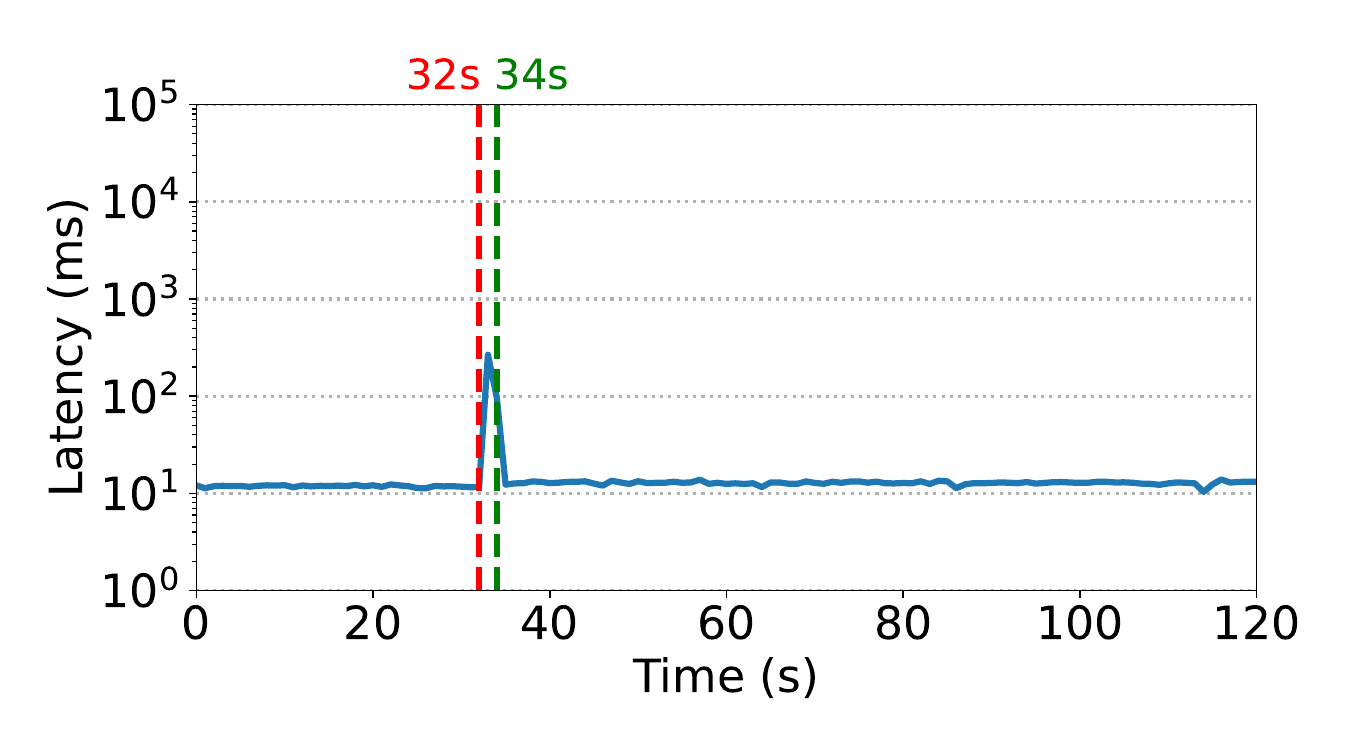}
    \end{subfigure}
    \hfill
    \begin{subfigure}[t]{0.24\textwidth}   
        \centering 
        \includegraphics[width=\linewidth]{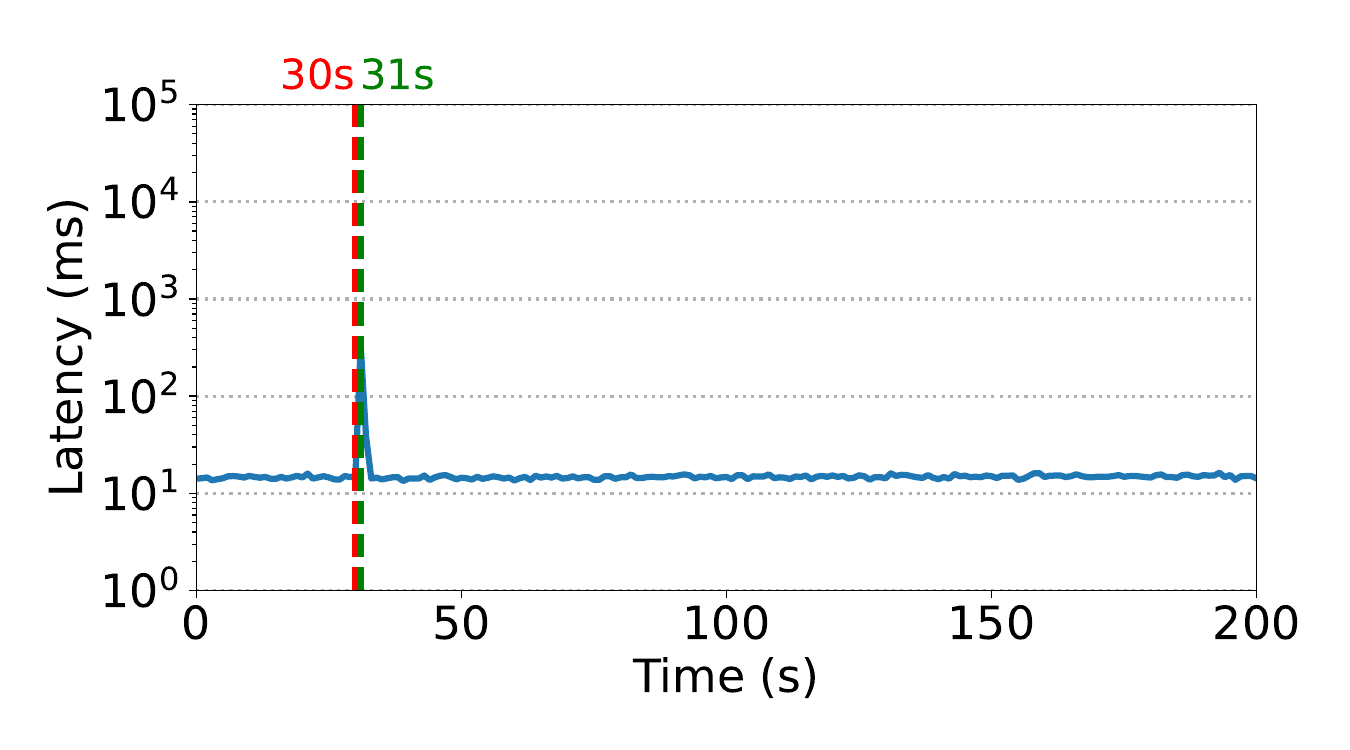}
    \end{subfigure}


    \begin{subfigure}[t]{0.24\textwidth}
        \centering
        \captionsetup{justification=centering}
        \includegraphics[width=\linewidth]{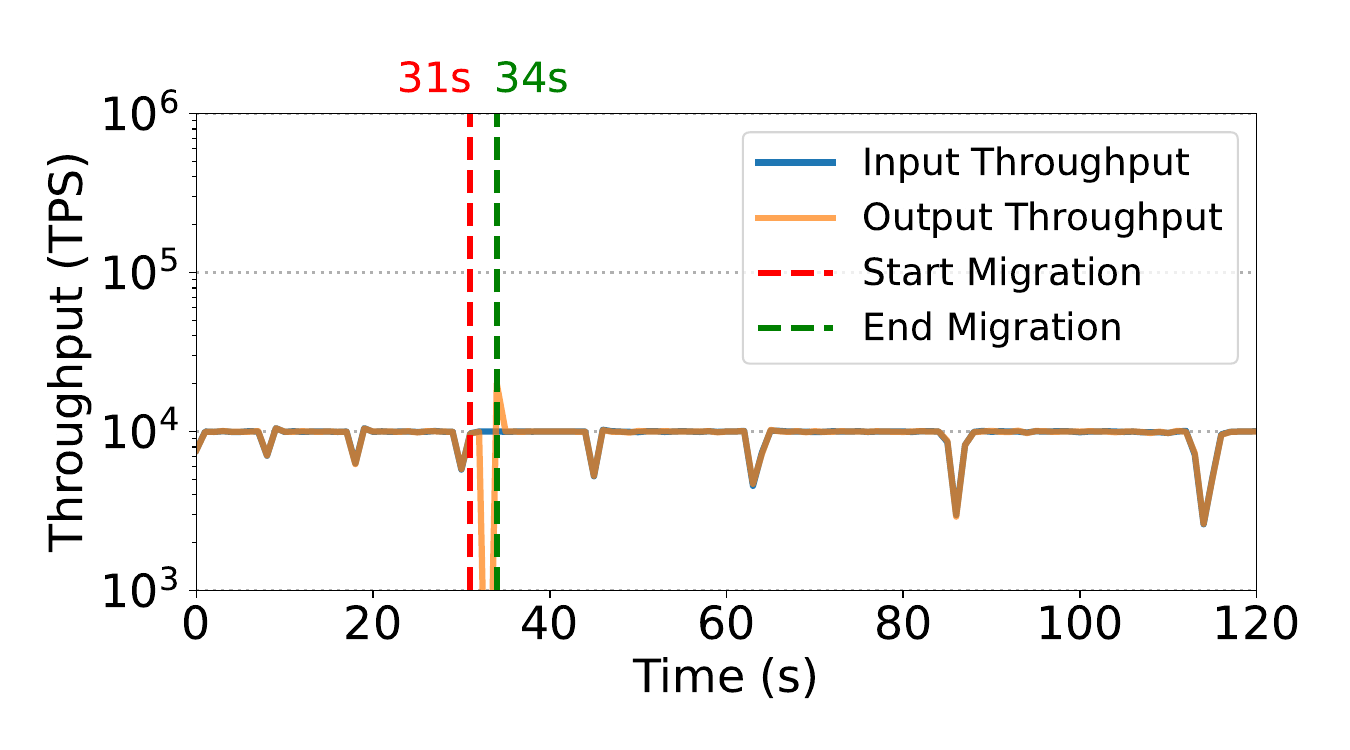}
        \vspace{-5mm}
        \caption*{{\footnotesize S\&R YCSB}}
    \end{subfigure}
    \hfill
    \begin{subfigure}[t]{0.24\textwidth}  
        \centering 
        \captionsetup{justification=centering}
        \includegraphics[width=\linewidth]{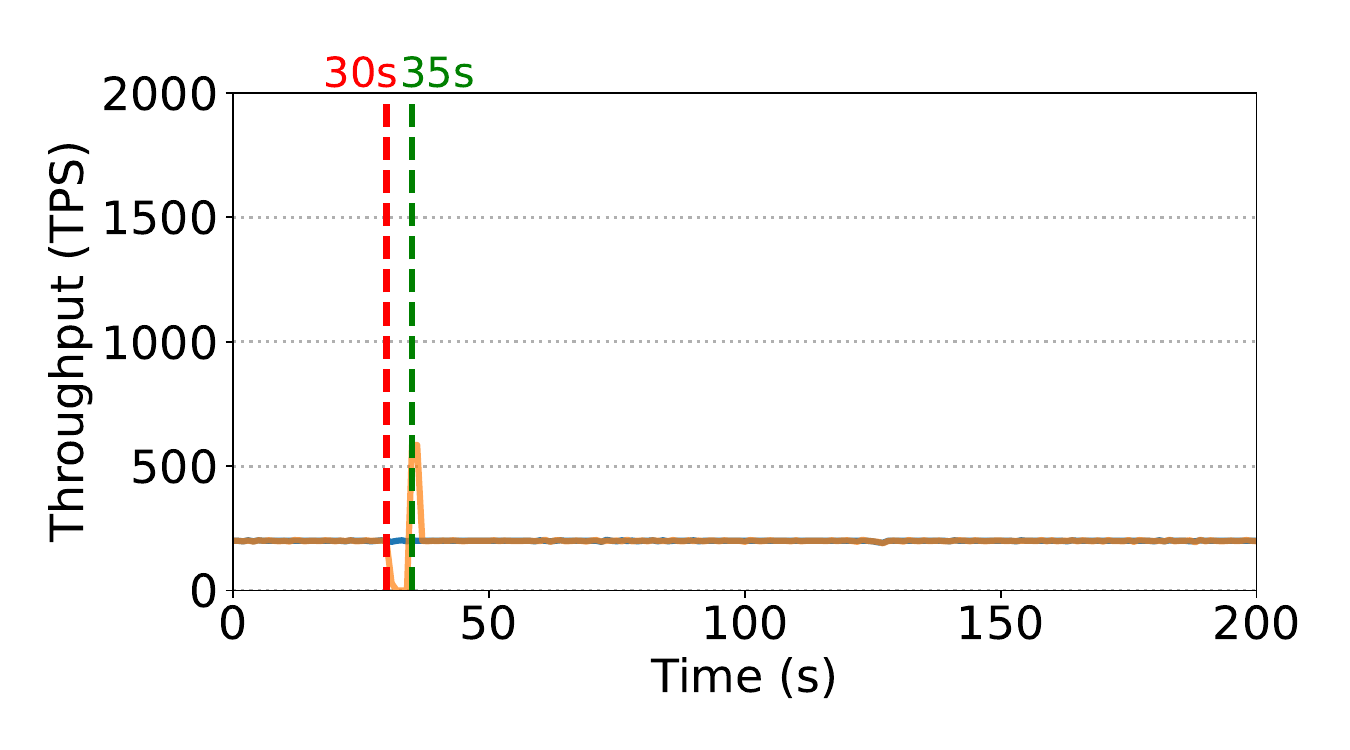}
        \vspace{-5mm}
        \caption*{{\footnotesize S\&R TPC-C}}
    \end{subfigure}
    \hfill
    \begin{subfigure}[t]{0.24\textwidth}   
        \centering 
        \captionsetup{justification=centering}
        \includegraphics[width=\linewidth]{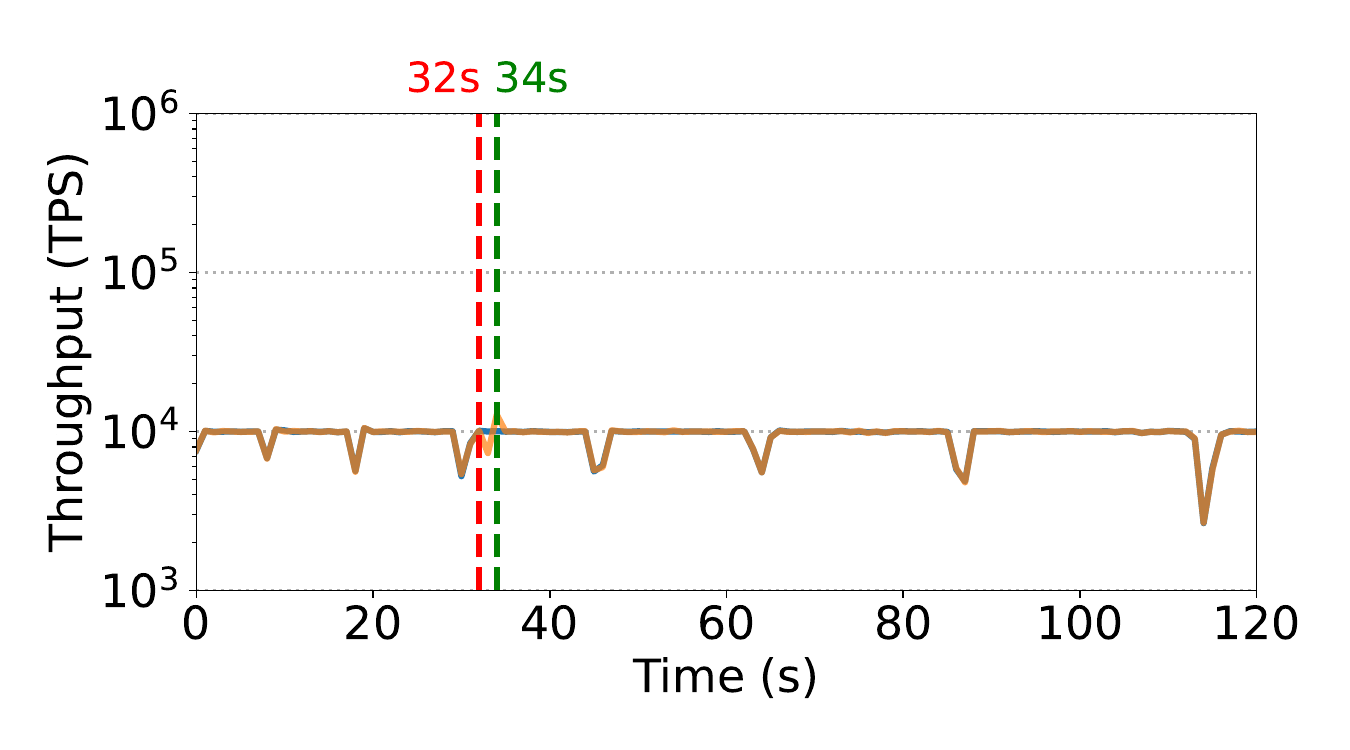}
        \vspace{-5mm}
        \caption*{{\footnotesize Online YCSB}}
    \end{subfigure}
    \hfill
    \begin{subfigure}[t]{0.24\textwidth}   
        \centering 
        \captionsetup{justification=centering}
        \includegraphics[width=\linewidth]{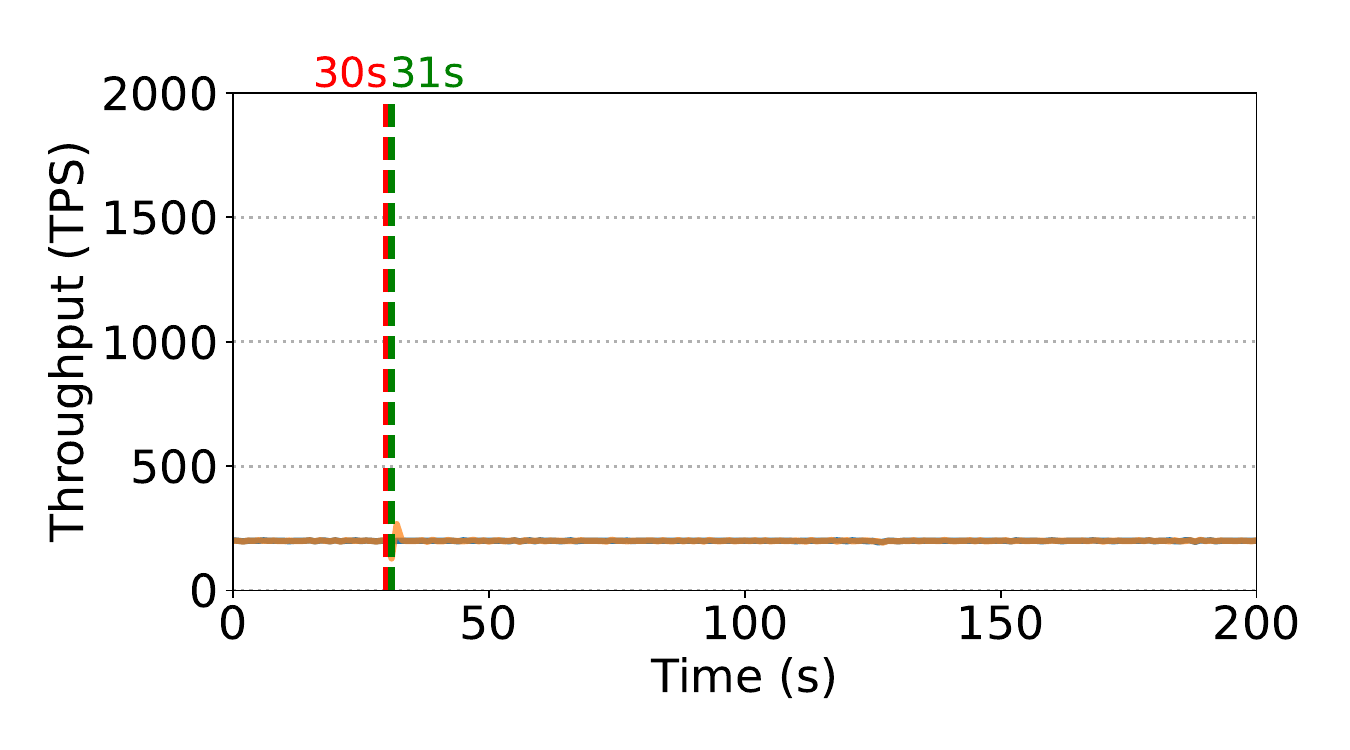}
        \vspace{-5mm}
        \caption*{{\footnotesize Online TPC-C}}
    \end{subfigure}
    \vspace{-2mm}
    \caption{Small State 1GB Scale Up: Latency (top row) and Throughput (bottom row).}
    \vspace{-2mm}
    \label{ch5:fig:scale_up_small_combined}
\end{figure*}

\subsection*{Results}

We have run YCSB and TPC-C workloads in scale-up and down scenarios with big and small state sizes.

\para{Scale-Down} In the scale-down scenario, we go from 16 Styx 1-CPU workers down to 12 in addition to repartitioning the state to the same number of partitions. In Figure~\ref{ch5:fig:scale_down_big_combined}, we observe that the S\&R method in both YCSB and TPC-C is affected by very high latency (tens of seconds), and 14-second downtime in YCSB and 43-second downtime in TPC-C while migrating and repartitioning 10 GB of data. The Online method displays a minor latency hike related to the rehash operation at the beginning of the migration phase in both workloads that do not exceed 10 seconds, and minimal downtime that is close to 5 seconds in YCSB and 10 seconds in TPC-C. In YCSB, both in small and large state, the migration takes around the same time with a minor advantage to the S\&R method. On the contrary, on TPC-C, the Online method takes twice as much since TPC-C contains more keys that need to be transferred, and the async migration is configured to 5 thousand keys per transactional epoch.

In Figure~\ref{ch5:fig:scale_down_small_combined} with the smaller state, we observe the same trends, but the migration impact is much smaller.

\para{Scale-Up} In the scale-up scenario, we go from 12 Styx 1-CPU workers up to 16 in addition to repartitioning the state to the same number of partitions. In Figures~\ref{ch5:fig:scale_up_big_combined} and~\ref{ch5:fig:scale_up_small_combined}
that display the big and small state scale-up experiments; we observe similar behavior to the scale-down experiments, which is to be expected since both migration methods are agnostic to scale-up/down semantics. The only difference is that migration takes longer since 16 threads perform the initial rehashing phase in the scale-down and 12 in the scale-up.

\para{Takeaways} In general, the Online migration method outperforms in all the critical performance indicators such as $i)$ downtime, where the Online migration is at least 4x faster than stop and restart, $ii)$ the peak mean latency does not go above 10 seconds, and once the hashes are computed and Styx catches up to the input, the transactions instantly drop to sub-second latencies. It is important to note that the only metric that the Online method falls behind is the migration time, where its importance lies in the fact that the fault tolerance mechanism is switched off during migration, and the rollback window in the case of failure is larger.

\section{Related Work}

\para{State Migration} Squall~\cite{squall} performs live state migration in transactional systems by locking involved partitions using a dedicated special transaction. It supports on-demand data movement but relies on DBMS-level deadlock handling and assumes range partitioning for optimizations. Clay~\cite{clay} incrementally migrates frequently co-accessed keys using a cost model to balance data movement and transaction performance during the migration. Albatross~\cite{DasNAA11} focuses on migrating state in shared storage systems by incrementally copying in-memory cache and active transaction state. To maintain consistency, Albatross relies on two-phase commit. Meces~\cite{meces} supports fine-grained, on-demand state migration targeting stream processing systems using markers. While ensuring exactly-once semantics, transactions are out-of-scope. Rhino~\cite{rhino} targets query reconfiguration on stream processing systems, maintaining state replicas and utilizing virtual nodes without transaction support.

\para{Autoscaling} Other works targeted dynamic reconfiguration in SPEs. DS2~\cite{ds2} is a control-based autoscaling solution that uses arrival and processing rates of operators to determine when scaling is needed. It focuses on dynamic reconfiguration for SPEs, scaling all operators in a single step by leveraging the topology of the streaming query. To achieve this, the optimal degree of parallelism per operator is calculated progressively. Dhalion~\cite{dhalion} is a control-based framework that uses a backpressure mechanism for rate control. It monitors backpressure signals such as load skew and slow instances to detect resource contention, which triggers scaling actions. DRS \cite{FuDMWYZ17} is a queuing-theory-based autoscaler designed to capture the impact of provisioned resources. It models operator behavior using queuing theory models under steady-state assumptions, offering a more structured framework for latency estimation than relying on backpressure or arrival rates. Lastly, the Horizontal Pod Autoscaler (HPA)\cite{kubernetesHPA} is the default autoscaling solution shipped with Kubernetes. HPA scales horizontally a deployment aiming at matching user-provided target values based on an observed metric, which can be user-defined (e.g., average CPU/memory utilization). However, Styx does not address factors that trigger autoscaling and is left for future work.
\chapter{Discussion \&\ Conclusion}
\label{conclusion}

\dropcap{I}{}n this thesis, we investigated the problem space of transactional cloud applications. While a few systems have tackled parts of this space, each with its advantages and limitations, significant challenges remain. We identified five key research gaps that motivated our work: $i)$ the lack of a principled abstraction and substrate for transactional cloud applications, $ii)$ the limitations of prior systems that hindered their suitability as general-purpose runtimes, $iii)$ the difficulty of developing correct and scalable transactional applications, $iv)$ the challenge of achieving high performance while ensuring strong transactional and fault-tolerance guarantees, and $v)$ the need for efficient and adaptive execution under dynamic workloads. In this chapter, we first summarize the main findings and contributions of this thesis in light of the research gaps identified. We then reflect on the broader implications of our work and discuss avenues for future research that build upon the foundations laid in this thesis.

\section{Main Findings}

\subsection*{Dataflows as a Substrate for Transactional Cloud Applications}

In \Cref{chapter1}, we presented our reasoning behind the selection of a dataflow engine as a suitable substrate for transactional cloud applications, considering the following research question: 

\vspace{2mm}

\rqs{RQ-1} What would be the optimal substrate for a system serving complex cloud applications?

\vspace{2mm}

\noindent To answer \textbf{RQ-1}, we analyzed the requirements of modern cloud applications from a developer's perspective and observed that event-driven microservices share structural similarities with stateful dataflow graphs. This led us to propose dataflow engines as a promising foundation for cloud runtimes, due to their ability to model asynchronous communication, stateful processing, and parallelism. This theoretical foundation guided the practical investigations in subsequent chapters.

\subsection*{The Limitations of Prior Approaches}

To verify our reasoning, in \Cref{chapter2} we enhanced Apache Flink Statefun, an existing SFaaS system, with transactional guarantees and named it T-Statefun. This was led by the following research question:

\vspace{2mm}

\rqs{RQ-2} Is it possible to use an existing SFaaS dataflow system for this purpose? If so, what are the limitations?

\vspace{2mm}

\noindent T-Statefun answered \textbf{RQ-2} and demonstrated that it is indeed feasible to retrofit a stream processing system with transactional capabilities. However, our findings revealed key limitations. Although T-Statefun outperformed existing SFaaS systems, its serializable protocol showed poor scalability under high-contention workloads, and its reliance on remote state access increased latency. Additionally, the system's API complexity made it difficult for non-expert developers to use it effectively, underscoring the need for better programmability.

\subsection*{Difficulty in Programming}

In \Cref{chapter3}, we tackled the programmability issue by introducing \textit{Stateflow}, a domain-specific language embedded in Python that allows developers to write object-oriented cloud applications using familiar imperative constructs. This work addressed the following research question:

\vspace{2mm}

\rqs{RQ-3} Can we design a domain-specific language that works on top of all stream processing systems, creating a simple and easy-to-use object-oriented API?

\vspace{2mm}

\noindent To answer \textbf{RQ-3}, we designed Stateflow as a DSL that compiles Python programs into a dataflow intermediate representation (IR). This abstraction decoupled application logic from the underlying execution engine, enabling the abstraction to be seamlessly compiled against different stream processors. By focusing on Python, we reduced the entry barrier for developers, enabling them to build complex transactional workflows without having to reason about concurrency or distributed systems internals.

\subsection*{Simple Highly Performant Transactional Cloud Applications}

In \Cref{chapter4}, we addressed the limitations in performance and flexibility of existing systems by designing and building Styx, a high-performance distributed dataflow runtime tailored to transactional SFaaS. The following research question drove this work:

\vspace{2mm}

\rqs{RQ-4} Can we build a system that enables developers to write transactional, data-intensive cloud applications without requiring expertise in distributed systems?

\vspace{2mm}

\noindent To answer \textbf{RQ-4}, we developed Styx, a streaming runtime that natively supports Stateflow applications with end-to-end transactional serializability and exactly-once guarantees. Styx achieved high throughput and low latency by applying forward-only deterministic transaction processing and avoiding expensive two-phase commits. It executed arbitrary orchestrations of function calls while preserving consistency and leveraging the dataflow execution model to enable parallelism and fault-tolerance. This work demonstrated that transactional guarantees and performance do not have to be mutually exclusive in cloud application runtimes.

\subsection*{Adaptivity and Efficiency}

Finally, in \Cref{chapter5}, we explored the dynamic scalability of Styx by incorporating elasticity and adaptivity, guided by the following research question:

\vspace{2mm}

\rqs{RQ-5} Can we give Styx elasticity properties, such as state migration, allowing it to become serverless? 

\vspace{2mm}

\noindent To address \textbf{RQ-5}, we extended Styx with a state migration mechanism that enables live relocation of operator state while preserving correctness. We adapted techniques from prior work on stream processor reconfiguration to work within the context of transactional function orchestrations. Our approach enabled Styx to scale applications elastically and recover from load imbalance without downtime, paving the way for a serverless execution model that supports long-lived, stateful cloud applications with fine-grained resource management.

\section{Limitations}

Despite the contributions of this thesis, there are some limitations we would like to address.

\para{High Availability} As discussed in \cref{chapter4}, Styx is designed to recover efficiently from failures using deterministic execution, input message replay, and periodic snapshotting. This enables rapid fault recovery and preserves exactly-once semantics. However, Styx does not yet support high-availability (HA) deployments, where failover is instantaneous, and downtime is imperceptible to users. In its current form, when a failure occurs, affected operators must restore state from durable storage and replay input logs, introducing recovery latency.
Achieving high availability would require mechanisms such as active-standby replication, where each operator has a replica that receives the same input stream and maintains an up-to-date state. Upon failure, the replica could take over immediately, avoiding the need to reload the state from disk and replay messages.

\para{Looser Coupling with the Replayable Message Queue} Styx currently assumes that input streams are ingested via Apache Kafka, which serves as both the ingress and egress mechanism and one of two durability layers. Kafka’s strong ordering guarantees and replayability simplify Styx’s fault tolerance model, enabling deterministic recovery through message replay. However, this design also introduces a tight coupling between Styx and Kafka’s delivery semantics, potentially limiting deployment flexibility and preventing seamless integration with other messaging systems.
In particular, Styx assumes that input queues are partitioned in a way that aligns with its internal parallelism, facilitating efficient sequencing and deterministic execution. These assumptions are made mainly for performance and simplicity, but they are not strictly necessary for correctness. Lifting these assumptions would require rethinking how causal order, partition alignment, and exactly-once semantics are enforced across input queues. Overall, while the current Kafka-centric design benefits from maturity and robustness, enabling broader compatibility with elastic or serverless ingress layers would increase Styx’s deployment portability and make it more suitable for diverse cloud environments.

\para{Analytical Queries}  
 Styx is purpose-built for transactional workloads that require low latency, fine-grained updates, and strong consistency guarantees. It does not currently support analytical queries or Hybrid Transactional and Analytical Processing (HTAP) workloads, which aim to unify real-time decision-making with historical insight. This limits Styx’s applicability in scenarios where the latest transactional state must be queried or analyzed in conjunction with larger historical datasets, for instance, in fraud detection, personalized recommendations, or operational dashboards.
Supporting HTAP workloads in a system like Styx is a non-trivial task. A key challenge lies in balancing the freshness of analytical results with the latency and throughput of transactional processing. Styx could explore integration with analytical backends or log-based data lake ingestion to expose transactional state externally, while maintaining its streaming semantics. Another approach could be the introduction of read-only views with bounded staleness guarantees or side-channel query engines that operate on compacted snapshots. Designing such hybrid execution models while preserving Styx’s exactly-once guarantees, deterministic behavior, and low overhead remains a direction for future work. However, it is beyond the scope of this thesis.

\para{Limited Adaptivity Capabilities} 
Although Styx supports elasticity through manual operator scaling and state migration, it currently lacks fully automated autoscaling mechanisms that are a requirement of modern serverless platforms. Users must manually provision and assign workers, and the system does not yet include dynamic operator placement, load-aware partitioning, or resource-aware scaling. Consequently, Styx cannot automatically respond to changes in workload intensity, nor can it scale to zero during periods of inactivity. This limits its ability to provide cost-efficient and responsive execution, especially in cloud-native and pay-as-you-go settings \cite{trends-and-problems-serverless}.

Nonetheless, the architectural choices in Styx, particularly its modular operator model, logical dataflow execution, and separation between control and data planes, lay a promising foundation for implementing adaptive, serverless behavior. For example, operator placement decisions could be informed by runtime statistics or predictive models, while idle operators could be garbage-collected and rehydrated using persisted state. These enhancements would move Styx toward a fully serverless execution model, where resource management becomes transparent, and users can focus entirely on application logic.

Exploring these capabilities remains outside the scope of this thesis but represents a direction for future research. Integrating reactive scaling policies would enable Styx to support workloads with highly variable demands, thereby creating a general-purpose, transactional serverless dataflow engine.

\para{Long-Running Transactions} Styx currently assumes that transactions are short-lived and bounded in both time and size. As a result, it does not support long-running transactions that can potentially span multiple input epochs. Because the system employs epoch-based processing with coordination barriers, a single long-running transaction can delay the entire epoch's progression, resulting in increased latency and reduced system throughput.

This limitation reflects a broader tension between isolation and liveness in distributed systems. Supporting long-running transactions with strict serializability often requires locking or coordination mechanisms that severely impact performance and availability. In response to this, alternative models such as SAGAs have been proposed \cite{sagas}, which decompose a long transaction into a series of smaller, compensatable steps. These approaches trade strong consistency guarantees for better scalability and resilience, particularly in cloud-native applications \cite{DBLP:conf/cidr/Helland07}.

For many real-world applications, strict isolation may not be necessary. For example, external API calls, human-in-the-loop workflows, or machine learning-based fraud detection can take seconds or even minutes to complete. In these cases, developers may prefer looser guarantees (e.g., eventual consistency, compensation mechanisms) to maintain responsiveness. Enabling support for such workloads would require a redesign of Styx’s execution model, likely through relaxed isolation semantics, compensation patterns, or asynchronous orchestration techniques, which remains an open area for future exploration.

\para{Styx Core Written in Python} Python enables rapid prototyping and lowers the barrier to entry for building distributed systems like Styx. However, it lacks low-level control over memory management, task scheduling, and high-performance networking. These limitations restrict the ability to fine-tune execution performance and concurrency behavior, especially in production environments. While the architectural principles and abstractions of Styx are language-agnostic, implementing the runtime in a systems programming language like Rust or C++ could significantly improve throughput and latency. Languages like Rust, in particular, offer compile-time guarantees that help prevent entire classes of concurrency-related bugs, making it a promising candidate for future reimplementation of the core dataflow engine.

Moreover, user applications in Styx do not need to be written in the same language as the runtime. With the growing maturity of WebAssembly (WASM), it becomes feasible to support applications authored in multiple languages that compile to a common execution target. This could expand Styx’s reach by decoupling the user-facing API from the runtime implementation, paving the way for a multi-language, safer, and more performant execution model.

\section{Future Directions}\label{concl:sec:lim_fw}

In this section, based on the insights and experience we gained from developing Stateflow and Styx, presented in this thesis, we identified open challenges regarding Stateflow/Styx and the field in general. Guided by these, we briefly discuss potential future directions.

\subsection*{Styx}

\para{Auto-Scaling}  
Future work could implement autoscaling mechanisms that adjust operator parallelism in response to workload characteristics. This includes both horizontal scaling (adding/removing workers) and shuffling based on key-based access patterns, ideally without requiring user intervention or manual tuning.

\para{Fault-Tolerant State Migration}  
State migration in existing systems is often performed with fault tolerance mechanisms temporarily disabled, limiting their applicability in highly available environments. A promising direction is the inception of new online, fault-tolerant migration techniques that enable state transfer between workers without pausing or compromising fault tolerance. Moreover, such approaches could open the door to stacked migrations, where new migration plans can be initiated before prior ones have fully completed, enabling faster and more adaptive reconfiguration.

\para{Analytical Workloads for Styx}  
Expanding Styx to support analytical workloads would bridge the gap between real-time transactional processing and decision-making, enabling a unified platform for Hybrid Transactional and Analytical Processing (HTAP). This integration is essential for modern cloud-native applications, which increasingly require the ability to react to operational events with strict transactional guarantees while simultaneously supporting low-latency analytical queries on fresh data.

Achieving HTAP in Styx would involve several system-level enhancements. First, introducing a columnar storage engine alongside the existing row or entity-based execution model would allow Styx to store and scan large volumes of historical data efficiently. Integrating HTAP capabilities also imposes new challenges on resource management and migration. For instance, state migration mechanisms must account for analytical query locality—migrating not only hot keys but also analytical materializations or summary data to where they are most needed. This reinforces the need for lightweight, non-disruptive online migration protocols, which preserve both fault tolerance and query consistency.

\para{Protocol for Long-Running Transactions}  
Currently, the transactional protocol supported in Styx is epoch-based, meaning that new transactions cannot begin until all transactions in an epoch are completed. A long-running transaction would prevent an epoch from completing, resulting in significantly increased latency in all subsequent epochs. To support long-running interactions, future work could explore a specialized protocol for managing either multi-epoch or long-running transactions. This might include compensation mechanisms, sagas, or distributed coordination strategies that maintain serializability without blocking progress.

\subsection*{Stateflow}

\para{Complete Domain-Specific Language for Cloud Applications}  
While Stateflow demonstrates the viability of a Python-embedded DSL for cloud dataflows, future work could formalize a standalone domain-specific language that captures application semantics declaratively. Such a language would provide static analysis, performance optimization, and correctness guarantees that exceed what Python currently offers.

\subsection*{Benchmarks}

\para{Benchmark for Transactional Cloud Applications}  
It is essential to have benchmarks that accurately represent real-world applications to improve existing systems and optimize new ones. However, the community lacks standardized benchmarks for evaluating transactional cloud-native applications. Designing a representative benchmark suite, including realistic workloads, SLAs, and cloud deployment models, would help contextualize systems like Styx and facilitate meaningful comparisons across architectures.



\thumbfalse

\chapter*{Bibliography}
\addcontentsline{toc}{chapter}{Bibliography}
\setheader{Bibliography}

\bibliographystyle{unsrt}
\bibliography{dissertation}

\listoffigures
\listoftables
\glsaddall

\chapter*{Acknowledgments}
\addcontentsline{toc}{chapter}{Acknowledgments}
\setheader{Acknowledgments}

Doing a PhD can often feel like a solitary endeavor; however, this journey would not have been possible without the support of many people, for which I am deeply grateful.

\para{Asterios} Thank you for all the “rope” you gave me to pursue the things I truly enjoyed during this PhD, as well as for the opportunities and life advice along the way.

\para{Geert-Jan} Thank you for our philosophical discussions and for your support during some of the most difficult moments of my PhD.

\para{Thesis Committee} Thank you for the time you dedicated to evaluating this thesis and for honoring me with your participation in the defense. Special thanks to Paris for always answering my questions about his own PhD and distributed systems in general.

\para{George C} You may have joined during my final years, but your critical assistance and fresh perspective made a profound difference. I can say with certainty that I would not be finishing this PhD without you.

\para{Marios} Thank you for your steady guidance throughout these years. Your clarity, patience, and calm perspective helped me navigate difficult moments and grow as a researcher and as a person.

\para{George S. \& Christos} We started as academic brothers, but I now consider you close to real brothers. Thank you for all the great moments, your tolerance, and your help.

\para{My Master Students} I hope I have supported your learning journey as much as you have enriched mine.

\para{WIS Group} Thank you to all colleagues in the WIS group for the supportive and stimulating environment, and for making daily life at the university genuinely enjoyable.

\vspace{2mm}

\noindent To everyone who opened this document to check if their name appears here, thank you. To my friends, both near and far, and to my family: your presence, messages, visits, and timely distractions kept me grounded and sane throughout this journey. Nothing I have achieved would have been possible without your love, patience, and faith in me.

\begin{flushright}
{\makeatletter\itshape
    Kyriakos \\
    Delft, January 2026
\makeatother}
\end{flushright}

\chapter*{Curriculum Vit\ae}
\addcontentsline{toc}{chapter}{Curriculum Vit\ae}
\setheader{Curriculum Vit\ae}

\makeatletter
\authors{\@firstname\ {\titleshape\@lastname}}
\makeatother

\noindent
\begin{longtable}{p{.225\textwidth} p{.70\textwidth}}
    08-12-1994 & Born in Chania, Greece
\end{longtable}

\section*{Professional Experience}

\begin{longtable}{p{.225\textwidth} p{.70\textwidth}}
    2025-present & Software Engineer, Ververica GmbH, Remote \\
    2023 & Research Intern, Huawei Technologies R\&D, The UK \\
    2019-2020 & Research Intern, ING Group, The Netherlands
\end{longtable}

\section*{Education}

\begin{longtable}{p{.225\textwidth} p{.70\textwidth}}
    2021-2025 & Doctor of Philosophy (Ph.D.), Computer Science \\
              & Delft University of Technology, The Netherlands \\ \\
    2018-2020& Master of Science (M.Sc.), Computer Science \\
              & Delft University of Technology, The Netherlands \\ \\
    2012-2018 & Diploma (M.Eng.), Electrical and Computer Engineering \\
              & Technical University of Crete, Greece
\end{longtable}
\chapter*{List of Publications}
\addcontentsline{toc}{chapter}{List of Publications}
\setheader{List of Publications}
\label{publications}

\begin{etaremune}{\small

\item[~1.] J. Arns, H. Ng, \textbf{K. Psarakis}, A. Katsifodimos, and P. Carbone. Event Horizon: Asymmetric Dependencies for Fast Geo-Distributed Operations, in Conference on Innovative Data Systems Research (CIDR), 2026.

\item[\faFileTextO~~2.] \textbf{K. Psarakis}, G. Christodoulou, G. Siachamis, M. Fragkoulis, and A. Katsifodimos. State Migration in Styx: Towards Serverless Transactional Functions, under review, 2025.

\item[\faFileTextO~~3.] \textbf{K. Psarakis}, O. Mraz, G. Christodoulou, G. Siachamis, M. Fragkoulis, and A. Katsifodimos. Styx in Action: Transactional Cloud Applications Made Easy, in Very Large Data Bases (VLDB), 2025.

\item[\faFileTextO~~4.] \textbf{K. Psarakis}, G. Christodoulou, G. Siachamis, M. Fragkoulis, and A. Katsifodimos. Styx: Transactional Stateful Functions on Streaming Dataflows, in  ACM Special Interest Group on Management of Data (SIGMOD), 2025.

\item[5.] R. Laigner, G. Christodoulou, \textbf{K. Psarakis}, A. Katsifodimos, and Y. Zhou. Transactional Cloud Applications: Status Quo, Challenges, and Opportunities, in ACM Special Interest Group on Management of Data (SIGMOD), 2025.

\item[\faFileTextO~~6.] \textbf{K. Psarakis}, G. Christodoulou, M. Fragkoulis, and A. Katsifodimos. Transactional Cloud Applications Go with the (Data)Flow, in Conference on Innovative Data Systems Research (CIDR), 2025.

\item[\faFileTextO~~7.] \textbf{K. Psarakis}, W. Zorgdrager, M. Fragkoulis, G. Salvaneschi, and A. Katsifodimos. Stateful Entities: Object-oriented Cloud Applications as Distributed Dataflows, in International Conference on Extending Database Technology (EDBT), 2024.

\item[8.] G. Siachamis, \textbf{K. Psarakis}, M. Fragkoulis, A. van Deursen, P. Carbone, and A. Katsifodimos. CheckMate: Evaluating Checkpointing Protocols for Streaming Dataflows, in IEEE International Conference on Data Engineering (ICDE), 2024.

\item[9.] G. Siachamis, G. Christodoulou, \textbf{K. Psarakis}, M. Fragkoulis, A. van Deursen, A. Katsifodimos. Evaluating Stream Processing Autoscalers, in ACM Conference on Distributed and Event‐Based Systems (DEBS), 2024.

\item[10.] G. Siachamis, \textbf{K. Psarakis}, M. Fragkoulis, O. Papapetrou, A. van Deursen, A. Katsifodimos. Adaptive Distributed Streaming Similarity Joins, in ACM Conference on Distributed and Event‐Based Systems (DEBS), 2023.

\item[11.] A. Ionescu, K. Patroumpas, \textbf{K. Psarakis}, G. Chatzigeorgakidis, D. Collarana, K. Barenscher, D. Skoutas, A. Katsifodimos, and S. Athanasiou. Topio: an Open-Source Web Platform for Trading Geospatial Data, in International Conference on Web Engineering (ICWE), 2023.
    
\item[\faTrophy~~12.] A. Ionescu, A. Alexandridou, L. Ikonomou, \textbf{K. Psarakis}, K. Patroumpas, G. Chatzigeorgakidis, D. Skoutas, S. Athanasiou, R. Hai, and A. Katsifodimos. Topio Marketplace: Search and Discovery of Geospatial Data, in International Conference on Extending Database Technology (EDBT), 2023.

\item[\faFileTextO~~13.] \textbf{K. Psarakis}, W. Zorgdrager, M. Fragkoulis, G. Salvaneschi, and A. Katsifodimos. Stateful Entities: Object-oriented Cloud Applications as Distributed Dataflows, in Conference on Innovative Data Systems Research (CIDR), 2023.

\item[\faFileTextO~~14.] M. de Heus, \textbf{K. Psarakis}, M. Fragkoulis, and A. Katsifodimos. Transactions across serverless functions leveraging stateful dataflows, in Elsevier Information Systems, 2022.

\item[15.] C. Koutras, \textbf{K. Psarakis}, G. Siachamis, A. Ionescu, M. Fragkoulis, A. Bonifati, and A. Katsifodimos. Valentine in Action: Matching Tabular Data at Scale, in Very Large Data Bases (VLDB), 2021.

\item[\faFileTextO~~\faTrophy~~16.] M. de Heus, \textbf{K. Psarakis}, M. Fragkoulis, and A. Katsifodimos. Distributed Transactions on Serverless Stateful Functions, in ACM Conference on Distributed and Event‐Based Systems (DEBS), 2021.

\item[17.] C. Koutras, G. Siachamis, A. Ionescu, \textbf{K. Psarakis}, J. Brons, M. Fragkoulis, C. Lofi, A. Bonifati, and A. Katsifodimos. Valentine: Evaluating matching techniques for dataset discovery, in IEEE International Conference on Data Engineering (ICDE), 2021.

}\end{etaremune}

\vspace{0.5cm}
\noindent
\faFileTextO~~Included in this thesis.\\
\faTrophy~~Won a best paper or demonstration award.
\chapter*{SIKS Dissertation Series}
\addcontentsline{toc}{chapter}{SIKS Dissertation Series}
\setheader{SIKS Dissertation Series}

\begin{xltabular}{\linewidth}{@{} l @{\hspace{0.5em}} l @{\hspace{1em}} X @{}}


2016

	&	 01	&	 Syed Saiden Abbas (RUN), Recognition of Shapes by Humans and Machines\\

	&	 02	&	 Michiel Christiaan Meulendijk (UU), Optimizing medication reviews through decision support: prescribing a better pill to swallow\\

	&	 03	&	 Maya Sappelli (RUN), Knowledge Work in Context: User Centered Knowledge Worker Support\\

	&	 04	&	 Laurens Rietveld (VUA), Publishing and Consuming Linked Data\\

	&	 05	&	 Evgeny Sherkhonov (UvA), Expanded Acyclic Queries: Containment and an Application in Explaining Missing Answers\\

	&	 06	&	 Michel Wilson (TUD), Robust scheduling in an uncertain environment\\

	&	 07	&	 Jeroen de Man (VUA), Measuring and modeling negative emotions for virtual training\\

	&	 08	&	 Matje van de Camp (TiU), A Link to the Past: Constructing Historical Social Networks from Unstructured Data\\

	&	 09	&	 Archana Nottamkandath (VUA), Trusting Crowdsourced Information on Cultural Artefacts\\

	&	 10	&	 George Karafotias (VUA), Parameter Control for Evolutionary Algorithms\\

	&	 11	&	 Anne Schuth (UvA), Search Engines that Learn from Their Users\\

	&	 12	&	 Max Knobbout (UU), Logics for Modelling and Verifying Normative Multi-Agent Systems\\

	&	 13	&	 Nana Baah Gyan (VUA), The Web, Speech Technologies and Rural Development in West Africa - An ICT4D Approach\\

	&	 14	&	 Ravi Khadka (UU), Revisiting Legacy Software System Modernization\\

	&	 15	&	 Steffen Michels (RUN), Hybrid Probabilistic Logics - Theoretical Aspects, Algorithms and Experiments\\

	&	 16	&	 Guangliang Li (UvA), Socially Intelligent Autonomous Agents that Learn from Human Reward\\

	&	 17	&	 Berend Weel (VUA), Towards Embodied Evolution of Robot Organisms\\

	&	 18	&	 Albert Mero\~{n}o Pe\~{n}uela (VUA), Refining Statistical Data on the Web\\

	&	 19	&	 Julia Efremova (TU/e), Mining Social Structures from Genealogical Data\\

	&	 20	&	 Daan Odijk (UvA), Context \& Semantics in News \& Web Search\\

	&	 21	&	 Alejandro Moreno C\'{e}lleri (UT), From Traditional to Interactive Playspaces: Automatic Analysis of Player Behavior in the Interactive Tag Playground\\

	&	 22	&	 Grace Lewis (VUA), Software Architecture Strategies for Cyber-Foraging Systems\\

	&	 23	&	 Fei Cai (UvA), Query Auto Completion in Information Retrieval\\

	&	 24	&	 Brend Wanders (UT), Repurposing and Probabilistic Integration of Data; An Iterative and data model independent approach\\

	&	 25	&	 Julia Kiseleva (TU/e), Using Contextual Information to Understand Searching and Browsing Behavior\\

	&	 26	&	 Dilhan Thilakarathne (VUA), In or Out of Control: Exploring Computational Models to Study the Role of Human Awareness and Control in Behavioural Choices, with Applications in Aviation and Energy Management Domains\\

	&	 27	&	 Wen Li (TUD), Understanding Geo-spatial Information on Social Media\\

	&	 28	&	 Mingxin Zhang (TUD), Large-scale Agent-based Social Simulation - A study on epidemic prediction and control\\

	&	 29	&	 Nicolas H\"{o}ning (TUD), Peak reduction in decentralised electricity systems - Markets and prices for flexible planning\\

	&	 30	&	 Ruud Mattheij (TiU), The Eyes Have It\\

	&	 31	&	 Mohammad Khelghati (UT), Deep web content monitoring\\

	&	 32	&	 Eelco Vriezekolk (UT), Assessing Telecommunication Service Availability Risks for Crisis Organisations\\

	&	 33	&	 Peter Bloem (UvA), Single Sample Statistics, exercises in learning from just one example\\

	&	 34	&	 Dennis Schunselaar (TU/e), Configurable Process Trees: Elicitation, Analysis, and Enactment\\

	&	 35	&	 Zhaochun Ren (UvA), Monitoring Social Media: Summarization, Classification and Recommendation\\

	&	 36	&	 Daphne Karreman (UT), Beyond R2D2: The design of nonverbal interaction behavior optimized for robot-specific morphologies\\

	&	 37	&	 Giovanni Sileno (UvA), Aligning Law and Action - a conceptual and computational inquiry\\

	&	 38	&	 Andrea Minuto (UT), Materials that Matter - Smart Materials meet Art \& Interaction Design\\

	&	 39	&	 Merijn Bruijnes (UT), Believable Suspect Agents; Response and Interpersonal Style Selection for an Artificial Suspect\\

	&	 40	&	 Christian Detweiler (TUD), Accounting for Values in Design\\

	&	 41	&	 Thomas King (TUD), Governing Governance: A Formal Framework for Analysing Institutional Design and Enactment Governance\\

	&	 42	&	 Spyros Martzoukos (UvA), Combinatorial and Compositional Aspects of Bilingual Aligned Corpora\\

	&	 43	&	 Saskia Koldijk (RUN), Context-Aware Support for Stress Self-Management: From Theory to Practice\\

	&	 44	&	 Thibault Sellam (UvA), Automatic Assistants for Database Exploration\\

	&	 45	&	 Bram van de Laar (UT), Experiencing Brain-Computer Interface Control\\

	&	 46	&	 Jorge Gallego Perez (UT), Robots to Make you Happy\\

	&	 47	&	 Christina Weber (UL), Real-time foresight - Preparedness for dynamic innovation networks\\

	&	 48	&	 Tanja Buttler (TUD), Collecting Lessons Learned\\

	&	 49	&	 Gleb Polevoy (TUD), Participation and Interaction in Projects. A Game-Theoretic Analysis\\

	&	 50	&	 Yan Wang (TiU), The Bridge of Dreams: Towards a Method for Operational Performance Alignment in IT-enabled Service Supply Chains\\

\midrule

2017

	&	 01	&	 Jan-Jaap Oerlemans (UL), Investigating Cybercrime\\

	&	 02	&	 Sjoerd Timmer (UU), Designing and Understanding Forensic Bayesian Networks using Argumentation\\

	&	 03	&	 Dani\"{e}l Harold Telgen (UU), Grid Manufacturing; A Cyber-Physical Approach with Autonomous Products and Reconfigurable Manufacturing Machines\\

	&	 04	&	 Mrunal Gawade (CWI), Multi-core Parallelism in a Column-store\\

	&	 05	&	 Mahdieh Shadi (UvA), Collaboration Behavior\\

	&	 06	&	 Damir Vandic (EUR), Intelligent Information Systems for Web Product Search\\

	&	 07	&	 Roel Bertens (UU), Insight in Information: from Abstract to Anomaly\\

	&	 08	& 	 Rob Konijn (VUA), Detecting Interesting Differences:Data Mining in Health Insurance Data using Outlier Detection and Subgroup Discovery\\

	&	 09	&	 Dong Nguyen (UT), Text as Social and Cultural Data: A Computational Perspective on Variation in Text\\

	&	 10	&	 Robby van Delden (UT), (Steering) Interactive Play Behavior\\

	&	 11	&	 Florian Kunneman (RUN), Modelling patterns of time and emotion in Twitter \#anticipointment\\

	&	 12	&	 Sander Leemans (TU/e), Robust Process Mining with Guarantees\\

 	&	 13	& 	 Gijs Huisman (UT), Social Touch Technology - Extending the reach of social touch through haptic technology\\

 	&	 14	&	 Shoshannah Tekofsky (TiU), You Are Who You Play You Are: Modelling Player Traits from Video Game Behavior\\

	&	 15	&	 Peter Berck (RUN),  Memory-Based Text Correction\\

	&	 16	&	 Aleksandr Chuklin (UvA), Understanding and Modeling Users of Modern Search Engines\\

	&	 17	&	 Daniel Dimov (UL), Crowdsourced Online Dispute Resolution\\

	&	 18	&	 Ridho Reinanda (UvA), Entity Associations for Search\\

	&	 19	& 	 Jeroen Vuurens (UT), Proximity of Terms, Texts and Semantic Vectors in Information Retrieval\\

	&	 20	&	 Mohammadbashir Sedighi (TUD), Fostering Engagement in Knowledge Sharing: The Role of Perceived Benefits, Costs and Visibility\\

	&	 21	&	 Jeroen Linssen (UT), Meta Matters in Interactive Storytelling and Serious Gaming (A Play on Worlds)\\

	&	 22	&	 Sara Magliacane (VUA), Logics for causal inference under uncertainty\\

	&	 23	&	 David Graus (UvA), Entities of Interest --- Discovery in Digital Traces\\

	&	 24	&	 Chang Wang (TUD), Use of Affordances for Efficient Robot Learning\\

	&	 25	&	 Veruska Zamborlini (VUA), Knowledge Representation for Clinical Guidelines, with applications to Multimorbidity Analysis and Literature Search\\

	&	 26	&	 Merel Jung (UT), Socially intelligent robots that understand and respond to human touch\\

	&	 27	&	 Michiel Joosse (UT), Investigating Positioning and Gaze Behaviors of Social Robots: People's Preferences, Perceptions and Behaviors\\

	&	 28	&	 John Klein (VUA), Architecture Practices for Complex Contexts\\

	&	 29	&	 Adel Alhuraibi (TiU), From IT-BusinessStrategic Alignment to Performance: A Moderated Mediation Model of Social Innovation, and Enterprise Governance of    IT"\\

	&	 30	&	 Wilma Latuny (TiU), The Power of Facial Expressions\\

	&	 31	&	 Ben Ruijl (UL), Advances in computational methods for QFT calculations\\

	&	 32	& 	 Thaer Samar (RUN), Access to and Retrievability of Content in Web Archives\\

	&	 33	&	 Brigit van Loggem (OU), Towards a Design Rationale for Software Documentation: A Model of Computer-Mediated Activity\\

	&	 34	&	 Maren Scheffel (OU), The Evaluation Framework for Learning Analytics \\

	&	 35	&	 Martine de Vos (VUA), Interpreting natural science spreadsheets \\

	&	 36	&	 Yuanhao Guo (UL), Shape Analysis for Phenotype Characterisation from High-throughput Imaging \\

	&	 37	&	 Alejandro Montes Garcia (TU/e), WiBAF: A Within Browser Adaptation Framework that Enables Control over Privacy \\

	&	 38	&	 Alex Kayal (TUD), Normative Social Applications \\

	&	 39	&	 Sara Ahmadi (RUN), Exploiting properties of the human auditory system and compressive sensing methods to increase   noise robustness in ASR \\

	&	 40	&	 Altaf Hussain Abro (VUA), Steer your Mind: Computational Exploration of Human Control in Relation to Emotions, Desires and Social Support For applications in human-aware support systems \\

	&	 41	&	 Adnan Manzoor (VUA), Minding a Healthy Lifestyle: An Exploration of Mental Processes and a Smart Environment to Provide Support for a Healthy Lifestyle\\

	&	 42	&	 Elena Sokolova (RUN), Causal discovery from mixed and missing data with applications on ADHD  datasets\\

	&	 43	&	 Maaike de Boer (RUN), Semantic Mapping in Video Retrieval\\

	&	 44	&	 Garm Lucassen (UU), Understanding User Stories - Computational Linguistics in Agile Requirements Engineering\\

	&	 45	&	 Bas Testerink	(UU), Decentralized Runtime Norm Enforcement\\

	&	 46	&	 Jan Schneider	(OU), Sensor-based Learning Support\\

	&	 47	&	 Jie Yang (TUD), Crowd Knowledge Creation Acceleration\\

	&	 48	&	 Angel Suarez (OU), Collaborative inquiry-based learning\\

\midrule

2018

	&	 01	&	 Han van der Aa (VUA), Comparing and Aligning Process Representations \\

	&	 02	&	 Felix Mannhardt (TU/e), Multi-perspective Process Mining \\

	&	 03	&	 Steven Bosems (UT), Causal Models For Well-Being: Knowledge Modeling, Model-Driven Development of Context-Aware Applications, and Behavior Prediction\\

	&	 04	&	 Jordan Janeiro (TUD), Flexible Coordination Support for Diagnosis Teams in Data-Centric Engineering Tasks \\

	&	 05	&	 Hugo Huurdeman (UvA), Supporting the Complex Dynamics of the Information Seeking Process \\

	&	 06	&	 Dan Ionita (UT), Model-Driven Information Security Risk Assessment of Socio-Technical Systems \\

	&	 07	&	 Jieting Luo (UU), A formal account of opportunism in multi-agent systems \\

	&	 08	&	 Rick Smetsers (RUN), Advances in Model Learning for Software Systems \\

	&	 09	&	 Xu Xie	(TUD), Data Assimilation in Discrete Event Simulations \\

	&	 10	&	 Julienka Mollee (VUA), Moving forward: supporting physical activity behavior change through intelligent technology \\

	&	 11	&	 Mahdi Sargolzaei (UvA), Enabling Framework for Service-oriented Collaborative Networks \\

	&	 12	&	 Xixi Lu (TU/e), Using behavioral context in process mining \\

	&	 13	&	 Seyed Amin Tabatabaei (VUA), Computing a Sustainable Future \\

	&	 14	&	 Bart Joosten (TiU), Detecting Social Signals with Spatiotemporal Gabor Filters \\

	&	 15	&	 Naser Davarzani (UM), Biomarker discovery in heart failure \\

	&	 16	&	 Jaebok Kim (UT), Automatic recognition of engagement and emotion in a group of children \\

	&	 17	&	 Jianpeng Zhang (TU/e), On Graph Sample Clustering \\

	&	 18	& 	 Henriette Nakad (UL), De Notaris en Private Rechtspraak \\

	&	 19	&	 Minh Duc Pham (VUA), Emergent relational schemas for RDF \\

	&	 20	&	 Manxia Liu (RUN), Time and Bayesian Networks \\

	&	 21	&	 Aad Slootmaker (OU), EMERGO: a generic platform for authoring and playing scenario-based serious games \\

	&	 22	&	 Eric Fernandes de Mello Ara\'{u}jo (VUA), Contagious: Modeling the Spread of Behaviours, Perceptions and Emotions in Social Networks \\

	&	 23	&	 Kim Schouten (EUR), Semantics-driven Aspect-Based Sentiment Analysis \\

	&	 24	&	 Jered Vroon (UT), Responsive Social Positioning Behaviour for Semi-Autonomous Telepresence Robots \\

	&	 25	&	 Riste Gligorov (VUA), Serious Games in Audio-Visual Collections \\

	&	 26	& 	 Roelof Anne Jelle de Vries (UT),Theory-Based and Tailor-Made: Motivational Messages for Behavior Change Technology \\

	&	 27	&	 Maikel Leemans (TU/e), Hierarchical Process Mining for Scalable Software Analysis \\

	&	 28	&	 Christian Willemse (UT), Social Touch Technologies: How they feel and how they make you feel \\

	&	 29	&	 Yu Gu (TiU), Emotion Recognition from Mandarin Speech \\

	&	 30	&	 Wouter Beek (VUA),  The "K" in "semantic web" stands for "knowledge": scaling semantics to the web \\

\midrule

2019

	&	 01	&	 Rob van Eijk (UL),Web privacy measurement in real-time bidding systems. A graph-based approach to RTB system classification \\

	&	 02	&	 Emmanuelle Beauxis Aussalet (CWI, UU), Statistics and Visualizations for Assessing Class Size Uncertainty \\

	&	 03	&	 Eduardo Gonzalez Lopez de Murillas (TU/e), Process Mining on Databases: Extracting Event Data from Real Life Data Sources \\

	&	 04	&	 Ridho Rahmadi (RUN), Finding stable causal structures from clinical data \\

	& 	 05	&	 Sebastiaan van Zelst (TU/e), Process Mining with Streaming Data \\

	&	 06	& 	 Chris Dijkshoorn (VUA), Nichesourcing for Improving Access to Linked Cultural Heritage Datasets \\

	&	 07	&	 Soude Fazeli (TUD), Recommender Systems in Social Learning Platforms \\

	& 	 08	&	 Frits de Nijs (TUD), Resource-constrained Multi-agent Markov Decision Processes \\

	&	 09	&	 Fahimeh Alizadeh Moghaddam (UvA), Self-adaptation for energy efficiency in software systems \\

	&	 10	&	 Qing Chuan Ye (EUR), Multi-objective Optimization Methods for Allocation and Prediction \\

	&	 11	&	 Yue Zhao (TUD), Learning Analytics Technology to Understand Learner Behavioral Engagement in MOOCs \\

	&	 12	&	 Jacqueline Heinerman (VUA), Better Together \\

	&	 13	&	 Guanliang Chen (TUD), MOOC Analytics: Learner Modeling and Content Generation \\

	&	 14	&	 Daniel Davis (TUD), Large-Scale Learning Analytics: Modeling Learner Behavior \& Improving Learning Outcomes in Massive Open Online Courses \\

	&	 15	&	 Erwin Walraven (TUD), Planning under Uncertainty in Constrained and Partially Observable Environments \\

	&	 16	&	 Guangming Li (TU/e), Process Mining based on Object-Centric Behavioral Constraint (OCBC) Models \\

	&	 17	&	 Ali Hurriyetoglu (RUN),Extracting actionable information from microtexts \\

	&	 18	&	 Gerard Wagenaar (UU), Artefacts in Agile Team Communication \\

	&	 19	&	 Vincent Koeman (TUD), Tools for Developing Cognitive Agents \\

	&	 20	&	 Chide Groenouwe (UU), Fostering technically augmented human collective intelligence \\

	&	 21	&	 Cong Liu (TU/e), Software Data Analytics: Architectural Model Discovery and Design Pattern Detection \\

	&	 22	&	 Martin van den Berg (VUA),Improving IT Decisions with Enterprise Architecture \\

	&	 23	&	 Qin Liu (TUD), Intelligent Control Systems: Learning, Interpreting, Verification\\

	&	 24	&	 Anca Dumitrache (VUA),  Truth in Disagreement - Crowdsourcing Labeled Data for Natural Language Processing\\

	&	 25	&	 Emiel van Miltenburg (VUA), Pragmatic factors in (automatic) image description \\

	&	 26	&	 Prince Singh (UT), An Integration Platform for Synchromodal Transport \\

	&	 27	&	 Alessandra Antonaci (OU), The Gamification Design Process applied to (Massive) Open Online Courses\\

	&	 28	&	 Esther Kuindersma (UL), Cleared for take-off: Game-based learning to prepare airline pilots for critical situations \\

	&	 29	&	 Daniel Formolo (VUA), Using virtual agents for simulation and training of social skills in safety-critical circumstances \\

	&	 30	&	 Vahid Yazdanpanah (UT), Multiagent Industrial Symbiosis Systems \\

	&	 31	&	 Milan Jelisavcic (VUA), Alive and Kicking: Baby Steps in Robotics \\

	&	 32	&	 Chiara Sironi (UM), Monte-Carlo Tree Search for Artificial General Intelligence in Games \\

	&	 33	&	 Anil Yaman (TU/e), Evolution of Biologically Inspired Learning in Artificial Neural Networks \\

	&	 34	&	 Negar Ahmadi (TU/e), EEG Microstate and Functional Brain Network Features for Classification of Epilepsy and PNES \\

	&	 35	&	 Lisa Facey-Shaw (OU), Gamification with digital badges in learning programming \\

	&	 36	&	 Kevin Ackermans (OU), Designing Video-Enhanced Rubrics to Master Complex Skills \\

	&	 37	&	 Jian Fang (TUD), Database Acceleration on FPGAs \\

	&	 38	&	 Akos Kadar (OU), Learning visually grounded and multilingual representations \\

\midrule

2020

	&	 01	&	 Armon Toubman (UL), Calculated Moves: Generating Air Combat Behaviour \\

	&	 02	&	 Marcos de Paula Bueno (UL), Unraveling Temporal Processes using Probabilistic Graphical Models \\

	&	 03	&	 Mostafa Deghani (UvA), Learning with Imperfect Supervision for Language Understanding \\

	&	 04	&	 Maarten van Gompel (RUN), Context as Linguistic Bridges \\

	&	 05	&	 Yulong Pei (TU/e), On local and global structure mining \\

	&	 06	&	 Preethu Rose Anish (UT), Stimulation Architectural Thinking during Requirements Elicitation - An Approach and Tool Support \\

	&	 07	&	 Wim van der Vegt (OU), Towards a software architecture for reusable game components \\

	&	 08	&	 Ali Mirsoleimani (UL),Structured Parallel Programming for Monte Carlo Tree Search \\

	&	 09	&	 Myriam Traub (UU), Measuring Tool Bias and Improving Data Quality for Digital Humanities Research \\

	&	 10	&	 Alifah Syamsiyah (TU/e), In-database Preprocessing for Process Mining \\

	&	 11	&	 Sepideh Mesbah (TUD), Semantic-Enhanced Training Data AugmentationMethods for Long-Tail Entity Recognition Models \\

	&	 12	&	 Ward van Breda (VUA), Predictive Modeling in E-Mental Health: Exploring Applicability in Personalised Depression Treatment \\

	&	 13	&	 Marco Virgolin (CWI), Design and Application of Gene-pool Optimal Mixing Evolutionary Algorithms for Genetic Programming \\

	&	 14	&	 Mark Raasveldt (CWI/UL), Integrating Analytics with Relational Databases \\

	&	 15	&	 Konstantinos Georgiadis (OU),  Smart CAT: Machine Learning for Configurable Assessments in Serious Games \\

	&	 16	&	 Ilona Wilmont (RUN), Cognitive Aspects of Conceptual Modelling \\

	&	 17	&	 Daniele Di Mitri (OU), The Multimodal Tutor: Adaptive Feedback from Multimodal Experiences \\

  	&	 18	&	 Georgios Methenitis (TUD), Agent Interactions \& Mechanisms in Markets with Uncertainties: Electricity Markets in Renewable Energy Systems \\

	&	 19	&	 Guido van Capelleveen (UT), Industrial Symbiosis Recommender Systems \\

	&	 20	&	 Albert Hankel (VUA), Embedding Green ICT Maturity in Organisations \\

	&	 21	&	 Karine da Silva Miras de Araujo (VUA), Where is the robot?: Life as it could be \\

	&	 22	&	 Maryam Masoud Khamis (RUN), Understanding complex systems implementation through a modeling approach: the case of e-government in Zanzibar \\

	&	 23	&	 Rianne Conijn (UT), The Keys to Writing: A writing analytics approach to studying writing processes using keystroke logging \\

	&	 24	&	 Lenin da N\'{o}brega Medeiros (VUA/RUN), How are you feeling, human? Towards emotionally supportive chatbots \\

	&	 25	&	 Xin Du (TU/e), The Uncertainty in Exceptional Model Mining \\

	&	 26	&	 Krzysztof Leszek Sadowski (UU), GAMBIT: Genetic Algorithm for Model-Based mixed-Integer opTimization \\

	&	 27	&	 Ekaterina Muravyeva (TUD), Personal data and informed consent in an educational context \\

	&	 28	&	 Bibeg Limbu (TUD), Multimodal interaction for deliberate practice: Training complex skills with augmented reality \\

	&	 29	&	 Ioan Gabriel Bucur (RUN), Being Bayesian about Causal Inference \\

	&	 30	&	 Bob Zadok Blok (UL), Creatief, Creatiever, Creatiefst \\

	&	 31	&	 Gongjin Lan (VUA), Learning better -- From Baby to Better \\

	&	 32	& 	 Jason Rhuggenaath (TU/e), Revenue management in online markets: pricing and online advertising \\

	&	 33	& 	 Rick Gilsing (TU/e), Supporting service-dominant business model evaluation in the context of business model innovation \\

	&	 34	&	 Anna Bon (UM), Intervention or Collaboration? Redesigning Information and Communication Technologies for Development \\

	&	 35	&	 Siamak Farshidi (UU), Multi-Criteria Decision-Making in Software Production \\

\midrule

2021

	&	 01	&	 Francisco Xavier Dos Santos Fonseca (TUD),Location-based Games for Social Interaction in Public Space \\

	&	 02	&	 Rijk Mercuur (TUD), Simulating Human Routines: Integrating Social Practice Theory in Agent-Based Models \\

	&	 03	&	 Seyyed Hadi Hashemi (UvA), Modeling Users Interacting with Smart Devices \\

	&	 04	&	 Ioana Jivet (OU), The Dashboard That Loved Me: Designing adaptive learning analytics for self-regulated learning \\

	&	 05	&	 Davide Dell'Anna (UU), Data-Driven Supervision of Autonomous Systems \\

	&	 06	&	 Daniel Davison (UT), "Hey robot, what do you think?" How children learn with a social robot \\

	&	 07	&	 Armel Lefebvre (UU), Research data management for open science \\

	&	 08	&	 Nardie Fanchamps (OU), The Influence of Sense-Reason-Act Programming on Computational Thinking \\

	&	 09	&	 Cristina Zaga (UT), The Design of Robothings. Non-Anthropomorphic and Non-Verbal Robots to Promote Children's Collaboration Through Play \\

	&	 10	&	 Quinten Meertens (UvA), Misclassification Bias in Statistical Learning \\

	&	 11	&	 Anne van Rossum (UL), Nonparametric Bayesian Methods in Robotic Vision \\

	&	 12	&	 Lei Pi (UL), External Knowledge Absorption in Chinese SMEs \\

	&	 13	&	 Bob R. Schadenberg (UT), Robots for Autistic Children: Understanding and Facilitating Predictability for Engagement in Learning \\

	&	 14	&	 Negin Samaeemofrad (UL), Business Incubators: The Impact of Their Support \\

	&	 15	& 	 Onat Ege Adali (TU/e), Transformation of Value Propositions into Resource Re-Configurations through the Business Services Paradigm  \\

	&	 16	&	 Esam A. H. Ghaleb (UM), Bimodal emotion recognition from audio-visual cues \\

	&	 17	&	 Dario Dotti (UM), Human Behavior Understanding  from motion and bodily cues using deep neural networks \\

	&	 18	&	 Remi Wieten (UU), Bridging the Gap Between Informal Sense-Making Tools and Formal Systems - Facilitating the Construction of Bayesian Networks and Argumentation Frameworks \\

	&	 19	&	 Roberto Verdecchia (VUA), Architectural Technical Debt: Identification and Management \\

	&	 20	&	 Masoud Mansoury (TU/e), Understanding and Mitigating Multi-Sided Exposure Bias in Recommender Systems \\

	&	 21	&	 Pedro Thiago Timb\'{o} Holanda (CWI), Progressive Indexes \\

	&	 22	&	 Sihang Qiu (TUD), Conversational Crowdsourcing \\

	&	 23	&	 Hugo Manuel Proen\c{c}a (UL), Robust rules for prediction and description \\

	&	 24	&	 Kaijie Zhu (TU/e), On Efficient Temporal Subgraph Query Processing \\

	&	 25	&	 Eoin Martino Grua (VUA), The Future of E-Health is Mobile: Combining AI and Self-Adaptation to Create Adaptive E-Health Mobile Applications \\

	&	 26	& 	 Benno Kruit (CWI/VUA), Reading the Grid: Extending Knowledge Bases from Human-readable Tables \\

	&	 27	&	 Jelte van Waterschoot (UT), Personalized and Personal Conversations: Designing Agents Who Want to Connect With You \\

	&	 28	&	 Christoph Selig (UL), Understanding the Heterogeneity of Corporate Entrepreneurship Programs \\

\midrule

2022

	&	 01	&	Judith van Stegeren (UT), Flavor text generation for role-playing video games \\

	&	 02	&	Paulo da Costa (TU/e), Data-driven Prognostics and Logistics Optimisation: A Deep Learning Journey \\

	&	 03	&	Ali el Hassouni (VUA), A Model A Day Keeps The Doctor Away: Reinforcement Learning For Personalized Healthcare \\

	&	 04	&	\"{U}nal Aksu (UU), A Cross-Organizational Process Mining Framework \\

	&	 05	&	Shiwei Liu (TU/e), Sparse Neural Network Training with In-Time Over-Parameterization \\

	&	 06	& 	Reza Refaei Afshar (TU/e), Machine Learning for Ad Publishers in Real Time Bidding \\

	&	 07	&	Sambit Praharaj (OU), Measuring the Unmeasurable? Towards Automatic Co-located Collaboration Analytics \\

	&	 08	&	Maikel L. van Eck (TU/e), Process Mining for Smart Product Design \\

	&	 09	&	Oana Andreea Inel (VUA), Understanding Events: A Diversity-driven Human-Machine Approach \\

	&	 10	&	Felipe Moraes Gomes (TUD), Examining the Effectiveness of Collaborative Search Engines \\

	&	 11	&	Mirjam de Haas (UT), Staying engaged in child-robot interaction, a quantitative approach to studying preschoolers' engagement with robots and tasks during second-language tutoring \\

	&	 12	&	Guanyi Chen (UU),  Computational Generation of Chinese Noun Phrases \\

	&	 13	&	Xander Wilcke (VUA), Machine Learning on Multimodal Knowledge Graphs: Opportunities, Challenges, and Methods for Learning on Real-World Heterogeneous and Spatially-Oriented Knowledge \\

	&	 14	&	Michiel Overeem (UU), Evolution of Low-Code Platforms \\

	&	 15	&	Jelmer Jan Koorn (UU), Work in Process: Unearthing Meaning using Process Mining \\

	&	 16	&	Pieter Gijsbers (TU/e), Systems for AutoML Research \\

	&	 17	&	Laura van der Lubbe (VUA), Empowering vulnerable people with serious games and gamification \\

	&	 18	&	Paris Mavromoustakos Blom (TiU), Player Affect Modelling and Video Game Personalisation \\

	&	 19	&	Bilge Yigit Ozkan (UU), Cybersecurity Maturity Assessment and Standardisation \\

	&	 20	&	Fakhra Jabeen (VUA), Dark Side of the Digital Media - Computational Analysis of Negative Human Behaviors on Social Media \\

	&	 21	&	Seethu Mariyam Christopher (UM), Intelligent Toys for Physical and Cognitive Assessments \\

	&	 22	&	Alexandra Sierra Rativa (TiU), Virtual Character Design and its potential to foster Empathy, Immersion, and Collaboration Skills in Video Games and Virtual Reality Simulations \\

	&	 23	&	Ilir Kola (TUD), Enabling Social Situation Awareness in Support Agents \\

	&	 24	&	Samaneh Heidari (UU), Agents with Social Norms and Values - A framework for agent based social simulations with social norms and personal values \\

	&	 25	&	Anna L.D. Latour (UL), Optimal decision-making under constraints and uncertainty \\

	&	 26	&	Anne Dirkson (UL), Knowledge Discovery from Patient Forums: Gaining novel medical insights from patient experiences \\

	&	 27	&	Christos Athanasiadis (UM), Emotion-aware cross-modal domain adaptation in video sequences \\

	&	 28	&	Onuralp Ulusoy (UU), Privacy in Collaborative Systems \\

	&	 29	&	Jan Kolkmeier (UT), From Head Transform to Mind Transplant: Social Interactions in Mixed Reality \\

	&	 30	&	Dean De Leo (CWI), Analysis of Dynamic Graphs on Sparse Arrays \\

	&	 31	&	Konstantinos Traganos (TU/e), Tackling Complexity in Smart Manufacturing with Advanced Manufacturing Process Management \\

	&	 32	&	Cezara Pastrav (UU), Social simulation for socio-ecological systems \\

	&	 33	&	Brinn Hekkelman (CWI/TUD), Fair Mechanisms for Smart Grid Congestion Management \\

	&	 34	&	Nimat Ullah (VUA), Mind Your Behaviour: Computational Modelling of Emotion \& Desire Regulation for Behaviour Change \\

	&	 35	&	Mike E.U. Ligthart (VUA), Shaping the Child-Robot Relationship: Interaction Design Patterns for a Sustainable Interaction \\

\midrule

2023

	&	 01	&	Bojan Simoski (VUA), Untangling the Puzzle of Digital Health Interventions \\

	&	 02	&	Mariana Rachel Dias da Silva (TiU), Grounded or in flight? What our bodies can tell us about the whereabouts of our thoughts \\

	&	 03	&	Shabnam Najafian (TUD), User Modeling for Privacy-preserving Explanations in Group Recommendations \\

	&	 04	&	Gineke Wiggers (UL), The Relevance of Impact: bibliometric-enhanced legal information retrieval \\

	&	 05	&	Anton Bouter (CWI), Optimal Mixing Evolutionary Algorithms for Large-Scale Real-Valued Optimization, Including Real-World Medical Applications \\

	&	 06	&	AntÃ³nio Pereira Barata (UL), Reliable and Fair Machine Learning for Risk Assessment \\

	&	 07	&	Tianjin Huang (TU/e), The Roles of Adversarial Examples on Trustworthiness of Deep Learning \\

	&	 08	&	Lu Yin (TU/e), Knowledge Elicitation using Psychometric Learning \\

	&	 09	&	Xu Wang (VUA), Scientific Dataset Recommendation with Semantic Techniques \\

	&	 10	&	Dennis J.N.J. Soemers (UM), Learning State-Action Features for General Game Playing \\

	&	 11	&	Fawad Taj (VUA), Towards Motivating Machines: Computational Modeling of the Mechanism of Actions for Effective Digital Health Behavior Change Applications \\

	&	 12	&	Tessel Bogaard (VUA), Using Metadata to Understand Search Behavior in Digital Libraries \\

	&	 13	&	Injy Sarhan (UU), Open Information Extraction for Knowledge Representation \\

	&	 14	&	Selma Čaušević (TUD), Energy resilience through self-organization \\

	&	 15	&	Alvaro Henrique Chaim Correia (TU/e), Insights on Learning Tractable Probabilistic Graphical Models \\

	&	 16	&	Peter Blomsma (TiU), Building Embodied Conversational Agents: Observations on human nonverbal behaviour as a resource for the development of artificial characters \\

	&	 17	&	Meike Nauta (UT), Explainable AI and Interpretable Computer Vision â€“ From Oversight to Insight \\

	&	 18	&	Gustavo Penha (TUD), Designing and Diagnosing Models for Conversational Search and Recommendation \\

	&	 19	&	George Aalbers (TiU), Digital Traces of the Mind: Using Smartphones to Capture Signals of Well-Being in Individuals \\

	&	 20	&	Arkadiy Dushatskiy (TUD), Expensive Optimization with Model-Based Evolutionary Algorithms applied to Medical Image Segmentation using Deep Learning \\

	&	 21	&	Gerrit Jan de Bruin (UL), Network Analysis Methods for Smart Inspection in the Transport Domain \\

	&	 22	&	Alireza Shojaifar (UU), Volitional Cybersecurity \\

	&	 23	&	Theo Theunissen (UU), Documentation in Continuous Software Development \\

	&	 24	&	Agathe Balayn (TUD), Practices Towards Hazardous Failure Diagnosis in Machine Learning \\

	&	 25	&	Jurian Baas (UU), Entity Resolution on Historical Knowledge Graphs \\

	&	 26	&	Loek Tonnaer (TU/e), Linearly Symmetry-Based Disentangled Representations and their Out-of-Distribution Behaviour \\

	&	 27	&	Ghada Sokar (TU/e), Learning Continually Under Changing Data Distributions \\

	&	 28	&	Floris den Hengst (VUA), Learning to Behave: Reinforcement Learning in Human Contexts \\

	&	 29	&	Tim Draws (TUD), Understanding Viewpoint Biases in Web Search Results \\

\midrule

2024

	&	 01	&	Daphne Miedema (TU/e), On Learning SQL: Disentangling concepts in data systems education \\

	&	 02	&	Emile van Krieken (VUA), Optimisation in Neurosymbolic Learning Systems \\

	&	 03	&	Feri Wijayanto (RUN), Automated Model Selection for Rasch and Mediation Analysis \\

	&	 04	&	Mike Huisman (UL), Understanding Deep Meta-Learning \\

	&	 05	&	Yiyong Gou (UM), Aerial Robotic Operations: Multi-environment Cooperative Inspection \& Construction Crack Autonomous Repair \\

	&	 06	&	Azqa Nadeem (TUD), Understanding Adversary Behavior via XAI: Leveraging Sequence Clustering to Extract Threat Intelligence \\

	&	 07	&	Parisa Shayan (TiU), Modeling User Behavior in Learning Management Systems \\

	&	 08	&	Xin Zhou (UvA), From Empowering to Motivating: Enhancing Policy Enforcement through Process Design and Incentive Implementation \\

	&	 09	&	Giso Dal (UT), Probabilistic Inference Using Partitioned Bayesian Networks \\

	&	 10	&	Cristina-Iulia Bucur (VUA), Linkflows: Towards Genuine Semantic Publishing in Science \\

	&	 11	&	withdrawn \\

	&	 12	&	Peide Zhu (TUD), Towards Robust Automatic Question Generation For Learning \\

	&	 13	&	Enrico Liscio (TUD), Context-Specific Value Inference via Hybrid Intelligence \\

	&	 14	&	Larissa Capobianco Shimomura (TU/e), On Graph Generating Dependencies and their Applications in Data Profiling \\

	&	 15	&	Ting Liu (VUA), A Gut Feeling: Biomedical Knowledge Graphs for Interrelating the Gut Microbiome and Mental Health \\

	&	 16	&	Arthur Barbosa CÃ¢mara (TUD), Designing Search-as-Learning Systems \\

	&	 17	&	Razieh Alidoosti (VUA), Ethics-aware Software Architecture Design \\

	&	 18	&	Laurens Stoop (UU), Data Driven Understanding of Energy-Meteorological Variability and its Impact on Energy System Operations \\

	&	 19	&	Azadeh Mozafari Mehr (TU/e), Multi-perspective Conformance Checking: Identifying and Understanding Patterns of Anomalous Behavior\\

	&	 20	&	Ritsart Anne Plantenga (UL), Omgang met Regels \\

	&	 21	&	Federica Vinella (UU), Crowdsourcing User-Centered Teams \\

	&	 22	&	Zeynep Ozturk Yurt (TU/e), Beyond Routine: Extending BPM for Knowledge-Intensive Processes with Controllable Dynamic Contexts \\

	&	 23	&	Jie Luo (VUA), Lamarck’s Revenge: Inheritance of Learned Traits Improves Robot Evolution \\

	&	 24	&	Nirmal Roy (TUD), Exploring the effects of interactive interfaces on user search behaviour \\

	&	 25	&	Alisa Rieger (TUD), Striving for Responsible Opinion Formation in Web Search on Debated Topics \\

	&	 26	&	Tim Gubner (CWI), Adaptively Generating Heterogeneous Execution Strategies using the VOILA Framework \\

	&	 27	&	Lincen Yang (UL), Information-theoretic Partition-based Models for Interpretable Machine Learning \\

	&	 28	&	Leon Helwerda (UL), Grip on Software: Understanding development progress of Scrum sprints and backlogs \\

	&	 29	&	David Wilson Romero Guzman (VUA), The Good, the Efficient and the Inductive Biases: Exploring Efficiency in Deep Learning Through the Use of Inductive Biases \\

	&	 30	&	Vijanti Ramautar (UU), Model-Driven Sustainability Accounting \\

	&	 31	&	Ziyu Li (TUD), On the Utility of Metadata to Optimize Machine Learning Workflows \\

	&	 32	&	Vinicius Stein Dani (UU), The Alpha and Omega of Process Mining \\

	&	 33	&	Siddharth Mehrotra (TUD), Designing for Appropriate Trust in Human-AI interaction \\

	&	 34	&	Robert Deckers (VUA), From Smallest Software Particle to System Specification - MuDForM: Multi-Domain Formalization Method \\

	&	 35	&	Sicui Zhang (TU/e), Methods of Detecting Clinical Deviations with Process Mining: a fuzzy set approach \\

	&	 36	&	Thomas Mulder (TU/e), Optimization of Recursive Queries on Graphs \\

	&	 37	&	James Graham Nevin (UvA), The Ramifications of Data Handling for Computational Models \\

	&	 38	&	Christos Koutras (TUD), Tabular Schema Matching for Modern Settings \\

	&	 39	&	Paola Lara Machado (TU/e), The Nexus between Business Models and Operating Models: From Conceptual Understanding to Actionable Guidance \\

	&	 40	&	Montijn van de Ven (TU/e), Guiding the Definition of Key Performance Indicators for Business Models \\

	&	 41	&	Georgios Siachamis (TUD), Adaptivity for Streaming Dataflow Engines \\

	&	 42	&	Emmeke Veltmeijer (VUA), Small Groups, Big Insights: Understanding the Crowd through Expressive Subgroup Analysis \\

	&	 43	&	Cedric Waterschoot (KNAW Meertens Instituut), The Constructive Conundrum: Computational Approaches to Facilitate Constructive Commenting on Online News Platforms \\

	&	 44	&	Marcel Schmitz (OU), Towards learning analytics-supported learning design \\

	&	 45	&	Sara Salimzadeh (TUD), Living in the Age of AI: Understanding Contextual Factors that Shape Human-AI Decision-Making \\

	&	 46	&	Georgios Stathis (Leiden University), Preventing Disputes: Preventive Logic, Law \& Technology \\

	&	 47	&	Daniel Daza (VUA), Exploiting Subgraphs and Attributes for Representation Learning on Knowledge Graphs \\

	&	 48	&	Ioannis Petros Samiotis (TUD), Crowd-Assisted Annotation of Classical Music Compositions \\

\midrule

2025

	&	 01	&	Max van Haastrecht (UL), Transdisciplinary Perspectives on Validity: Bridging the Gap Between Design and Implementation for Technology-Enhanced Learning Systems \\

	&	 02	&	Jurgen van den Hoogen (JADS), Time Series Analysis Using Convolutional Neural Networks \\

	&	 03	&	Andra-Denis Ionescu (TUD), Feature Discovery for Data-Centric AI \\

	&	 04	&	Rianne Schouten (TU/e), Exceptional Model Mining for Hierarchical Data \\

	&	 05	&	Nele Albers (TUD), Psychology-Informed Reinforcement Learning for Situated Virtual Coaching in Smoking Cessation \\

	&	 06	&	Daniël Vos (TUD), Decision Tree Learning: Algorithms for Robust Prediction and Policy Optimization \\

	&	 07	&	Ricky Maulana Fajri (TU/e), Towards Safer Active Learning: Dealing with Unwanted Biases, Graph-Structured Data, Adversary, and Data Imbalance \\

	&	 08	&	Stefan Bloemheuvel (TiU), Spatio-Temporal Analysis Through Graphs: Predictive Modeling and Graph Construction \\

	&	 09	&	Fadime Kaya (VUA), Decentralized Governance Design - A Model-Based Approach \\

	&	 10	&	Zhao Yang (UL), Enhancing Autonomy and Efficiency in Goal-Conditioned Reinforcement Learning \\

	&	 11	&	Shahin Sharifi Noorian (TUD), From Recognition to Understanding: Enriching Visual Models Through Multi-Modal Semantic Integration \\

	&	 12	&	Lijun Lyu (TUD), Interpretability in Neural Information Retrieval \\

	&	 13	&	Fuda van Diggelen (VUA), Robots Need Some Education: on the complexity of learning in evolutionary robotics \\

	&	 14	&	Gennaro Gala (TU/e), Probabilistic Generative Modeling with Latent Variable Hierarchies \\

	&	 15	&	Michiel van der Meer (UL), Opinion Diversity through Hybrid Intelligence \\

	&	 16	&	Monika Grewal (TUD), Deep Learning for Landmark Detection, Segmentation, and Multi-Objective Deformable Registration in Medical Imaging \\

	&	 17	&	Matteo De Carlo (VUA), Real Robot Reproduction: Towards Evolving Robotic Ecosystems \\

	&	 18	&	Anouk Neerincx (UU), Robots That Care: How Social Robots Can Boost Children's Mental Wellbeing \\

	&	 19	&	Fang Hou (UU), Trust in Software Ecosystems \\

	&	 20	&	Alexander Melchior (UU), Modelling for Policy is More Than Policy Modelling (The Useful Application of Agent-Based Modelling in Complex Policy Processes) \\

	&	 21	&	Mandani Ntekouli (UM), Bridging Individual and Group Perspectives in Psychopathology: Computational Modeling Approaches using Ecological Momentary Assessment Data \\

	&	 22	&	Hilde Weerts (TU/e), Decoding Algorithmic Fairness: Towards Interdisciplinary Understanding of Fairness and Discrimination in Algorithmic Decision-Making \\

	&	 23	&	Roderick van der Weerdt (VUA), IoT Measurement Knowledge Graphs: Constructing, Working and Learning with IoT Measurement Data as a Knowledge Graph \\

	&	 24	&	Zhong Li (UL), Trustworthy Anomaly Detection for Smart Manufacturing \\

	&	 25	&	Kyana van Eijndhoven (TiU), A Breakdown of Breakdowns: Multi-Level Team Coordination Dynamics under Stressful Conditions \\

	&	 26	&	Tom Pepels (UM), Monte-Carlo Tree Search is Work in Progress \\

	&	 27	&	Danil Provodin (JADS, TU/e), Sequential Decision Making Under Complex Feedback \\

	&	 28	&	Jinke He (TUD), Exploring Learned Abstract Models for Efficient Planning and Learning \\

	&	 29	&	Erik van Haeringen (VUA), Mixed Feelings: Simulating Emotion Contagion in Groups \\

	&	 30	&	Myrthe Reuver (VUA), A Puzzle of Perspectives: Interdisciplinary Language Technology for Responsible News Recommendation \\

	&	 31	&	Gebrekirstos Gebreselassie Gebremeskel (RUN), Spotlight on Recommender Systems: Contributions to Selected Components in the Recommendation Pipeline \\

	&	 32	&	Ryan Brate (UU), Words Matter: A Computational Toolkit for Charged Terms \\

	&	 33	&	Merle Reimann (VUA), Speaking the Same Language: Spoken Capability Communication in Human-Agent and Human-Robot Interaction \\

	&	 34	&	Eduard C. Groen (UU), Crowd-Based Requirements Engineering \\

	&	 35	&	Urja Khurana (VUA), From Concept To Impact: Toward More Robust Language Model Deployment \\

	&	 36	&	Anna Maria Wegmann (UU), Say the Same but Differently: Computational Approaches to Stylistic Variation and Paraphrasing \\

	&	 37	&	Chris Kamphuis (RUN), Exploring Relations and Graphs for Information Retrieval \\

	&	 38	&	Valentina Maccatrozzo (VUA), Break the Bubble: Semantic Patterns for Serendipity \\

	&	 39	&	Dimitrios Alivanistos (VUA), Knowledge Graphs \& Transformers for Hypothesis Generation: Accelerating Scientific Discovery in the Era of Artificial Intelligence \\

	&	 40	&	Stefan Grafberger (UvA), Declarative Machine Learning Pipeline Management via Logical Query Plans \\

	&	 41	&	Mozhgan Vazifehdoostirani (TU/e), Leveraging Process Flexibility to Improve Process Outcome - From Descriptive Analytics to Actionable Insights \\

	&	 42	&	Margherita Martorana (VUA), Semantic Interpretation of Dataless Tables: a metadata-driven approach for findable, accessible, interoperable and reusable restricted access data \\

	&	 43	&	Krist Shingjergji (OU), Sense the Classroom - Using AI to Detect and Respond to Learning-Centered Affective States in Online Education \\

	&	 44	&	Robbert Reijnen (TU/e), Dynamic Algorithm Configuration for Machine Scheduling Using Deep Reinforcement Learning \\

	&	 45	&	Anjana Mohandas Sheeladevi (VUA), Occupant-Centric Energy Management: Balancing Privacy, Well-being and Sustainability in Smart Buildings \\

	&	 46	&	Ya Song (TU/e), Graph Neural Networks for Modeling Temporal and Spatial Dimensions in Industrial Decision-making \\

	&	 47	&	Tom Kouwenhoven (UL), Collaborative Meaning-Making. The Emergence of Novel Languages in Humans, Machines, and Human-Machine Interactions \\

	&	 48	&	Evy van Weelden (TiU), Integrating Virtual Reality and Neurophysiology in Flight Training \\

	&	 49	&	Selene BÃ¡ez SantamarÃ­a (VUA), Knowledge-centered conversational agents with a drive to learn \\

	&	 50	&	Lea Krause (VUA), Contextualising Conversational AI \\

	&	 51	&	Jiaxu Zhao (TU/e), Understanding and Mitigating Unwanted Biases in Generative Language Models \\

	&	 52	&	Qiao Xiao (TU/e), Model, Data and Communication Sparsity for Efficient Training of Neural Networks \\

	&	 53	&	Gaole He (TUD), Towards Effective Human-AI Collaboration: Promoting Appropriate Reliance on AI Systems \\

	&	 54	&	Go Sugimoto (VUA), MISSING LINKS Investigating the Quality of Linked Data and its Tools in Cultural Heritage and Digital Humanities \\

	&	 55	&	Sietze Kai Kuilman (TUD), AI that Glitters is Not Gold: Requirements for Meaningful Control of AI Systems \\

	&	 56	&	Wijnand van Woerkom (UU), A Fortiori Case-Based Reasoning: Formal Studies with Applications in Artificial Intelligence and Law \\

	&	 57	&	Syeda Amna Sohail (UT), Privacy-Utility Trade-Off in Healthcare Metadata Sharing and Beyond: A Normative and Empirical Evaluation at Inter and Intra Organizational Levels \\

	&	 58	&	Junhan Wen (TUD), "From iMage to Market": Machine-Learning-Empowered Fruit Supply \\

	&	 59	&	Mohsen Abbaspour Onari (TU/e), From Explanation to Trust: Modeling and Measuring Trust in Explainable Decision Support \\

	&	 60	&	Marcel Jurriaan Robeer (UU), Beyond Trust: A Causal Approach to Explainable AI in Law Enforcement \\

	&	 61	&	Shuai Wang (VUA), Links in Large Integrated Knowledge Graphs: Analysis, Refinement, and Domain Applications \\

	&	 62	&	Khaleel Asyraaf Mat Sanusi (OU), Augmenting a learning model within immersive learning environments for psychomotor skills \\

	&	 63	&	Rashid Zaman (TU/e), Online Conformance Checking on Degraded Data \\

	&	 64	&	Jens d'Hondt (TU/e), Effective and Efficient Multivariate Similarity Search \\

	&	 65	&	Aswin Balasubramaniam (UT), Disentangling Runner Drone Interaction Potentialities \\

\midrule

2026

	&	 01	&	Pei-Yu Chen (TUD), Human-Agent Alignment Dialogues: Eliciting User Information at Runtime for Personalized Behavior Support \\

	&	 02	&	Hezha Hassan Mohammedkhan (TiU), Estimating Body Measurements of Children from 2D Images: Towards the Automatic Detection of Malnutrition \\

	&	 03	&	Kyriakos Psarakis (TUD), Democratizing Scalable Cloud Applications: Transactional Stateful Functions on Streaming Dataflows \\


\end{xltabular}

\end{document}